\documentclass[a4paper,12pt]{amsart}

\usepackage{kantlipsum} 
\calclayout

\usepackage{amsmath}    
\usepackage{amssymb}    
\usepackage{bm}         
\usepackage{bbm}
\usepackage{graphicx}   
\usepackage{hyperref}   
\usepackage{natbib}    
\usepackage{epsf}
\usepackage{lscape}
\usepackage{enumerate}
\usepackage{enumitem}
\usepackage{color}
\usepackage{algorithm}
\usepackage[]{algorithmic}
\usepackage[utf8]{inputenc}
\usepackage[T1]{fontenc}
\usepackage{xr}  
\usepackage{cancel}
\bibpunct{(}{)}{;}{a}{,}{,}

\newtheorem{thm}{Theorem}
\newtheorem{lemma}[thm]{Lemma}
\newtheorem{cor}{Corollary}
\newtheorem{prop}[thm]{Proposition}

\newtheorem{defn}{Definition}

\newcommand{\bx}{{\bf x}}

\newcommand{\by}{{\bf y}}

\newcommand{\bz}{{\bf z}}

\newcommand{\wbar}{\overline{w}}
\newcommand{\Wbar}{\overline{W}}

\newcommand{\btheta}{\bm{\theta}}

\newcommand{\talpha}{\tilde{\alpha}}
\newcommand{\tvartheta}{\tilde{\vartheta}}

\begin{document}

\title{Tractable Bayesian variable selection: beyond normality}

\author{David Rossell}
\address{Universitat Pompeu Fabra, Department of Business and Economics, Barcelona (Spain)}

\author{Francisco J. Rubio}
\address{London School of Hygiene \& Tropical Medicine, London (UK)}

\date{}  

\begin{abstract}
Bayesian variable selection often assumes normality,
but the effects of model misspecification are not sufficiently understood.
There are sound reasons behind this assumption, particularly for large $p$:
ease of interpretation, analytical and computational convenience.
More flexible frameworks exist, including semi- or non-parametric models,
often at the cost of some tractability.
We propose a simple extension of the Normal model that allows for skewness and thicker-than-normal tails but preserves tractability.
It leads to easy interpretation and a log-concave likelihood that facilitates optimization and integration.
We characterize asymptotically parameter estimation and Bayes factor rates, in particular studying the effects of model misspecification.
Under suitable conditions misspecified Bayes factors are consistent
and induce sparsity at the same asymptotic rates than under the correct model.
However, the rates to detect signal are altered by an exponential factor,
often resulting in a loss of sensitivity.
These deficiencies can be ameliorated by inferring the error distribution from the data,
a simple strategy that can improve inference substantially.
Our work focuses on the likelihood and can thus be combined with any likelihood penalty
or prior, but here we focus on non-local priors to induce extra sparsity
and ameliorate finite-sample effects caused by misspecification.
Our results highlight the practical importance of focusing on the likelihood
rather than solely on the prior, when it comes to Bayesian variable selection.
The methodology is available in R package `mombf'.

\vspace*{.3in}


\end{abstract}

\maketitle

\section{Introduction}

The rise of high-dimensional problems has generated a renewed interest in simple models.
Beyond the obvious issue that modest sample sizes
limit the number of parameters that can be learned accurately,
simple models remain a central choice due to their
analytical and computational tractability,
ease of interpretation, and the fact that they often work well in practice.
There is, however, a pressing need to seek extensions which,
while retaining the aforementioned advantages, incorporate additional flexibility and can be studied
without unrealistically assuming that the posed model is correct. 
Ideally such extensions should detect when the added
flexibility is not needed so that one can fall back onto simpler models.
We focus on canonical variable selection in linear
regression from a Bayesian standpoint, although some results may also be useful for penalized likelihood methods.
Given that the number of models to consider is exponential in the number of variables,
it is highly convenient to adopt error models that lead to fast within-model calculations,
\textit{e.g.}~closed forms or fast approximations for the integrated likelihood.
Our work is based on two-piece distributions,
an easily interpretable family that has a long history and which we fully characterize
in the linear model case (synthesizing and extending current results) 
under model misspecification.
Our main contributions are showing that two-piece errors (specifically when applied to the Normal and Laplace families)
lead to tractable inference, proposing simple computational algorithms,
and characterizing variable selection under model misspecification,
including when this likelihood is combined with non-local priors (NLPs, \cite{johnson:2010}).
We show that in the presence of asymmetries or heavy tails the Normal model incurs a significant loss of power, and propose a formal strategy to detect such departures from normality.
When these departures are negligible our model collapses onto Normal errors,
for which closed-form expressions are often available.

To fix ideas, we consider the linear regression model
\begin{align}
y = X \theta + \epsilon,
\label{eq:lm}
\end{align}
where $y=(y_1,\ldots,y_n)^T$ is the observed outcome for $n$ individuals,
$X$ is an $n \times p$ matrix with potential predictors,
$\theta=(\theta_1,\ldots,\theta_p)^T \in \mathbb{R}^p$ are regression coefficients
and $\epsilon=(\epsilon_1,\ldots,\epsilon_n)^T$ are independent and identically distributed (id) errors
(see Section \ref{ssec:bf_rates} for a discussion on non-id errors).
The goal is to determine the non-zero coefficients in $\theta$
under an arbitrary data-generating distribution for the $\epsilon_i$'s,
building a framework that remains convenient for large $p$.
Let $\gamma_j=\mbox{I}(\theta_j \neq 0)$ for $j=1,\ldots,p$ be variable inclusion indicators
and $p_\gamma=\sum_{j=1}^{p} \gamma_j$ the number of active variables.
To consider that residuals may be asymmetric and/or have thicker-than-normal tails
$\gamma_{p+1}=1$ denotes the presence of asymmetry ($\gamma_{p+1}=0$ otherwise)
and $\gamma_{p+2}=1$ that of thick tails ($\gamma_{p+2}=0$ for Normal tails).
Thus $\gamma=(\gamma_1,\ldots,\gamma_{p+2})$ denotes the assumed model.
$X_{\gamma}$ and $\theta_\gamma$ are the corresponding submatrix of $X$ and subvector of $\theta$, respectively.
We denote the $i^{th}$ row in $X$ and $X_\gamma$ by $x_i^T \in \mathbb{R}^p$
and $x_{\gamma i}^T \in \mathbb{R}^{p_\gamma}$.

There are a number of proposals to relax the normality assumption.
Within the frequentist literature
\cite{wanghansheng:2007} proposed median regression with LASSO penalties (LASSO-LAD) and
\cite{wanglang:2009} with rank-based SCAD penalties.
\cite{arslan:2012} extended LASSO median regression by weighting observations
and \cite{fan:2014} considered adaptive LASSO quantile regression.
These approaches are formally connected to assuming either Laplace or asymmetric Laplace errors.
There are also model-free M-estimation methods, \textit{e.g.}~combining Huber's loss with an adaptive LASSO penalty \cite{lambertlacroix:2011},
sparse trimmed-means LASSO \cite{alfons:2013},
and non-negative garrote extensions to induce robustness to outliers 
\cite{gijbels:2015}.
Theoretical characterizations also exist,
\textit{e.g.}~\cite{mendelson:2014} proved the consistency and asymptotic normality of high-dimensional M-estimators
and \cite{loh:2015} extended the results to generalized M-estimators with non-convex loss functions.
Within the Bayesian framework,
\cite{gottardo:2007} and \cite{wang:2016} consider variable selection after transforming $y_i$ and/or $x_i$,
the former allowing for $t$ errors and the latter inducing NLPs on $\theta$
via the transformation's Jacobian. 
While certainly interesting, the transformed conditional mean $E(y_i \mid x_i)$ is no longer linear in $x_i$
and parameter interpretation and prior elicitation is less straightforward.
Our main interest is in linear predictors with simple error distributions.
Along these lines, \cite{yu:2013} proposed Gibbs sampling for model choice in Bayesian quantile regression using a latent scale augmentation,
and \cite{yan_kottas:2015} extended Azzalini's skew Normal to Laplace errors
within Bayesian quantile regression, which leads to easily-implementable MCMC, and induced sparsity via LASSO penalties.
Related to our work \cite{rubiogenton:2016} and \cite{rubioyu:2016} employ skew-symmetric and two-piece errors
in linear regression, respectively, albeit the set of covariates is fixed and they focus on prediction and censored responses.
Yet another possible avenue is to pose highly flexible errors,
\textit{e.g.}~\cite{chung:2009} set a non-parametric model to simultaneously learn the effect of $x_i$
on the mean and on the shape of the residual distribution.
\cite{kundu:2014} proposed variable selection with non-parametric symmetric residuals,
for which notably \cite{chae:2016} proved high-dimensional model selection consistency and concentration rates under model misspecification.
Most Bayesian work uses Markov Chain Monte Carlo (MCMC) for parameter estimation and computation of
marginal likelihoods and does not collapse onto the Normal model when warranted by the data,
hampering its computational scalability as $p$ or $n$ grow,
further the theoretical study is typically M-closed.

In contrast, we show that simpler parametric error models equipped with efficient analytical
approximations to the integrated likelihood
achieve selection consistency under model misspecification,
and embed these models within a framework that when appropriate collapses onto normality.
We also show that model misspecification can markedly decrease the sensitivity to detect truly active variables,
\textit{e.g.}~under asymmetry or heavy tails.
Our results complement the examples in \cite{gruenwald:2014},
where the presence of inliers favoured the addition of spurious variables
(see also Figure 1 in \cite{kundu:2014}).
We show that asymptotically misspecified Bayes factors to discard spurious models essentially multiply the correct
Bayes factor by a constant term, but when detecting true signals this term is exponential in $n$.
That is, asymptotically model misspecification has more serious effects on sensitivity than on false positives.
For finite $n$, false positives can be an important issue. We use the example in
\cite{gruenwald:2014} to illustrate how such finite $n$ effects can be reduced
by penalizing small coefficients via NLPs (Section \ref{ssec:nonid_sim}).

Before presenting our approach we clarify our main contributions relative to earlier work in two-piece distributions.
\cite{rubio:2014} showed that Jeffreys priors and their associated posteriors for location-scale two-piece models are improper,
and that the (improper) independence Jeffreys prior leads to a proper posterior.
\cite{rubioyu:2016} extended the study to linear regression, again under improper priors.
Unfortunately, improper priors cannot be used for Bayesian model selection as they lead to the well-known Jeffreys-Lindley-Bartlett paradox.
There is also literature (\text{e.g.} \cite{arellanovalle:2005}) on MLE consistency and asymptotic normality in the case with no covariates.
Checks of the large sample theory technical conditions are however hard to come by,
which are non-standard due to the non-existence of certain derivatives.
Our two-piece likelihood properties, specifically log-concativity and asymptotic analysis under model misspecification are, to our knowledge, new.
As well as our results on Bayes factors, indeed the main theme of our paper: model selection.
The M-estimation technical machinery for the theorems is also of interest
as an avenue for asymptotic analysis of Bayesian model selection under misspecification.
Finally optimization and integration algorithms
built on interior-point methods are newly developed here to scale with $n$ and $p$.
A particular case of our framework provides a new approach to Bayesian quantile regression.
We also propose a novel strategy to infer the error model from the data.

The manuscript is structured as follows.
Section \ref{sec:loglikelihood} reviews two-piece distributions
and establishes the concavity of the log-likelihood in the asymmetric Normal and Laplace cases.
Section \ref{sec:prior} proposes a prior formulation based on NLPs
that enforces sparsity and discards degrees of asymmetry that are irrelevant in practice.
Section \ref{sec:paramestim} tackles maximum likelihood and posterior mode estimation,
specifically giving asymptotic distributions and optimization algorithms
that capitalize on likelihood tractability.
Section \ref{sec:modsel} outlines a framework to select both variables and the residual distribution,
proposes fast approximations to the integrated likelihood and characterizes asymptotically the associated Bayes factors.
Section \ref{sec:results} shows results on simulated and experimental data,
and Section \ref{sec:conclusion} offers concluding remarks.
The supplementary material contains all proofs and further results.
R code to reproduce our results is also provided as a supplement to this article.

\section{Log-likelihood}
\label{sec:loglikelihood}

We recall the definition of a two-piece distribution for model \eqref{eq:lm} and predictors $X_\gamma$.

\begin{defn}
A random variable $y_i \in \mathbb{R}$ following a two-piece distribution
with location $x_{\gamma i}^T\theta_\gamma$, scale $\sqrt{\vartheta} \in \mathbb{R}^+$ and asymmetry $\alpha$ has density function
$s(y_i;x_{\gamma i}^T\theta_\gamma,\vartheta,\alpha) =$
\begin{eqnarray}\label{eq:TP}\\
\dfrac{2}{\sqrt{\vartheta}[a(\alpha)+b(\alpha)]}
\left[f\left(\dfrac{y_i-x_{\gamma i}^T\theta_\gamma}{\sqrt{\vartheta} a(\alpha)} \right) I(y_i<x_{\gamma i}^T\theta_\gamma) + f\left(\dfrac{y_i-x_{\gamma i}^T \theta_\gamma}{\sqrt{\vartheta}b(\alpha)} \right) I(y_i \geq x_{\gamma i}^T \theta_\gamma) \right],\nonumber
\end{eqnarray}
where $f(\cdot)$ is a symmetric unimodal density with mode at $0$ and support on $\mathbb R$,
and $a(\alpha),b(\alpha) \in \mathbb{R}^+$.
\end{defn}
Two-piece distributions induce asymmetry by (continuously) merging two symmetric densities that have the same mode
$x_{\gamma i}^T \theta_\gamma$ but different scale parameters $\sqrt{\vartheta}a(\alpha)$, $\sqrt{\vartheta}b(\alpha)$ on each side of the mode.
Some popular parameterizations are
the inverse scale factors $\{a(\alpha),b(\alpha)\}=\{\alpha,1/\alpha\}$ for $\alpha \in {\mathbb R}^+$ \cite{fernandez:1998}
or the epsilon-skew parameterization $\{a(\alpha),b(\alpha)\}=\{1-\alpha,1+\alpha\}$ for $\alpha \in [-1,1]$ \cite{mudholkar:2000}.
We adopt the latter as it leads to orthogonality in the expected log-likelihood hessian between $\alpha$ and $\vartheta$,
also it allows easy interpretation as the total variation distance between
$s(y_i;x_{\gamma i}^T \theta_\gamma,\vartheta,\alpha)$ and its symmetric counterpart $s(y_i; x_{\gamma i}^T \theta_\gamma,\vartheta,0)$
is $|\alpha|/2$ \cite{dette:2016}.
Further, a classical skewness coefficient proposed by Arnold-Groeneveld
defined as $\mbox{AG}=1-2F(x_{\gamma i}^T \theta_\gamma) \in [-1,1]$ for a univariate random variable with mode at $x_{\gamma i}^T \theta_\gamma$
and cumulative distribution function $F()$,
is equal to $\mbox{AG}=-\alpha$ \cite{rubio:2014}.

Two-piece distributions are appealing for regression given that the mode of $s()$ is $x_{\gamma i}^T \theta_\gamma$,
its mean (when defined) depends on $x_{\gamma i}$ only through $x_{\gamma i} \theta_\gamma$
and its variance is proportional to $\vartheta$ (see below for specific expressions),
facilitating interpretation and prior elicitation.
Despite these properties and them being a classical strategy with a fascinating history,
proposed at least as early as 1897 and rediscovered multiple times \cite{wallis:2014},
their popularity has been limited due to practical concerns,
\textit{e.g.}~log-likelihood maximization may be hampered by discontinuous gradients or hessians.
For this reason we focus on two-piece Normal and Laplace errors,
for which we prove log-concavity and thus analytical and computational tractability,
giving a practical mechanism to capture asymmetry and heavier-than-normal tails.
Specifically, the two-piece Normal is obtained by letting $f(z)=N(z;0,1)$ in (\ref{eq:TP}) be the standard Normal density,
and gives $E(y_i \mid x_{\gamma i})= x_{\gamma i}^T\theta_\gamma- \alpha\sqrt{8\vartheta/\pi}$,
$\mbox{Var}(y_i \mid  x_{\gamma i})= \vartheta [(3-8/\pi) \alpha^2 +1]$
and a median that is also linear in $x_{\gamma i}$ \cite{mudholkar:2000}.
The corresponding likelihood has the simple expression
$\log L_1(\theta_\gamma,\vartheta,\alpha) =$
\begin{align}
-\frac{n}{2} \log (2\pi) -\frac{n}{2} \log (\vartheta) - \frac{1}{2\vartheta}
\left( \sum_{i \in A(\theta_\gamma)}^{} \frac{(y_i-x_{\gamma i}^T\theta_\gamma)^2}{(1+\alpha)^2}
+ \sum_{i \not\in A(\theta_\gamma)}^{} \frac{(y_i-x_{\gamma i}^T\theta_\gamma)^2}{(1-\alpha)^2} \right)
= \nonumber \\
=-\frac{n}{2} \log (2\pi) -\frac{n}{2}\log (\vartheta) - \frac{1}{2\vartheta}
(y - X_\gamma\theta_\gamma)^T W^2 (y - X_\gamma\theta_\gamma).
\label{eq:skewnorm_loglhood}
\end{align}
where $A(\theta_\gamma)= \left\{  i: y_i < x_{\gamma i}^T\theta_\gamma \right\}$ are the observations with negative residuals,
$W=\mbox{diag}(w)$, $w_i=|1+\alpha|^{-1}$ if $i \in A(\theta_\gamma)$ and $w_i=|1-\alpha|^{-1}$ if $i \not\in A(\theta_\gamma)$.
For later convenience we denote by $\wbar$ the signed weight vector
with $\wbar_i=w_i$ if $i \in A(\theta_\gamma)$ and $\wbar_i=-w_i$ if $i \not \in A(\theta_\gamma)$,
by $w^k=(w_1^k,\ldots,w_n^k)$ the element-wise $k^{th}$ power of a vector,
$\wbar^k=(\mbox{sign}(\wbar_1) |\wbar_1|^k,\ldots,\mbox{sign}(\wbar_n) |\wbar_n|^k)^T$
and $\Wbar^k=\mbox{diag}(\wbar^k)$.
Note that (\ref{eq:skewnorm_loglhood}) is linked to asymmetric least square regression 
and is the Normal likelihood for $\alpha=0$.

The two-piece Laplace is obtained by setting $f(z)=0.5 \exp(-\vert z \vert)$ in (\ref{eq:TP}).
This distribution is more commonly referred to as asymmetric Laplace,
we denote it $y_i \sim \mbox{AL}(x_{\gamma i}^T \theta_\gamma, \vartheta, \alpha)$
and note that $E(y_i \mid x_{\gamma i},\theta_\gamma,\vartheta,\alpha)= x_{\gamma i}^T\theta_\gamma - 2 \alpha \sqrt{\vartheta}$
and $\mbox{Var}(y_i \mid x_{\gamma i})= 2\vartheta(1+\alpha^2)$ \cite{arellanovalle:2005}.
For coherency from here onwards, we also refer to the two-piece Normal as asymmetric Normal and denote $y_i \sim \mbox{AN}(x_{\gamma i}^T \theta_\gamma, \vartheta, \alpha)$.
The asymmetric Laplace log-likelihood is
$\log L_2(\theta_\gamma,\vartheta,\alpha) =$
\begin{align}
-n\log(2) -\frac{n}{2} \log(\vartheta) - \frac{1}{\sqrt{\vartheta}}
\left( \sum_{i \in A(\theta_\gamma)}^{} \frac{\vert y_i-x_{\gamma i}^T\theta_\gamma\vert}{1+\alpha}
+ \sum_{i \not\in A(\theta_\gamma)}^{} \frac{\vert y_i-x_{\gamma i}^T\theta_\gamma\vert}{1-\alpha} \right).
\label{eq:skewlap_loglhood}
\end{align}
The symmetric Laplace case is obtained for $\alpha=0$, in which case optimization of (\ref{eq:skewlap_loglhood})
with respect to $\theta_\gamma$ is equivalent to median regression, whereas for fixed $\alpha \neq 0$ it leads to quantile regression.
Hence a particular case of our framework is obtained when conditioning upon asymmetric Laplace errors
with a fixed $\alpha$, this leads to Bayesian quantile regression for the quantile $\tau=(1+\alpha)/2$.
Fixing $\alpha$ can be interesting in certain applications, is implemented in our software and
illustrated in the DLD data (Section \ref{ssec:dld}).
However by default we recommend treating $\alpha$ as a parameter to be learnt from the data.
This reduces sensitivity to model misspecification: conditioning upon non-optimal $\alpha$ increases the KL-divergence between the assumed model class and the data-generating truth, which may decrease power to detect truly active variables (Proposition \ref{prop:bfrates} and follow-up discussion).
Further, we propose a framework to infer the error distribution, clearly there one wishes to use the best-fitting $\alpha$.
Finally, each $\alpha$ conditioned upon may lead to different selected variables, this can be interesting but in applications one often is more interested in global variable selection.

Our first results regarding the tractability of (\ref{eq:skewnorm_loglhood})-(\ref{eq:skewlap_loglhood})
are given in Propositions \ref{prop:lhood_skewnorm_properties}-\ref{prop:lhood_skewlap_properties}
(Proposition \ref{prop:lhood_skewnorm_properties}(i) was already shown by \cite{mudholkar:2000}).

\begin{prop}
The asymmetric Normal log-likelihood in (\ref{eq:skewnorm_loglhood}) satisfies:
\begin{enumerate}[label=(\roman*)]
\item Its gradient is continuous and is given by

\begin{align}
g_1(\theta_\gamma,\vartheta,\alpha)=
\begin{pmatrix}
\frac{1}{\vartheta} X_\gamma^T W^2 (y - X_\gamma\theta_\gamma) \\
-\frac{n}{2 \vartheta} + \frac{1}{2 \vartheta^2}(y-X_\gamma\theta_\gamma)^T W^2 (y-X_\gamma\theta_\gamma) \\
\frac{1}{\vartheta} (y - X_\gamma\theta_\gamma)^T \Wbar^3 (y - X_\gamma\theta_\gamma)
\end{pmatrix}.
\nonumber
\end{align}

\item Its Hessian with respect to $\theta_\gamma$ is continuous everywhere except on the zero Lebesgue measure set
$\{ \theta_\gamma \in \mathbb{R}^p: x_{\gamma i}^T\theta_\gamma = y_i \mbox{ for some } i=1,\ldots,n \}$,
and is $H_1(\theta_\gamma,\vartheta,\alpha)= \vartheta^{-1} \times$
\begin{align}
\begin{pmatrix}
-X_\gamma^T W^2 X_\gamma
& \frac{1}{\vartheta} X_\gamma^T W^2 (X_\gamma\theta_\gamma - y)
& -2 X_\gamma^T \Wbar^3 (y - X_\gamma\theta_\gamma) \\
\frac{1}{\vartheta} (X_\gamma\theta_\gamma - y)^T W^2 X_\gamma
 & \frac{n}{2\vartheta} - \frac{(y-X_\gamma\theta_\gamma)^T\Wbar^2 (y-X_\gamma\theta_\gamma)}{\vartheta^2}
 & - \frac{1}{\vartheta} (y - X_\gamma\theta_\gamma)^T \Wbar^3 (y - X_\gamma\theta_\gamma) \\
-2 (y - X_\gamma\theta_\gamma)^T \Wbar^3 X_\gamma
& - \frac{1}{\vartheta} (y - X_\gamma\theta_\gamma)^T \Wbar^3 (y - X_\gamma\theta_\gamma)
& -3 (y-X_\gamma\theta_\gamma)^T W^4 (y-X_\gamma\theta_\gamma)
\end{pmatrix},
\nonumber
\end{align}

\item If $\mbox{rank}(X_\gamma)=p_\gamma$, then
$H_1(\theta_\gamma,\vartheta,\alpha)$ is strictly negative definite with respect to $(\theta_\gamma,\alpha)$ and
(\ref{eq:skewnorm_loglhood}) has a unique maximum $(\widehat{\theta_\gamma},\widehat{\vartheta},\widehat{\alpha})$.
Alternatively, if $\mbox{rank}(X_\gamma)<p_\gamma$, then $H_1(\theta_\gamma,\vartheta,\alpha)$ is negative semidefinite.
\end{enumerate}
\label{prop:lhood_skewnorm_properties}
\end{prop}

The implication is that, analogously to Normal errors, when $X_\gamma$ has full rank (\ref{eq:skewnorm_loglhood}) is continuous and concave
almost everywhere in $(\theta_\gamma,\alpha)$.
This fact, combined with $\log L_1$ having a continuous gradient, guarantees overall concavity and hence a unique maximum
(see the proof for a formal argument).
Further, inspection of \eqref{prop:lhood_skewnorm_properties} reveals that
$\log L_1$ is locally quadratic as a function of $\theta_\gamma$ within regions of constant $A(\theta_\gamma)$
and that its maximizer with respect to $(\theta_\gamma,\alpha)$ does not depend on $\vartheta$,
two observations that facilitate optimization.

Proposition \ref{prop:lhood_skewlap_properties} shows that, although $\log L_2$ is piecewise-linear in $\theta_\gamma$ and thus has a singular hessian,
one can prove concavity and uniqueness of a maximum in terms of $(\theta_\gamma,\alpha)$ as in Proposition \ref{prop:lhood_skewnorm_properties},
extending the well-known result of concavity with respect to only $\theta_\gamma$ \cite{koenker:2005}. 
In Sections \ref{sec:paramestim}-\ref{sec:modsel} we describe how this result facilitates computation,
in particular leading to simple optimization and analytical approximations to integrated likelihoods,
and asymptotic characterizations.


\begin{prop}
The asymmetric Laplace log-likelihood in (\ref{eq:skewlap_loglhood}) satisfies:
\begin{enumerate}[label=(\roman*)]
\item It is continuously differentiable with gradient

\begin{align}
g_2(\theta_\gamma,\vartheta,\alpha)= \vartheta^{-\frac{1}{2}} \times
\begin{pmatrix}
- X_\gamma^T \wbar \\
-\frac{n}{2\vartheta^{\frac{1}{2}}}  + \frac{1}{2\vartheta}
w^T \vert y - X_\gamma\theta_\gamma \vert \\
 \vert y - X_\gamma\theta_\gamma \vert^T \wbar^2
\end{pmatrix},
\nonumber
\end{align}

\noindent except on the zero Lebesgue measure set
$\{ \theta_\gamma \in \mathbb{R}^p: x_{\gamma i}^T\theta_\gamma = y_i \mbox{ for some } i=1,\ldots,n \}$,
where the gradient is undefined.

\item Its Hessian with respect to $\theta_\gamma$ is continuous everywhere except on the zero Lebesgue measure set
$\{ \theta_\gamma \in \mathbb{R}^p: x_{\gamma i}^T\theta_\gamma = y_i \mbox{ for some } i=1,\ldots,n \}$,
and is $H_2(\theta_\gamma,\vartheta,\alpha)= \vartheta^{-1/2} \times$

\begin{align}
\begin{pmatrix}
0 &
\frac{1}{2\vartheta} X_\gamma^T \wbar
& X_\gamma^T w^2 \\
\frac{1}{2\vartheta} \wbar^T X_\gamma
& \frac{n}{2\vartheta^{\frac{3}{4}}} - \frac{3}{4\vartheta^2} w^T \vert y-X_\gamma \theta_\gamma \vert
& -\frac{1}{2\vartheta} \vert y - X_\gamma\theta_\gamma \vert^T \wbar^2 \\
(X_\gamma^T w^2)^T
& -\frac{1}{2\vartheta} \vert y - X_\gamma\theta_\gamma \vert^T \wbar^2
& -2 \vert y - X_\gamma \theta_\gamma \vert^T \wbar^3
\end{pmatrix}.
\nonumber
\end{align}

\item If $\mbox{rank}(X_\gamma)=p_\gamma$, then \eqref{eq:skewlap_loglhood} is strictly concave in $(\theta_\gamma,\alpha)$
and has a unique maximum $(\widehat{\theta_\gamma},\widehat{\vartheta},\widehat{\alpha})$.
Alternatively, if $\mbox{rank}(X_\gamma)<p_\gamma$, then it is non-strictly concave in $(\theta_\gamma,\alpha)$.
\end{enumerate}

\label{prop:lhood_skewlap_properties}
\end{prop}

Parameter estimates maximizing \eqref{eq:skewnorm_loglhood}-\eqref{eq:skewlap_loglhood} can be interpreted as the best-fitting linear model under weighted least-squares or weighted least absolute deviations, respectively. Different weights are assigned to observations on each side of the estimated $x_i^T \theta$. The weights are determined by $\alpha$, which captures residual asymmetry and converges to a unique KL-optimal value (Section \ref{ssec:asymp_mle}). Selected variables can be interpreted in a similar fashion, essentially as defining the smallest model amongst those minimizing each criterion (Section \ref{ssec:bf_rates}). That is, variable selection can be understood in terms of optimal variable configurations under well-known criteria.

\section{Prior formulation}
\label{sec:prior}

We complete the Bayesian model via priors on the model indicators $\gamma$
and the model-specific parameters $(\theta_\gamma, \alpha)$.
For $p(\gamma)$ by default we adopt the standard Beta-Binomial$(a_\gamma,b_\gamma)$ prior \cite{scott:2010}
where $a_\gamma,b_\gamma>0$ are known constants (by default $a_\gamma=b_\gamma=1$),
although our implementation also incorporates uniform and Binomial priors.
The four posed residual distributions (Normal, asymmetric Normal, Laplace and asymmetric Laplace)
are assigned equal prior probability independently from the variable inclusions.
Therefore
\begin{align}
p(\gamma)= \frac{1}{4} \frac{B(a_\gamma + \sum_{j=1}^{p} \gamma_j, b_\gamma + p - \sum_{j=1}^{p} \gamma_j)}{B(a_\gamma,b_\gamma)},
\label{eq:priormodel}
\end{align}
where $B()$ is the Beta function.
Any model with $p_\gamma>n$ is assigned $p(\gamma)=0$, as it would result in data interpolation.

Regarding $p(\theta_\gamma \mid \gamma)$,
given that the mode, mean and median of $y_i$ are linear in $x_{\gamma i}^T \theta_\gamma$
the usual prior specification strategies under Normal errors remain sensible.
The possibilities are too numerous to list here, see \textit{e.g.}~\cite{bayarri:2012} or \cite{mallickh:2013} and references therein.
We focus on the class of NLPs introduced by \cite{johnson:2010},
as these lead to stronger sparsity than conventional (local) priors and
(under suitable conditions) consistency of posterior model probabilities in high-dimensional Normal regression
where $p=o(n)$ \cite{johnson:2012} or $\log p=o(n)$ \cite{shin:2015}.
However our theory also applies to local priors.
The basic intuition is that, under model $\gamma$, all elements in $\theta_\gamma$
are assumed to be non-zero. Thus, $p(\theta_\gamma \mid \gamma)$ should vanish as any element in $\theta_\gamma$ approaches 0.
We focus on two specific choices \cite{johnson:2012,rossell:2013b}
\begin{align}
p_M(\theta_\gamma\mid \vartheta, \gamma) = \prod_{\gamma_j=1}^{}
  \frac{\theta_j^2}{k g_{\theta} \vartheta} N(\theta_j; 0, g_{\theta} k \vartheta),
\label{eq:pmom} \\
p_E(\theta_\gamma\mid \vartheta, \gamma) = \prod_{\gamma_j=1}^{} \exp \left\{
    \sqrt{2} - \frac{g_{\theta} k \vartheta}{\theta_j^2} \right\}
 N(\theta_j;0, g_{\theta} k \vartheta),
\label{eq:pemom}
\end{align}
called product MOM and eMOM priors (respectively), where $g_{\theta}$ is a known prior dispersion.
For Normal or asymmetric Normal errors $k=1$, and for the Laplace or asymmetric Laplace $k=2$
as then $\mbox{Var}(\epsilon_i)$ is proportional to $2\vartheta$.
Along the same lines for the scale parameter we set a standard inverse gamma
$p(\vartheta \mid \gamma)= \mbox{IG}(\vartheta; a_{\vartheta}/2,k b_{\vartheta}/2)$
(in our examples $a_{\vartheta}=b_{\vartheta}=0.01$).
MOM vanishes at a quadratic speed around the origin
and accelerates polynomial Bayes factor sparsity rates,
whereas eMOM vanishes exponentially and leads to quasi-exponential rates \cite{johnson:2010,rossell:2017},
a result we extend here for our new class of models and under model misspecification (Section \ref{sec:modsel}).
In our examples, we follow the default recommendation in \cite{johnson:2010}
and set $g_{\theta}=0.348,0.119$ for MOM and eMOM (respectively),
under the rationale that they assign 0.01 prior probability to $| \theta_i/\sqrt{\vartheta} |< 0.2$,
{\it i.e.}~effect sizes often deemed practically irrelevant.
Naturally, whenever prior information is available we recommend using it to set $g_\theta$.
The supplementary material describes a third prior class called iMOM that provides a thick-tailed counterpart to the eMOM.
Although the iMOM is implemented in our software, we do not consider it further here
given that its performance was very similar to the eMOM
but it has the unappealing property of leading to non-convex optimization (akin to other thick-tailed priors, \textit{e.g.}~Cauchy),
and when considering $p(\alpha)$ (see below) it leads to a density that diverges on the boundary
($\alpha=-1$ or $\alpha=1$).

To set $p(\alpha \mid \gamma_{p+1}=1)$ ($\alpha=0$ under $\gamma_{p+1}=0$)
we reparameterize $\talpha=\operatorname{atanh}(\alpha)\in{\mathbb R}$ as in \cite{rubio:2014}.
These authors proposed $0.5 (1+\alpha) \sim \mbox{Beta}(2,2)$,
which places the prior mode at $\alpha=0$ and thus defines a local prior.
Our goal here is to detect situations where the degree of asymmetry is practically relevant
and to otherwise allow the posterior to collapse on the symmetric model.
To achieve this, we consider
$p_M(\talpha \mid \gamma_{p+1}=1)= \talpha^2 \phi(\talpha/\sqrt{g_\alpha})/\sqrt{g_\alpha}$,
and $p_E(\talpha \mid \gamma_{p+1}=1)= e^{\sqrt{2}- g_\alpha/\talpha^2} N(\talpha;0,g_\alpha)$,
where $g_{\alpha} \in \mathbb{R}^+$ is a fixed prior dispersion parameter.
To set $g_\alpha$, by default we consider that Arnold-Groeneveld asymmetry coefficients $|\alpha|< 0.2$ are often practically irrelevant.
Thus, we set $g_\alpha$ such that $P(|\alpha| \geq 0.2)=0.99$.
Also, note that $\alpha=2$ gives a total variation distance of $|\alpha|/2=0.1$,
{\it i.e.} the largest difference
$|P(\epsilon_i \in A \mid \alpha=0)-P(\epsilon_i \in A \mid \alpha)|$ for any set $A$ is 0.1,
which we typically view as irrelevant.
Since $\mbox{atanh}(0.2)=0.203$, a direct calculation gives that $P(|\talpha| \geq 0.203)= 0.99$ when
$g_\alpha=0.357,0.122$ under MOM and eMOM.
To assess sensitivity in our examples, we also considered $g_\alpha$ such that $P(|\alpha| \geq 0.1)=0.99$
(total variation distance=0.05), giving $g_\alpha=0.087,0.030$.
Figure \ref{fig:NLPalpha} depicts $p(\alpha)$ under these settings.
Our results showed that variable selection is typically robust to choices of $g_\alpha$ within this range.

\begin{figure}[h]
\begin{center}
\begin{tabular}{c}
\includegraphics[width=8cm,height=6cm]{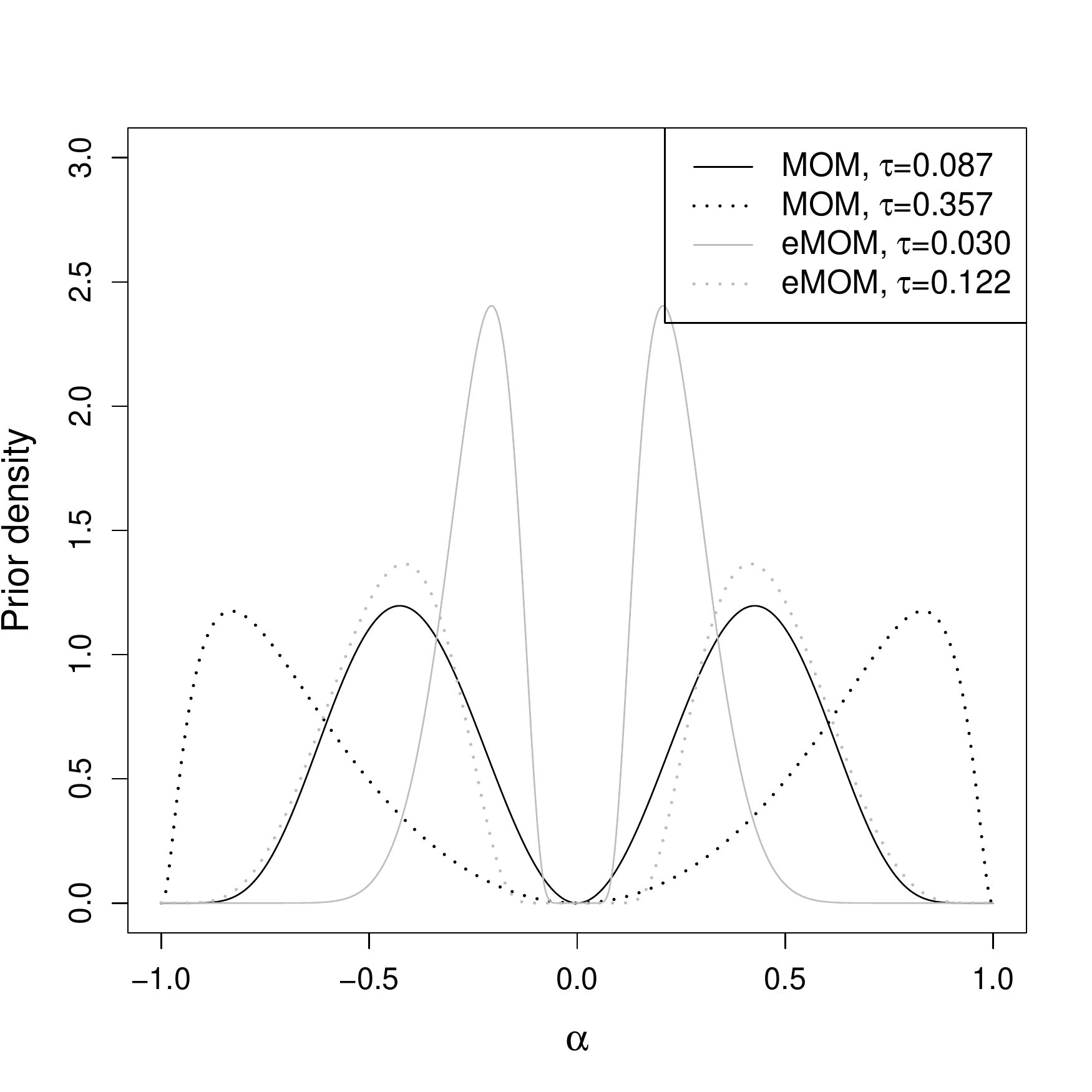}
\end{tabular}
\end{center}
\caption{Default priors for $\alpha$.}
\label{fig:NLPalpha}
\end{figure}

\section{Parameter estimation}
\label{sec:paramestim}

We obtain some results for parameter estimation under a given $\gamma$ that are also useful to establish variable selection rates
(see Section \ref{sec:modsel} for results on Bayesian model averaging).
Section \ref{ssec:asymp_mle} gives the limiting distribution of
$(\widehat{\theta}_\gamma,\widehat{\vartheta}_\gamma,\widehat{\alpha}_\gamma)= \arg\max_{\theta_\gamma,\vartheta,\alpha} \log L_k(\theta_\gamma,\vartheta,\alpha)$
as $n \rightarrow \infty$ for asymmetric Normal $(k=1)$ and Laplace $(k=2)$
when data are generated from (\ref{eq:lm}) but the error model may be misspecified.
Briefly, as is typically the case, we obtain parameter estimation consistency
and asymptotic normality, albeit there is a loss of efficiency and an underestimation of uncertainty.
Section \ref{ssec:mle_optim} presents novel optimization algorithms for maximum likelihood and posterior mode estimation
designed to improve the computational scalability of current related methods.

\subsection{Asymptotic distributions}
\label{ssec:asymp_mle}

We lay out technical conditions for our asymptotic results to hold.
\begin{enumerate}[label=\bfseries A\arabic*.]
\item The parameter space $\Gamma\subset {\mathbb R}^p\times{\mathbb R}_+ \times(-1,1)$ is compact and convex.

\item Data are truly generated as $y_i= x_i^T \theta^* + \epsilon_i$ for some $\theta^* \in \mathbb{R}^p$, fixed $p_{\gamma^*}=\sum_{j=1}^p \mbox{I}(\theta_j^* \neq 0)$ and $\epsilon_i$ are \emph{i.i.d.} and independent of $x_i$. Let the data-generating $y_i \vert x_i \stackrel{i.i.d.}{\sim} S_0(\cdot \vert x_i)$ with density $s_0(y_i \mid x_i)>0$ for all $y_i$.

\item For all $\gamma$ there is some $n_0$ such that $X_\gamma^TX_\gamma$ is strictly positive definite almost surely for all $n>n_0$.

\item
Denote by $x_i \stackrel{i.i.d.}{\sim} \Psi(\cdot)$ the generating process of the covariates (which can be either stochastic or deterministic).
\begin{eqnarray*}
\int  \vert y_1\vert^j dS_0(y_1\vert x_1) d\Psi(x_1)<\infty,\\
\int  \vert \vert x_1 \vert \vert^j d\Psi(x_1)<\infty,
\end{eqnarray*}
where $j=1,2, \text{ or } 4$, and we specify the order $j$ of interest in each of the results below, and $\vert\vert \cdot \vert \vert$ denotes the Euclidean distance $\vert\vert z \vert\vert= (\sum_{}^{} z_i^2)^{\frac{1}{2}}$.

\item For $\eta \in \Gamma$
\begin{eqnarray*}
\int\dfrac{\partial}{\partial \eta_j}\left[\int m_{\eta}(y_1,x_1) dS_0(y_1\vert x_1) \right]d\Psi(x_1) &=& \dfrac{\partial}{\partial \eta_j} \int\int m_{\eta}(y_1,x_1) dS_0(y_1\vert x_1) d\Psi(x_1),\\
\int\dfrac{\partial^2}{\partial \eta_i \eta_j}\left[\int m_{\eta}(y_1,x_1) dS_0(y_1\vert x_1) \right]d\Psi(x_1) &=& \dfrac{\partial^2}{\partial \eta_i \eta_j} \int\int m_{\eta}(y_1,x_1) dS_0(y_1\vert x_1) d\Psi(x_1).
\end{eqnarray*}

\end{enumerate}

These conditions are in line with those in classical robust regression, \textit{e.g.}~see \cite{huber:1973} or \cite{koenker:1982}.
Condition A1 is made out of technical convenience, naturally one may take an arbitrarily large $\Gamma$. Condition A2 states that data truly arise from a linear model, where the key assumption is that the residuals are independent. Extensions to non-id errors are discussed in Section \ref{ssec:bf_rates}. Condition A3 holds whenever the rows of $X$ are regarded as a deterministic sequence satisfying the condition, or for instance when $x_i$ are independent and identically distributed from an underlying distribution of fixed dimension with positive-definite $\mbox{Cov}(x_1)$, as then $X^TX$ converges almost surely to a positive-definite matrix by the strong law of large numbers. We focus on fixed $p$, extensions to $p$ growing with $n$ are possible along the lines in \cite{mendelson:2014}, but its detailed treatment is beyond the scope of this paper. Condition A4 requires existence of moments up to a certain order.
Condition A5 requires being able to exchange integration and differentiation, and is needed only to prove asymptotic normality.

Our results summarize and extend classical studies focusing on $\theta_\gamma$ in least squares, median and quantile regression to consider the whole parameter vector $(\theta_\gamma,\vartheta,\alpha)$. Briefly, \cite{eicker:1964} and \cite{srivastava:1971} showed that the least squares estimator ($k=1,\alpha=0$) satisfies  $\sqrt{n} V^T (\widehat{\theta}_{\gamma}-\theta_0) \stackrel{D}{\longrightarrow} N(0, \mbox{Var}(\epsilon_1) I)$,
where $\theta_0$ minimizes Kullback-Leibler divergence to the data-generating truth and $VV^T=X_\gamma^TX_\gamma/n$,
assuming that $\mbox{Var}(\epsilon_1)< \infty$ and minimum conditions on $X_\gamma^TX_\gamma$. 
To our knowledge, the asymmetric Normal has been much less studied, \textit{e.g.}~\cite{kimber:1985}, \cite{mudholkar:2000} and \cite{arellanovalle:2005} considered the case with no covariates and no checks of the conditions required by large sample theory are shown, which are non-trivial given that $H_1(\theta_{\gamma},\vartheta,\alpha)$ is discontinuous. Regarding Laplace errors ($k=2,\alpha=0$), \cite{pollard:1991} and \cite{knight:1999} showed $2f_0 \sqrt{n}V^T(\widehat{\theta}_{\gamma}-\theta_0) \stackrel{D}{\longrightarrow} N(0,I)$, where $f_0=p(\epsilon_0)$ and $\epsilon_0$ is the median of $s_0(\epsilon_i)$, under mild conditions on $X_\gamma^TX_\gamma$ and $f_0>0$.
\cite{koenker:1994} generalized the result to the asymmetric Laplace, obtaining $2f_0\sqrt{n/(1-\alpha^2)}V^T(\widehat{\theta}_{\gamma}-\theta_0) \stackrel{D}{\longrightarrow} N(0,I)$, where $f_0=p(\epsilon)$ evaluated at the $\tau^{th}$ quantile $\epsilon= S_0^{-1}(\tau)$, where
in our parameterization $\tau=(1+\alpha)/2$.
Proposition \ref{prop:Consistency} establishes the consistency of the maximum likelihood estimator
$\widehat{\eta}_\gamma=(\widehat{\theta}_\gamma,\widehat{\vartheta}_\gamma,\widehat{\alpha}_\gamma)$
to the Kullback-Leibler optimal parameter values,
whereas Proposition \ref{prop:AsympNorm} gives asymptotic normality.

\begin{prop}\label{prop:Consistency}
Assume Conditions A1--A4 with $p<n$, where $j=2$ in A4 when $k=1$ and $j=1$ when $k=2$.
Then, the function $M_k(\theta_\gamma,\vartheta,\alpha) =  {\mathbb E}[\log L_k(y_1\vert x_1^{T}\theta_\gamma,\vartheta,\alpha)]$ has a unique maximizer $(\theta_\gamma^*,\vartheta_\gamma^*,\alpha_\gamma^*) = \operatorname{argmax}_{\Gamma} M_k(\theta_\gamma,\vartheta,\alpha)$. Moreover, the maximum likelihood estimator $(\widehat{\theta}_\gamma,\widehat{\vartheta}_\gamma,\widehat{\alpha}_\gamma)\stackrel{P}{\rightarrow} (\theta_\gamma^*,\vartheta_\gamma^*,\alpha_\gamma^*)$ as $n\rightarrow \infty$.
\end{prop}

%
%

\begin{prop}\label{prop:AsympNorm}
Assume Conditions A1--A5, with $j=4$ in A4 when $k=1$ and $j=2$ when $k=2$. Denote $\eta = (\theta_\gamma,\vartheta,\alpha)$, $m_{\eta}(y_1,x_1)=\log s_k(y_1\vert x_1^{T}\theta_{\gamma},\vartheta,\alpha)$, $P m_{\eta} = {\mathbb E}\left[m_{\eta}(y_1,x_1)\right]$, and $\eta_\gamma^* = (\theta_\gamma^*,\vartheta_\gamma^*,\alpha_\gamma^*) = \operatorname{argmax}_{\Gamma} P m_{\eta}$.
Then, the sequence $\sqrt{n}(\widehat{\eta}_\gamma-\eta_\gamma^*)$  is asymptotically Normal with mean $0$ and covariance matrix $V_{\eta_\gamma^*}^{-1} {\mathbb E}[ \dot{m}_{\eta_\gamma^*} \dot{m}_{\eta_\gamma^*}^T] V_{\eta_\gamma^*}^{-1}$, where $\dot{m}_{\eta_\gamma^*}$ is the gradient of $m_{\eta}(\cdot)$, with respect to $\eta$, evaluated at $\eta_\gamma^*$ and $V_{\eta_\gamma^*}$ is the second derivative matrix of $P m_{\eta}$ evaluated at $\eta_\gamma^*$.
\end{prop}


The sandwich covariance $V_{\eta_\gamma^*}^{-1} {\mathbb E}[ \dot{m}_{\eta_\gamma^*} \dot{m}_{\eta_\gamma^*}^T] V_{\eta_\gamma^*}^{-1}$
is typically an inflated version of that obtained when the true model is assumed ($V_{\eta_\gamma^*}^{-1}$),
implying the well-known consequence of model misspecification
that parameter estimation suffers a loss of efficiency and uncertainty is underestimated.
To gain insight, Corollary \ref{cor:asymp_mle_misspec} gives specific asymptotic variances under various model misspecification cases.
For instance, when truly $\epsilon_i \sim N(0,\vartheta)$ wrongly assuming Laplace errors increases the variance by a factor $\pi/2$,
and a similar phenomenon is observed when ignoring the presence of residual asymmetry.
We defer discussion of the implications for variable selection to Section \ref{sec:modsel} and the examples in Section \ref{sec:results}.
\begin{cor}
The asymptotic distribution of $\widehat{\theta}_\gamma$ obtained by maximizing either the Normal, ANormal, Laplace or ALaplace likelihood
is $V(\widehat{\theta}_\gamma-\theta_\gamma^*) \stackrel{D}{\longrightarrow} N\left(0, v I \right)$, for some $v>0$.
The asymptotic variances $v$, when $\epsilon_i$ truly arise \emph{i.i.d.}~under four specific distributions, are given below.
\begin{center}
\begin{tabular}{|c|c|c|c|c|}\hline
                                & \multicolumn{4}{|c|}{Maximized log-likelihood} \\ \hline
True model                      & Normal & ANormal & Laplace & ALaplace \\ \hline
$N(0,\vartheta)$                & $\vartheta$  & $\vartheta$ & $\frac{\pi}{2} \vartheta$ & $\frac{\pi}{2} \vartheta$\\

$\mbox{AN}(0,\vartheta,\alpha)$ & $\vartheta (1+0.454 \alpha^2)$ & $\vartheta(1-\alpha^2)$ $(\star)$ & $\frac{\pi}{2} \vartheta k_{\alpha}$ & $\dfrac{\pi}{2}\vartheta (1-{\alpha^*_\gamma}^2)$\\

$L(0,\vartheta)$                & $2\vartheta$ & $2\vartheta$ & $\vartheta$ & $\vartheta$\\

$\mbox{AL}(0,\vartheta,\alpha)$ & $2\vartheta (1+\alpha^2)$ & $2\vartheta w_{\alpha,\alpha^*_\gamma}$ ($\star$) &$\vartheta (1+|\alpha|)^2$ & $\vartheta (1-\alpha^2)$\\ \hline
\end{tabular}
\end{center}
where $k_{\alpha}= \exp\left\{ \left[ \Phi^{-1}\left( \frac{1}{2(1+|\alpha|)} \right)  \right]^2  \right\} \geq 1$, $w_{\alpha,\alpha^*_\gamma} = \dfrac{(1+\alpha)^2-2 \alpha  \left(1+\alpha^*_\gamma\right)}{\left(1-\alpha ^2\right)^2} \in [0,1]$, and $\alpha^*_\gamma$ is as in Proposition \ref{prop:AsympNorm}. Cases marked $(\star)$ 
were derived assuming that covariates have zero mean. 
\label{cor:asymp_mle_misspec}
\end{cor}

\subsection{Optimization}
\label{ssec:mle_optim}

We outline simple, efficient algorithms to obtain
$(\widehat{\theta}_\gamma,\widehat{\vartheta}_\gamma,\widehat{\alpha}_\gamma)= \arg\max_{\theta_\gamma,\vartheta,\alpha} \log L_k(\theta_\gamma,\vartheta,\alpha)$, where $k \in \{1,2\}$ are the asymmetric Normal and Laplace log-likelihoods \eqref{eq:skewnorm_loglhood}-\eqref{eq:skewlap_loglhood}. We also consider the corresponding posterior modes
$(\tilde{\theta}_\gamma,\tilde{\vartheta}_\gamma,\tilde{\alpha}_\gamma)= \arg\max_{\theta_\gamma,\vartheta,\alpha} \log L_k(\theta_\gamma,\vartheta,\alpha)
+ \log p(\theta_\gamma,\vartheta,\alpha \mid \gamma)$, where $p(\theta_\gamma,\vartheta,\alpha \mid \gamma)$ is the prior density (Section \ref{sec:prior}).
The algorithms are useful to obtain parameter estimates or Laplace approximations to the integrated likelihood.
\cite{mudholkar:2000} and \cite{arellanovalle:2005} gave an algorithm to obtain $\widehat{\theta}_\gamma$
for $\log L_1$ in the case with no covariates ($p_\gamma=1$).
To tackle point discontinuities in the derivatives
their algorithm requires solving $n$ separate optimization problems, which does not scale up with increasing $n$,
or alternatively using method of moments estimators.
Maximum likelihood estimation of $\theta_\gamma$ under the asymmetric Laplace and fixed $\alpha$ is connected to quantile regression (see below).
Regarding Bayesian frameworks, most rely on MCMC for parameter estimation
but this is too costly when we wish to consider a potentially large number of models.
Instead, we propose a generic framework for jointly obtaining $(\widehat{\theta}_\gamma,\widehat{\vartheta}_\gamma,\widehat{\alpha}_\gamma)$
or $(\tilde{\theta}_\gamma,\tilde{\vartheta}_\gamma,\tilde{\alpha}_\gamma)$ applicable to both the asymmetric Normal and Laplace.
The key result we exploit is concavity of the log-likelihood given by
Propositions \ref{prop:lhood_skewnorm_properties}-\ref{prop:lhood_skewlap_properties},
which allows iteratively optimizing first $\theta_\gamma$ and then $(\vartheta,\alpha)$.
Optimization with respect to $(\vartheta,\alpha)$ has closed form,
whereas setting $\theta_\gamma$ can be seen as weighted least squares for the asymmetric Normal
and as quantile regression for the asymmetric Laplace.
The latter task of maximizing $\log L_2$ with respect to $\theta_\gamma$ is a classical problem that can be framed as linear programming,
for which simplex and interior-point methods are available.
However, these are not applicable to the posterior mode as the target is no longer piecewise linear
and even efficient implementations have computational complexity greater than cubic in $p$ and supra-linear in $n$
\cite{koenker:2005}.

We outline two simple algorithms that have lower complexity
and can be readily adapted to obtain the posterior mode.
Briefly, in Algorithm \ref{alg:mle_lma}, Step 2 follows from setting first derivatives to zero
and directly extends \cite{mudholkar:2000} (Proposition 4.4) and \cite{arellanovalle:2005} (Section 4.2).
Step 3 is essentially a Levenberg-Marquardt algorithm \cite{levenberg:1944,marquardt:1963} exploiting gradient continuity.
$g_\theta$ and $H_\theta$ denote the gradient and hessian with respect to $\theta_\gamma$
as in Propositions \ref{prop:lhood_skewnorm_properties}-\ref{prop:lhood_skewlap_properties},
where for $\log L_2()$ we use the asymptotic hessian $X^TX/(\vartheta (1-\alpha^2))$.
Its updates are in between those of a Newton-Raphson and gradient descent algorithms
and can be interpreted as restricting the Newton-Raphson step to a trust region
where the quadratic approximation is accurate \cite{sorensen:1982}.
For large regularization parameter $\lambda$ the update $\delta$ converges to the gradient algorithm,
which by continuity is guaranteed to increase the target,
whereas for small $\lambda$ it converges to the Newton-Raphson algorithm,
achieving quadratic convergence as $\theta_\gamma^{(t)}$ approaches the optimum.

\begin{algorithm}
{\bf Optimization via Levenberg-Marquardt}\label{alg:mle_lma}
\begin{enumerate}
\item Initialize $\widehat{\theta}_\gamma^{(0)}= (X^TX)^{-1}X^Ty$, $\lambda=0$.
Set $t=1$
\item Let $s_1=\sum_{i \in A(\widehat{\theta}_\gamma^{(t-1)})}^{}  |y_i-x_{\gamma i}^T\widehat{\theta}_\gamma^{(t-1)}|^{3-k}$,
$s_2=\sum_{i \not\in A(\widehat{\theta}_\gamma^{(t-1)})}^{}  |y_i-x_{\gamma i}^T\widehat{\theta}_\gamma^{(t-1)}|^{3-k}$.
Update
\begin{align}
\widehat{\alpha}^{(t)}=
\frac{ s_1^{\frac{k}{2+k}} - s_2^{\frac{k}{2+k}} }{ s_1^{\frac{k}{2+k}} + s_2^{\frac{k}{2+k}}};
\widehat{\vartheta}^{(t)}= \frac{1}{4n^{k}} \left( s_1^{\frac{k}{2+k}} + s_2^{\frac{k}{2+k}} \right)^{2+k}.
\nonumber
\end{align}

\item Propose $m= \theta_\gamma^{(t-1)} + \delta$, where
\begin{align}
\delta= -\left(H_\theta + \lambda \mbox{diag}(H_\theta)\right)^{-1} g_\theta, \nonumber
\end{align}
and $g_\theta, H_\theta$ are the subsets of
$g_k(\widehat{\theta}_\gamma^{(t-1)},\widehat{\vartheta}^{(t)},\widehat{\alpha}^{(t)})$ and
$H_k(\widehat{\theta}_\gamma^{(t-1)},\widehat{\vartheta}^{(t)},\widehat{\alpha}^{(t)})$
corresponding to $\theta_\gamma$.
If $\log L_k(m,\vartheta^{(t)},\alpha^{(t)})>\log L_k(\theta_\gamma^{(t-1)},\vartheta^{(t)},\alpha^{(t)})$
set $\theta_\gamma^{(t)}=m$ and $\lambda= \lambda/2$, else update $\lambda= 1+\lambda$ and repeat Step 3.
\end{enumerate}
\end{algorithm}

Given a good initial guess $\widehat{\theta}_\gamma^{(0)}$,
the fact that $\log L_k$ are locally well approximated by a quadratic function in $\theta_\gamma$
($\log L_1$ is exactly locally quadratic) results in Algorithm \ref{alg:mle_lma} usually converging after a few iterations.
As usual, with second-order optimization each iteration requires a matrix inversion that is costly when $p$ is large.
As an alternative, Algorithm \ref{alg:mle_cda} uses coordinate descent to optimize each $\theta_{\gamma j}$ sequentially,
which only requires univariate updates,
where updating the set $A(\theta_\gamma)$ for each $\theta_{\gamma j}$ implies that Step 3 has cost $O(n p)$.
In contrast, Algorithm \ref{alg:mle_lma} determines $A(\theta_\gamma)$ once per iteration and performs matrix inversion,
with total cost $O(n + p^3)$ per iteration. Hence, although Algorithm \ref{alg:mle_lma} usually requires fewer iterations than Algorithm \ref{alg:mle_cda},
for large $p$ the latter is typically preferrable.
A related study of computational cost is offered in \cite{breheny:2011} in the context of penalized likelihood optimization,
who found that coordinate descent is often preferrable to multivariate updates.
These results show that, contrary to historical beliefs, two-piece distributions lead to convenient optimization.
R package mombf \cite{mombf:2016} incorporates both algorithms but our examples are based on Algorithm \ref{alg:mle_cda}, the results were essentially identical
to those of Algorithm \ref{alg:mle_lma} but the running time was substantially shorter.

We adapted both algorithms to find the posterior mode by simply redefining $g_k$ and $H_k$
to be the gradient and Hessian of $\log L_k(\theta_\gamma,\vartheta,\alpha) + \log p(\theta_\gamma,\vartheta,\alpha \mid \gamma)$.
The corresponding expressions are in Supplementary Section \ref{ssec:deriv_logprior}.
We remark that due to the penalty around the origin NLPs such as $p_M()$ and $p_E()$ in (\ref{eq:pmom})-(\ref{eq:pemom}) are not log-concave,
however this is not an issue as they are symmetric and log-concave in each quadrant (fixed $\mbox{sign}(\theta_\gamma,\alpha)$).
Thus $\log p(\theta_\gamma,\vartheta,\alpha \mid y,\gamma)$ is concave in each quadrant,
its unique global mode lies in the same quadrant as the maximum likelihood estimator
and we may initialize the algorithm at
$(\tilde{\theta}_\gamma^{(0)},\tilde{\vartheta}^{(0)},\tilde{\alpha}^{(0)})=(\widehat{\theta}_\gamma,\widehat{\vartheta}_\gamma,\widehat{\alpha}_\gamma)$.
Convergence is typically achieved after a few iterations.

\begin{algorithm}
{\bf Optimization via coordinate descent}\label{alg:mle_cda}
\begin{enumerate}
\item Set an arbitrary $c>1$ and initialize $\theta_\gamma^{(0)}$, $\lambda=0$ as in Algorithm \ref{alg:mle_lma}.

\item Update $(\widehat{\vartheta}^{(t)},\widehat{\alpha}^{(t)})$ as in Algorithm \ref{alg:mle_lma}.

\item For $j=1,\ldots,p_\gamma$,
let $m= \theta_{\gamma j}^{(t-1)} - \frac{g_j}{h_{jj}(1 + \lambda)}$,
where $g_j$ is the $j^{th}$ element in $g_1(\theta_\gamma)$ and $h_{jj}$ the $(j,j)$ element in $H_1(\theta_\gamma)$
at $\theta_\gamma=(\theta_{\gamma 1}^{(t)},\ldots,\theta_{\gamma j-1}^{(t)},\theta_{\gamma j}^{(t-1)},\ldots,\theta_{\gamma p_\gamma}^{(t-1)})$.
If $L_k$ evaluated at $\theta_{\gamma j}^{(t)}=m$ increases, set $\theta_{\gamma j}^{(t)}=m$, $\lambda= \lambda/c$,
else iteratively update $\lambda= c+\lambda$ and $m$ until $L_k$ increases.
\end{enumerate}
\end{algorithm}


\section{Model selection}
\label{sec:modsel}

Under a standard Bayesian framework $p(\gamma \mid y)= p(y \mid \gamma) p(\gamma) / p(y)$,
with integrated likelihood
\begin{align}
p(y \mid \gamma)= \int L_1(\theta_\gamma,\vartheta,0) p(\theta_\gamma,\vartheta) d\theta_\gamma d\vartheta,
\mbox{ if } \gamma_{p_\gamma+1}=0,\gamma_{p_\gamma+2}=0, \nonumber \\
p(y \mid \gamma)= \int L_1(\theta_\gamma,\vartheta,\alpha) p(\theta_\gamma,\vartheta,\alpha) d\theta_\gamma d\vartheta d\alpha,
\mbox{ if } \gamma_{p_\gamma+1}=1,\gamma_{p_\gamma+2}=0, \nonumber \\
p(y \mid \gamma)= \int L_2(\theta_\gamma,\vartheta,0) p(\theta_\gamma,\vartheta) d\theta_\gamma d\vartheta,
\mbox{ if } \gamma_{p_\gamma+1}=0,\gamma_{p_\gamma+2}=1, \nonumber \\
p(y \mid \gamma)= \int L_2(\theta_\gamma,\vartheta,\alpha) p(\theta_\gamma,\vartheta,\alpha) d\theta_\gamma d\vartheta d\alpha,
\mbox{ if } \gamma_{p_\gamma+1}=1,\gamma_{p_\gamma+2}=1.
\label{eq:marglhood}
\end{align}
Section \ref{ssec:intlhood} discusses how to compute $p(y \mid \gamma)$
and Section \ref{ssec:bf_rates} the asymptotic properties of the associated Bayes factors
and Bayesian model averaging, along with a discussion on model misspecification and to what extent
these results can be generalized to non-identically distributed errors (e.g. under heteroscedasticity or hetero-asymmetry).
Section \ref{ssec:modelsearch} outlines a stochastic model search algorithm that can be used when $p$ is too large for exhaustive enumeration of the $2^{p+2}$ models.

\subsection{Integrated likelihood}
\label{ssec:intlhood}

Computing (\ref{eq:marglhood}) in the case $\gamma_{p+1}=\gamma_{p+2}=0$ corresponds to Normal linear regression,
 for which existing methods are typically available,
\textit{e.g.}~\cite{johnson:2012} gave closed-form expressions for the MOM and Laplace approximations for the eMOM.
The three remaining cases require numerical evaluation, for which we propose Laplace and Monte Carlo approximations.
The former are appealing due to log-likelihood concavity and asymptotic normality (Section \ref{sec:paramestim}).
Indeed, in our examples they delivered very similar inference and were orders of magnitude faster than Monte Carlo.
Hence, by default we recommend Laplace approximations over Monte Carlo, except in small $p$ situations where the latter is still practical.
To ensure that the parameter support is on the real numbers Laplace approximations are based on the reparameterization
$\eta=(\theta_\gamma,\log(\vartheta),\mbox{atanh}(\alpha))$ and given by
\begin{align}
\widehat{p}(y \mid \gamma)= \exp\{\log L_k(\tilde{\eta}) + \log p(\tilde{\eta}) \} \frac{(2\pi)^{\sum_{j=1}^{p+2} \gamma_j/2}}{|H_k(\tilde{\eta})|^{1/2}},
\label{eq:laplace_approx}
\end{align}
where $k=1,2$ for $\gamma_{p+2}=0,1$ respectively,
$\tilde{\eta}$ and $H_k(\tilde{\eta})$ are the posterior mode and hessian of $\log L_k(\eta) + \log p(\eta)$.
The specific expressions are given in Supplementary Section \ref{sec:approx_marglhood}.
Expression (\ref{eq:laplace_approx}) simply requires the posterior mode (Algorithms \ref{alg:mle_lma}-\ref{alg:mle_cda}) and evaluating the hessian.
The latter is straightforward for $k=1$, but for $k=2$ it is singular in $\theta_\gamma$, requiring some care.
The reasoning behind (\ref{eq:laplace_approx}) is to approximate the log-integrand in (\ref{eq:marglhood})
by a smooth function that has strictly positive definite hessian in $\theta_\gamma$,
which is facilitated in our setting by $\log L_2$ concavity and asymptotic normality.
We found that a simple yet effective strategy is to replace $H_2$ by the asymptotic expected hessian $\overline{H}_2$
obtained under independent asymmetric Laplace errors.

Although we did not find the following concern to be a practical issue in our examples,
we remark that in principle $\overline{H}_2$ may underestimate the underlying uncertainty in $\theta_\gamma$
and thus inflate $|\overline{H}_2|$,
\textit{e.g.}~under truly non-Laplacian independent and identically distributed errors one needs to add a multiplicative constant (Section \ref{ssec:asymp_mle}),
whereas independent but heteroscedastic errors require a matrix-reweighting adjustment \cite{kocherginsky:2005}.
Typical strategies to improve the estimated curvature rely either on direct estimation under the assumption of independent and identically distributed errors, or indirect estimation via inversion of score tests, although these only provide univariate confidence intervals
and their cost does not scale well with $p$, or sampling-based methods such as bootstrap or Monte Carlo.
As a practical alternative here we consider that the goal is really to approximate the actual curvature of $\log L_2$,
which can be easily done with a few point evaluations of $\log L_2$ in a neighbourhood of $\tilde{\eta}_\gamma$.
Briefly, we consider the adjustment $D \overline{H}_2 D$,
where $D$ is a diagonal matrix such that its element $d_{ii}$ gives the best approximation
of $\log L_2$ as a quadratic function of $\theta_i$ in the least squares sense.
$D \overline{H}_2 D$ matches the actual curvature in $\log L_2$ and is thus less dependent on asymptotic theory than other strategies,
and has the advantage that $D$ can be computed quickly.
See Supplementary Section \ref{ssec:approx_logl_alapl} for further details and Supplementary Figure \ref{fig:quadapprox} for an example.
Given that the unadjusted $\overline{H}_2$ performed well in our examples and the associated results
were practically indistinguishable to those based on Monte Carlo, unless otherwise stated our results are based on $\overline{H}_2$.

As our Monte Carlo alternative, we implemented an importance sampling estimator based on multivariate T draws
and covariance matching the asymptotic posterior covariance.
Specifically, let $\eta^{(b)} \sim T_3(\tilde{\eta}, \tilde{H}_k^{-1}/3)$ for $b=1,\ldots,B$ where $B$ is a large integer,
then
\begin{align}
\widehat{p}_I(y \mid \gamma)= B^{-1} \sum_{b=1}^{B} L_k(\eta^{(b)}) p(\eta^{(b)}) / T_3(\eta^{(b)}; \tilde{\eta}, \tilde{H_k}^{-1}/3).
\label{eq:is}
\end{align}

We remark that NLPs are multimodal in $(\theta_\gamma,\alpha)$, thus some care is needed when using Laplace approximations.
To give an honest characterization of the properties of our preferred computational method,
in Section \ref{sec:modsel} we obtain asymptotic rates for Bayes factors based on $\hat{p}(y \mid \gamma)$ in \eqref{eq:laplace_approx}.
\cite{rossell:2017} studied the discrepancies between $p(y \mid \gamma)$ and $\hat{p}(y \mid \gamma)$ for MOM, iMOM and eMOM priors
and Normal errors.
Briefly, given that secondary modes vanish asymptotically for truly active covariates but not for spurious covariates,
$\hat{p}(y \mid \gamma)$ imposes a stronger penalty on spurious variables than $p(y \mid \gamma)$,
however for such models $p(y \mid \gamma)$ decreases fast enough that both approximations typically lead to very similar inference.

\subsection{Bayes factor rates}
\label{ssec:bf_rates}

Let $\gamma^*=(\mbox{I}(\theta_1^* \neq 0),\ldots,\mbox{I}(\theta_p^* \neq 0),\mbox{I}(\alpha^*\neq 0),\mbox{I}(k^*=2))$
be the optimal model,
that is $(\theta^*,\vartheta^*,\alpha^*,k^*)= \arg\max_{\Gamma,k} M_k(\theta,\vartheta,\alpha)$ maximize
the expected log-likelihood across $k=1,2$, and the expectation is with respect to the data-generating density in Condition A1.
We indicate by $\gamma^* \subset \gamma$ that $\gamma^*$ is a submodel of $\gamma$, {\it i.e.} $\gamma_j^* \leq \gamma_j$ for $j=1,\ldots,p+1$,
and by $\gamma^* \not\subset \gamma$ that $\gamma_j^*>\gamma_j$ for some $j$.
If the data were truly generated from the assumed error distribution, it is well-known that the Bayes factor in favour of $\gamma$
decreases exponentially with $n$ when $\gamma^* \not\subset \gamma$ ($\gamma$ is missing important variables).
Conversely when $\gamma$ adds spurious variables to $\gamma^*$
the Bayes factor is only $O_p(n^{-(p_\gamma-p_{\gamma^*})/2})$ under local priors,
an imbalance that is ameliorated under NLPs, which achieve faster polynomial or quasi-exponential rates
depending on their chosen parametric form \cite{johnson:2010,johnson:2012}.
Proposition \ref{prop:bfrates} gives an extension under model misspecification, the first result of this kind for NLPs.
{\textcolor{blue}{We remark that the rates apply directly to the Laplace approximations \eqref{eq:laplace_approx}.
As studied by \cite{rossell:2017} (Supplementary Section 5, Supplementary Figure 8), when $\gamma$ contains spurious parameters
the non-local posterior $p(\theta_\gamma,\alpha_\gamma \mid \gamma,y)$
can have non-vanishing multimodality, in which case Laplace approximations $\hat{p}(y \mid \gamma)$ underestimate $p(y \mid \gamma)$ even as $n \rightarrow \infty$.
In our experience this is not a major concern (e.g. Table 3S compares Laplace with importance sampling estimates),
but we find it preferrable to characterize inference under our recommended computational framework, i.e. for $\hat{p}(y \mid \gamma)$.
}}
A critical condition for Proposition \ref{prop:bfrates} is that the prior density be strictly positive at the optimum,
$p(\theta_{\gamma^*}^*,\alpha_{\gamma^*}^* \mid \gamma^*)>0$,
which is trivially satisfied by pMOM and peMOM priors.
It also holds for local priors, which for simplicity we define as $p(\theta_\gamma,\vartheta,\alpha\mid \gamma)>0$
for all $(\theta_\gamma,\vartheta,\alpha) \in \Gamma_\gamma$ and we assume to be continuous.

\begin{prop}\label{prop:bfrates}
Suppose that Conditions A1-A3 hold, fixed $p_\gamma,p_{\gamma^*}$ and $n \rightarrow \infty$.
If $\gamma^* \not\subset \gamma$ then
{\textcolor{blue}{
$\frac{1}{n}\log (\widehat{p}(y\mid \gamma)/\widehat{p}(y\mid \gamma^*))
\stackrel{P}{\longrightarrow} -a_1$
for local, pMOM and peMOM priors and some constant $a_1>0$.
}}
Conversely, if $\gamma^* \subset \gamma$ then
$\widehat{p}(y\mid \gamma)/\widehat{p}(y\mid \gamma^*)=O_p(b_n)$
where $b_n= n^{-(p_\gamma-p_{\gamma^*})/2}$ for local priors,
$b_n= n^{-3(p_\gamma-p_{\gamma^*})/2}$ for the pMOM prior, and
$b_n=e^{-c\sqrt{n}}$ for the peMOM prior where $c>0$.

\end{prop}

\begin{cor}
Let $E(\theta_i \mid y)= \sum_{\gamma}^{} E(\theta_i \mid y, \gamma) p(\gamma \mid y)$
be Bayesian model averaging estimates,
$r^+= \mbox{max}_{p_\gamma=p_{\gamma^*}+1} p(\gamma)/p(\gamma^*)$,
$r^{-}= \mbox{max}_{p_\gamma \leq p_{\gamma^*}} p(\gamma)/p(\gamma^*)$,
where $p(\gamma)$ is non-increasing in $p_\gamma$ and $\log r^-= O(n)$.
Under the conditions in Proposition \ref{prop:bfrates},
if $\theta_i^* = 0$ then $E(\theta_i \mid y)= r^+ O_p(n^{-2})$ under the pMOM prior
and $r^+ O_p(e^{-c\sqrt{n}})$ under the peMOM prior.
If $\theta_i^* \neq 0$ then
$E(\theta_i \mid y)= \theta_i^* + O_p(n^{-1/2})$ under the pMOM and peMOM priors.
\label{cor:bma}
\end{cor}

Proposition \ref{prop:bfrates} implies model selection consistency
with Bayes factor rates that have the same functional form as when the correct model is assumed.
We emphasize that this does not imply that there is no cost due to assuming an incorrect model:
the coefficient $a_1$ in the exponential or those in the polynomial rates are affected.
The constant $a_1$ determines how quickly one can detect truly active variables (asymptotically)
and is given by the KL divergence between the assumed model class and the data-generating truth.
That is, under the true model $a_1$ takes a different value than under a misspecified model
and hence the ratio of the correct versus misspecified Bayes factors to detect signals is essentially exponential in $n$.
In contrast, when $\gamma^* \subset \gamma$ this ratio converges to a constant,
hence the effects of model misspecification on false positives vanishes asymptotically.
We remark that, for finite $n$, misspecification can have a marked effect on false positives,
see Section \ref{ssec:nonid_sim} for examples.
Corollary \ref{cor:bma} is the trivial implication that Bayes factors also drive parameter estimation shrinkage
in a Bayesian model averaging setting \cite{rossell:2017}.
When $\theta_i^*=0$, the shrinkage is $1/n^2$ or $e^{-\sqrt{n}}$ for pMOM and peMOM respectively,
in contrast to $1/n$ for local priors and $1/\sqrt{n}$ for the unregularized MLE,
times a term given by model prior probabilities.

We remark that Conditions A1-A3 for Proposition \ref{prop:bfrates} assume independent and identically distributed (id) errors.
It is possible to relax these conditions, particularly that of id errors.
Loosely speaking, the three main ingredients in the proof are that
$(\hat{\theta}_\gamma,\hat{\alpha}_\gamma,\hat{\vartheta}_\gamma) \stackrel{P}{\longrightarrow} (\theta_\gamma^*,\alpha_\gamma^*,\vartheta_\gamma^*)$
(MLE consistency),
{\textcolor{blue}{
that asymptotically
$P\left(n^{-p_\gamma} |H_k(\tilde{\eta}_\gamma)| \in [c_1,c_2] \right) \longrightarrow 1$ for some constants $c_1>0,c_2>0$,
}}
and that the likelihood ratio statistic between $\gamma^*$ and a supra-model $\gamma$ is bounded in probability.
The MLE and likelihood ratio conditions hold quite generally for non-id errors, in particular the latter is satisfied
whenever its limiting distribution is say a chi-square or mixture of chi-squares.
Regarding $H_k$, under independent but non-id errors the ALaplace model has
$H_2^{-1}= s (X_\gamma^T F_\gamma X_\gamma)^{-1} (X_\gamma^T X_\gamma)(X_\gamma^T F_\gamma X_\gamma)^{-1},$
where $s>0$ is a constant depending on $\alpha$ and $F_\gamma$ an $n \times n$ diagonal matrix accounting for each observation's variance \cite{kocherginsky:2005}. The Laplace model is a particular case of this result.
Under Normal errors the MLE has the non-asymptotic covariance $H_1^{-1}= (X_\gamma^T X_\gamma)^{-1} X_\gamma^T \mbox{Cov}(\epsilon) X_\gamma (X_\gamma^T X_\gamma)^{-1}$, and similarly for the asymmetric least squares criterion implied by the two-piece Normal.
Provided that $\max_{i=1,\ldots,n} \mbox{Var}(\epsilon_i)$ is bounded or grows at a slower-than-polynomial rate with $n$
{\textcolor{blue}{
and the eigenvalues of $n (X_\gamma^T X_\gamma)^{-1}$ lie between two positive constants,
then $P\left(n^{-p_\gamma} |H_k(\tilde{\eta}_\gamma)| \in [c_1,c_2] \right) \longrightarrow 1$ for some $c_1>0,c_2>0$.
Relaxing the independence assumption requires more care,
\textit{e.g.} under very strong dependence $|H_k|$ could grow at a slower rate than $n^{p_\gamma}$.}}
We remark that these observations are simply meant to provide intuition,
obtaining precise conditions for Proposition \ref{prop:bfrates}
under non-iid settings is an interesting question for future research.

From the discussion above model misspecification affects sensitivity via the constant $a_1$.
In our experience, typically there is a loss of power.
Fully characterizing this issue theoretically is complicated as $a_1$ depends
on the unknown data-generating truth, but it is possible to provide some intuition.
Consider an arbitrary variable configuration
$(\gamma_1,\ldots,\gamma_p) \not\subset (\gamma_1^*,\ldots,\gamma_p^*)$ that is missing some truly active variables.
Suppose that, as in Condition A1, truly $\epsilon_i \sim s_0(\epsilon_i)=s(\epsilon_i \mid \xi_0)$ for some
error density family $s(\epsilon_i \mid \xi)$, $\xi \in \Xi$, and fixed $\xi_0 \in \Xi$.
Denote by $L_0(\theta_\gamma,\xi)$ the likelihood under the correct $\epsilon_i \sim s(\epsilon_i \mid \xi)$
and $p_0(y \mid \gamma)= \int L_0(y \mid \theta_\gamma,\xi) p(\theta_\gamma,\xi) d\theta_\gamma d\xi$ the associated
integrated likelihood under some prior $p(\theta_\gamma,\xi)>0$.
The interest is in comparing the correct Bayes factor $p_0(y \mid \gamma^*)/p_0(y \mid \gamma)$
to the misspecified $\hat{p}(y \mid \gamma^*)/\hat{p}(y \mid \gamma)$.
Under fairly general conditions
$$\log (p_0(y \mid \gamma^*)/p_0(y \mid \gamma)) \approx n \mbox{D}_0(p_0(y \mid \theta_\gamma^*,\xi_\gamma^*,\gamma)),$$
plus lower order terms analogous to those in Proposition \ref{prop:bfrates} when $\gamma \not\subset \gamma^*$,
where $\mbox{D}_0(p_0(y \mid \theta_\gamma^*,\xi_\gamma^*,\gamma))$ is the Kullback-Leibler divergence between the data-generating $p_0(y \mid \theta^*_{\gamma^*},\xi_0,\gamma^*)$
and the KL-optimal $p_0(y \mid \theta_\gamma^*,\xi_\gamma^*,\gamma)$ under $\gamma$.
Trivial algebra gives
\begin{align}
\log \left( \frac{p_0(y \mid \gamma^*)/p_0(y \mid \gamma)}{\hat{p}(y \mid \gamma^*)/\hat{p}(y \mid \gamma)} \right) \approx
n \left( \mbox{D}_0(p_0(y \mid \theta_\gamma^*,\xi_\gamma^*,\gamma))
+ \mbox{D}_0(p(y \mid \eta^*_{\gamma^*},\gamma^*))
- \mbox{D}_0(p(y \mid \eta^*_{\gamma},\gamma)) \right).
\label{eq:bf_losspower}
\end{align}
The sign of the right hand side in \eqref{eq:bf_losspower} determines whether the misspecified Bayes factor
has lower or greater asymptotic power than the correct Bayes factor.
A precise study of \eqref{eq:bf_losspower} deserves separate treatment,
but the expression can be loosely interpreted as a type of triangle inequality.
If the divergence due to simultaneously using the wrong error distribution and $\gamma$ instead of $\gamma^*$,
$\mbox{D}_0(p(y \mid \eta^*_{\gamma},\gamma))$,
is smaller than the sum of the divergences due to only using the wrong error distribution plus that of only using $\gamma$ instead of $\gamma^*$.
Then, misspecifiying the error distribution results in slower (but still exponential) Bayes factor rates to detect truly active variables.
To our knowledge there is no guarantee that \eqref{eq:bf_losspower} is positive in general for any
assumed model and data-generating truth, however in all our examples misspecified Bayes factors
exhibited such a loss of power, suggesting that this is often the case.

\subsection{Model exploration}
\label{ssec:modelsearch}

Algorithm \ref{alg:gibbs_modelsearch} describes a novel
Gibbs sampling that can be used when $p_\gamma$ is too large for exhaustive enumeration of all $2^{p_\gamma+2}$ models.
Although conceptually simple, Algorithm \ref{alg:gibbs_modelsearch} extends a method that delivered good results for high-dimensional variable selection
under Normal errors \cite{johnson:2012},
and is designed to spend most iterations in the Normal model whenever it is a good enough approximation.
That is, as illustrated in our examples the computational effort adapts automatically to the nature of the data,
so that the cost associated to abandoning the Normal model is only incurred when this is required to improve inference.
Our implementation also allows the user to fix $(\gamma_{p+1},\gamma_{p+2})$, so that one can condition on Normal,
asymmetric Normal, Laplace or asymmetric Laplace errors whenever this is desired.

The number of iterations $T$ should ideally be large enough for the chain to converge,
see for instance \cite{johnson:2013} for a discussion of formal convergence diagnostics based on coupling methods.
In practice, it usually suffices to monitor some posterior quantities of interest.
For instance, in the setting of variable selection with NLPs \cite{rossell:2017}
found useful to set $T$ large enough so that sampling-based estimates of $p(\gamma_j=1 \mid y)$
are close enough to estimates based on renormalizing posterior probabilities across the models visited so far.

\begin{algorithm}
{\bf Gibbs model space search.}\label{alg:gibbs_modelsearch}
\begin{enumerate}
\item Let $\gamma_{p+1}^{(0)}=\gamma_{p+2}^{(0)}=0$ and set $\gamma_1^{(0)},\ldots,\gamma_p^{(0)}$
using the greedy forward-backward initialization algorithm in \cite{johnson:2012}.
Set $t=1$.

\item For $j=1,\ldots,p$, update $\gamma_j^{(t)}=1$ with probability
$$
\frac{p(\gamma_1^{(t)},\ldots,\gamma_{j-1}^{(t)},1,\gamma_{j+1}^{(t-1)},\ldots,\gamma_p^{(t-1)} \mid y)}
{\sum_{\gamma_j=0}^{1} p(\gamma_1^{(t)},\ldots,\gamma_{j-1}^{(t)},\gamma_j,\gamma_{j+1}^{(t-1)},\ldots,\gamma_p^{(t-1)} \mid y)}.
$$

\item Update $(\gamma_{p+1}^{(t)},\gamma_{p+2}^{(t)})= (l,m)$ with probability
$$
\frac{p(\gamma_1^{(t)},\ldots,\gamma_p^{(t)},l,m \mid y)}
{\sum_{\gamma_{p+1}=0}^{1} \sum_{\gamma_{p+2}=0}^{1} p(\gamma_1^{(t)},\ldots,\gamma_p^{(t)},\gamma_{p+1},\gamma_{p+2} \mid y)}.
$$
If $t \leq T$, set $t=t+1$ and go back to Step 2, otherwise stop.
\end{enumerate}
\end{algorithm}

\section{Results}
\label{sec:results}

We studied via simulations the practical implications of model misspecification on variable selection,
both on small and large $p$ (Sections \ref{ssec:lowdim_sim}-\ref{ssec:highdim_sim}),
as well as the ability of our framework to detect asymmetries ($\gamma_{p+1}=1$) and heavier-than-normal tails ($\gamma_{p+2}=1$).
{\textcolor{blue}{The heteroscedastic errors simulation in Section \ref{ssec:nonid_sim} and}}
the DLD example in Section \ref{ssec:dld} also illustrates how to perform quantile regression for multiple fixed quantile levels as a particular case of our framework.

Computations were carried out using function modelSelection in R package mombf 1.9.2 \cite{mombf:2016},
using default prior settings (Section \ref{sec:prior}) and Laplace approximations to $p(y \mid \gamma)$ unless otherwise stated.
Although our goal is to build a Bayesian framework to cope with simple departures from normality,
for comparison we included some penalized likelihood methods with available R implementation:
standard LASSO penalties on least squares regression (LASSO-LS, \cite{tibshirani:1996}),
LASSO penalties on least absolute deviation (LASSO-LAD, \cite{wanglang:2009}), SCAD penalties on least squares \cite{fan:2001},
and LASSO penalties on quantile regression (LASSO-QR, \cite{wu:2009}).
For LASSO-LS, LASSO-LAD, LASSO-QR and SCAD we set the penalization parameter with 10-fold cross-validation
using functions mylars, rq.lasso.fit and ncvreg in R packages parcor 0.2.6, rqPen 1.5.1 and ncvreg 3.4.0 (respectively) with default parameters.
LASSO-LAD corresponds to setting the 0.5 quantile in rq.lasso.fit, whereas for LASSO-QR we set the optimal quantile $(1+\alpha)/2$ where
$\alpha$ is the data-generating truth. That is, we performed a conservative comparison where results for LASSO-QR may be slightly optimistic.
All R code is provided in supplementary files.

\subsection{Low-dimensional simulation}
\label{ssec:lowdim_sim}

\begin{figure}
\begin{center}
\begin{tabular}{cc}
$\epsilon_i \sim $ Normal & $\epsilon_i \sim $ ANormal \\
\includegraphics[width=0.48\textwidth]{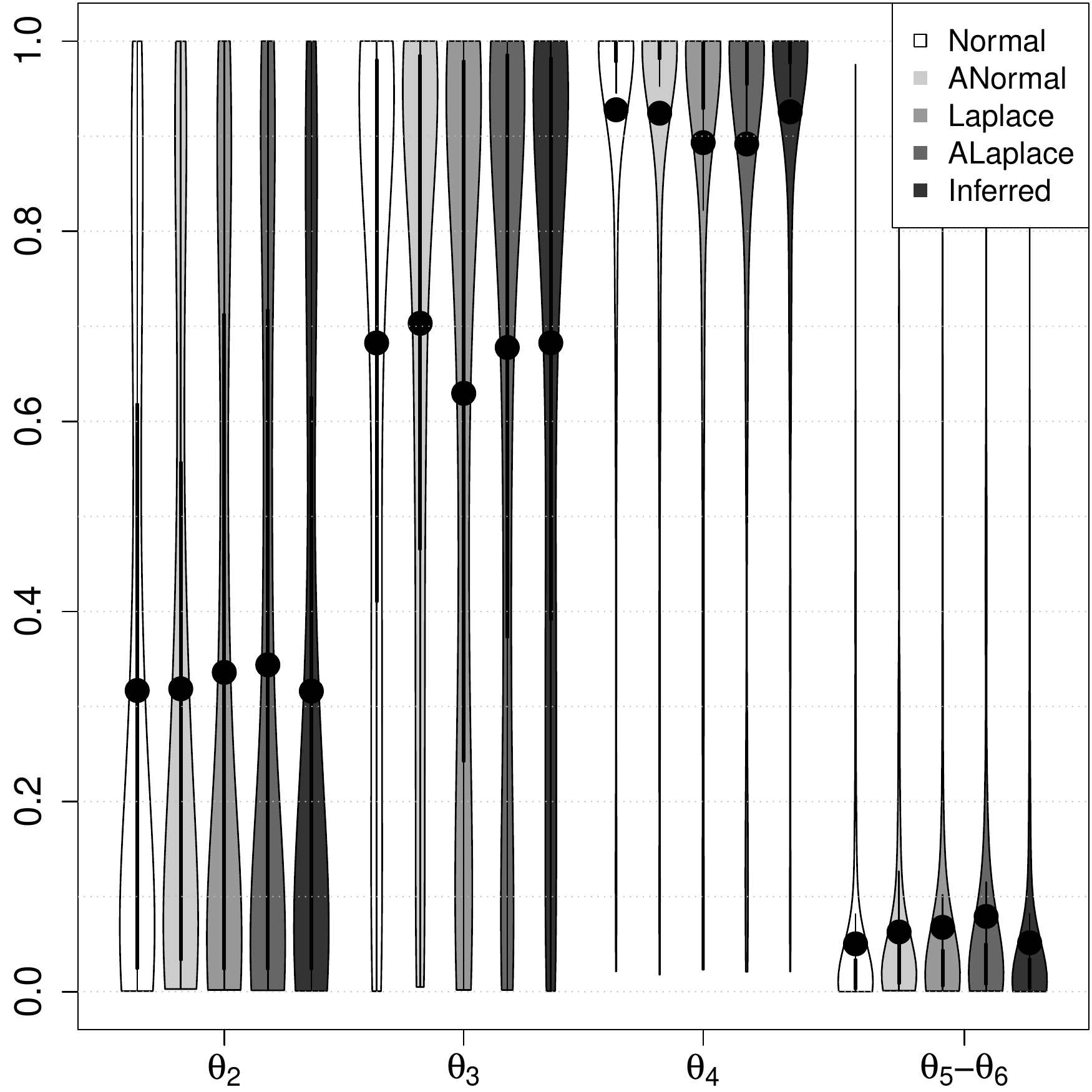} &
\includegraphics[width=0.48\textwidth]{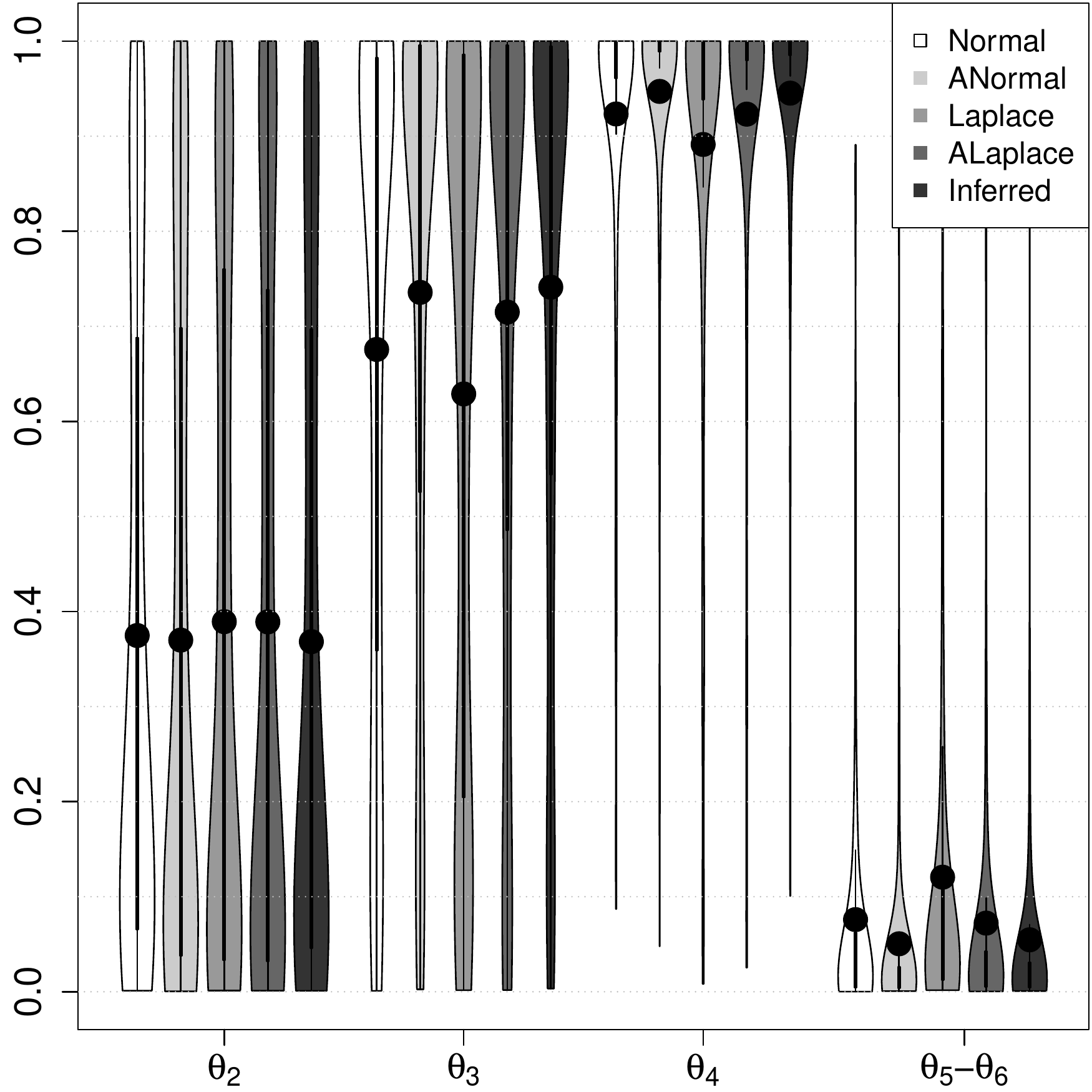} \\
$\epsilon_i \sim $ Laplace & $\epsilon_i \sim $ ALaplace \\
\includegraphics[width=0.48\textwidth]{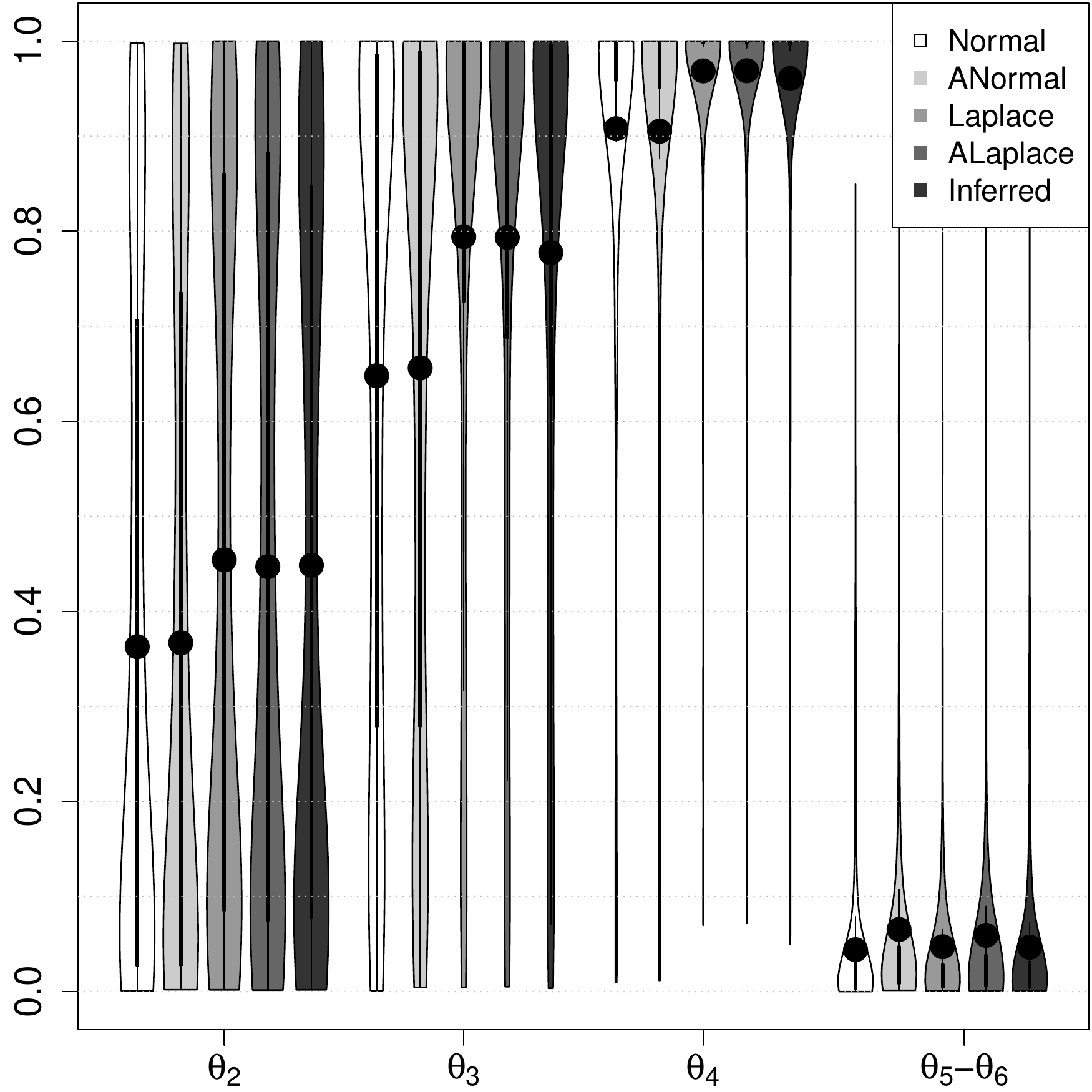} &
\includegraphics[width=0.48\textwidth]{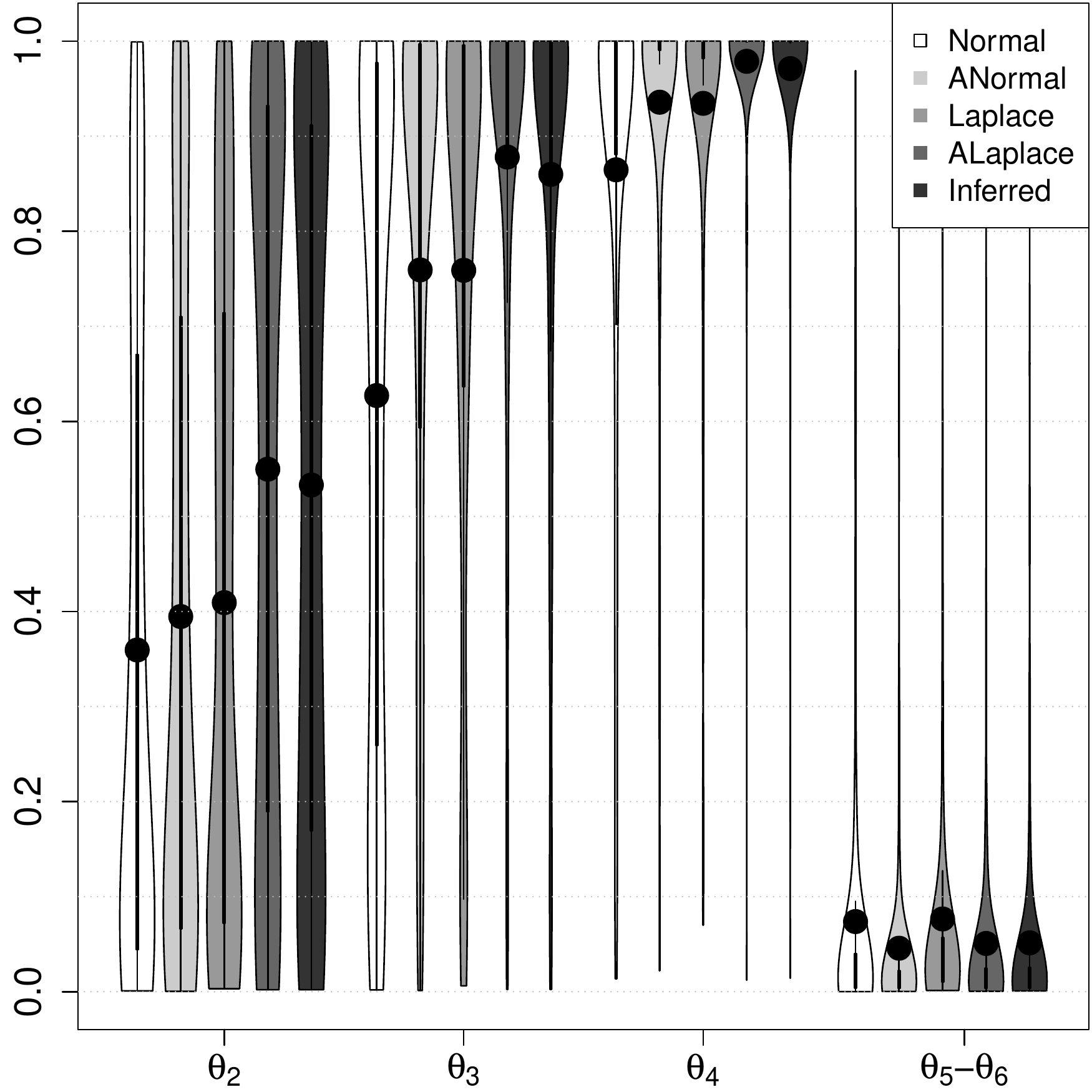} \\
\end{tabular}
\end{center}
\caption{$P(\theta_i \neq 0 \mid y)$ for simulation with constant $\vartheta=2$, $\alpha=0,0.5$.
$P(\theta_i \neq 0 \mid y)$ for $p=6$, $\theta=(0,0.5,1,1.5,0,0)$,
$n=100$, $\rho_{ij}=0.5$. Black circles show the mean.}
\label{fig:simres_margpp}
\end{figure}

We started by simulating 200 data sets from a linear model with Normal residuals, each with
$n=100$, $p=6$, $\theta=(0,0.5,1,1.5,0,0)$ ($\theta_1=0$ corresponds to the intercept), $\vartheta=2$.
Covariate values were generated from a multivariate Normal centered at 0,
with unit variances and all pairwise correlations $\rho_{ij}=0.5$.
We compared the results under assumed Normal, asymmetric Normal, Laplace and asymmetric Laplace errors,
and also when inferring the residual distribution with our framework (Section \ref{sec:modsel}).
Throughout, we used MOM priors with default $g_{\theta}=0.348$, $g_{\alpha}=0.357$ and uniform model probabilities $p(\gamma) \propto 1$.
Given that $p$ is small we enumerated and computed $p(\gamma \mid y)$ for all models.
Figure \ref{fig:simres_margpp} (top left) shows the marginal posterior probabilities $p(\gamma_j=1 \mid y)$.
These were almost identical under assumed Normal and asymmetric Normal errors.
Both models were preferrable to Laplace or asymmetric Laplace errors, mainly in giving higher $p(\gamma_j=1 \mid y)$ for truly active variables.

We repeated the simulation study, this time generating $\epsilon_i \sim \mbox{AN}(0,2,-0.5)$,
$\epsilon_i \sim L(0,2)$ and finally $\epsilon_i \sim\mbox{AL}(0,2,-0.5)$.
Here we observed more marked differences than under $\epsilon_i \sim N(0,2)$,
specifically failing to account for thick tails caused a substantial drop in $p(\gamma_j=1 \mid y)$ for truly active predictors.
As an example, when truly $\epsilon_i \sim \mbox{AL}(0,4,-0.5)$ the mean $p(\theta_3 \neq 0 \mid y)$ increased from 0.63
under assumed Normal errors to 0.89 under asymmetric Laplace errors.
These results suggest that wrongly assuming Normal errors may has more pronounced consequences
on inference than using more robust error distributions.
Interestingly, including asymmetry in the model had no noticeable adverse effects on inference
even when residuals were truly symmetric, and improved power when residuals were truly asymmetric.
Hence the reasoning for adopting symmetric models seems mostly computational.

Our framework based on inferring $(\gamma_{p+1},\gamma_{p+2})$ showed a highly competitive behaviour,
usually fairly close to assuming the true distribution (Figure \ref{fig:simres_margpp}).
The mean posterior probability assigned to the true error distribution was always $>0.8$ (Supplementary Table \ref{tab:perror_p5}),
indicating that the desired departures from normality were effectively detected.

We repeated all the analyses above first using Monte Carlo estimates of $p(y \mid \gamma)$ based on $B=10,000$ importance samples,
and then again using our alternative default $g_\alpha=0.087$.
Supplementary Table \ref{tab:perror_p5} shows that inference on the error distribution remained remarkably stable,
albeit as expected reducing $g_\alpha=0.357$ to $0.087$ increases slightly $p(\alpha \neq 0 \mid y)$ in all settings.
Supplementary Figures \ref{fig:simres_margpp_priorskew1}-\ref{fig:simres_margpp_mc} show $p(\gamma_j=1 \mid y)$.
These are virtually indistinguishable from those in Figure \ref{fig:simres_margpp},
indicating that the results are robust to these implementation details.

Finally, we assessed the behaviour of the least-squares initialization in Algorithms \ref{alg:mle_lma}-\ref{alg:mle_cda} under different data-generating mechanisms, specifically in terms of CPU times.
Table \ref{tab:runtime_sim1} gives mean times across $10,000$ independent simulations with $p=6$ and increasing data-generating truths $\alpha^*=0, -0.25,-0.5,-0.75$, both for two-piece Normal and two-piece Laplace errors. These are for the whole model-fitting process, including exhaustive model enumeration and computation of posterior model probabilities.
The time increases were of roughly 25\% from $\alpha^*=0$ to $\alpha^*=-0.75$. This is as expected, under asymmetry least-squares is a poorer initial $\hat{\btheta}^{(0)}$. The increase is however mild, indicating that a larger fraction of the computation cost arises from other operations (e.g. matrix inversion after the mode has been found). These results support that our $\hat{\btheta}^{(0)}$ is not particularly problematic. One could certainly consider alternative $\hat{\btheta}^{(0)}$, say median regression or trimmed least squares, but these are typically costlier that least-squares hence the overall gains are likely to be moderate at best.

\subsection{Non-identically distributed errors}
\label{ssec:nonid_sim}

\begin{figure}
\begin{center}
\begin{tabular}{cc}
$\epsilon_i \sim $ Normal & $\epsilon_i \sim $ ANormal \\
\includegraphics[width=0.48\textwidth]{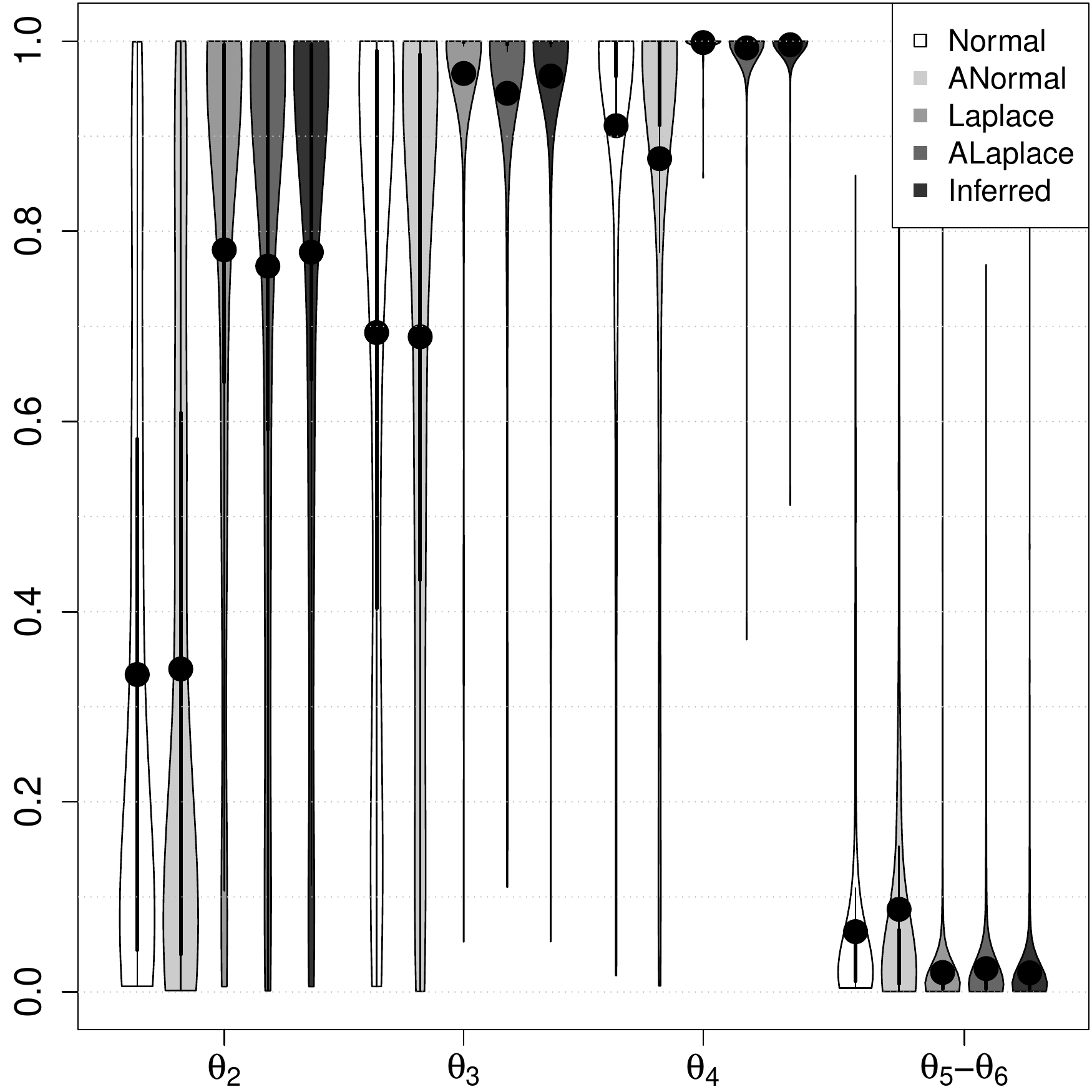} &
\includegraphics[width=0.48\textwidth]{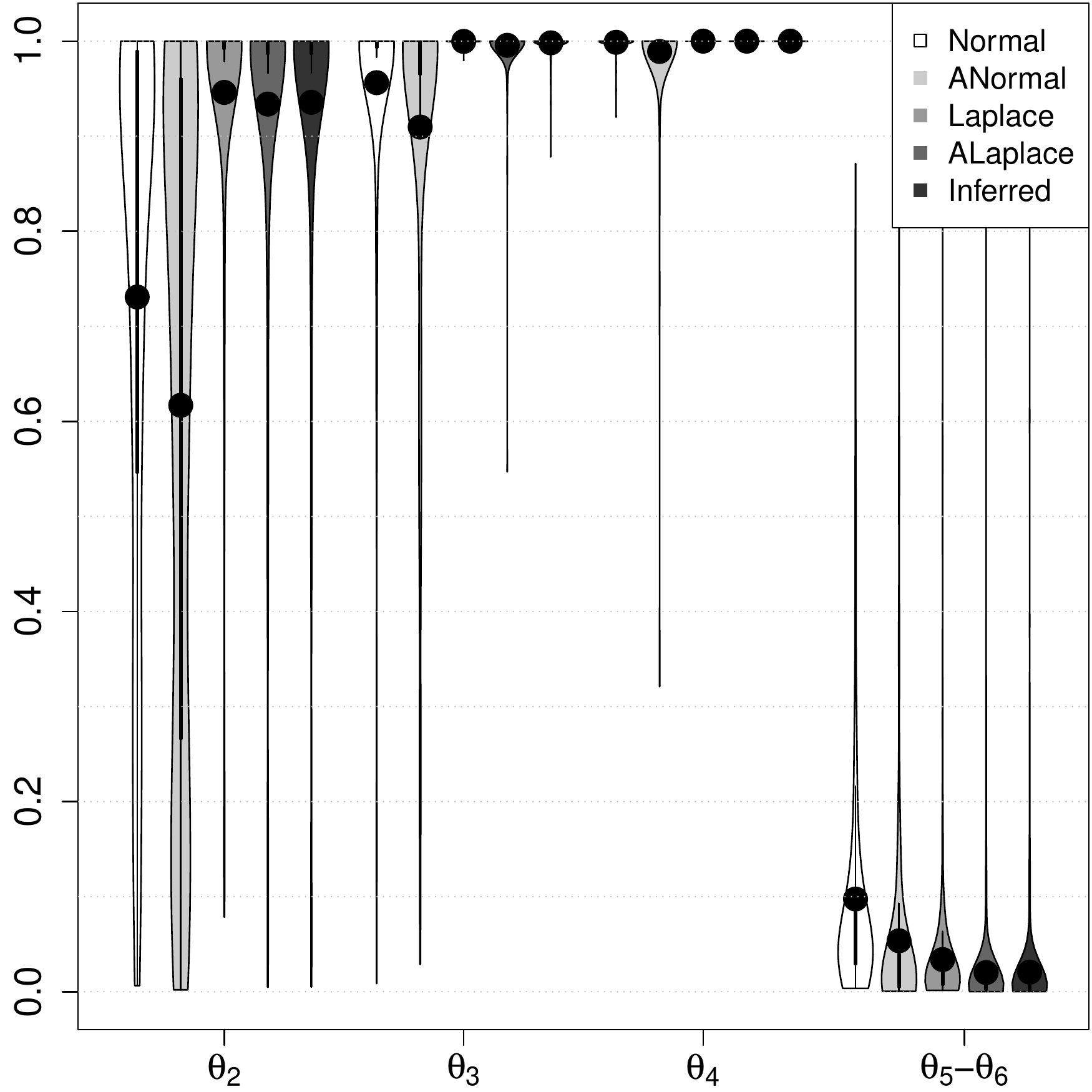} \\
$\epsilon_i \sim $ Laplace & $\epsilon_i \sim $ ALaplace \\
\includegraphics[width=0.48\textwidth]{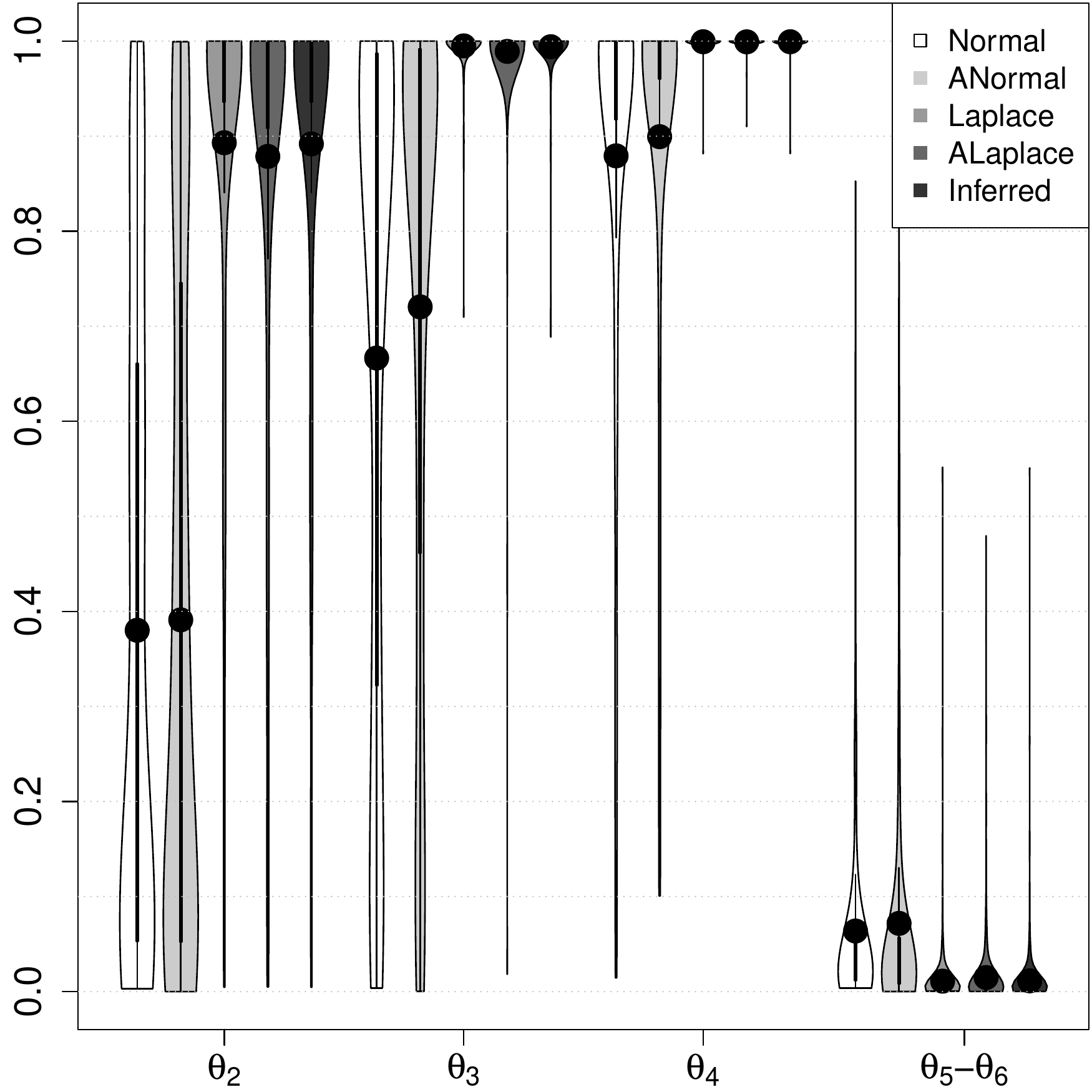} &
\includegraphics[width=0.48\textwidth]{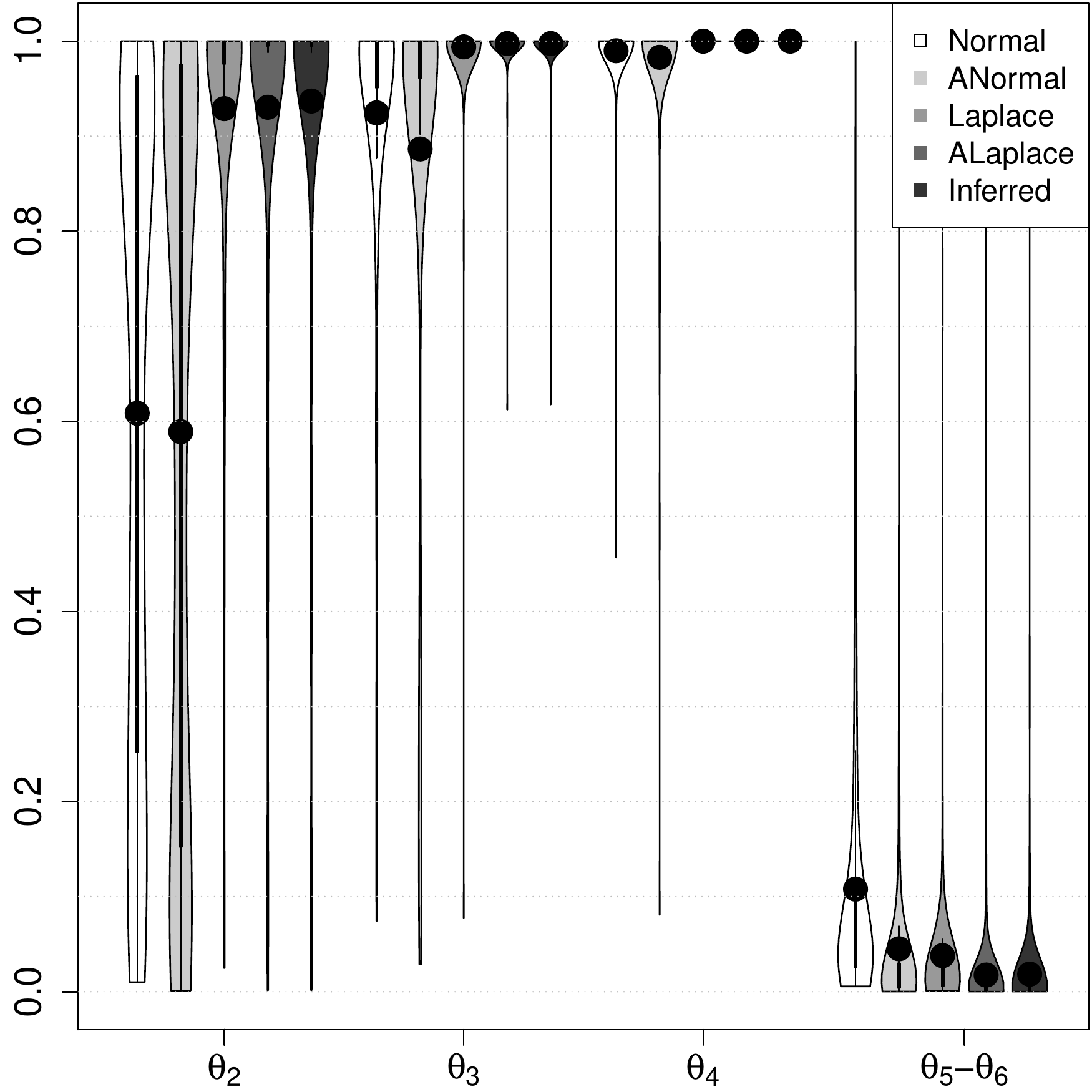} \\
\end{tabular}
\end{center}
\caption{$P(\theta_i \neq 0 \mid y)$ for simulation with $\vartheta_i \propto e^{x_i^T \theta}$,
constant $\alpha=0,-0.5$. $p=6$, $\theta=(0,0.5,1,1.5,0,0)$,
$n=100$, $\rho_{ij}=0.5$. Black circles show the mean.}
\label{fig:simres_hetsk_margpp}
\end{figure}

We investigate the effect of deviations from the identically distributed errors assumption.
We repeated the simulations in Section \ref{ssec:lowdim_sim}
under heteroscedastic and hetero-asymmetric errors,
and reproduced a pathological example reported by \cite{gruenwald:2014}.
Under heteroscedasticity, we set $\tilde{\epsilon}_i= e^{x_i^T \theta} \epsilon_i/c$
where $c$ was set such that $\mbox{Var}(\tilde{\epsilon}_i)=\mbox{Var}(\epsilon_i)$,
so that the signal-to-noise was comparable to our earlier simulations.
{\textcolor{blue}{This example mimics that used by \cite{koenker:2005} (Figure 1.6) to illustrate
the potential interest of conditioning upon multiple quantile levels,}}
except that ours has a stronger (exponential) association between mean and variance.
{\textcolor{blue}{We first apply our framework without conditioning on $\alpha$.}}
Figure \ref{fig:simres_hetsk_margpp} shows $P(\gamma_j=1 \mid y)$ for $p=6$.
The main feature is that the Laplace and asymmetric Laplace models clearly outperform
the Normal model both in sensitivity and specificity.
For instance, when truly $\theta_2^*=0.5$ the mean $P(\gamma_2=1 \mid y)$ increased from 0.33 to 0.78
under assumed Normal and Laplace residuals respectively.
The mean for truly inactive $\theta_5^*=\theta_6^*=0$ decreased from 0.063 to 0.021.
Interestingly, inferring the error model chose Laplace errors even when these were truly Normal and
showed a highly competitive performance (Supplementary Table \ref{tab:infererror_heterosk}).
Intuitively, heteroscedasticity gives an overabundance of residuals at the origin and at the tails relative to a homoscedastic Normal.
Such errors are better captured by a Laplace model.

{\textcolor{blue}{
Next, following \cite{koenker:2005} we assessed the performance of quantile regression at fixed quantile levels $q=0.05,0.25,0.75,0.95$.
The usual motivation for conditioning upon multiple quantiles is to consider that each quantile
could potentially depend on a different subset of predictors.
This corresponds to conditioning upon asymmetric Laplace errors and fixed $\alpha=2q-1$ (Section \ref{sec:loglikelihood}).
The marginal posterior inclusion probabilities in Table \ref{tab:multiplequantile_heterosk} show that $q=0.5$
(the KL-optimal value) led to substantially higher sensitivity than say $q=0.05$ or $q=0.95$.
We remark that
under our heteroscedastic data-generating truth the $q^{th}$ conditional quantile is $\bx_i^T\btheta + z_q \sqrt{e^{\bx_i^T\btheta} /c}$
where $z_q$ is the $q^{th}$ standard Normal quantile.
The results illustrate that, in this and similar situations where all quantiles depend on the same subset of variables,
inferring $\alpha$ can lead to better variable selection than conditioning upon poor choices of $\alpha$.
Naturally, under more complex scenarios where quantiles do depend on different variable subsets,
conditioning upon multiple $\alpha$ can provide a richer description of the dependence of $y_i$ on $\bx_i$.
}}

Our second simulation scenario considered the presence of non-constant asymmetry.
Specifically, we generated $\mbox{tanh}(\alpha_i) \sim N(\mbox{atanh}(\bar{\alpha}),1/4^2)$
where the median asymmetry is $\bar{\alpha}=0,-0.5$ as before.
Under this setting when $\bar{\alpha}=0$ then $\alpha_i \in (-0.45,0.45)$ with 0.95 probability
and when $\bar{\alpha}=-0.5$ it is $(-0.78,-0.06)$, i.e. there is substantial variation in asymmetry.
Supplementary Figure \ref{fig:simres_hetsym_margpp} displays $P(\gamma_j=1 \mid y)$ for $p=6$.
These results are qualitatively similar to those in Figure \ref{fig:simres_margpp}
where $\alpha_i$ was held fixed.
We remark that although in these examples non-constant asymmetry was not a concern,
its impact could be more serious in other settings, \textit{e.g.}~under strong dependencies between the asymmetry and the mean.
See Section \ref{sec:conclusion} for some further discussion.

Finally, we mimic the example in \cite{gruenwald:2014}, Section 5.1.2.
The authors set $(y_i,x_{i1},\ldots,x_{ip})=(0,0,\ldots,0)$ with probability 0.5
and $y_i= x_i^T \theta^* + \epsilon_i$ with probability 0.5,
where $x_{ij} \sim N(0,1)$, $\theta^*=(0.1,0.1,0.1,0.1,0.1,0,\ldots,0)$ and $\epsilon_i \sim N(0,\vartheta)$.
This extreme case of non-id errors is interesting in that the degeneracy at the origin
results in inliers, rather than the more commonly considered outliers in $y_i$ or leverage points in $x_i$.
We selected variables under assumed Normal errors for $p=n=50$,
for this $(n,p)$ the authors reported a particularly large inflation of false positives (as $n \rightarrow \infty$ these disappeared).
Specifically we set $\vartheta^*=2$, Zellner's $p(\theta_\gamma \mid \gamma)= N(\theta_\gamma; 0, n (X_\gamma^T X_\gamma)^{-1})$
and the Beta-Binomial(1,1) prior for $p(\gamma)$.
The posterior mode selected a striking 21.3 out of the 45 spurious variables (mean across 100 independent simulations),
confirming their findings (Supplementary Table \ref{tab:gruenwald}).
Under a pMOM prior the mean false positives decreased to 12.1  when
conditioning on Normal errors and further to 10.5 when inferring the error model.
Interestingly under the peMOM prior and Normal errors the mean false positives were only 2.9.
All methods showed similar sensitivity, selecting roughly 3 out of the 5 active variables.
This example illustrates that, while serious model misspecification can have marked effects for finite $n$,
these can be partially mitigated by adopting priors that penalize small coefficients
and flexible error models. In this particular example the exponential peMOM penalties
were more effective than the pMOM penalties in lowering false positives.

\subsection{High dimensional simulation}
\label{ssec:highdim_sim}

\begin{figure}
\begin{center}
\begin{tabular}{cc}
$\epsilon_i \sim N(0,4)$ & $\epsilon_i \sim \mbox{AN}(0,4,-0.5)$ \\
\includegraphics[width=0.48\textwidth]{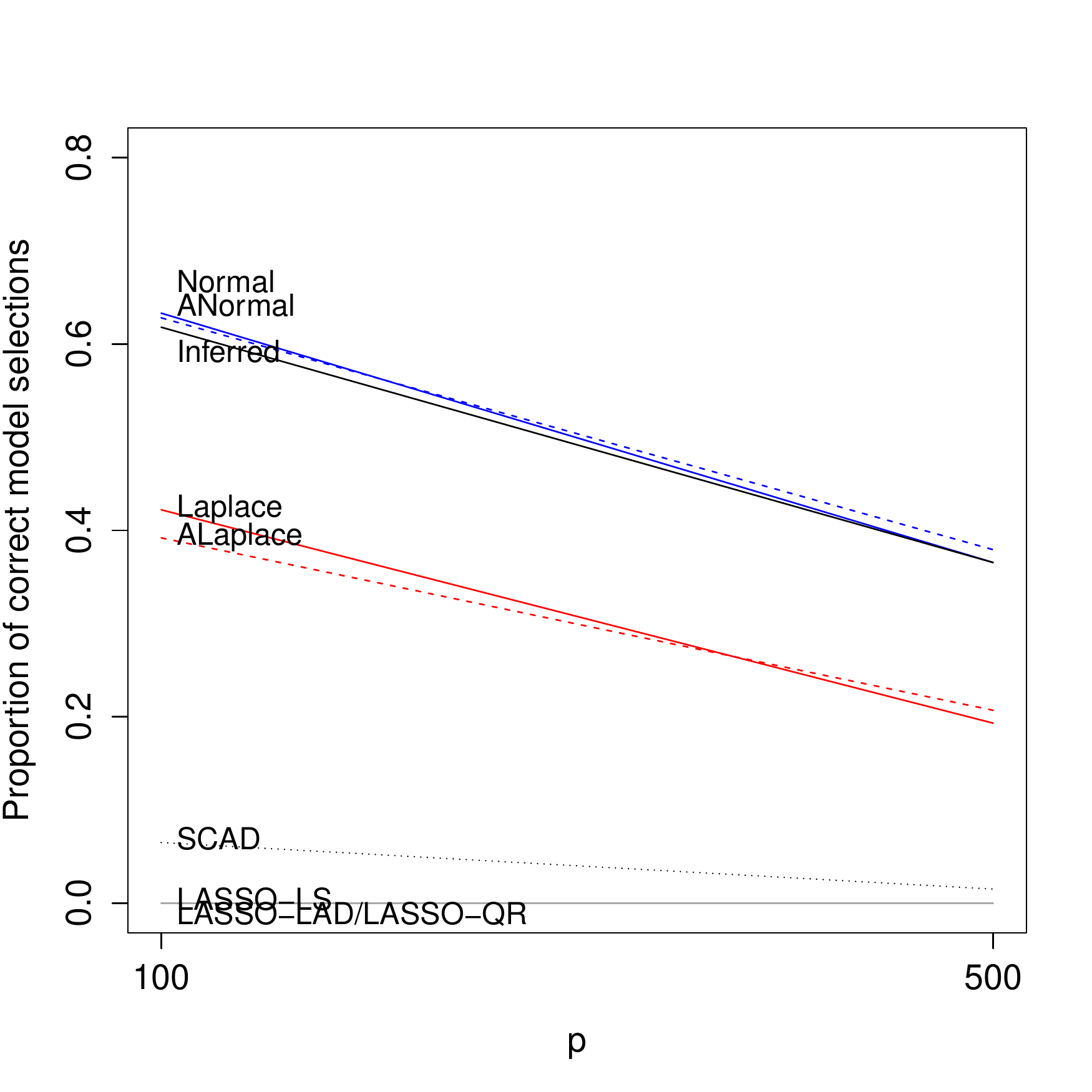} &
\includegraphics[width=0.48\textwidth]{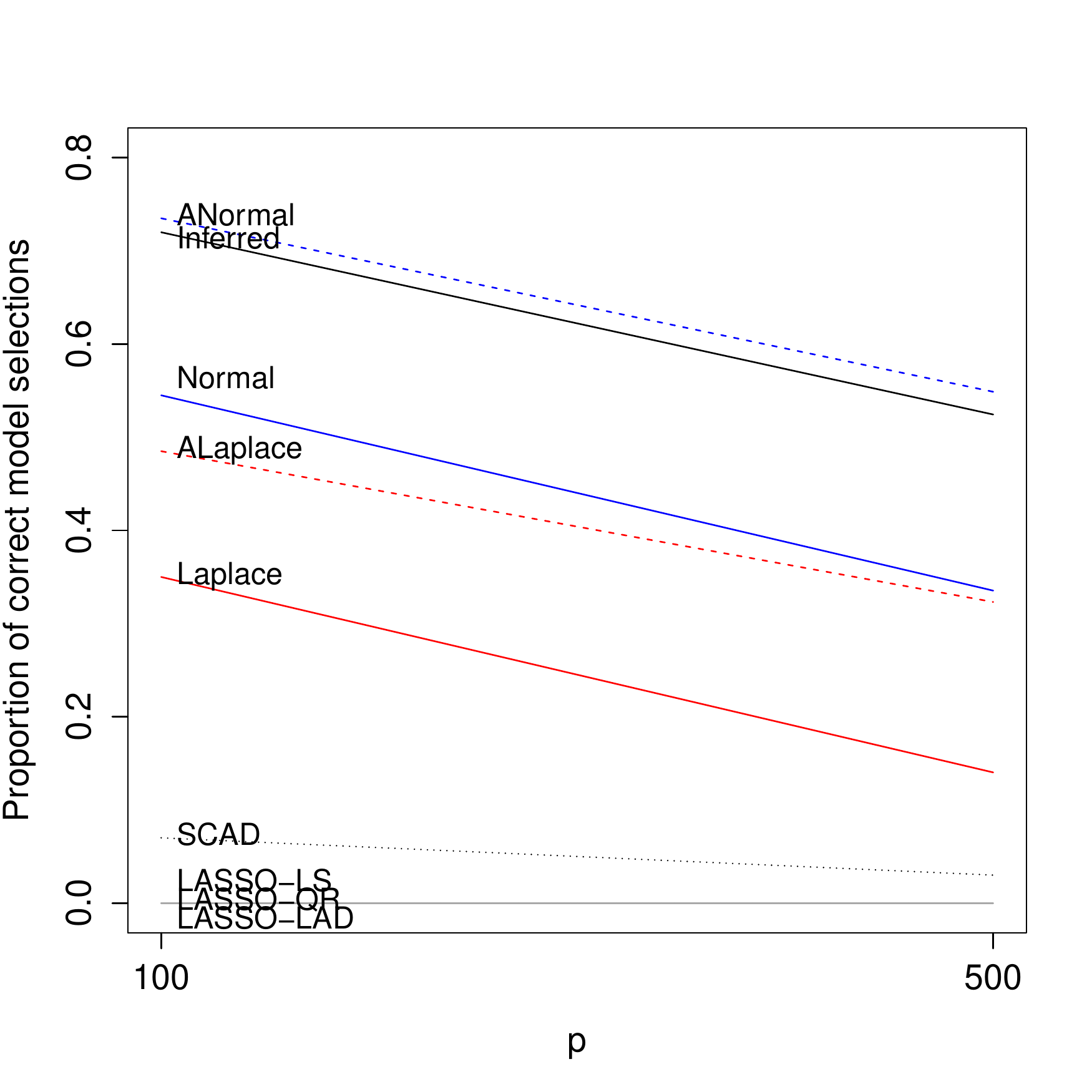} \\
$\epsilon_i \sim L(0,4)$ & $\epsilon_i \sim \mbox{AL}(0,4,-0.5)$ \\
\includegraphics[width=0.48\textwidth]{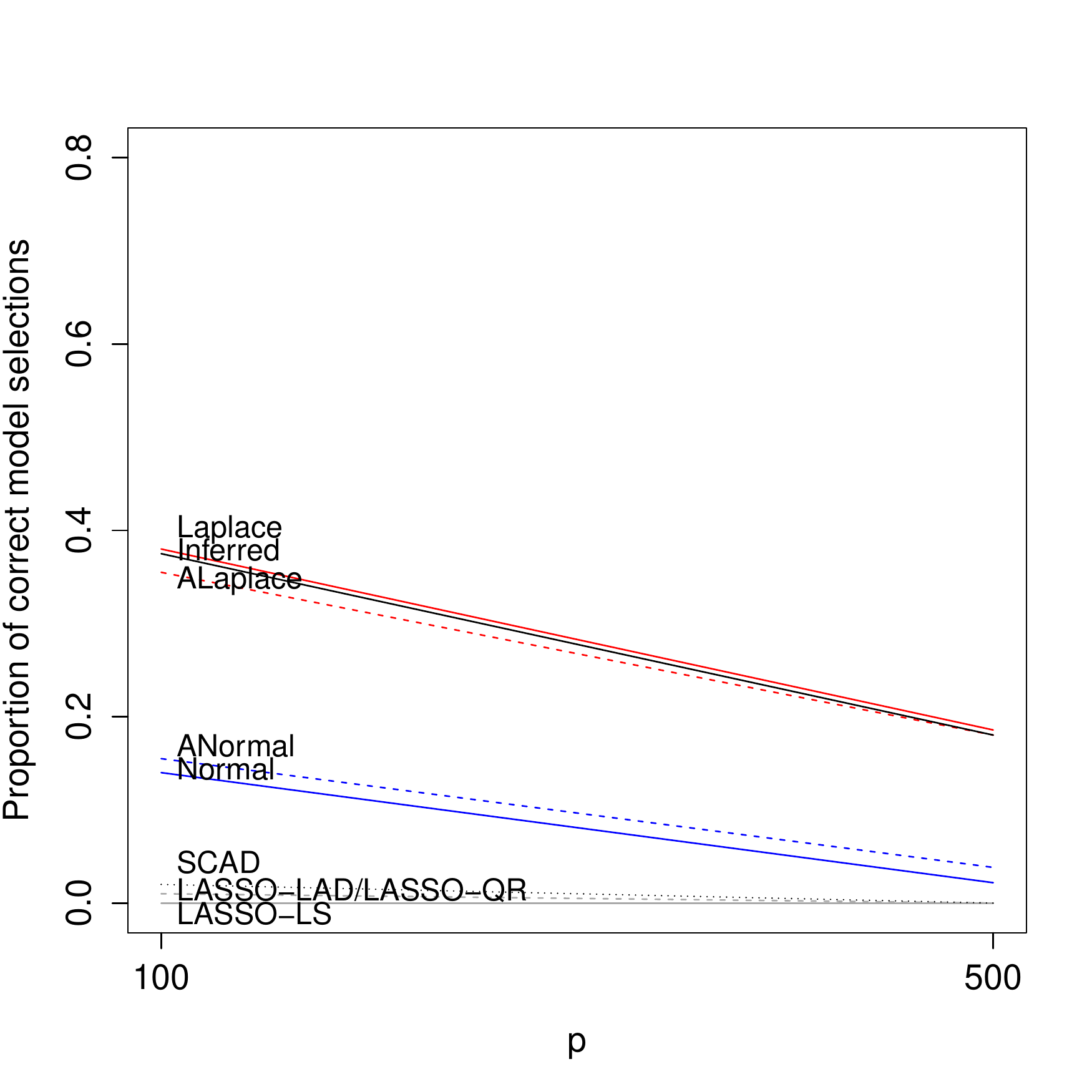} &
\includegraphics[width=0.48\textwidth]{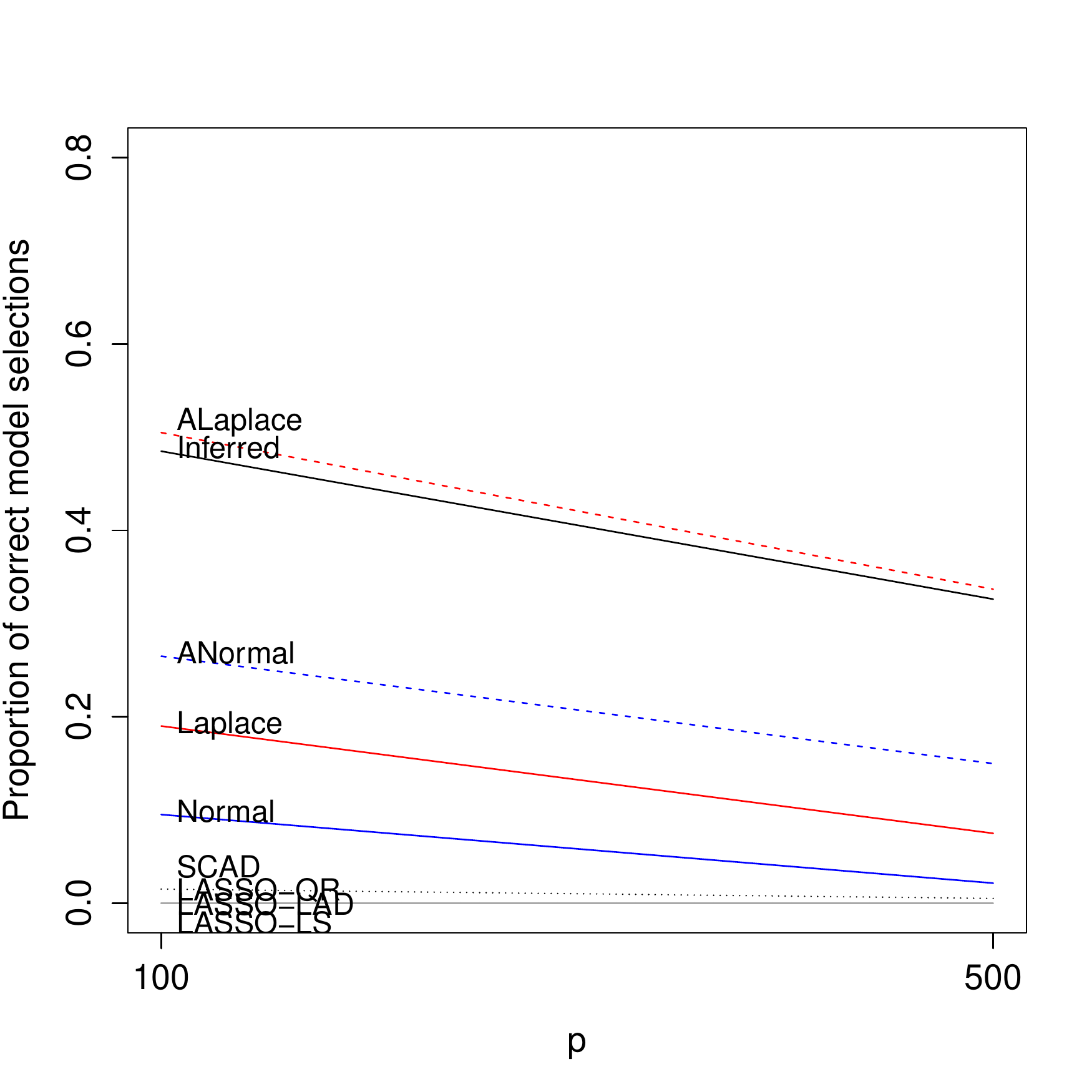} \\
\end{tabular}
\end{center}
\caption{Proportion of correct model selections $p(\widehat{\gamma}= \gamma_0 \mid y)$.
$\vartheta=1$, $\theta=(0,0.5,1,1.5,0,\ldots,0)$, $n=100$, $\rho_{ij}=0.5$.}
\label{fig:pcorrect_vtheta1}
\end{figure}

We repeated the simulation study in Section \ref{ssec:lowdim_sim}
with $\theta=(0,0.5,1,1.5,0,\ldots,0)$ by adding 95 spurious predictors for a total of $p=100$ covariates,
and subsequently 400 more spurious predictors for a total $p=500$.
Given that the model space is too large for a full enumeration,
we run the Gibbs algorithm in Section \ref{ssec:modelsearch} with $T=10,000$ iterations.
To initialize the chain we used the greedy Gibbs algorithm from \cite{johnson:2012},
which starts at $\gamma=(0,\ldots,0)$ and keeps adding or removing individual covariates until a local mode is found.
We set $p(\gamma)$ to the default Beta-Binomial(1,1) and left all other settings as in Section \ref{ssec:lowdim_sim}.

We conducted one first set of simulations under $\vartheta=1$.
Figure \ref{fig:pcorrect_vtheta1} shows the proportion of simulations in which the posterior mode
$\widehat{\gamma}= \arg\max_\gamma p(\gamma \mid y)$ was equal to the simulation truth $\gamma_0=(0,1,1,1,0,\ldots,0)$.
The main finding was that assuming the wrong error distribution had a marked detrimental impact on Bayesian variable selection,
particularly in the presence of asymmetries or thicker-than-normal tails.
Supplementary Table \ref{tab:simres_vtheta1} gives the exact figures, as well as the number of false and true positives.
All Bayesian formulations compared favourably to LASSO-LS, LASSO-LAD, LASSO-QR and LASSO-SCAD,
mainly due to the latter incurring an inflated number of false positives.
This is in agreement with earlier findings \cite{johnson:2012,rossell:2017} when comparing NLPs to penalized likelihoods,
and likely partially related to the fact that cross-validation focuses on predictive ability and thus tends to favour the
inclusion of a few spurious covariates.
Interestingly, in our study LASSO-LAD showed little advantages over LASSO-LS, even under truly Laplace errors.
LASSO-QR did improve slightly upon LASSO-SCAD when truly $\alpha^* \neq 0$ both in sensitivity and specificity.
Analogously to the $p=6$ case in Figure \ref{fig:simres_margpp}, when $p=101,501$ the marginal inclusion probabilities
for truly active variables suffered a drop when ignoring the presence of asymmetry or heavy tails
(Supplementary Figures \ref{fig:simres_margpp_p100_vartheta1}-\ref{fig:simres_margpp_p500_vartheta1}).
Our framework to infer the error distribution delivered highly competitive inference.

Supplementary Table \ref{tab:runtime_sim} indicates CPU times for $p=100$.
The Normal model exhibited lower times under truly Normal or Laplace errors,
likely due to the availability of closed-form expressions for $p(\gamma \mid y)$.
The presence of asymmetry encouraged the inclusion of an intercept term under the Normal model,
the associated increase in model dimension cancelled the computational savings.
Times for our inferred residuals framework were highly competitive under all scenarios.

To emulate a situation with lower signal-to-noise ratio we repeated the simulation study under $\vartheta=2$.
The results are shown in Supplementary Table \ref{tab:simres_vtheta2}
and Supplementary Figures \ref{fig:simres_margpp_p100_vartheta2}-\ref{fig:simres_margpp_p500_vartheta2}.
Briefly, the performance of all methods suffered in this more challenging setting
due to a drop in the power to detect truly active predictors,
however their relative performances were largely analogous to those for $\vartheta=1$.

\subsection{TGFB data}
\label{ssec:tgfb}

\begin{figure}
\begin{tabular}{cc}
\includegraphics[width=0.5\textwidth]{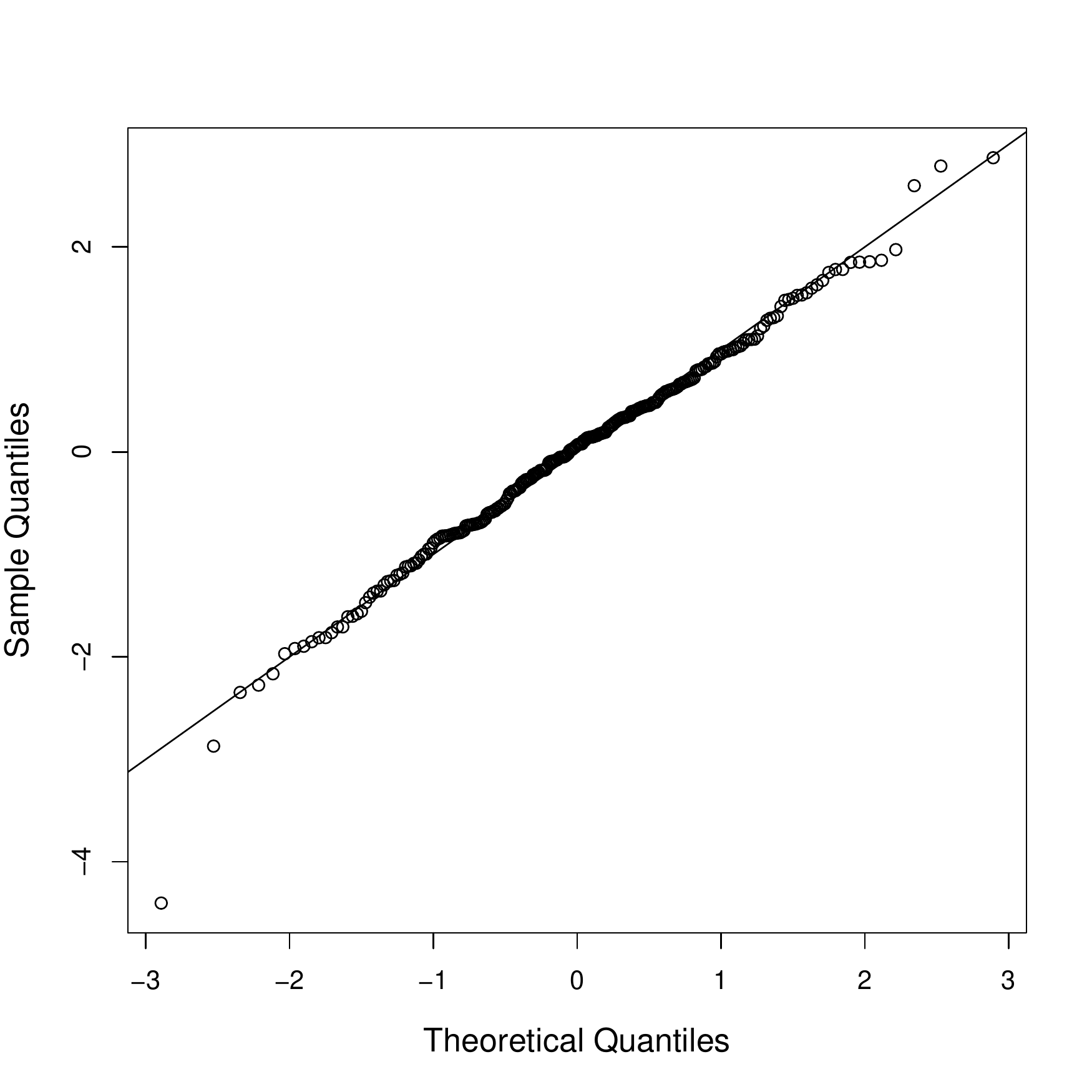} &
\includegraphics[width=0.5\textwidth]{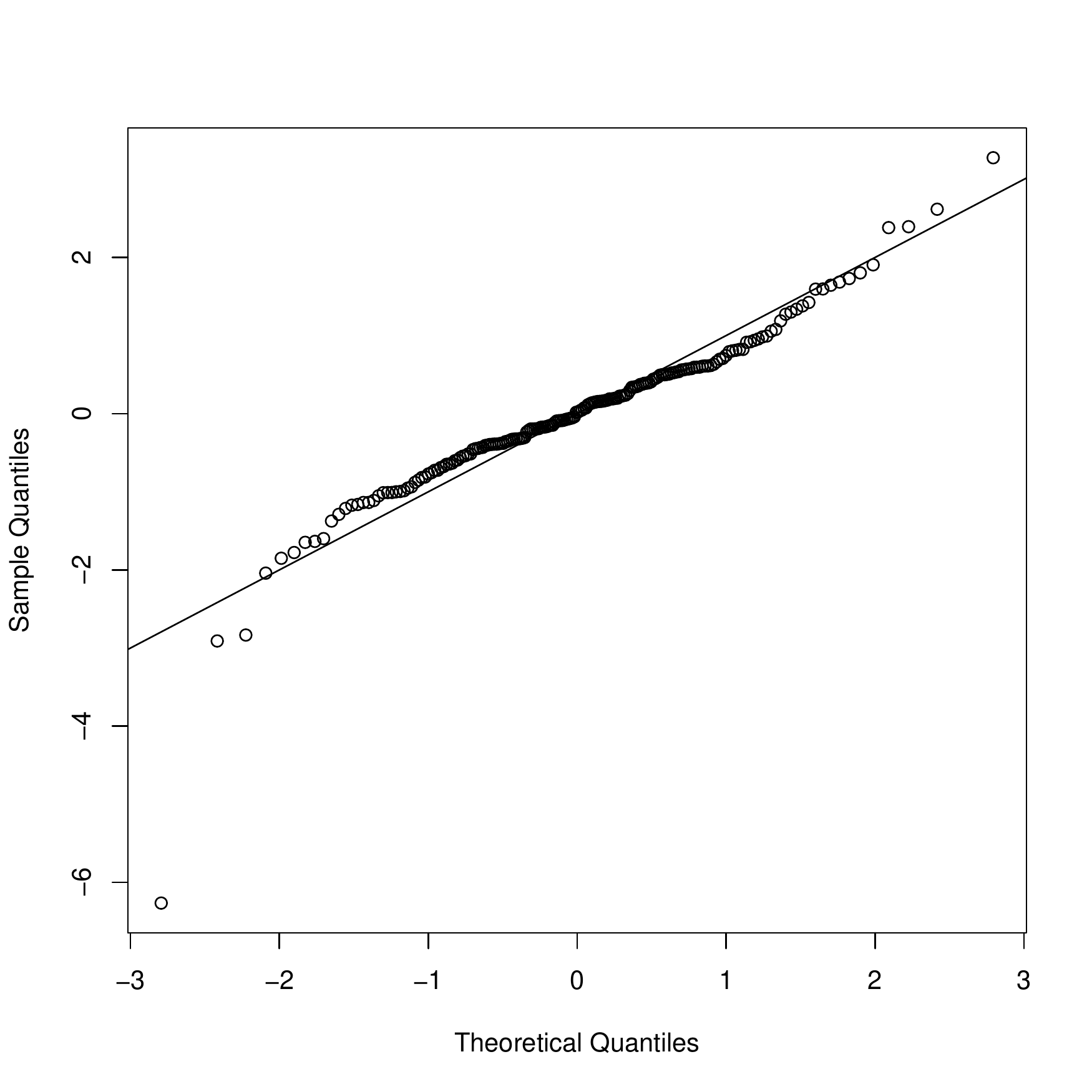} \\
\end{tabular}
\caption{QQ Normal plot for TGFB (left) DLD (right) data.}
\label{fig:qqnorm}
\end{figure}

We illustrate our methodology with the human microarray gene expression data in colon cancer patients from \cite{calon:2012}.
Briefly, following upon \cite{rossell:2017}, we aim to detect which
amongst $p=$10,172 candidate genes have an effect on the
expression levels of TGFB, a gene known to play an important role in colon cancer progression.
These data contain moderately correlated covariates with absolute
Pearson correlations ranging in (0,0.956) and 99\% of them being in the interval (0,0.375).
Both response and predictors were standardized to zero mean and unit variance.
The dataset and further information are provided in \cite{rossell:2017}.


\begin{table}
\begin{center}
\begin{tabular}{|l|cc|} \hline
Gene symbol   & \multicolumn{2}{|c|}{$p(\gamma \mid y)$} \\
             & Normal & Inferred \\ \hline
  ARL4C,AOC3,URB2,FAM89B,PCGF2,CCDC102B & 0.299 & 0.304 \\
            ARL4C,CNRIP1,AOC3,PCGF2 & 0.165 & 0.167 \\
                 ARL4C,CNRIP1,PCGF2 & 0.161 & 0.163 \\
       ARL4C,CNRIP1,AOC3,PCGF2,RPS6KB2 & 0.045 & 0.046 \\
            ARL4C,AOC3,PCGF2,CCDC102B & 0.028 & 0.028 \\
       ARL4C,AOC3,FAM89B,PCGF2,CCDC102B & 0.025 & 0.025 \\ \hline
\end{tabular}
\end{center}
\caption{TGFB data. Highest probability models under Normal and inferred error distribution.}
\label{tab:top6_tgfb}
\end{table}

We start by considering inference under the Normal model, i.e. conditional on $\gamma_{p+1}=\gamma_{p+2}=0$.
We ran 1,000 Gibbs iterations ({\it i.e.}~$10^3\times 10,172$ model updates),
which was deemed sufficient for practical convergence (see supplementary material in \cite{rossell:2017}). 
Table \ref{tab:top6_tgfb} shows the highest posterior probability models.
The top model included the 6 genes ARL4C, AOC3, URB2, FAM89B, PCGF2, CCDC102B
and had an estimated $p(\gamma \mid y)=0.299$.
Alternatively, selecting variables with marginal $p(\gamma_j=1 \mid y)>0.5$
\cite{barbieri:2004} returned 5 out of these 6 genes ($p(\gamma_j=1 \mid y)=0.482$ for FAM89B).
Briefly, according to genecards.org FAM89B is a TGFB regulator, ARL4C and PCGF2 have been related to various cancer types and
AOC3 is used to alleviate cancer symptoms, reinforcing the plausibility that these genes may be indeed related to TGFB.
URB2 and CCDC102B have no known relation to cancer, although the latter is connected to ARL4D in the STRING interaction networks.

We next considered the possibility that the Normal model might not be adequate for these data.
As an exploratory check, a quantile-quantile plot based on the residuals under the top model
revealed no strong departure from normality (Figure \ref{fig:qqnorm}).
Although this is somewhat reassuring one cannot discard a lack of normality under a different set of predictors,
as the top model was selected under assumed normality.
To conduct a more formal analysis we run Algorithm \ref{alg:gibbs_modelsearch} ($T=1,000$ iterations) now including $\gamma_{p+1},\gamma_{p+2}$.
The posterior probabilities for Normal, asymmetric Normal, Laplace and asymmetric Laplace errors
were 0.998, 0.0002, 0.0018 and 1.3$\times 10^{-27}$, respectively.
The six top models and their posterior probabilities closely matched those under the assumed Normal model (Table \ref{tab:top6_tgfb}),
and the correlation between marginal inclusion probabilities under Normal and inferred residuals was 0.96.
These results support that
our framework to infer $(\gamma_{p+1},\gamma_{p+2})$ in Algorithm \ref{alg:gibbs_modelsearch} is able to detect when errors are approximately Normal.

\subsection{DLD data}
\label{ssec:dld}

We consider another genomics study by \cite{yuan:2016}.
In contrast to Section \ref{ssec:tgfb},
here RNA-sequencing was used to measure gene expression, a newer and more precise technology than microarrays.
The study included 100 colorectal, 36 prostate, and 6 pancreatic cancer and 50 healthy control patients, for a total of $n=192$ patients.
Briefly, the authors used a measure of expression called RPM. RPM considers the number of reads mapped to a given gene relative to the gene length
and may exhibit heavy tails or asymmetries, even after log or other transformations.
We focus on the 58 messenger RNA genes identified in the exRNA species diversity analysis provided by the authors in Supplementary Table S1.
To illustrate our methodology, we consider predicting the expression of gene DLD based on the remaining 57 genes
and the 3 binary variables indicating the patient type (colorectal, prostate, pancreatic).
According to genecards.org, the protein encoded by DLD can perform mechanistically distinct functions, it can regulate the energy metabolism
and has been found to be associated with dehydrogenase and leukocyte adhesion defficiencies.

We first applied our methodology conditioning on Normal errors ($\gamma_{p+1}=\gamma_{p+2}=0$).
We used 10,000 Gibbs iterations.
The highest posterior probability model had $p(\gamma \mid y)=0.58$ and contained 5 genes
(C6orf226, ECH1, CSF2RA, FBXL19, RRP1B),
however, its residuals showed a clear departure from normality (Figure \ref{fig:qqnorm}, right).
We run again our Gibbs algorithm, this time inferring $\gamma_{p+1}$ and $\gamma_{p+2}$.
The analysis returned an overwhelming $p(\gamma_{p+1}=1,\gamma_{p+2}=0 \mid y)=0.999$ in favour of Laplace residuals.
The top model had posterior probability 0.36 and contained the same 5 predictors plus an extra gene MTMR1.
MTMR1 encodes a protein related to the myotubularin family containing consensus sequences for protein tyrase phosphatases,
whereas the response gene DLD has a post-translational modification based on tyrosine phosphorylation,
thus giving a plausible biological mechanism connecting MTMR1 and DLD.
Supplementary Table \ref{tab:dld_margpp} lists the six largest marginal variable inclusion probabilities under Normal and inferred error distribution.

So far, we treated $\alpha$ as a parameter to be learnt from the data.
We now condition upon asymmetric Laplace errors and fixed $\alpha=-0.5,0,0.5$.
This leads to quantile regression for the $(1+\alpha)/2=0.25,0.5,0.75$ percentiles (Section \ref{sec:loglikelihood}).
Supplementary Table \ref{tab:dld_top5models} displays the top 5 models for each $\alpha$.
Briefly, five genes (C6orf226, CSF2RA, ECH1, RRP1B and FBXL19) featured in the top model for all $\alpha$'s,
the first four with marginal inclusion probability $>0.99$.
FBXL19 had higher probability under $\alpha=0$ than $\alpha=-0.5,0.5$ (0.783 vs. 0.516 and 0.467 respectively).
MTMR1 featured in the top model only for $\alpha=0$ (marginal probability 0.619).
Given the biological plausibility that MTMR1 is related to DLD,
these results suggest that setting $\alpha=0$ (the value inferred from the data) may have led to higher power to detect MTMR1
than conditioning on Normal or asymmetric Laplace residuals with $\alpha=-0.5$ or $\alpha=0.5$.
This is in agreement with Proposition \ref{prop:bfrates}
and our simulations in Sections \ref{ssec:lowdim_sim}-\ref{ssec:highdim_sim}.

\section{Discussion}
\label{sec:conclusion}

Most efforts in Bayesian variable selection focus either on the Normal model or on flexible alternatives that require MCMC.
Our framework represents a middle-ground to add flexibility in a parsimonious manner that remains analytically and computationally tractable,
facilitating applications where either $p$ is large or $n$ is too moderate to fit more complex models accurately.
Our results show that model misspecification is a non-ignorable issue with important consequences for model selection.
Bayes factor rates typically retain the same functional dependence on $n$ (\textit{e.g.}~polynomial or exponential) as when the model is correctly specified,
however the coefficients governing these rates do change.
Specifically, the ratio of the correct \textit{vs.}~misspecified Bayes factors to detect truly active variables grows exponentially with $n$
when a triangle-type inequality holds, signaling the potential for an important drop in sensitivity.
Our empirical studies support this finding: failing to account for simple forms of asymmetry or heavy tails
reduced the proportion of correct model selections by several folds.
Misspecification also has an effect on false positives.
Although here the ratio of correct vs. misspecified Bayes factors is essentially a constant,
the effect can be noticeable for finite $n$.
Hence it is important to consider flexible likelihoods and, when possible, also adopt false positive correction mechanisms
for finite $n$. As a possible venue for the latter, we illustrated in an example how non-local priors helped discard small
spurious parameters arising from misspecification. A more detailed study would be interesting future work.

Other future avenues include extensions to allow for polynomial error tails, dependent errors, heteroscedasticity or covariate-dependent asymmetry.
We remark that fully non-parametric strategies already exist, \textit{e.g.}~\cite{chung:2009}.
The challenge is to build models that provide an intermediate level of flexibility while giving tractable variable selection.
For instance, allowing the variance or asymmetry to depend on $x_i$ is an interesting task for which there is no unique agreed-upon solution. One possibility is to let $\vartheta_i= \exp\left( x_i^T\beta\right)$, where $\vert \beta \vert \leq \vert \theta \vert$, akin to what \cite{daye:2012} for Normal errors. The authors found that the log-likelihood for $\beta$ for fixed $\theta$ is log-concave, and so is that for $\theta$ under fixed $\beta$, enabling fast optimization. It would be interesting to develop similar strategies for the asymmetry and non-Normal errors. An issue here would be dealing with the increased problem dimensionality due to selecting variables also for $\beta$.
Another interesting venue stemming from our work is posing non-parametric models that can collapse onto simple parametric forms
when the extra flexibility is not needed.
Again the idea is to strike a balance between the tractability offered by simple models
and the ultimate goal of providing accurate inference.
Other extensions are developing more advanced optimization or model search strategies,
our goal here was to illustrate that even relatively simple methods can be competitive.
Such computational issues are particularly meaningful in increasingly challenging settings, \textit{e.g.}~large graphical or spatio-temporal models.
Overall, we hope to have provided a basic framework that others can build on to tackle these exciting applications.


\section*{Supplementary material. Tractable Bayesian variable selection: beyond normality}

\date{}  
\maketitle

\section{iMOM prior}
\label{sec:imom}

The product iMOM prior density on $\theta_\gamma$ \cite{johnson:2012} is given by
\begin{align}
p_I(\theta_\gamma\mid \vartheta, \gamma) = \prod_{\gamma_j=1}^{} \frac{(g_{\theta} \vartheta)^{\frac{1}{2}}}{\sqrt{\pi} \theta_j^2}
\exp\left\{ - \frac{g_{\theta} \vartheta}{\theta_j^2} \right\},
\label{eq:pimom}
\end{align}
where by default $g_\theta=0.133$ assigns $p(|\theta/\vartheta^{1/2}|>0.2)=0.99$.
Regarding the asymmetry parameter $\talpha=\mbox{atanh}(\alpha)$,
the prior is $p_I(\talpha \mid \gamma_{p+1}=1)= \talpha^{-2} \sqrt{g_\alpha/\pi} e^{-g_\alpha/\talpha^2}$,
and the default prior dispersions are $g_\alpha=0.033$ to obtain $P(|\talpha| \geq 0.1)= 0.99$
and $g_\alpha=0.136$ for $P(|\alpha| \geq 0.2)=0.99$.

\section{Proofs}

For simplicity, we drop $\gamma$ from the notation in the proof of
Propositions \ref{prop:lhood_skewnorm_properties}-\ref{prop:AsympNorm} and Corollay \ref{cor:asymp_mle_misspec},
given that all arguments are conditional on a given model $\gamma$.

\subsection{Proof of Proposition \ref{prop:lhood_skewnorm_properties}}

We start by stating a useful lemma stating that positive definite hessian plus continuous gradient guarantees concavity.

\begin{lemma}
Let $f(\theta)$ be a function with continuous gradient $g(\theta)$, for all $\theta$,
and negative definite hessian $H(\theta)$ almost everywhere with respect to the Lebesgue measure.
Then, $f(\theta)$ is strictly concave.
If $H(\theta)$ is negative semidefinite, then $f(\theta)$ is concave.
\label{prop:negdef_contgrad_uniquemax}
\end{lemma}

\begin{proof}
Let $\theta_1$ and $\theta_2$ be two arbitrary values and denote
$\theta_w= (1-w)\theta_1 + w\theta_2$ where $w \in [0,1]$.
Define $h(w)= -f(\theta_w)$,
to show that $f(\theta)$ is concave it suffices to see that $h(w)$ is convex for arbitrary $(w,\theta_1,\theta_2)$.
Straightforward algebra shows that
$\frac{\partial}{\partial w} h(w)= -g(\theta_w) (\theta_2 - \theta_1)$
and further derivation shows that
\begin{align}
\frac{\partial^2}{\partial w^2} h(w)= -(\theta_2 - \theta_1)^T H(\theta_w) (\theta_2 - \theta_1) > 0,
\nonumber
\end{align}
since $H(\theta)$ is negative definite ($\geq 0$ for negative semidefinite).

The second derivative $\frac{\partial^2}{\partial w^2} h(w)>0$ almost everywhere
and the first derivative $\frac{\partial}{\partial w} h(w)$ is continuous, which implies that
$\frac{\partial}{\partial w} h(w)$ is strictly increasing in $w$ and hence $h(w)$ is strictly convex
(non-strictly convex when $H(\theta)$ is negative semidefinite).
\end{proof}

\vspace{3mm}

\underline{Proof of Proposition \ref{prop:lhood_skewnorm_properties}, Part (i)}

The gradient $g_1(\theta,\vartheta,\alpha)$ follows from straightforward algebra,
which is obviously continuous with respect to $\vartheta \in \mathbb{R}^+$ and $\alpha \in [-1,1]$.
To see continuity of $g_1(\theta,\vartheta,\alpha)$ with respect to $\theta$,
consider increasing a single $\theta_j$ for some $j \in \{1,\ldots,p\}$ and fix the remaining elements in $\theta$,
which we denote $\theta_{(-j)}$. Also denote $x_{i(-j)}$ the subvector of $x_i$ obtained by removing $x_{ij}$.
Clearly, $\log L_1(\theta,\vartheta,\alpha)$ is quadratic in $\theta_j$ with coefficients that stay constant until
$\theta_j$ increases beyond a value $t$ such that an observation $i^*$ is added to or removed from $A(\theta)$,
{\it i.e.}~$y_{i^*} < x_{i^*(-j)}^T \theta_{(-j)} + x_{i^*j} \theta_j$ for $\theta_j \leq t$
and $y_{i^*} > x_{i^*(-j)}^T \theta_{(-j)} + x_{i^*j} \theta_j$ for $\theta_j > t$.
Taking the limit of the contribution of $i^*$ to $\log \left( L_1(\theta,\vartheta,\alpha) \right)$
as either $\theta_j \rightarrow t^{-}$ or $\theta_j \rightarrow t^{+}$ we obtain
$$
\mathop {\lim }\limits_{\theta_j \to t^+} \frac{(y_{i^*} - x_{i^*}\theta_i)^2}{(1+\alpha)^2}=
\mathop {\lim }\limits_{\theta_j \to t^-} \frac{(y_{i^*} - x_{i^*}\theta_i)^2}{(1-\alpha)^2}= 0,
$$
{\it i.e.}~$\log \left( L_1(\theta,\vartheta,\alpha) \right)$ is continuous.
Similarly, taking the limits for the contribution to the first partial derivative with respect to $\theta_j$ gives
$$
\mathop {\lim }\limits_{\theta_j \to t^+} \frac{2(y_{i^*} - x_{i^*}\theta_i)}{(1+\alpha)^2}=
\mathop {\lim }\limits_{\theta_j \to t^-} \frac{2(y_{i^*} - x_{i^*}\theta_i)}{(1-\alpha)^2}= 0,
$$
which proves that $g_1(\theta,\vartheta,\alpha)$ is continuous.

\vspace{3mm}
\underline{Proof of Proposition \ref{prop:lhood_skewnorm_properties}, Part (ii)}

The form of $H_1(\theta, \vartheta, \alpha)$ follows from easy algebra.

\vspace{3mm}
\underline{Proof of Proposition \ref{prop:lhood_skewnorm_properties}, Part (iii)}

We start by noting that the maximum of the asymmetric-normal log-likelihood with respect to $(\theta,\alpha)$
does not depend on $\vartheta$, hence we simply need to see that
\begin{align}
H=
\begin{pmatrix}
X^T W^2 X  & 2 X^T \Wbar^3 (y - X\theta) \\
2 (y - X\theta)^T \Wbar^3 X  & 3 (y-X\theta)^T W^2 (y-X\theta)
\end{pmatrix},
\\
\nonumber
\end{align}
is positive definite for almost all $(\theta,\alpha)$.
Once we show this, by Part (i) and Lemma \ref{prop:negdef_contgrad_uniquemax} we have that there is a unique maximum.

To see that $H$ is positive definite, we shall show that all its leading principal minors are positive.
Note that $X^TW^2X$ is the gram matrix corresponding to $W X$ and is hence positive definite when $\mbox{rank}(WX)=p$,
or equivalently when $\mbox{rank}(X)=p$ given that the effect of $W$ is to simply re-scale the rows of $X$.
If $\mbox{rank}(WX)<p$ then $X^TW^2X$ is positive semidefinite.
Therefore, we just need to check that $\mbox{det}(H)>0$.
Now, the usual formula for determinant based on submatrices gives that
$\mbox{det}(H)= \mbox{det}(X^TW^2X) \mbox{det} (B)$, where $B=$
\begin{align}
3 (y-X\theta)^T W^4 (y-X\theta) - 4 (y - X\theta)^T \Wbar^3 X (X^TW^2X)^{-1} X^T \Wbar^3 (y - X\theta) \nonumber \\
=3(y-X\theta)^TW^2 \left( I - \frac{4}{3} \Wbar X (X^TW^2X)^{-1} X^T \Wbar) \right) W^2 (y-X\theta),
\label{eq:detB_step1}
\end{align}
is a scalar, $I$ is the $n \times n$ identity matrix,
as usual $W$ is an $n \times n$ diagonal matrix with entries $1/(1 \pm \alpha)^2$ where the $\pm$ depends on whether
$i \in A(\theta)$ or $i \not\in A(\theta)$,
and similarly $\Wbar$ is diagonal with entries $\pm (1 \pm \alpha)$.
All that is left is to see that $B>0$.
For ease of notation let us define $Z= \Wbar X$,
given that $\Wbar \Wbar= \mbox{diag}(1/(1 \pm \alpha)^2)=W^2$ we can write
\begin{align}
B= 3(y-X\theta)^TW^2 \left( I - \frac{4}{3} Z (Z^TZ)^{-1} Z^T) \right) W^2 (y-X\theta)= \nonumber \\
4(y-X\theta)^TW^2 (I - Z (Z^TZ)^{-1} Z^T) W^2 (y-X\theta)
- (y-X\theta)^TW^2 W^2(y-X\theta) > 0 \nonumber \\
\Leftrightarrow
4 \frac{(y-X\theta)^TW^2 (I - Z (Z^TZ)^{-1} Z^T) W^2 (y-X\theta) }{(y-X\theta)^TW^2 W^2(y-X\theta)} - 1 > 0.
\label{eq:detB_step2}
\end{align}
To complete the proof, note that $a=W^2 (y-X\theta) \in \mathbb{R}^n$ is simply a vector
and that the hat matrix $Z (Z^TZ)^{-1} Z^T$ is symmetric and idempotent, which implies that it has $\mbox{rank}(Z)$ eigenvalues equal to 1
and $n - \mbox{rank}(Z)$ eigenvalues equal to 0.
Thus $I - Z (Z^TZ)^{-1} Z^T$ has $n - \mbox{rank}(Z)$ eigenvalues equal to 1 and the remaining $\mbox{rank}(Z)$ eigenvalues equal to 0.
Given that $n> \mbox{rank}(Z)$ by assumption,
$I - Z (Z^TZ)^{-1} Z^T$ has at least one non-zero eigenvalue, which allows us to bound
\begin{align}
\mbox{min}_{a \in \mathbb{R}^n} \frac{a (I - Z(Z^TZ)^{-1}Z^T) a}{a^Ta} \geq 1,
\nonumber
\end{align}
which from (\ref{eq:detB_step2}) gives that $B \geq 3$ and hence that $H$ is positive definite.

\subsection{Proof of Proposition \ref{prop:lhood_skewlap_properties}}

Parts (i) and (ii) follow from straightforward algebra.
For Part (iii) we first show that $\log L_2(\theta,\vartheta,\alpha)$ is (non-strictly) concave in $(\theta,\alpha)$
and then that when $\mbox{rank}(X)=p$ it is strictly concave.
To see non-strict concavity note that
$-|y_i - x_i^T\theta|/ (\sqrt{\vartheta} (1+\alpha))=
-\mbox{max}\{y_i-x_i^T\theta,x_i^T\theta-y_i\} / (\sqrt{\vartheta} (1+\alpha))$ is the maximum of two (non-strictly) concave functions
in $(\theta,\alpha)$ and hence also concave,
from which it follows that $L_3(\theta,\vartheta,\alpha)$ is a sum of concave functions and thus concave.

For ease of notation let $\eta=(\theta,\vartheta,\alpha)$, we now show that $\log L_2(\eta)$ is stricly concave
at any arbitrary $\eta_1=(\theta_1,\vartheta,\alpha_1)$ as long as $\mbox{rank}(X)=p$.
It is useful to note that $H_2(\theta,\vartheta,\alpha)$ is strictly negative definite in $\alpha$,
as the corresponding minor $- 2 \vert W^3 (y - X \theta) \vert / \sqrt{\vartheta}<0$.
From the definition of concavity and continuity of the log-likelihood,
if $\log L_2(\eta)$ were concave but non-strictly concave at $\eta=\eta_1$
then for some $\eta_2=(\theta_2,\vartheta,\alpha_2) \neq \eta_1$ we would have that
$\log L_2(a \eta_1 + (1-a) \eta_2)= a\log L_2(\eta_1) + (1-a)\log L_2(\eta_2)$ for all $a \in [0,1]$,
{\it i.e.}~$\log L_2(\eta)$ would be locally linear (in fact, constant) along the direction defined by $\eta_2-\eta_1$,
and in particular $\log L_2(\eta_1)=\log L_2(\eta_2)$.
From its form
\begin{align}
\log L_2(\eta)= -\frac{n}{2} \log(\vartheta) - \frac{1}{\vartheta}
\left( \frac{\sum_{i \in A(\theta)}^{} |y_i - x_i^T\theta|}{1+\alpha} +
\frac{\sum_{i \not\in A(\theta)}^{} |y_i - x_i^T\theta|}{1-\alpha} \right),
\nonumber
\end{align}
is locally linear in $\theta$ but clearly non-linear in $\alpha$,
implying that $\alpha_2=\alpha_1$.
More formally, it is easy to see that for fixed $\theta_1 \neq \theta_2$
the roots of $\log L_2(\eta_1)= \log L_2(\eta_2)$ in terms of $\alpha_2$
are given by the roots of a quadratic polynomial that are not linear in $\theta_2$,
thus the only possible linear solution is $\alpha_2=\alpha_1$.
The problem is hence reduced to showing that there is no $\theta_2$ sufficiently close to $\theta_1$ such that
\begin{align}
|W (y - X\theta_1)|= |W (y - X\theta_2)|,
\label{eq:weighted_abs_erros_concproof}
\end{align}
where $|\cdot|$ denotes the $L_1$ norm
and as usual $W$ is a diagonal matrix with $(i,i)$ element $(1+\alpha)^{-1}$ if $i \in A(\theta_1)$
and $(1-\alpha)^{-1}$ if $i \not\in A(\theta_1)$, where we note that
$A(\theta_2)=A(\theta_1)$ for $\theta_2$ sufficiently close to $\theta_1$ and thus
the same weighting matrix $W$ can be used in left and right hand sides of (\ref{eq:weighted_abs_erros_concproof}).
Expression \eqref{eq:weighted_abs_erros_concproof} is the $L_1$ error function featuring in median regression
with re-scaled $\tilde{y}= W y$ and $\tilde{X}= W X$,
which is concave as long as $p=\mbox{rank}(W X)=\mbox{rank}(X)$, as we wished to prove.

\subsection{Proof of Proposition \ref{prop:Consistency}}
\subsection*{\underline{Two-piece normal errors ($k=1$)}}
The proof strategy is as follows: we first show that the average log-likelihood $M_n(\theta_\gamma,\vartheta,\alpha) = \dfrac{1}{n}\log L_1(\theta_\gamma,\vartheta,\alpha)$ converges to its expected value $M(\theta_\gamma,\vartheta,\alpha)$ uniformly across $(\theta_\gamma,\vartheta,\alpha)\in\Gamma$, and later show that $M(\theta_\gamma,\vartheta,\alpha)$ has a unique maximum $(\theta_\gamma^*,\vartheta_\gamma^*,\alpha_\gamma^*)$, which jointly satisfy the conditions in Theorem 5.7 from \cite{vandervaart:1998} for consistency of $(\widehat{\theta}_\gamma,\widehat{\vartheta}_\gamma,\widehat{\alpha}_\gamma) \stackrel{P}{\longrightarrow} (\theta_\gamma^*,\vartheta_\gamma^*,\alpha_\gamma^*)$.

We remark that Condition A3 is met for instance by deterministic sequences $\{x_i\}$ satisfying the stated positive-definiteness condition and also by $x_i \stackrel{i.i.d.}{\sim} \Psi$ as long as $E(x_1 x_1^T)=\Sigma$ for some positive definite $\Sigma$, since then $n^{-1} X^TX \stackrel{a.s.}{\longrightarrow} \Sigma$ by the strong law of large numbers, and given that eigenvalues are continuous functions of $X^TX$ by the continuous mapping theorem $X^TX$ is positive definite almost surely as $n \rightarrow \infty$. Finally, $\Gamma$ is assumed to contain the maximizer $(\theta_\gamma^*,\vartheta_\gamma^*,\alpha_\gamma^*)$.

By the law of large numbers and the \emph{i.i.d.}~assumption, we have that $M_n(\theta_\gamma,\vartheta,\alpha) \stackrel{P}{\rightarrow} M(\theta_\gamma,\vartheta,\alpha)$, for each $(\theta_\gamma,\vartheta,\alpha)\in\Gamma$. Next, we prove that the limit $M$ is finite for all $(\theta_\gamma,\vartheta,\alpha)\in\Gamma$.
\begin{eqnarray*}
&& \vert M(\theta_\gamma,\vartheta,\alpha)\vert = \Big\vert {\mathbb E}\left[\log s_1(y_1\vert x_1^{T}\theta_\gamma,\vartheta,\alpha)\right] \Big\vert \leq  {\mathbb E}\left[\vert\log s_1(y_1\vert x_1^{T}\theta_\gamma,\vartheta,\alpha)\vert\right] \\
&=& \int \int \vert \log s_1(y_1\vert x_1^{T}\theta_\gamma,\vartheta,\alpha)\vert dS_0(y_1\vert  x_1) d\Psi( x_1)\\
&=&  \int \int_{y_1< x_1^{T}\theta_\gamma} \Bigg\vert \log \dfrac{1}{\sqrt{\vartheta}} \phi\left(\dfrac{y_1- x_1^{T}\theta_\gamma}{\sqrt{\vartheta}(1+\alpha)}\right) \Bigg\vert dS_0(y_1\vert  x_1) d\Psi( x_1)\\
&+& \int \int_{y_1\geq x_1^{T}\theta_\gamma} \Bigg\vert \log \dfrac{1}{\sqrt{\vartheta}} \phi\left(\dfrac{y_1- x_1^{T}\theta_\gamma}{\sqrt{\vartheta}(1-\alpha)}\right) \Bigg\vert dS_0(y_1\vert  x_1) d\Psi( x_1).
\end{eqnarray*}
For the first term in the last inequality we obtain, by integrating over the whole space, assumption A4 with $j=2$, and the triangle inequality, the following upper bound
\begin{eqnarray*}
&& \int \int \Bigg\vert \log \dfrac{1}{\sqrt{\vartheta}} \phi\left(\dfrac{y_1- x_1^{T}\theta_\gamma}{\sqrt{\vartheta}(1+\alpha)}\right) \Bigg\vert dS_0(y_1\vert  x_1) d\Psi( x_1)\\
&\leq& \vert \log\sqrt{2\pi\vartheta} \vert + \int\int \dfrac{(y_1- x_1^{T}\theta_\gamma)^2}{2\vartheta(1+\alpha)^2} dS_0(y_1\vert  x_1) d\Psi( x_1) <\infty.
\end{eqnarray*}
Analogously for the second term. Now, let $\vartheta=\vartheta^\star$ be an arbitrary fixed value for the (squared) scale parameter. The aim now is to first show that the average log-likelihood $M_n(\theta_\gamma,\vartheta^\star,\alpha)= n^{-1} \log L_1(\theta_\gamma,\vartheta^\star,\alpha)$ converges to its expected value $M(\theta_\gamma,\vartheta^\star,\alpha)$ uniformly in $(\theta_\gamma,\alpha)$, which implies that $(\widehat{\theta}_\gamma,\widehat{\alpha}_\gamma) \stackrel{P}{\longrightarrow} (\theta_\gamma^*,\alpha_\gamma^*)$, and to then exploit that $\widehat{\vartheta}_\gamma$ and $\vartheta_\gamma^*$ have simple expressions to show that $\widehat{\vartheta}_\gamma \stackrel{P}{\longrightarrow} \vartheta_\gamma^*$. To see that $M_n(\theta_\gamma,\vartheta^\star,\alpha)$ converges to $M(\theta_\gamma,\vartheta^\star,\alpha)$ uniformly in $(\theta_\gamma,\alpha)$
we use the result in Proposition \ref{prop:lhood_skewnorm_properties} that for positive-definite $X^TX$ (which holds for $n>n_0$) we have that $M_n(\theta_\gamma,\vartheta^\star,\alpha)$ is a sequence of concave functions in $(\theta_\gamma,\alpha)$, which by the convexity lemma in \cite{pollard:1991} (see also Theorem 10.8 from \cite{rockafellar:2015}) implies that
\begin{eqnarray}\label{SC1}
\sup_{(\theta_\gamma,\alpha) \in K} \left\vert M_n(\theta_\gamma,\vartheta^\star,\alpha) - M(\theta_\gamma,\vartheta^\star,\alpha)\right\vert \stackrel{P}{\longrightarrow} 0,
\end{eqnarray}
for each compact set $K\subseteq\Gamma$, and also that $M(\theta_\gamma,\vartheta^\star,\alpha)$ is finite and concave in $(\theta_\gamma,\alpha)$ and thus has a unique maximum $(\theta_\gamma^*,\alpha_\gamma^*)$. That is, for a distance measure $d()$ and every $\varepsilon>0$ we have
\begin{eqnarray}\label{SC2}
\sup_{d((\theta_\gamma^*,\vartheta^\star,\alpha_\gamma^*),(\theta,\vartheta^\star,\alpha)) \geq \varepsilon } M(\theta_\gamma,\vartheta^\star,\alpha) <M(\theta_\gamma^*,\vartheta^\star,\alpha_\gamma^*).
\end{eqnarray}
The consistency of $(\widehat{\theta}_\gamma,\widehat{\alpha}_\gamma) \stackrel{P}{\longrightarrow} (\theta_\gamma^*,\alpha_\gamma^*)$ follows directly from \eqref{SC1} and \eqref{SC2} together with Theorem 5.7 from \cite{vandervaart:1998}. To see that $\widehat{\vartheta}_\gamma \stackrel{P}{\longrightarrow} \vartheta_\gamma^*$, note first that from
\begin{align}\label{eq:mfun_tpnorm}
M(\theta_\gamma,\vartheta^\star,\alpha) = -\log\left(\sqrt{2\pi\vartheta^\star}\right)
&- \dfrac{1}{2\vartheta^\star} \int \Biggl[ \dfrac{(y_1- x_1^T\theta_\gamma)^2}{(1+\alpha)^2} I(y_1< x_1^T\theta_\gamma) \\
&+ \dfrac{(y_1- x_1^T\theta_\gamma)^2}{(1-\alpha)^2} I(y_1\geq x_1^T\theta_\gamma)\Biggr] dS_0(y_1\vert x_1) d\Psi(x_1),
\nonumber
\end{align}
we see that $(\theta_\gamma^*,\alpha_\gamma^*)$ does not depend on $\vartheta^\star$, thus $(\theta_\gamma^*,\alpha_\gamma^*)$ is a global maximum. From \eqref{eq:mfun_tpnorm} $M(\theta_\gamma^*,\vartheta,\alpha_\gamma^*)$ trivially has the maximizer
\begin{eqnarray*}
\vartheta_\gamma^*= \int \left[\dfrac{(y_1- x_1^T\theta_\gamma^*)^2}{(1+\alpha_\gamma^*)^2} I(y_1< x_1^T\theta_\gamma^*)
+ \dfrac{(y_1- x_1^T\theta_\gamma^*)^2}{(1-\alpha_\gamma^*)^2} I(y_1\geq x_1^T\theta_\gamma^*) \right]dS_0(y_1\vert x_1) d\Psi(x_1),
\end{eqnarray*}
and from the likelihood equations we have that
\begin{align}\label{eq:mle_vartheta_tpnorm}
\widehat{\vartheta}_{\gamma}= \frac{1}{n} \left(
\sum_{i=1}^{n} \frac{(y_i-x_i^T \widehat{\theta}_{\gamma})^2}{(1+\widehat{\alpha}_{\gamma})^2} I(y_i \leq x_i^T\widehat{\theta}_{\gamma}) +
\frac{(y_i-x_i^T \widehat{\theta}_{\gamma})^2}{(1-\widehat{\alpha}_{\gamma})^2} I(y_i>x_i^T\widehat{\theta}_{\gamma})\right).
\end{align}
In order to simplify notation, let us define
\begin{eqnarray*}
\rho(y_i,x_i,\theta_{\gamma},\alpha)= \frac{(y_i-x_i^T \theta_{\gamma})^2}{(1+\alpha)^2} I(y_i \leq x_i^T\theta_{\gamma}) +
\frac{(y_i-x_i^T \theta_{\gamma})^2}{(1-\alpha)^2} I(y_i>x_i^T\theta_\gamma).
\end{eqnarray*}
Then, by the triangle inequality
\begin{eqnarray*}
\left\vert \widehat{\vartheta}_{\gamma} - \vartheta_{\gamma}^* \right\vert  \leq \left\vert \widehat{\vartheta}_{\gamma} - \dfrac{1}{n}\sum_{i=1}^n\rho(y_i,x_i,\theta_{\gamma}^*,\alpha_{\gamma}^*) \right\vert + \left \vert \dfrac{1}{n}\sum_{i=1}^n\rho(y_i,x_i,\theta_\gamma^*,\alpha_\gamma^*) - \vartheta_\gamma^* \right\vert.
\end{eqnarray*}
For the second term it follows, by the law of large numbers, that
$$\left\vert \dfrac{1}{n}\sum_{i=1}^n\rho(y_i,x_i,\theta_\gamma^*,\alpha_\gamma^*) - \vartheta_\gamma^* \right\vert \stackrel{P}{\rightarrow} 0.$$
For the first term we have
\begin{eqnarray*}
\left\vert \widehat{\vartheta}_\gamma - \dfrac{1}{n}\sum_{i=1}^n\rho(y_i,x_i,\theta_\gamma^*,\alpha_\gamma^*) \right\vert &=& \left\vert M_n(\widehat{\theta}_\gamma,1/2,\widehat{\alpha}_\gamma) - M_n(\theta_\gamma^*,1/2,\alpha_\gamma^*) \right\vert\\
&\leq& \left\vert M_n(\widehat{\theta}_\gamma,1/2,\widehat{\alpha}_\gamma) - M(\widehat{\theta}_\gamma,1/2,\widehat{\alpha}_\gamma) \right\vert \\
&+& \left\vert  M(\widehat{\theta}_\gamma,1/2,\widehat{\alpha}_\gamma) - M(\theta_\gamma^*,1/2,\alpha_\gamma^*) \right\vert\\
&+& \left\vert  M(\theta_\gamma^*,1/2,\alpha_\gamma^*) - M_n(\theta_\gamma^*,1/2,\alpha_\gamma^*) \right\vert\\
&\leq& 2 \sup_{(\theta_\gamma,\alpha) \in \Gamma} \left\vert M_n(\theta_\gamma,1/2,\alpha) - M(\theta_\gamma,1/2,\alpha)\right\vert \\
&+& \left\vert  M(\widehat{\theta}_\gamma,1/2,\widehat{\alpha}) - M(\theta_\gamma^*,1/2,\alpha_\gamma^*) \right\vert.
\end{eqnarray*}
By using \eqref{SC1}, the consistency of $(\widehat{\theta}_\gamma,\widehat{\alpha}_\gamma)$, and the continuous mapping theorem it follows that $\left\vert \widehat{\vartheta}_\gamma - \dfrac{1}{n}\sum_{i=1}^n\rho(y_i,x_i,\theta_\gamma^*,\alpha_\gamma^*)\right\vert \stackrel{P}{\rightarrow} 0$. Consequently, $\widehat{\vartheta} \stackrel{P}{\rightarrow} \vartheta_\gamma^*$, which completes the proof.

\subsection*{\underline{Two-piece Laplace errors ($k=2$)}}

The proof strategy is analogous to that with $k=1$. Denote $M_n(\theta_\gamma,\vartheta,\alpha) = \dfrac{1}{n}\log L_2(\theta_\gamma,\vartheta,\alpha)$. By the law of large numbers, we have that $M_n(\theta_\gamma,\vartheta,\alpha) \stackrel{P}{\rightarrow} M(\theta_\gamma,\vartheta,\alpha)$, for each $(\theta_\gamma,\vartheta,\alpha)\in\Gamma$. Moreover,
\begin{eqnarray*}
&& \vert M(\theta_\gamma,\vartheta,\alpha)\vert = \Big\vert {\mathbb E}\left[\log s_2(y_1\vert x_1^{T}\theta_\gamma,\vartheta,\alpha)\right] \Big\vert \leq  {\mathbb E}\left[\vert\log s_2(y_1\vert x_1^{T}\theta_\gamma,\vartheta,\alpha)\vert\right] \\
&=& \int \vert \log s_2(y_1\vert x_1^{T}\theta_\gamma,\vartheta,\alpha)\vert dS_0(y_1\vert x_1) d\Psi(x_1)\\
&=&  \int_{y<x^{T}\theta_\gamma} \Bigg\vert \log \dfrac{1}{\sqrt{\vartheta}} f\left(\dfrac{y_1-x_1^{T}\theta_\gamma}{\sqrt{\vartheta}(1+\alpha)}\right) \Bigg\vert dS_0(y_1\vert x_1) d\Psi(x_1)\\
&+& \int_{y_1\geq x_1^{T}\theta_\gamma} \Bigg\vert \log \dfrac{1}{\sqrt{\vartheta}} f\left(\dfrac{y_1-x_1^{T}\theta_\gamma}{\sqrt{\vartheta}(1-\alpha)}\right) \Bigg\vert dS_0(y_1\vert x_1) d\Psi(x_1),
\end{eqnarray*}
where $f(z)=0.5\exp(-\vert z\vert)$. For the first term in the last inequality we have, by integrating over the whole space and the triangle inequality, the following upper bound
\begin{eqnarray*}
&&\int \Bigg\vert \log \dfrac{1}{\sqrt{\vartheta}} f\left(\dfrac{y_1-x_1^{T}\theta_\gamma}{\sqrt{\vartheta}(1+\alpha)}\right) \Bigg\vert dS_0(y_1\vert x_1) d\Psi(x_1)\\
&\leq& \vert \log 2\sqrt{\vartheta} \vert + \int \dfrac{\vert y_1-x_1^{T}\theta_\gamma\vert}{\sqrt{\vartheta}(1+\alpha)} dS_0(y_1\vert x_1) d\Psi(x_1)< \infty,
\end{eqnarray*}
where the finiteness follows by assumption A4 with $j=1$. An analogous result is obtained for the second term. Now, let $\vartheta=\vartheta^\star$ be an arbitrary fixed value for the (squared) scale parameter. From Proposition \ref{prop:lhood_skewlap_properties}, it follows that for positive-definite $X^TX$ (which is guaranteed by assumption A2, for $n>n_0$) we have that $M_n(\theta_\gamma,\vartheta^\star,\alpha)$ is concave in $(\theta,\alpha)$, which by the convexity lemma in \cite{pollard:1991} implies that
\begin{eqnarray}\label{SCL1}
\sup_{(\theta_\gamma,\alpha) \in K} \left\vert M_n(\theta_\gamma,\vartheta^\star,\alpha) - M(\theta_\gamma,\vartheta^\star,\alpha)\right\vert \stackrel{P}{\longrightarrow} 0,
\end{eqnarray}
for any compact set $K \subseteq \Gamma$, and also that $M(\theta_\gamma,\vartheta^\star,\alpha)$ is concave in $(\theta_\gamma,\alpha)$ and thus has a unique maximum $(\theta_\gamma^*,\alpha_\gamma^*)$. That is, for a distance measure $d()$ and every $\varepsilon>0$ we have
\begin{eqnarray}\label{SCL2}
\sup_{d((\theta_\gamma^*,\vartheta^\star,\alpha_\gamma^*),(\theta_\gamma,\vartheta^\star,\alpha)) \geq \varepsilon } M(\theta_\gamma,\vartheta^\star,\alpha) <M(\theta_\gamma^*,\vartheta^\star,\alpha_\gamma^*).
\end{eqnarray}
The consistency of $(\widehat{\theta}_\gamma,\widehat{\alpha}_\gamma) \stackrel{P}{\longrightarrow} (\theta_\gamma^*,\alpha_\gamma^*)$ follows directly from \eqref{SCL1} and \eqref{SCL2} together with Theorem 5.7 from \cite{vandervaart:1998}. To see that $\widehat{\vartheta}_\gamma \stackrel{P}{\longrightarrow} \vartheta_\gamma^*$, note first that from
\begin{align}\label{eq:mfun_tplap}
M(\theta_\gamma,\vartheta^\star,\alpha) = -\log\left(2\sqrt{\vartheta^\star}\right)
&- \dfrac{1}{\sqrt{\vartheta^\star}} \int \Biggl[\dfrac{\vert y_1-x_1^T\theta_\gamma\vert}{1+\alpha} I(y_1<x_1^T\theta_\gamma) \\
&+  \dfrac{\vert y_1-x_1^T\theta_\gamma\vert}{1-\alpha} I(y_1\geq x_1^T\theta_\gamma)\Biggr] dS_0(y_1\vert x_1) \Psi(x_1),
\nonumber
\end{align}
we see that $(\theta_\gamma^*,\alpha_\gamma^*)$ does not depend on $\vartheta^\star$, thus $(\theta_\gamma^*,\alpha_\gamma^*)$ is a global maximum. From \eqref{eq:mfun_tpnorm} $M(\theta_\gamma^*,\vartheta,\alpha_\gamma^*)$ trivially has the maximizer
\begin{eqnarray*}
\vartheta_\gamma^*= \left\{\int \left[\dfrac{\vert y_1-x_1^T\theta_\gamma^*\vert}{1+\alpha_\gamma^*} I(y_1<x_1^T\theta_\gamma^*)
+  \dfrac{\vert y_1-x_1^T\theta_\gamma^*\vert}{1-\alpha_\gamma^*} I(y_1\geq x_1^T\theta_\gamma^*) \right]dS_0(y_1\vert x_1) d\Psi(x_1)\right\}^2,
\end{eqnarray*}
and from the likelihood equations we have that
\begin{align}\label{eq:mle_vartheta_tplap}
\widehat{\vartheta}_\gamma=\left[ \frac{1}{n} \left(
\sum_{i=1}^{n} \frac{\vert y_i-x_i^T \widehat{\theta}_\gamma\vert}{1+\widehat{\alpha}_\gamma} I(y_i \leq x_i^T\widehat{\theta}_\gamma) +
\frac{\vert y_i-x_i^T \widehat{\theta}_\gamma\vert}{1-\widehat{\alpha}_\gamma} I(y_i>x_i^T\widehat{\theta}_\gamma)\right)\right]^2.
\end{align}
Let us define
\begin{eqnarray*}
\rho(y_i,x_i,\theta_\gamma,\alpha)= \frac{\vert y_i-x_i^T \theta_\gamma\vert}{1+\alpha} I(y_i \leq x_i^T\theta_\gamma) +
\frac{\vert y_i-x_i^T \theta_\gamma\vert}{1-\alpha} I(y_i>x_i^T\theta_\gamma).
\end{eqnarray*}
Then, by the triangle inequality
\begin{eqnarray*}
\left\vert \sqrt{\widehat{\vartheta}_\gamma} - \sqrt{\vartheta_\gamma^*} \right\vert  \leq \left\vert \sqrt{\widehat{\vartheta}_\gamma} - \dfrac{1}{n}\sum_{i=1}^n\rho(y_i,x_i,\theta_\gamma^*,\alpha_\gamma^*) \right\vert + \left \vert \dfrac{1}{n}\sum_{i=1}^n\rho(y_i,x_i,\theta_\gamma^*,\alpha_\gamma^*) - \sqrt{\vartheta_\gamma^*} \right\vert.
\end{eqnarray*}
For the second term in the right-hand side of the last equation, it follows, by the law of large numbers and the continuous mapping theorem, that
$$\left\vert \dfrac{1}{n}\sum_{i=1}^n\rho(y_i,x_i,\theta_\gamma^*,\alpha_\gamma^*) - \sqrt{\vartheta_\gamma^*} \right\vert \stackrel{P}{\rightarrow} 0.$$
For the first term we have
\begin{eqnarray*}
\left\vert \sqrt{\widehat{\vartheta}_\gamma} - \dfrac{1}{n}\sum_{i=1}^n\rho(y_i,x_i,\theta_\gamma^*,\alpha_\gamma^*) \right\vert &=& \left\vert M_n(\widehat{\theta}_\gamma,1,\widehat{\alpha}_\gamma) - M_n(\theta_\gamma^*,1,\alpha_\gamma^*) \right\vert\\
&\leq& \left\vert M_n(\widehat{\theta}_\gamma,1,\widehat{\alpha}_\gamma) - M(\widehat{\theta}_\gamma,1,\widehat{\alpha}_\gamma) \right\vert \\
&+& \left\vert   M(\widehat{\theta}_\gamma,1,\widehat{\alpha}_\gamma) -  M(\theta_\gamma^*,1,\alpha_\gamma^*) \right\vert\\
&+& \left\vert   M(\theta_\gamma^*,1,\alpha_\gamma^*) - M_n(\theta_\gamma^*,1,\alpha_\gamma^*) \right\vert\\
&\leq& 2 \sup_{(\theta_\gamma,\alpha) \in \Gamma} \left\vert M_n(\theta_\gamma,1,\alpha) - M(\theta_\gamma,1,\alpha)\right\vert\\
&+& \left\vert   M(\widehat{\theta}_\gamma,1,\widehat{\alpha}_\gamma) -  M(\theta_\gamma^*,1,\alpha_\gamma^*) \right\vert.
\end{eqnarray*}
By using \eqref{SCL1}, the consistency of $(\widehat{\theta}_\gamma,\widehat{\alpha}_\gamma)$, and the continuous mapping theorem it follows that $\left\vert \sqrt{\widehat{\vartheta}_\gamma} - \dfrac{1}{n}\sum_{i=1}^n\rho(y_i,x_i,\theta_\gamma^*,\alpha_\gamma^*)\right\vert \stackrel{P}{\rightarrow} 0$. Consequently, $\widehat{\vartheta}_\gamma \stackrel{P}{\rightarrow} \vartheta_\gamma^*$, which completes the proof.

\subsection{Proof of Proposition \ref{prop:AsympNorm}}

\subsection*{\underline{Two-piece normal errors ($k=1$)}}
The proof technique consists of showing first that $\dot{m}_{\eta}(y_1,x_1)$ is dominated by an $L^2$ function (square integrable), $K(y_1,x_1)$, for $\eta$ in a neighborhood of $\eta_\gamma^*$. Then, we prove that the function $P m_{\eta}$ admits a second-order Taylor expansion at $\eta_\gamma^*$ and that the matrix $V_{\eta_\gamma^*}$ is nonsingular. Finally, we appeal to the consistency result in Proposition \ref{prop:Consistency} in order to apply Theorem 5.23 of \cite{vandervaart:1998} to prove the asymptotic normality of $\widehat{\eta}_\gamma$.

We first note that under assumptions A1--A4, where A4 is assumed to be satisfied for $j=4$ throughout, Proposition \ref{prop:Consistency} implies the existence and uniqueness of $\eta_\gamma^* $. The gradient of $m_{\eta}(y_1,x_1)$, which is given by (i) in Proposition \ref{prop:lhood_skewnorm_properties} (with $n=1$), is bounded for all $\eta\in\Gamma$ and for each $(y_1,x_1)$, due to the compactness of $\Gamma$. Now, a direct application of the Minkowski inequality implies that $\vert \vert \dot{m}_{\eta}(y_1,x_1) \vert \vert$ is upper bounded by the sum of the absolute values of the entries of $\dot{m}_{\eta}(y_1,x_1)$. Let us now define $K(y_1,x_1)=\sup_{\eta\in{\mathcal B}_{\eta_\gamma^*}} \vert \vert \dot{m}_{\eta}(y_1,x_1) \vert \vert$, where ${\mathcal B}_{\eta_\gamma^*}\subset \Gamma$ is any neighborhood of $\eta_\gamma^*$, whose projection over $\theta$ coincides with ${\mathcal B}_{\theta_\gamma^*}$. Thus, from the expression of $\dot{m}_{\eta}(y_1,x_1)$ together with assumption A4, it follows that
\begin{eqnarray*}
\int K(y_1,x_1)^2 dS_0(y_1\vert x_1) d\Psi(x_1) < \infty,
\end{eqnarray*}
Then, by using the mean value theorem and the Cauchy-Schwartz inequality, it follows that for $\eta_1, \eta_2\in{\mathcal B}_{\eta_0}$, with probability 1,
\begin{eqnarray*}
\vert m_{\eta_1}(y_1,x_1) - m_{\eta_2}(y_1,x_1) \vert &=& \vert \dot{m}_{\eta_\star}(y_1,x_1)^{T} (\eta_1-\eta_2)\vert\\
&\leq& \vert \vert \dot{m}_{\eta_\star}(y_1,x_1) \vert \vert \cdot \vert \vert \eta_1-\eta_2 \vert \vert
\\&\leq& K(y_1,x_1) \cdot \vert \vert \eta_1-\eta_2 \vert \vert,
\end{eqnarray*}
where $\eta_\star = (1-c)\eta_1+c\eta_2$, for some $c\in(0,1)$.

Now, for each $ x_1$:
\begin{eqnarray*}
P m_{\eta\vert x_1} = {\mathbb E}[m_{\eta}\vert  x_1] &=& - \dfrac{1}{2}\log(2\pi) -\dfrac{1}{2}\log(\vartheta) \\ &-&\dfrac{1}{2\vartheta(1+\alpha)^2}\int_{-\infty}^{ x_1^{T}\theta_\gamma} (y_1- x_1^{T}\theta_\gamma)^2 dS_0(y_1\vert  x_1) \\
&-&\dfrac{1}{2\vartheta(1-\alpha)^2}\int_{ x_1^{T}\theta_\gamma}^{\infty} (y_1- x_1^{T}\theta_\gamma)^2  dS_0(y_1\vert  x_1).
\end{eqnarray*}
Thus, the gradient of $P m_{\eta\vert x_1}$ is given by
\begin{eqnarray*}
\dfrac{\partial}{\partial \theta_\gamma}P m_{\eta\vert x_1} &=& -\dfrac{ x_1}{\vartheta(1+\alpha)^2}I_1 + \dfrac{ x_1}{\vartheta(1-\alpha)^2}I_2,\\
\dfrac{\partial}{\partial \vartheta}P m_{\eta\vert x_1} &=& -\dfrac{1}{2\vartheta} + \dfrac{I_3}{2\vartheta^2(1+\alpha)^2}  + \dfrac{I_4}{2\vartheta^2(1-\alpha)^2},\\
\dfrac{\partial}{\partial \alpha}P m_{\eta\vert x_1} &=&  \dfrac{I_3}{\vartheta(1+\alpha)^3} - \dfrac{I_4}{\vartheta(1-\alpha)^3},
\end{eqnarray*}
Then, the second derivative matrix is given by
\begin{eqnarray*}
\dfrac{\partial^2}{\partial \theta_\gamma^2}P m_{\eta\vert x_1} &=& -\dfrac{ x_1  x_1^{T}[(1+\alpha)^2-4\alpha S_0( x_1^{T} \theta_\gamma\vert  x_1)]}{\vartheta(1-\alpha^2)^2},\\
\dfrac{\partial^2}{\partial \vartheta^2}P m_{\eta\vert x_1} &=& \dfrac{1}{2\vartheta^2} - \dfrac{I_3}{\vartheta^3(1+\alpha)^2}- \dfrac{I_4}{\vartheta^3(1-\alpha)^2},\\
\dfrac{\partial^2}{\partial \alpha^2}P m_{\eta\vert x_1} &=& -\dfrac{3 I_3}{\vartheta (1+\alpha)^4} - \dfrac{3 I_4}{\vartheta (1-\alpha)^4},\\
\dfrac{\partial^2}{\partial \vartheta \partial \theta_\gamma}P m_{\eta\vert x_1} &=& \dfrac{ x_1}{\vartheta^2(1+\alpha)^2}I_1 - \dfrac{ x_1}{\vartheta^2(1-\alpha)^2}I_2,\\
\dfrac{\partial^2}{\partial \alpha \partial \theta_\gamma}P m_{\eta\vert x_1} &=& \dfrac{2 x_1}{\vartheta(1+\alpha)^3}I_1 + \dfrac{2 x_1}{\vartheta(1-\alpha)^3}I_2,\\
\dfrac{\partial^2}{\partial \vartheta \partial \alpha }P m_{\eta\vert x_1} &=& -\dfrac{I_3}{\vartheta^2(1+\alpha)^3} + \dfrac{I_4}{\vartheta^2(1-\alpha)^3},
\end{eqnarray*}
where $I_1 = \int_{-\infty}^{ x_1^T\theta_\gamma} S_0(y_1\vert x_1)dy_1$, and $I_2 = \int_{ x_1^T\theta_\gamma}^{\infty} \left[1-S_0(y_1\vert x_1)\right]dy_1$, $I_3 = \int_{-\infty}^{ x_1^{T}\theta_\gamma} (y_1- x_1^{T}\theta_\gamma)^2 dS_0(y_1\vert  x_1)$, and $I_4 = \int_{ x_1^{T}\theta_\gamma}^{\infty} (y_1- x_1^{T}\theta_\gamma)^2  dS_0(y_1\vert  x_1)$. These entries are finite for all $\eta \in \Gamma$ by assumption A4. Note that $P m_{\eta} = {\mathbb E}[P m_{\eta\vert x_1}]$, where the expectation is taken over $x_1$. Assumptions A1--A4 together with Proposition \eqref{prop:Consistency} imply that $P m_{\eta}$ is finite and that this expectation is concave and has a unique maximum at $\eta_\gamma^*$. From assumption A5,
\begin{eqnarray*}
\dfrac{\partial}{\partial \theta_\gamma} P m_{\eta}\Bigg\vert_{\eta=\eta_\gamma^*} &=& {\mathbb E}\left[ \dfrac{\partial}{\partial \theta_\gamma}P m_{\eta\vert x_1}\right]\Bigg\vert_{\eta=\eta_\gamma^*} = 0,\\
\dfrac{\partial}{\partial \alpha} P m_{\eta}\Bigg\vert_{\eta=\eta_\gamma^*} &=& {\mathbb E}\left[ \dfrac{\partial}{\partial \alpha}P m_{\eta\vert x_1}\right]\Bigg\vert_{\eta=\eta_\gamma^*} = 0 ,
\end{eqnarray*}
which in turn imply that $\dfrac{\partial^2}{\partial \vartheta \partial \theta_\gamma}P m_{\eta}=0$ and $\dfrac{\partial^2}{\partial \vartheta \partial \alpha }P m_{\eta} = 0$ at $\eta = \eta_\gamma^*$. Thus, the matrix of second derivatives evaluated at $\eta_\gamma^*$ has the following structure:
\begin{eqnarray*}
V_{\eta} = \left( \begin{array}{ccc}
\dfrac{\partial^2}{\partial \theta_\gamma^2}P m_{\eta} & 0 & \dfrac{\partial^2}{\partial \vartheta \partial \alpha }P m_{\eta} \\
0 & \dfrac{\partial^2}{\partial \vartheta^2}P m_{\eta} & 0 \\
\dfrac{\partial^2}{\partial \vartheta \partial \alpha }P m_{\eta} & 0 & \dfrac{\partial^2}{\partial \alpha^2}P m_{\eta} \end{array} \right).
\end{eqnarray*}
Consequently, the determinant of this matrix is given by
\begin{eqnarray*}
\operatorname{det}V_{\eta} = \dfrac{\partial^2}{\partial \vartheta^2}P m_{\eta}\times \operatorname{det}\left( \begin{array}{cc}
\dfrac{\partial^2}{\partial \theta_\gamma^2}P m_{\eta} & \dfrac{\partial^2}{\partial \vartheta \partial \alpha }P m_{\eta} \\
\dfrac{\partial^2}{\partial \vartheta \partial \alpha }P m_{\eta} & \dfrac{\partial^2}{\partial \alpha^2}P m_{\eta} \end{array} \right).
\end{eqnarray*}
The determinant on the right-hand side of this expression, evaluated at $\eta_\gamma^*$, is non-zero since the $P m_{\eta}$ is concave with respect to $(\theta_\gamma,\alpha)$, as shown in Proposition \ref{prop:Consistency}. Moreover, the fact that the first derivative $\dfrac{\partial}{\partial \vartheta}P m_{\eta} =0$ at $\eta=\eta_\gamma^*$ together with the fact that $\eta_\gamma^*$ is the unique maximizer implies that $ \dfrac{\partial^2}{\partial \vartheta^2}P m_{\eta}\neq 0$. Consequently, the matrix of second derivatives of $P m_{\eta}$ is nonsingular at $\eta_\gamma^*$. The asymptotic normality result follows by Theorem 5.23 from \cite{vandervaart:1998}.

\subsection*{\underline{Two-piece Laplace errors ($k=2$)}}
First, we note that under assumptions A1--A4, where $j=2$ in A4 throughout, Proposition \ref{prop:Consistency} implies the existence and uniqueness of $\eta_\gamma^* $. The gradient of $m_{\eta}(y_1,x_1)$, which is given by (i) in Proposition \ref{prop:lhood_skewlap_properties} (with $n=1$), is bounded for almost all $\eta\in\Gamma$ and for each $(y_1,x_1)$, due to the compactness of $\Gamma$. Now, a direct application of the Minkowski inequality implies that $\vert \vert \dot{m}_{\eta}(y_1,x_1) \vert \vert$ is upper bounded almost surely by the sum of the absolute values of the entries of $\dot{m}_{\eta}(y_1,x_1)$. Let us now define $K(y_1,x_1)=\sup_{\eta\in{\mathcal B}_{\eta_\gamma^*}} \vert \vert \dot{m}_{\eta}(y_1,x_1) \vert \vert$, where ${\mathcal B}_{\eta_\gamma^*}\subset \Gamma$ is any neighborhood of $\eta_\gamma^*$, whose projection over $\theta_\gamma$ coincides with ${\mathcal B}_{\theta_\gamma^*}$. Thus, from the expression of $\dot{m}_{\eta}(y_1,x_1)$ together with assumption A4, it follows that
\begin{eqnarray*}
\int K(y_1,x_1)^2 dS_0(y_1\vert x_1) d\Psi(x_1) < \infty,
\end{eqnarray*}
Then, by using the mean value theorem and the Cauchy-Schwartz inequality, it follows that for $\eta_1, \eta_2\in{\mathcal B}_{\eta_\gamma^*}$, with probability 1,
\begin{eqnarray*}
\vert m_{\eta_1}(y_1,x_1) - m_{\eta_2}(y_1,x_1) \vert &=& \vert \dot{m}_{\eta_\star}(y_1,x_1)^{T} (\eta_1-\eta_2)\vert\\
&\leq& \vert \vert \dot{m}_{\eta_\star}(y_1,x_1) \vert \vert \cdot \vert \vert \eta_1-\eta_2 \vert \vert
\\&\leq& K(y_1,x_1) \cdot \vert \vert \eta_1-\eta_2 \vert \vert,
\end{eqnarray*}
where $\eta_\star = (1-c)\eta_1+c\eta_2$, for some $c\in(0,1)$.

Now, for each $ x_1$:
\begin{eqnarray*}
P m_{\eta\vert x_1} = {\mathbb E}[m_{\eta}\vert x_1] &=& - \log(2) - \dfrac{1}{2}\log(\vartheta) -\dfrac{1}{\sqrt{\vartheta}(1+\alpha)}\int_{-\infty}^{x_1^{T}\theta_\gamma} S_0(y_1\vert x_1)dy_1 \\
&-&\dfrac{1}{\sqrt{\vartheta}(1-\alpha)}\int_{x_1^{T}\theta_\gamma}^{\infty} 1-S_0(y_1\vert x_1)dy_1.
\end{eqnarray*}
Then, the gradient of $P m_{\eta\vert x_1}$ is given by
\begin{eqnarray*}
\dfrac{\partial}{\partial \theta_\gamma}P m_{\eta\vert x_1} &=& -\dfrac{x_1 S_0(x_1^{T}\theta_\gamma \vert x_1)}{\sqrt{\vartheta}(1+\alpha)} + \dfrac{x_1 [1- S_0(x_1^{T}\theta_\gamma \vert x_1)]}{\sqrt{\vartheta}(1-\alpha)},\\
\dfrac{\partial}{\partial \vartheta}P m_{\eta\vert x_1} &=& -\dfrac{1}{2\vartheta} + \dfrac{I_1}{2\vartheta^{3/2}(1+\alpha)} + \dfrac{I_2}{2\vartheta^{3/2}(1-\alpha)},\\
\dfrac{\partial}{\partial \alpha}P m_{\eta\vert x_1} &=&  \dfrac{I_1}{\sqrt{\vartheta}(1+\alpha)^2} - \dfrac{I_2}{\sqrt{\vartheta}(1-\alpha)^2},
\end{eqnarray*}
where $I_1 = \int_{-\infty}^{x_1^T\theta_\gamma} S_0(y_1\vert x_1)dy$, and $I_2 = \int_{x_1^T\theta_\gamma}^{\infty} 1-S_0(y_1\vert x_1)dy_1$, which are finite by assumption A4. Then, the second derivative matrix is given by
\begin{eqnarray*}
\dfrac{\partial^2}{\partial \theta_\gamma^2}P m_{\eta\vert x_1} &=& -\dfrac{2x_1x_1^{T} s_0(x_1^{T}\theta_\gamma \vert x_1)}{\sqrt{\vartheta}(1-\alpha^2)},\\
\dfrac{\partial^2}{\partial \vartheta^2}P m_{\eta\vert x_1} &=& \dfrac{1}{2\vartheta^2} - \dfrac{3I_1}{4\vartheta^{5/2}(1+\alpha)} - \dfrac{3I_2}{4\vartheta^{5/2}(1-\alpha)},\\
\dfrac{\partial^2}{\partial \alpha^2}P m_{\eta\vert x_1} &=&  -\dfrac{2I_1}{\sqrt{\vartheta}(1+\alpha)^3} - \dfrac{2I_2}{\sqrt{\vartheta}(1-\alpha)^3},\\
\dfrac{\partial^2}{\partial \vartheta \partial \theta_\gamma}P m_{\eta\vert x_1} &=& \dfrac{x_1 S_0(x_1^{T}\theta_\gamma \vert x_1)}{2\vartheta^{3/2}(1+\alpha)} - \dfrac{x_1 [1- S_0(x_1^{T}\theta_\gamma \vert x_1)]}{2\vartheta^{3/2}(1-\alpha)},\\
\dfrac{\partial^2}{\partial \alpha \partial \theta_\gamma}P m_{\eta\vert x_1} &=& \dfrac{x_1 S_0(x_1^{T}\theta_\gamma \vert x_1)}{\sqrt{\vartheta}(1+\alpha)^2} + \dfrac{x_1 [1- S_0(x_1^{T}\theta_\gamma \vert x_1)]}{\sqrt{\vartheta}(1-\alpha)^2},\\
\dfrac{\partial^2}{\partial \vartheta \partial \alpha }P m_{\eta\vert x_1} &=& -\dfrac{I_1}{2\vartheta^{3/2}(1+\alpha)^2} + \dfrac{I_2}{2\vartheta^{3/2}(1-\alpha)^2}.
\end{eqnarray*}
These entries are finite for all $\eta \in \Gamma$ by assumption A4. Note that $P m_{\eta} = {\mathbb E}[P m_{\eta\vert x_1}]$, where the expectation is taken over $x_1$. Assumptions A1--A4 together with Proposition \ref{prop:Consistency}, imply that $P m_{\eta}$ is finite and that this expectation is concave and has a unique maximum at $\eta_\gamma^*$. From assumption A5,
\begin{eqnarray*}
\dfrac{\partial}{\partial \theta_\gamma} P m_{\eta}\Bigg\vert_{\eta=\eta_\gamma^*} &=& {\mathbb E}\left[ \dfrac{\partial}{\partial \theta_\gamma}P m_{\eta\vert x_1}\right]\Bigg\vert_{\eta=\eta_\gamma^*} = 0,\\
\dfrac{\partial}{\partial \alpha} P m_{\eta}\Bigg\vert_{\eta=\eta_\gamma^*} &=& {\mathbb E}\left[ \dfrac{\partial}{\partial \alpha}P m_{\eta\vert x_1}\right]\Bigg\vert_{\eta=\eta_\gamma^*} = 0 ,
\end{eqnarray*}
which in turn imply that $\dfrac{\partial^2}{\partial \vartheta \partial \theta_\gamma}P m_{\eta}=0$ and $\dfrac{\partial^2}{\partial \vartheta \partial \alpha }P m_{\eta} = 0$ at $\eta=\eta_\gamma^*$. Thus, it follows that the matrix of second derivatives evaluated at $\eta_\gamma^*$ has the structure:
\begin{eqnarray*}
V_{\eta} = \left( \begin{array}{ccc}
\dfrac{\partial^2}{\partial \theta_\gamma^2}P m_{\eta} & 0 & \dfrac{\partial^2}{\partial \vartheta \partial \alpha }P m_{\eta} \\
0 & \dfrac{\partial^2}{\partial \vartheta^2}P m_{\eta} & 0 \\
\dfrac{\partial^2}{\partial \vartheta \partial \alpha }P m_{\eta} & 0 & \dfrac{\partial^2}{\partial \alpha^2}P m_{\eta} \end{array} \right).
\end{eqnarray*}
Consequently, the determinant of this matrix is given by
\begin{eqnarray*}
\operatorname{det}V_{\eta} = \dfrac{\partial^2}{\partial \vartheta^2}P m_{\eta}\times \operatorname{det}\left( \begin{array}{cc}
\dfrac{\partial^2}{\partial \theta_\gamma^2}P m_{\eta} & \dfrac{\partial^2}{\partial \vartheta \partial \alpha }P m_{\eta} \\
\dfrac{\partial^2}{\partial \vartheta \partial \alpha }P m_{\eta} & \dfrac{\partial^2}{\partial \alpha^2}P m_{\eta} \end{array} \right).
\end{eqnarray*}
The determinant on the right-hand side of this expression, evaluated at $\eta_\gamma^*$, is non-zero since the $P m_{\eta}$ is concave with respect to $(\theta_\gamma,\alpha)$, as shown in Proposition \ref{prop:Consistency}. Moreover, the fact that the first derivative $\dfrac{\partial}{\partial \vartheta}P m_{\eta} =0$ at $\eta=\eta_\gamma^*$ together with the fact that $\eta_\gamma^*$ is the unique maximizer implies that $ \dfrac{\partial^2}{\partial \vartheta^2}P m_{\eta}\neq 0$. Consequently, the matrix of second derivatives of $P m_{\eta}$ is nonsingular at $\eta_\gamma^*$. The asymptotic normality result follows by Theorem 5.23 from \cite{vandervaart:1998}.

\subsection{Proof of Corollary \ref{cor:asymp_mle_misspec}}
\label{proof:asymp_mle_lapl}
The result when $\epsilon_i\sim L(0,\vartheta)$ follows directly from \cite{pollard:1991} Theorem 1, hence it suffices to find the
expression for $f_0$ under each assumed residual distribution.
The median for a general two-piece distribution with a mode at 0 is given by
$\sqrt{\vartheta} (1+\alpha) F^{-1}\left( \frac{1}{2(1+\alpha)}  \right)$ if $\alpha >0$
and $\sqrt{\vartheta} (1-\alpha) F^{-1} \left( \frac{(1-2\alpha)}{2(1-\alpha)} \right)$ if $\alpha\leq 0$,
where $F(\cdot)$ is the cdf of the standard underlying distribution with mode 0, $\vartheta=1$
(\cite{arellanovalle:2005}, Expression (9)).

When $\epsilon_i \sim N(0,\vartheta)$ we have $m=0$ and hence $f_0=N(0;0,\vartheta)=1/(\sqrt{2\pi\vartheta})$.
When $\epsilon_i \sim \mbox{AN}(0,\vartheta,\alpha)$ we have
$m=\sqrt{\vartheta}(1+\alpha) \Phi^{-1}(0.5/(1+\alpha))$ if $\alpha>0$
and $m=\sqrt{\vartheta}(1-\alpha) \Phi^{-1}(0.5(1-2 \alpha)/(1-\alpha))$ if $\alpha <0$,
where $\Phi^{-1}(\cdot)$ is the inverse standard cdf,
and hence
$f_0= \exp \left\{ -\frac{1}{2} \left( \Phi^{-1}\left(\frac{0.5}{1+|\alpha|}\right)  \right)^2 \right\} \frac{1}{\sqrt{2\pi\vartheta}}$.
For the Laplace and Asymmetric Laplace, we note that the inverse cdf of the standard Laplace distribution
evaluated at a quantile $q \in [0,1]$ is
$F^{-1}(q)= \log(2q)$ if $q<0.5$ and $F^{-1}(q)= - \log(2(1-q))$ if $q \geq 0.5$.
When $\epsilon_i \sim L(0,\vartheta)$ we have $m=0$ and $f_0=1/(2\sqrt{\vartheta})$.
Finally, when $\epsilon_i \sim \mbox{AL}(0,\vartheta,\alpha)$ we have
$m=-\sqrt{\vartheta}(1+\alpha) \log (1+\alpha)$ if $\alpha>0$ and
$m=\sqrt{\vartheta}(1-\alpha) \log (1-\alpha)$ if $\alpha<0$,
from which it follows that $f_0= \frac{1}{2\sqrt{\vartheta}} \exp \left\{ -\log(1+|\alpha|)  \right\}= \frac{1}{2\sqrt{\vartheta} (1+|\alpha|)}$.

The results for the true Normal model follows by using classic asymptotic results on least square estimators (see \emph{e.g.} \cite{newey:1987})

\subsection{Proof of Proposition \ref{prop:bfrates}}
\label{ssec:proof_bfrates}

We provide the proof for the asymmetric Normal and asymmetric Laplace ($\alpha \neq 0$), their symmetric counterparts follow as particular cases.
$\eta_\gamma=(\theta_\gamma,\vartheta_\gamma,\alpha_\gamma)$ denotes the parameter vector under model $\gamma$,
$\widehat{\eta}_\gamma$ the MLE and $\tilde{\eta}_\gamma$ the posterior mode for a given observed $(y,X)$.
Further, $M_k(\eta_\gamma)= E (\log L_k(\eta_\gamma))$ where the expectation is with respect to the data-generating truth
and $\eta_\gamma^*= \arg\max_{\eta \in \Gamma_\gamma} M_k(\eta_\gamma)$ is the optimal parameter value under $\gamma$.
We wish to characterize the asymptotic behaviour of the Laplace-approximated Bayes factors
\begin{align}
\frac{\widehat{p}(y \mid \gamma)}{\widehat{p}(y \mid \gamma^*)}=
e^{\log L_k(\tilde{\eta}_\gamma) - \log L_k(\tilde{\eta}_{\gamma^*})} \times
\frac{p(\tilde{\eta}_\gamma \mid \gamma)}{p(\tilde{\eta}_{\gamma^*} \mid \gamma^*)} \times
(2 \pi)^{\frac{p_\gamma - p_{\gamma^*}}{2}} \times
\frac{ \left| H_k(\tilde{\eta}_{\gamma^*}) \right|^{\frac{1}{2}}}{\left| H_k(\tilde{\eta}_{\gamma}) \right|^{\frac{1}{2}}},
\label{eq:laplaceapprox_bf}
\end{align}
when $(y,X)$ arise from the data-generating model in Condition A1, which may differ from the assumed model.
{\textcolor{black}{
The term $(2 \pi)^{\frac{p_\gamma - p_{\gamma^*}}{2}}$ is a constant since $p_\gamma$ and $p_{\gamma^*}$ are fixed.
The expression for $H_1$ is given by \eqref{eq:skewnorm_logl_hessian_repar}
and recall that for $H_2$ we are taking the asymptotic covariance in \eqref{eq:alaplace_expected_hessian}.
Hence
$$
|H_2(\tilde{\eta}_\gamma)|=n^{p_\gamma} \left\vert\frac{1}{n} H_2(\tilde{\eta}_\gamma)\right\vert=
n^{p_\gamma}
\begin{vmatrix}
-\frac{1}{n}X^TX \dfrac{1}{\tilde{\vartheta}_\gamma(1-\tilde{\alpha}_\gamma^2)} &
\frac{\overline{x}}{\sqrt{\tilde{\vartheta}_\gamma}(1-\tilde{\alpha}_\gamma^2)}
& 0\\
\frac{\overline{x}}{\sqrt{\tilde{\vartheta}_\gamma}(1-\tilde{\alpha}_\gamma^2)} &
-\frac{1}{4 \tilde{\vartheta}_\gamma^2} & 0 \\
0 & 0 & -\frac{2}{1-\tilde{\alpha}_\gamma^2}
\end{vmatrix}.
$$
The determinant converges in probability to a negative constant
since $\tilde{\eta}_\gamma \stackrel{P}{\longrightarrow} \eta_\gamma^*$ by Proposition \ref{prop:Consistency}, together with
the continuous mapping theorem and the asymptotic Hessian (the limiting $-H_2$) being positive definite.
An analogous argument applies to $H_1$, hence
$n^{\frac{p_\gamma-p_{\gamma^*}}{2}} |H_k(\tilde{\eta}_{\gamma^*})|^{\frac{1}{2}}/ | H_k(\tilde{\eta}_{\gamma})|^{\frac{1}{2}}
\stackrel{P}{\longrightarrow} \tilde{a}_3$
for some constant $\tilde{a}_3>0$. In other words, $|H_k(\tilde{\eta}_{\gamma^*})|^{\frac{1}{2}}/ | H_k(\tilde{\eta}_{\gamma})|^{\frac{1}{2}}
= O_p \left(n^{\frac{p_{\gamma^*}-p_\gamma}{2}}\right)$.
}

The proof strategy is to first show that when $M_k(\eta_{\gamma}^*)-M_k(\eta_{\gamma^*}^*)<0$ ({\it i.e.} $\gamma^* \not\subset \gamma$)
the log-first term of the right hand in \eqref{eq:laplaceapprox_bf}
{\textcolor{black}{
behaves asymptotically in probability as $-na_1$, for some constant $a_1>0$,
and the logarithm of the second term converges in probability to a constant $a_2$. Thus,
$$\frac{1}{n} \log\left(\frac{\widehat{p}(y \mid \gamma)}{\widehat{p}(y \mid \gamma^*)}\right)
= -a_1(1+o_p(1)) + \frac{1}{n} \left( a_2 + \frac{p_\gamma-p_{\gamma^*}}{2} ( a_3 - \log(n)) \right) \stackrel{P}{\longrightarrow} -a_1$$
where $a_3=\log(2\pi) - \log(\tilde{a}_3)$, as we wish to prove.
}}
Subsequently we shall show that when $M_k(\eta_{\gamma}^*)-M_k(\eta_{\gamma^*}^*)=0$ (the case $\gamma^* \subset \gamma$)
the first term is essentially the likelihood ratio test statistic and is $O_p(1)$,
whereas, analogously to the results in \cite{johnson:2010} and \cite{rossell:2017},
the second term converges to a positive constant under local priors,
but it is $O_p(\tilde{b}_n)$ where $\tilde{b}_n= n^{p_{\gamma^*}-p_\gamma}$ under the pMOM prior
and $\tilde{b}_n= e^{-c\sqrt{n}}$ for some $c>0$ under the peMOM prior.
{\textcolor{black}{
This gives
$$
\frac{\widehat{p}(y \mid \gamma)}{\widehat{p}(y \mid \gamma^*)}=
e^{O_p(1)} O_p(\tilde{b}_n) O_p\left(n^{\frac{p_{\gamma^*}-p_\gamma}{2}}\right)=
O_p(b_n)
$$
}}
where $b_n= n^{\frac{p_{\gamma^*}-p_\gamma}{2}}$ for local priors,
$b_n= n^{3(p_{\gamma^*}-p_\gamma)/2}$ for the pMOM prior and $b_n= e^{-c\sqrt{n}} n^{\frac{p_{\gamma^*}-p_\gamma}{2}}$ for the peMOM prior, as we wish to prove.

Consider first the case when $\gamma^* \not\subset \gamma$, which implies $M_k(\eta_\gamma^*) - M_k(\eta_{\gamma^*}^*)<0$.
Then by continuity of $p(\eta_\gamma \mid \gamma)$ we have that
$p(\tilde{\eta}_\gamma \mid \gamma) \stackrel{P}{\longrightarrow} p(\eta_\gamma^* \mid \gamma) \geq 0$,
and analogously $p(\tilde{\eta}_{\gamma^*} \mid \gamma^*) \stackrel{P}{\longrightarrow} p(\eta^*_{\gamma^*} \mid \gamma^*)>0$
(strict positivity is ensured by the assumption of prior positivity at $\eta^*_{\gamma^*}$).
{\textcolor{black}{
Hence
$
p(\tilde{\eta}_\gamma \mid \gamma)/p(\tilde{\eta}_{\gamma^*} \mid \gamma^*) \stackrel{P}{\longrightarrow} a_2
$
for some constant $a_2 \geq 0$.
Note that $a_2=0$ when $\theta_\gamma^*$ contains some zeroes and hence a non-local prior would take the value
$p(\eta_\gamma^* \mid \gamma)=0$,
but this gives even faster Bayes factor rates in favor of $\gamma^*$.
}}
Regarding $\log L_k(\tilde{\eta}_\gamma) - \log L_k(\tilde{\eta}_{\gamma^*})$,
the law of large numbers and uniform convergence of $\log L_k$ to its expected value shown in Proposition \ref{prop:Consistency} give that
\begin{align}
\frac{1}{n} (\log L_k(\tilde{\eta}_\gamma) - \log L_k(\tilde{\eta}_{\gamma^*})) \stackrel{P}{\longrightarrow}
(M_k(\eta_\gamma^*) - M_k(\eta_{\gamma^*}^*)) < 0,
\label{eq:lln_lrt}
\end{align}
{\textcolor{black}{
hence the constant $a_1$ defined above is $a_1=M_k(\eta_{\gamma^*}^*) - M_k(\eta_\gamma^*) >0$.
}}

Next consider the case when $\gamma^* \subset \gamma$, which implies $M_k(\eta_\gamma^*) - M_k(\eta_{\gamma^*}^*)=0$.
Since $\tilde{\eta}_\gamma \stackrel{P}{\longrightarrow} \eta_\gamma^*$ by Proposition \ref{prop:Consistency},
we have that under a local prior
\begin{align}
\frac{p(\tilde{\eta}_\gamma \mid \gamma)}{p(\tilde{\eta}_{\gamma^*} \mid \gamma^*)} \stackrel{P}{\longrightarrow}
\frac{p(\eta_\gamma \mid \gamma)}{p(\eta_{\gamma^*} \mid \gamma^*)}>0.
\label{eq:limit_ratiopriors}
\end{align}
Under a non-local prior we still have
$p(\eta_{\gamma^*} \mid \gamma^*)>0$ but in contrast $p(\eta_\gamma \mid \gamma)=0$.
Thus, it is necessary to characterize the rate at which the latter term vanishes.
Briefly, following the proof of Theorem 1 in \cite{koenker:1982}, the fact that $\log L_k$ converges uniformly to its expectation
(see the proof of our Proposition \ref{prop:Consistency})
and consistency of $\tilde{\eta}_\gamma \stackrel{P}{\longrightarrow} \eta_\gamma^*$ give that $\log L_k$ can be approximated
by a quadratic function plus a term that is $o_p(1)$.
Then, the argument leading to \cite{rossell:2017}, Proposition 2(i), gives that
$\tilde{\theta}_{\gamma j}-\widehat{\theta}_{\gamma j}=O_p(1/n)$ and thus $\tilde{\theta}_{\gamma j}=O_p(n^{-1/2})$ under the pMOM prior $p_M$,
whereas $\tilde{\theta}_{\gamma j}=O_p(n^{-1/4})$ under the peMOM prior $p_E$.
It follows that $\pi_M(\tilde{\theta}_\gamma)= O_p(1) \prod_{\theta_{\gamma j}^* \neq 0}^{} \tilde{\theta}_{\gamma j}^2= O_p(n^{-(p_\gamma - p_{\gamma^*})})$,
and $\pi_E(\tilde{\eta})= O_p(1) \prod_{\theta_{\gamma j}^* \neq 0}^{} e^{O_p(1)/\tilde{\theta}_{\gamma}^2}= O_p(e^{-c\sqrt{n}})$ for some $c>0$, as desired.

To conclude the proof,
since $\log L_k(\tilde{\eta}_\gamma) - \log L_k(\tilde{\eta}_{\gamma^*})= \lambda(y) + o_p(1)$
where $\lambda(y)=\log L_k(\widehat{\eta}_\gamma) - \log L_k(\widehat{\eta}_{\gamma^*})$ is the likelihood ratio (LR) statistic,
it only remains to show that $\lambda(y)=O_p(1)$.
The strategy is to see that $\lambda(y)= \lambda(y;\vartheta_\gamma^*)(1 + o_p(1))$,
where $\lambda(y;\vartheta_\gamma^*)=
\log L_k(\widehat{\theta}_\gamma,\vartheta_\gamma^*,\widehat{\alpha}_\gamma) - \log L_k(\widehat{\theta}_{\gamma^*},\vartheta_\gamma^*,\widehat{\alpha}_\gamma)$
is the LR obtained by plugging in the oracle $\vartheta_\gamma^*=\vartheta_{\gamma^*}^*$,
then use classical results to prove that $\lambda(y;\vartheta_\gamma^*)=O_p(1)$.
Taking derivatives of the likelihoods
(Expressions (\ref{eq:skewnorm_loglhood}) and (\ref{eq:skewlap_loglhood}) in the main paper)
shows that for $k=1$ the MLE must satisfy
$$
\widehat{\vartheta}_\gamma= \frac{1}{n} \left(
\sum_{i \in A(\theta)}^{} \frac{(y_i-x_i^T \widehat{\theta}_\gamma)^2}{(1+\hat{\alpha})^2} +
\sum_{i \not\in A(\theta)}^{} \frac{(y_i-x_i^T \widehat{\theta}_\gamma)^2}{(1-\hat{\alpha})^2}
 \right)=
\frac{1}{n} (y-X_\gamma \widehat{\theta}_\gamma)^T W_{\hat{\theta}_\gamma,\hat{\alpha}}^2 (y - X_\gamma \widehat{\theta}_\gamma)
,
$$
whereas for $k=2$ it satisfies
$$
\widehat{\vartheta}_\gamma^{\frac{1}{2}}= \frac{1}{n} \left(
\sum_{i \in A(\theta)}^{} \frac{|y_i-x_i^T \widehat{\theta}_\gamma|}{(1+\hat{\alpha})} +
\sum_{i \not\in A(\theta)}^{} \frac{|y_i-x_i^T \widehat{\theta}_\gamma|}{(1-\hat{\alpha})}
 \right)=
\frac{1}{n} | W_{\hat{\theta}_\gamma,\hat{\alpha}}^{\frac{1}{2}} (y-X_\gamma \widehat{\theta}_\gamma)|.
$$
Plugging $\widehat{\vartheta}_\gamma$ into the likelihoods gives 
\begin{align}
\lambda(y)&= -\frac{n}{2} \log \left( \frac{\widehat{\vartheta}_\gamma}{\widehat{\vartheta}_{\gamma^*}} \right)=
-\frac{n}{2} \log \left( 1 + \frac{\widehat{\vartheta}_{\gamma} - \widehat{\vartheta}_{\gamma^*}}{\widehat{\vartheta}_{\gamma^*}} \right)=
-\frac{n}{2} \frac{\widehat{\vartheta}_{\gamma} - \widehat{\vartheta}_{\gamma^*}}{\widehat{\vartheta}_{\gamma^*}} (1 + o_p(1)) \nonumber \\
&= -\frac{n}{2} \frac{\widehat{\vartheta}_{\gamma} - \widehat{\vartheta}_{\gamma^*}}{\vartheta_{\gamma^*}^*} (1 + o_p(1))=
\lambda(y; \vartheta_\gamma^*) (1 + o_p(1))
\label{eq:lrt_mlevartheta}
\end{align}
since by Proposition \ref{prop:Consistency} $\widehat{\vartheta}_{\gamma^*} \stackrel{P}{\longrightarrow} \vartheta_{\gamma^*}^* >0$
and $(\widehat{\vartheta}_{\gamma} - \widehat{\vartheta}_{\gamma^*})/\widehat{\vartheta}_{\gamma^*} \stackrel{P}{\longrightarrow} 0$.

{\textcolor{black}{Finally we show that $\lambda(y; \vartheta_\gamma^*)=O_p(1)$,
which implies $\lambda(y; \vartheta_\gamma^*) (1 + o_p(1))= O_p(1)$ and completes the proof.
For ease of notation when $k=1$ define
$Z_n(\gamma)= (y - X_\gamma \widehat{\theta}_\gamma)^T W^2_{\hat{\theta}_\gamma,\hat{\alpha}} (y-X_\gamma \widehat{\theta}_\gamma)$
and $Z(\gamma)= (y - X_\gamma \widehat{\theta}_\gamma)^T W^2_{\theta_\gamma^*,\alpha^*} (y-X_\gamma \widehat{\theta}_\gamma)$,
and when $k=2$ let
$Z_n(\gamma)= |W_{\hat{\theta}_\gamma,\hat{\alpha}}^{\frac{1}{2}}(y - X_\gamma \widehat{\theta}_\gamma)|$,
$Z(\gamma)= |W_{\theta^*_\gamma,\alpha^*}^{\frac{1}{2}}(y - X_\gamma \widehat{\theta}_\gamma)|$.
Then by definition
\begin{align}
\lambda(y; \vartheta_\gamma^*)= \frac{Z_n(\gamma^*) - Z_n(\gamma)}{2\vartheta_{\gamma}^*}.
\end{align}
Now, note that
$Z_n(\gamma)= Z(\gamma) + Z(\gamma) (Z_n(\gamma)-Z(\gamma))/Z(\gamma)= Z(\gamma) (1+o_p(1))$,
since Proposition \ref{prop:Consistency} gives that
$\frac{1}{n} Z_n(\gamma) \stackrel{P}{\longrightarrow} \vartheta_\gamma$,
$\frac{1}{n} Z(\gamma) \stackrel{P}{\longrightarrow} \vartheta_\gamma$
and hence $(Z_n(\gamma)-Z(\gamma))/Z(\gamma) \stackrel{P}{\longrightarrow} 0$.
Following the same argument $Z_n(\gamma^*)= Z(\gamma^*) (1+o_p(1))$, hence
\begin{align}
\lambda(y; \vartheta_\gamma^*)=
\frac{Z(\gamma^*) - Z(\gamma) +o_p(Z(\gamma^*) - Z(\gamma)) }{2\vartheta_\gamma^*}
=\frac{Z(\gamma^*) - Z(\gamma)}{2\vartheta_\gamma^*} (1+o_p(1)).
\label{eq:lrt_knownweights}
\end{align}
The term $(Z(\gamma^*)-Z(\gamma))/\vartheta_\gamma^*$ is the LR test statistic for fixed $(\vartheta_\gamma^*,\alpha_\gamma^*)$
comparing $\gamma$ and $\gamma^* \subset \gamma$.
}}
When $k=2$ this is a quantile regression LR test statistic,
which \cite{koenker:1982} showed to be asymptotically $\chi_{p_\gamma - p_{\gamma^*}}^2$ (after rescaling by a constant)
precisely under our Conditions A2-A3.
{\textcolor{black}{
When $k=1$,
$(Z(\gamma^*)-Z(\gamma))/\vartheta_\gamma^*$ is the LR test statistic for a weighted least squares problem
regressing $\tilde{y}=W_{\theta^*,\alpha^*} y$ on $\widetilde{X}=W_{\theta^*,\alpha^*}X$,
which can be shown to be $O_p(1)$ under the conditions in Proposition \ref{prop:AsympNorm}.
Briefly, as usual for any $\gamma$ the total sum of squares can be decomposed as
$\tilde{y}^T \tilde{y}= \widehat{\theta}_\gamma^T \widetilde{X}_\gamma^T \widetilde{X}_\gamma \widehat{\theta}_\gamma
+ (\tilde{y} - \widetilde{X}_\gamma \widehat{\theta}_\gamma)^T (\tilde{y} - \widetilde{X}_\gamma \widehat{\theta}_\gamma)$, hence
$Z(\gamma^*) - Z(\gamma)=$
\begin{align}
(\tilde{y} - \widetilde{X}_{\gamma^*} \widehat{\theta}_{\gamma^*})^T (\tilde{y}-\widetilde{X}_{\gamma^*} \widehat{\theta}_{\gamma^*})
- (\tilde{y} - \widetilde{X}_{\gamma} \widehat{\theta}_{\gamma})^T (\tilde{y}-\widetilde{X}_{\gamma} \widehat{\theta}_{\gamma})
= \widehat{\theta}_\gamma^T \widetilde{X}_\gamma^T \widetilde{X}_\gamma \widehat{\theta}_\gamma
- \widehat{\theta}_{\gamma^*}^T \widetilde{X}_{\gamma^*}^T \widetilde{X}_{\gamma^*} \widehat{\theta}_{\gamma^*}.
\label{eq:waldtest_tpnorm}
\end{align}
Without loss of generality let $\widetilde{X}_\gamma=(\widetilde{X}_{\gamma^*},\widetilde{X}_{\gamma \setminus \gamma^*})$,
where $\widetilde{X}_{\gamma \setminus \gamma^*}$ are the columns in $\widetilde{X}_\gamma$ not contained in $\widetilde{X}_{\gamma^*}$.
Let $R=(I - \widetilde{X}_{\gamma^*} (\widetilde{X}_{\gamma^*}^T \widetilde{X}_{\gamma^*})^{-1} \widetilde{X}_{\gamma^*}) \widetilde{X}_{\gamma \setminus \gamma^*}$
be orthogonal to the projection of $\widetilde{X}_{\gamma}$ onto $\widetilde{X}_{\gamma^*}$,
then clearly $\widetilde{X}_{\gamma^*}^T R=0$ and $(\widetilde{X}_{\gamma^*},R)$ span the column space of $\widetilde{X}_{\gamma}$.
Hence
$\widehat{\theta}_\gamma^T \widetilde{X}_\gamma^T \widetilde{X}_\gamma \widehat{\theta}_\gamma=
\widehat{\theta}_{\gamma^*}^T \widetilde{X}_{\gamma^*}^T \widetilde{X}_{\gamma^*} \widehat{\theta}_{\gamma^*}
+ \widehat{\theta}_R^T R^T R \widehat{\theta}_R,
$
where $\widehat{\theta}_R= (R^TR)^{-1}R^Ty$,
giving that $Z(\gamma^*)-Z(\gamma)= \widehat{\theta}_R^T R^T R \widehat{\theta}_R$.
By Proposition \ref{prop:AsympNorm},
$\sqrt{n} \widehat{\theta}_R \stackrel{D}{\longrightarrow} N(0,\vartheta_\gamma^* V)$ for a fixed positive-definite matrix $V$.
}

{\textcolor{black}{
To conclude, our Conditions A3-A4 guarantee $\frac{1}{n} R^T R \stackrel{P}{\longrightarrow} \Sigma_R$ for some fixed $\Sigma_R$
and by the continuous mapping theorem
$\sqrt{n} \Sigma_R^{\frac{1}{2}} \widehat{\theta}_R \stackrel{D}{\longrightarrow} N(0,\vartheta_\gamma^* \Sigma_R^{\frac{1}{2}} V \Sigma_R^{\frac{1}{2}})$.
Hence $\frac{n}{\vartheta_\gamma^*} \widehat{\theta}_R^T \Sigma_R \widehat{\theta}_R \stackrel{D}{\longrightarrow} Q$,
where $Q=O_p(1)$ is a sum of re-scaled central chi-square random variables with 1 degree of freedom.
By Slutsky's theorem
$\frac{Z(\gamma^*) - Z(\gamma)}{\vartheta_\gamma^*}= \frac{n}{\vartheta_\gamma^*} \widehat{\theta}_R^T (\frac{1}{n} R^TR) \widehat{\theta}_R \stackrel{D}{\longrightarrow} Q$,
as we wished to prove.
}}

\subsection{Proof of Corollary \ref{cor:bma}}

The proof runs analogous to \cite{rossell:2017}, Proposition 3(ii).
Briefly, the BMA estimate is $E(\theta_i \mid y)=$
\begin{align}
E(\theta_i \mid \gamma^*,y) p(\gamma^* \mid y)
+ \sum_{\gamma^* \subset \gamma}^{} E(\theta_i \mid \gamma,y) p(\gamma \mid y)
+ \sum_{\gamma^* \not\subset \gamma}^{} E(\theta_i \mid \gamma,y) p(\gamma \mid y).
\label{eq:bma1}
\end{align}

Suppose that $\theta_i^* \neq 0$.
From Proposition \ref{prop:AsympNorm}, the difference between the MLE under $\gamma$ and $\theta_i^*$ is $O_p(1/\sqrt{n})$,
and it can be shown that the difference between a Laplace approximation to $E(\theta_i \mid \gamma,y)$ and the MLE is $O_p(1/\sqrt{n})$
hence $E(\theta_i \mid \gamma,y) - \theta_i^*= O_p(1/\sqrt{n})$.
Since $p(\gamma^* \mid y) \stackrel{P}{\longrightarrow} 1$ by Proposition \ref{prop:bfrates},
we have that $E(\theta_i \mid \gamma^*,y) p(\gamma^* \mid y)= \theta_i^* + O_p(1/\sqrt{n})$.
If $\theta_i^* = 0$ then by definition $E(\theta_i \mid \gamma^*,y) p(\gamma^* \mid y)=0$.

Consider the second term in \eqref{eq:bma1} where $\gamma^* \subset \gamma$, 
$$p(\gamma \mid y) \leq 1/(1+ B_{\gamma^*,\gamma} p(\gamma^*)/p(\gamma)) < B_{\gamma,\gamma^*} p(\gamma)/p(\gamma^*)= O_p(b_n^{(k)}) p(\gamma)/p(\gamma^*) \leq O_p(b_n^{(k)}) r^+,$$
where $B_{\gamma^*,\gamma}$ is the Bayes factor between $\gamma^*$ and $\gamma$.
From Proposition \ref{prop:bfrates}, we have that $b_n^{(k)}=n^{-(p_\gamma-p_{\gamma^*})/2}$ for a local prior,
$b_n^{(k)}=n^{-3(p_\gamma-p_{\gamma^*})/2}$ for the pMOM prior,
and $b_n^{(k)}= e^{-c\sqrt{n}}$, for some $c>0$, for the peMOM and piMOM priors.
Also, $E(\theta_i \mid \gamma,y) = \theta_i^* + O_p(1/\sqrt{n})$.
Therefore, if $\theta_i^* \neq 0$, we have
$E(\theta_i \mid \gamma,y) p(\gamma \mid y)= O_p(b_n^{(k)}) r^+$.
If $\theta_i^*=0$, then
$E(\theta_i \mid \gamma,y) p(\gamma \mid y)= O_p(b_n^{(k)}/\sqrt{n}) p(\gamma)/p(\gamma^*) \leq O_p(b_n^{(k)}/\sqrt{n}) r^+$.
The case for $\gamma^* \not\subset \gamma$ proceeds similarly by noting that by Proposition \ref{prop:bfrates}
we have $B_{\gamma,\gamma^*} r^-= O_p(e^{-cn}) r^-= O_p(b_n^{(k)})$ for some $c>0$,
since $e^{-cn} r^-= O(b_n^{(k)})$ by assumption.

Combining the previous results it follows that, if $\theta_i^* \neq 0$, then
\begin{align}
E(\theta_i \mid y)= \theta_i^* + O_p(1/\sqrt{n}) + O_p(b_n^{(k)}) r^+= \theta_i^* + O_p(1/\sqrt{n}),
\label{eq:bma2}
\end{align}
since $b_n^{(k)} r^+= O_p(1/\sqrt{n})$ by the assumption that $r^+$ does not increase with $n$.
Conversely if $\theta_i^*= 0$, then
\begin{align}
E(\theta_i \mid y)= O_p(b_n^{(k)}/\sqrt{n}) r^+,
\label{eq:bma3}
\end{align}
giving the desired result.

\section{Approximations to the integrated likelihood}
\label{sec:approx_marglhood}

For ease of notation, we drop the subindex $k$ denoting the set of active variables
and let $\theta=(\theta_1,\ldots,\theta_{|k|})$ be their coefficients.
Both the Laplace and Importance Sampling approximations require maximizing and evaluating the hessian of $h_l(\theta,\vartheta,\talpha)= \log L(\theta,\tvartheta,\talpha) + \log p(\theta,\tvartheta,\talpha)$, where $L(\cdot)$ and $p(\cdot)$ are the appropriate likelihood and prior density.
Denote by $g_l(\theta,\tvartheta,\talpha)$ the gradient of $h_l(\cdot)$ and by $H_l(\theta,\tvartheta,\talpha)$ its hessian, Algorithm \ref{alg:postmode} finds the posterior mode.

\begin{algorithm}
{\bf Posterior mode via Newton-Raphson}\label{alg:postmode}
\begin{enumerate}
\item Initialize $(\theta^{(0)},\tvartheta^{(0)},\talpha^{(0)})=(\widehat{\theta},\log (\widehat{\vartheta}),\mbox{atanh}(\widehat{\alpha}))$ where $(\widehat{\theta},\widehat{\vartheta},\widehat{\alpha})$ is the MLE given by Algorithm \ref{alg:mle_lma}. Set $t=1$ and repeat Steps 2-3 until $e$ is below some small tolerance (default $10^{-5}$).

\item Update $(\theta^{(t)},\tvartheta^{(t)},\talpha^{(t)})=$
$$(\theta^{(t-1)},\tvartheta^{(t-1)},\talpha^{(t-1)}) - H_l^{-1}(\theta^{(t-1)},\tvartheta^{(t-1)},\talpha^{(t-1)}) g_l(\theta^{(t-1)},\tvartheta^{(t-1)},\talpha^{(t-1)}).$$

\item Compute $e=||(\theta^{(t)},\tvartheta^{(t)},\talpha^{(t)}) - (\theta^{(t-1)},\tvartheta^{(t-1)},\talpha^{(t-1)})||^{\infty}$ where $||\bz||^{\infty}$ is the largest element of $\bz$ in absolute value. Set $t=t+1$.
\end{enumerate}
\end{algorithm}
As usual, in the event that $(\theta^{(t)},\tvartheta^{(t)},\talpha^{(t)})$ does not increase $h_l(\cdot)$, Step 2 can be adjusted by adding a constant $\lambda$ to the diagonal of $H_l(\cdot)$, which for large $\lambda$ gives the direction of the gradient and is guaranteed to decrease $h_l(\cdot)$. However, we observed that this is extremely rare in practice. Usually, the simple Newton step increases $h_l(\cdot)$ at each iteration and converges to the maximum in a few iterations.

Both $g_l(\cdot)$ and $H_l(\cdot)$ are the sum of a term coming from the log-likelihood plus a term coming from the log-prior density. The exact expressions are given below separately.

As an alternative to Algorithm \ref{alg:postmode}, we also provide Algorithm \ref{alg:postmode_cda} based on Coordinate Descent
({\it i.e.}~successive univariate optimization).
Note that the Newton steps to update $\theta_j$ and $\alpha$ are in the direction of the gradient and are hence guaranteed to increase the
objective function for small enough $\lambda$.
Step 2 takes advantage of the fact that the maximizer with respect to
$\tvartheta$ for fixed $(\theta,\alpha)$ is available in closed form.

\begin{algorithm}
{\bf Posterior mode via CDA}\label{alg:postmode_cda}
\begin{enumerate}
\item Initialize $\theta^{(0)}$ to the least squares estimate, $\alpha^{(0)}=0$, $t=0$.

\item For the MOM prior set
$\tvartheta^{(t)}= \log \left( s/(n+p+3a_{\vartheta}) \right)$, where

$$s= \left( b_{\vartheta} + \theta^{(t)^T} \theta^{(t)} +
\sum_{i \in A(\theta)}^{} \frac{(y_i-x_i^T\theta^{(t)})^2}{(1+\alpha^{(t)})^2}
+ \sum_{i \not\in A(\theta)}^{} \frac{(y_i-x_i^T\theta^{(t)})^2}{(1-\alpha^{(t)})^2}
\right).$$

For eMOM and iMOM use a Newton-Raphson step.

\item For $j=1,\ldots,p$

\begin{enumerate}
\item Set $\lambda=1$ and $\theta^*= \theta_j^{(t-1)} - \lambda g^*/h^*$,
where $g^*$ and $h^*$ are the first and second derivatives of
$f(\theta_j)= \log L_1(\theta_1^{(t-1)},\ldots,\theta_{j-1}^{(t-1)},\theta_j,\theta_{j+1}^{(t)},\ldots,\theta_p^{(t)},\vartheta^{(t)},\alpha) + \log p(\theta_j \mid \vartheta)$
evaluated at $\theta_j=\theta_j^{(t-1)}$.

\item If $f(\theta^*)>f(\theta_j^{(t-1)})$ set $\theta_j^{(t)}=\theta^*$, else set $\lambda= 0.5 \lambda$
and repeat Step 3-(1).
\end{enumerate}

\item Let $\tilde{\alpha}^*= \talpha^{(t-1)} - \lambda g^*/h^*$, where
$g^*$ and $h^*$ are the first and second derivatives of
$f(\talpha)= \log L_1(\theta^{(t)},\vartheta^{(t)},\talpha) + \log p(\talpha)$ at $\talpha= \talpha^{(t-1)}$.
If $f(\alpha^*)>f(\alpha^{(t-1)})$ set $\alpha^{(t)}=\alpha^*$, else set $\lambda= 0.5 \lambda$ and repeat Step 4.

\item Compute $e= \max |(\theta^{(t)},\tvartheta^{(t)},\talpha^{(t)}) - (\theta^{(t-1)},\tvartheta^{(t-1)},\talpha^{(t-1)})|$.
If $e<10^{-5}$ stop, else set $t=t+1$ and go back to Step 1.
\end{enumerate}
\end{algorithm}

\subsection{Derivatives of the log-likelihood}
\label{ssec:deriv_logl}

\subsubsection{\underline{Two-piece Normal}}

Under the re-parameterization $\tvartheta=\log(\vartheta)$, $\talpha=\mbox{atanh}(\alpha)$ the two-piece Normal log-likelihood (\ref{eq:skewnorm_loglhood}) has gradient

\begin{align}
\begin{pmatrix}
\frac{1}{\exp(\tvartheta)}
 X^T W (y - X\theta) \\
-\frac{n}{2} + \frac{1}{2\exp(\tvartheta)}
(y - X\theta)^T W (y - X\theta) \\
\frac{1}{2\exp(\tvartheta)}
(y - X\theta)^T W^{\star} (y - X\theta)
\end{pmatrix},
\nonumber
\end{align}
where as usual $W=\mbox{diag}(w)$,
$w_i=[1+\tanh(\talpha)]^{-2}$ if $i \in A(\theta)$ and $w_i=[1-\tanh(\talpha)]^{-2}$ if $i \not \in A(\theta)$,
and $W^{\star}=\mbox{diag}(w^*)$ with $w_i^{\star}=-\frac{2 \text{sech}^2(\talpha)}{(\tanh (\talpha)+1)^3}$ if $i \in A(\theta)$ and $w_i^{\star}=\frac{2 \text{sech}^2(\talpha)}{(1-\tanh (\talpha))^3}$  if $i \not \in A(\theta)$.
Its Hessian is given by

\begin{align}
-e^{-\tvartheta}
\begin{pmatrix}
 X^T W X &
 X^T W (y - X\theta) &
 X^T W^{\star} (y - X\theta)\\
&
\frac{1}{2} (y - X\theta)^T W (y - X\theta) &
-\frac{1}{2} (y - X\theta)^T W^{\star} (y - X\theta)\\
&
&
\frac{1}{2} (y - X\theta)^T W^{\star\star} (y - X\theta) \\
\end{pmatrix},
\label{eq:skewnorm_logl_hessian_repar}
\end{align}

\noindent where $W^{\star\star}=\mbox{diag}(w^{\star\star})$, with
$w_i^{\star\star}=2 e^{-4 \talpha} \left(e^{2 \talpha}+2\right)$ if $i \in A(\theta)$ and $w_i^{\star\star}=2 e^{2 \talpha}+4 e^{4 \talpha}$ if $i \not \in A(\theta)$.

\subsubsection{\underline{Two-piece Laplace}}

The asymmetric Laplace $\log L_2(\theta,\tilde{\vartheta},\tilde{\alpha})$, where
$\tvartheta=\log(\vartheta)$, $\talpha=\mbox{atanh}(\alpha)$ has gradient
$$
\begin{pmatrix}
- e^{-\tvartheta/2} X^T \wbar \\
-\frac{n}{2} + \frac{1}{2} e^{-\tvartheta/2} w^T |y-X \theta| \\
e^{-\tvartheta/2} |y-X \theta|^T \wbar^*
\end{pmatrix},
$$
and hessian
\begin{align}
e^{-\tvartheta/2} \times
\begin{pmatrix}
0 &
\frac{1}{2} X^T \wbar &
 X^T w^* \\
\frac{1}{2} \wbar^T X &
-\frac{1}{4} w^T |y - X \theta| &
-\frac{1}{2} |y - X \theta|^T \wbar^* \\
(X^T w^*)^T &
-\frac{1}{2} |y - X \theta|^T \wbar^* &
-2 |y - X \theta|^T w^*
\end{pmatrix},
\label{eq:skewlapl_logl_hessian_repar}
\end{align}
where
$w_i=\wbar_i=(1+\alpha)^{-1}$,
$w_i^*=\wbar_i^*=e^{-2\alpha}$ if $i \in A(\theta)$,
and $w_i=(1-\alpha)^{-1}$, $\wbar_i=-w_i$
$w_i^*=e^{2\alpha}$, $\wbar_i^*=-w_i^*$
if $i \not\in A(\theta)$.
Naturally, symmetric Laplace errors are the particular case $\alpha=0$
and give $w_{i}=w_{i}^*=1$.

\subsubsection{\underline{Expected two-piece Laplace log-likelihood}}

We derive $\overline{L}_2=E(\log L_2(\eta))$, where $\eta=(\theta,\vartheta,\alpha)$
and its derivatives under the data-generating model
$y_i= x_i^T \theta_0 + \epsilon_i$, for some $\theta_0 \in \mathbb{R}^p$
where $\epsilon_i$ are independent across $i=1,\ldots,n$ and arise from an arbitrary probability density function $s_0(y_i\vert x_i)$.
After some algebra and noting that $\epsilon_i=y_i-x_i^T \theta_0$ gives

\begin{eqnarray*}
\overline{L}_2= \int_{}^{}\!\, \log L_2(\eta) s_0(\epsilon\vert x) d\epsilon= - n\log(2) - \dfrac{n}{2}\log(\vartheta)
&-\sum_{i=1}^{n} \dfrac{1}{\sqrt{\vartheta}(1+\alpha)}\int_{-\infty}^{x_i^{T}(\theta-\theta_0)} S_0(\epsilon_i)d\epsilon_i \\
&-\sum_{i=1}^{n} \dfrac{1}{\sqrt{\vartheta}(1-\alpha)}\int_{x_i^{T}(\theta-\theta_0)}^{\infty} 1-S_0(\epsilon_i)d\epsilon_i,
\end{eqnarray*}
where $S_0(\epsilon_i) = S_0(\epsilon_i\vert 0)$ is the cumulative probability function associated to $s_0(\epsilon_i) = s_0(\epsilon_i\vert 0)$, where $0$ indicates a zero covariate vector.
Then taking derivatives we obtain
\begin{eqnarray*}
\dfrac{\partial}{\partial \theta}\overline{L}_2 &=& \sum_{i=1}^{n}
-\dfrac{x_i S_0(x_i^{T}(\theta - \theta_0))}{\sqrt{\vartheta}(1+\alpha)} + \dfrac{x_i [1- S_0(x_i^{T}(\theta- \theta_0))]}{\sqrt{\vartheta}(1-\alpha)},\\
\dfrac{\partial}{\partial \vartheta}\overline{L}_2 &=& \sum_{i=1}^{n}
-\dfrac{1}{2\vartheta} + \dfrac{I_{i1}}{2\vartheta^{3/2}(1+\alpha)} + \dfrac{I_{i2}}{2\vartheta^{3/2}(1-\alpha)},\\
\dfrac{\partial}{\partial \alpha}\overline{L}_2 &=& \sum_{i=1}^{n} \dfrac{I_{i1}}{\sqrt{\vartheta}(1+\alpha)^2} - \dfrac{I_{i2}}{\sqrt{\vartheta}(1-\alpha)^2},
\end{eqnarray*}
where $I_{i1} = \int_{-\infty}^{x_i^T(\theta-\theta_0)} S_0(\epsilon_i)d\epsilon_i$, $I_{i2} = \int_{x_i^T(\theta- \theta_0)}^{\infty} 1-S_0(\epsilon_i)d\epsilon_i$.
The second derivatives are

\begin{eqnarray*}
\dfrac{\partial^2}{\partial \theta^2} \overline{L}_2
&=& - \sum_{i=1}^{n} \dfrac{2x_ix_i^{T} s_0(x_i^{T}(\theta-\theta_0))}{\sqrt{\vartheta}(1-\alpha^2)},\nonumber \\
\dfrac{\partial^2}{\partial \vartheta^2} \overline{L}_2 &=&
\sum_{i=1}^{n} \dfrac{1}{2\vartheta^2} - \dfrac{3I_{i1}}{4\vartheta^{5/2}(1+\alpha)} - \dfrac{3I_{i2}}{4\vartheta^{5/2}(1-\alpha)},\nonumber \\
\dfrac{\partial^2}{\partial \alpha^2} \overline{L}_2 &=&
-\sum_{i=1}^{n} \dfrac{2I_{i1}}{\sqrt{\vartheta}(1+\alpha)^3} - \dfrac{2I_{i2}}{\sqrt{\vartheta}(1-\alpha)^3},\nonumber \\
\dfrac{\partial^2}{\partial \vartheta \partial \theta} \overline{L}_2 &=&
\sum_{i=1}^{n} \dfrac{x_i S_0(x_i^{T}(\theta - \theta_0))}{2\vartheta^{3/2}(1+\alpha)} - \dfrac{x_i [1- S_0(x_i^{T}(\theta - \theta_0))]}{2\vartheta^{3/2}(1-\alpha)},\\
\dfrac{\partial^2}{\partial \alpha \partial \theta} \overline{L}_2 &=&
\sum_{i=1}^{n} \dfrac{x_i S_0(x_i^{T}(\theta - \theta_0))}{\sqrt{\vartheta}(1+\alpha)^2} + \dfrac{x_i [1- S_0(x_i^{T}(\theta - \theta_0))]}{\sqrt{\vartheta}(1-\alpha)^2},\\
\dfrac{\partial^2}{\partial \vartheta \partial \alpha } \overline{L}_2 &=&
-\sum_{i=1}^{n}\dfrac{I_{i1}}{2\vartheta^{3/2}(1+\alpha)^2} + \dfrac{I_{i2}}{2\vartheta^{3/2}(1-\alpha)^2}.\\
\end{eqnarray*}

Simple inspection reveals that $(\partial/\partial \theta) \overline{L}_2=0$ implies
$(\partial^2/\partial \theta \partial \vartheta) \overline{L}_2=0$,
and likewise $(\partial/\partial \alpha) \overline{L}_2=0$ implies
$(\partial^2/\partial \theta \partial \alpha) \overline{L}_2=0$.
Since the maximum likelihood estimator $(\widehat{\theta},\widehat{\vartheta},\widehat{\alpha})$ converges in probability to the maximizer of $\overline{L}_2$,
these second derivatives evaluated at $(\widehat{\theta},\widehat{\vartheta},\widehat{\alpha})$ also converge in probability to 0.

We wish to find an asymptotic expression for the remaining
second derivatives evaluated at $(\widehat{\theta},\widehat{\vartheta},\widehat{\alpha})$
when the data-generating truth is $\epsilon_i \sim \mbox{AL}(x_i^T \theta_0,\vartheta_0,\alpha_0)$
for some $(\theta_0,\vartheta_0,\alpha_0)$.
Given that $(\widehat{\theta},\widehat{\vartheta},\widehat{\alpha}) \stackrel{P}{\longrightarrow} (\theta_0,\vartheta_0,\alpha_0)$,
the expressions above require evaluating the density of an asymmetric Laplace
$s_0(0)=1/(2\sqrt{\vartheta_0})$ and its cumulative probability function $S_0(0)=(1+\alpha_0)/2$.
Similarly, direct integration gives $I_{i1}= \sqrt{\vartheta_0}(1+\alpha_0)^2/2$ and $I_{i2}= \sqrt{\vartheta_0}(1-\alpha_0)^2/2$.

\begin{align}
\dfrac{\partial^2}{\partial \theta^2} \overline{L}_2 \stackrel{P}{\longrightarrow}
&-X^TX \dfrac{1}{\vartheta_0(1-\alpha_0^2)},\nonumber \\
\dfrac{\partial^2}{\partial \vartheta^2} \overline{L}_2 \stackrel{P}{\longrightarrow}
&\dfrac{n}{2\vartheta_0^2} - \dfrac{3n (1+\alpha_0)}{8\vartheta_0^{2}} - \dfrac{3 (1-\alpha_0)}{8\vartheta_0^{2}}=
-\frac{n}{4 \vartheta_0^2},\nonumber \\
\dfrac{\partial^2}{\partial \alpha^2} \overline{L}_2 \stackrel{P}{\longrightarrow}
&-\dfrac{n}{1+\alpha_0} - \dfrac{n}{1-\alpha_0}= -\frac{2n}{1-\alpha_0^2},\nonumber \\
\dfrac{\partial^2}{\partial \alpha \partial \theta} \overline{L}_2 \stackrel{P}{\longrightarrow}
& \frac{n \overline{x}}{\sqrt{\vartheta_0}} \left(\dfrac{1}{2(1+\alpha_0)} + \dfrac{1}{2(1-\alpha_0)} \right)=
\frac{n \overline{x}}{\sqrt{\vartheta_0}(1-\alpha_0^2)}.
\label{eq:alaplace_expected_hessian}
\end{align}

\subsection{Derivatives of the log-prior density}
\label{ssec:deriv_logprior}

The log-prior density is $\log p(\theta,\tvartheta)= \log p(\theta \mid \tvartheta) + \log p(\tvartheta)$ when $\talpha=0$ under the assumed model and $\log p(\theta,\tvartheta,\talpha)= \log p(\theta,\tvartheta) + \log p(\talpha)$ when $\talpha \neq 0$, where $p(\theta \mid \tvartheta)$ and $p(\talpha)$ are the pMOM, piMOM or peMOM priors and $p(\tvartheta)= \mbox{IG}(e^{\tvartheta}; a_{\vartheta}/2,b_{\vartheta}/2) e^{\vartheta}$. For ease of notation let $\theta^{-a}$ be the vector with elements $\theta_i^{-a}$ for $i=1,\ldots,|k|$.

\subsubsection{\underline{pMOM prior}} Straightforward algebra gives

\begin{align}
\nabla \log p_M(\theta,\tvartheta,\talpha)&= \begin{pmatrix}
2 \theta^{-1} - \theta e^{-\tvartheta}/g_{\theta} \\
-\frac{3|k|+a_{\vartheta}}{2} + (\theta^T\theta/g_{\theta} + b_{\vartheta}) e^{-\tvartheta}/2 \\
2 \talpha^{-1} - \talpha g_{\alpha}^{-1} \\
\end{pmatrix}, \nonumber \\
\nabla^2 \log p_M(\theta,\tvartheta,\talpha)&= \begin{pmatrix}
\mbox{diag}( -2\theta^{-2} - e^{-\tvartheta}/g_{\theta})  & \theta e^{-\tvartheta}/g_{\theta} & 0 \\
\theta^T e^{-\tvartheta}/g_{\theta} & - e^{-\tvartheta} (\theta^T\theta/g_{\theta}+b_{\vartheta})/2 & 0\\
0 & 0 & -2\talpha^{-2} - g_{\alpha}^{-1} \nonumber
\end{pmatrix},
\end{align}

\subsubsection{\underline{piMOM prior}} We obtain

\begin{align}
\nabla \log p_I(\theta,\tvartheta,\talpha)&= \begin{pmatrix}
-2 \theta^{-1} + 2 g_{\theta} e^{\tvartheta} \theta^{-3} \\
(|k| - a_{\vartheta})/2 +b_{\vartheta}e^{-\tvartheta}/2  - g_{\theta} e^{\tvartheta} \sum_{i}^{} \theta_i^{-2} \\
-2 \talpha^{-1} - 2 g_{\alpha} \talpha^{-3}
\end{pmatrix}, \nonumber \\
\nabla^2 \log p_I(\theta,\tvartheta,\talpha)&= \begin{pmatrix}
\mbox{diag}(2\theta^{-2} - 6 g_{\theta} e^{\tvartheta} \theta^{-4}) & 2g_{\theta} e^{\tvartheta} \theta^{-3} & 0 \\
(-2g_{\theta} e^{\tvartheta} \theta^{-3})^T & - b_{\vartheta} e^{-\tvartheta}/2 -e^{\tvartheta} g_{\theta} \sum_{i}^{} \theta_i^{-2} & 0 \\
0 & 0 & 2\talpha^{-2} + 6 g_{\alpha} \talpha^{-4} \nonumber
\end{pmatrix}.
\end{align}

\subsubsection{\underline{peMOM prior}} We obtain

\begin{align}
\nabla \log p_E(\theta,\tvartheta,\talpha)= \begin{pmatrix}
2 g_{\theta} e^{\tvartheta} \theta^{-3} - \theta e^{-\tvartheta} g_{\theta}^{-1} \\
-(|k| + a_{\vartheta})/2 +(b_{\vartheta}+\theta^T\theta/g_{\theta})e^{-\tvartheta}/2  - g_{\theta} e^{\tvartheta} \sum_{i}^{} \theta_i^{-2} \\
2 g_{\alpha} \talpha^{-3} - \talpha g_{\alpha}^{-1}
\end{pmatrix}, \nonumber
\end{align}
and $\nabla^2 \log p_E(\theta,\tvartheta,\talpha)=$
\begin{align}
\begin{pmatrix}
\mbox{diag}(- 6 g_{\theta} e^{\tvartheta} \theta^{-4} - e^{-\tvartheta}g_{\theta}^{-1}) & 2g_{\theta} e^{\tvartheta} \theta^{-3} + \theta e^{-\tvartheta} g_{\theta}^{-1} & 0 \\
(2g_{\theta} e^{\tvartheta} \theta^{-3} + \theta e^{-\tvartheta} g_{\theta}^{-1})^T & -(b_{\vartheta} + \theta^T\theta/g_{\theta}) e^{-\tvartheta}/2  -e^{\tvartheta} g_{\theta} \sum_{i}^{} \theta_i^{-2} & 0 \\
0 & 0 & - 6 g_{\alpha} \talpha^{-4} - g_{\alpha}^{-1} \nonumber
\end{pmatrix}.
\end{align}

\subsection{Quadratic approximation to asymmetric Laplace log-likelihood}
\label{ssec:approx_logl_alapl}

The goal is to approximate the curvature of the one-dimensional function
$f(\lambda)=\log L_2(\theta_\lambda,\widehat{\vartheta},\widehat{\alpha})$ around $\lambda=0$,
where $\theta_\lambda=(\widehat{\theta}_1,\ldots,\widehat{\theta}_{j-1},\widehat{\theta}_j+\lambda,\widehat{\theta}_{j+1},\ldots,\widehat{\theta}_p)$
is fixed to the maximum likelihood estimator except for the $j^{th}$ regression parameter, which is a function of $\lambda \in \mathbb{R}$.
Given that $f(0)$ is known and that its derivative at $\lambda=0$ is 0 ($\widehat{\theta}$ is a maximum) we seek $h_j^*<0$ such that
$f(\lambda)-f(0) \approx 0.5 h_j^* \lambda^2$.
Our strategy is to evaluate $f(\lambda_k)$ on a grid $\lambda_k$ for $k=1,\ldots,K$
and use the least-squares estimate
$h_j^*= 2 \sum_{k=1}^{K} \lambda_k^2 (f(\lambda_k)-f(0))/\sum_{k=1}^{K} \lambda_k^4$,
where the form of $\log L_2$ gives the simple expression
$$
f(\lambda_k)-f(0)= - \frac{1}{\sqrt{\widehat{\vartheta}}}
\sum_{i=1}^{n} \vert r_i - \lambda_k x_{ij} \vert
\left( \frac{\mbox{I}(r_i \leq \lambda_k x_{ij})}{1+\widehat{\alpha}}
+ \frac{\mbox{I}(r_i> \lambda_k x_{ij})}{1-\widehat{\alpha}} \right),
$$
and $r_i=y_i - x_i^T \widehat{\theta}$.
Once $h_1^*,\ldots,h_p^*$ have been obtained we let
$D=\mbox{diag}(h_1^*/\bar{h}_{11},\ldots,h_p^*/\bar{h}_{pp})$
where $\overline{H}= (X^TX)/(\widehat{\vartheta} (1-\widehat{\alpha}^2))$ is the asymptotic hessian under asymmetric Laplace errors,
and we approximate the hessian of $\log L_2(\theta,\widehat{\vartheta},\widehat{\alpha})$
around $\theta=\widehat{\theta}$ with $H^*= D^{\frac{1}{2}} \overline{H} D^{\frac{1}{2}}$.
The construction ensures that the diagonal elements in $H^*$ are $h_1^*,\ldots,h_p^*$,
i.e. the quadratic approximation matches the actual curvature of $\log L_2$ along each canonical axis.
From Section \ref{sec:paramestim} the correlation structure borrowed from $\overline{H}$
remains asymptotically valid as long as the residuals are independent and identically distributed,
however in our experience the approximation usually suffices for practical purposes even when these assumptions is violated.

The problem has been thus reduced to choosing the grid $\lambda_1,\ldots,\lambda_K$.
One naive option is to take the $n$ points of non-differentiability $\lambda=r_i/x_{ij}$,
however, by the nature of least squares, this strategy tends to approximate better $f(\lambda)$ for large $\lambda^2$
and we are interested in local approximations around $\lambda=0$,
further evaluating $f(\lambda)$ at $n$ points requires $O(n^2)$ operations for each $j=1,\ldots,p$
and is thus computationally costly.
Instead we evaluate $f(\lambda)$ only at the $K=2$ points
given by the endpoints of the asymptotic 95\% confidence interval
$\lambda=\{ -1.96 \overline{v}_j, 1.96\overline{v}_j)$ where $\overline{v}_j$ is the $j^{th}$ diagonal element in $\overline{H}^{-1}$.
This simple strategy ensures that the approximation holds locally around $\lambda=0$
in the sense of having non-negligible likelihood, requires only $O(n)$ operations
and we have observed to deliver reasonably accurate approximations in practice.
Our approximation is similar in spirit to the rank-based score test inversion used to obtain confidence intervals
in quantile regression, which has been amply described to deliver fairly precise intervals,
with the important difference that rank inversion requires an ordering of observations
that scales poorly with $p$ and $n$.

Supplementary Figure \ref{fig:quadapprox} shows an example with the likelihood $L_2$ (scaled to $(0,1)$)
and the two quadratic approximations based on the asymptotic covariance
and its least-squares adjustment for an intercept-only model ($p=1$) and $n=200$.
When residuals were truly generated from an asymmetric Laplace (left panel) the two quadratic approximations were essentially identical,
however under truly normally distributed residuals the asymptotic covariance over-estimated the curvature.

\begin{figure}
\begin{center}
\begin{tabular}{cc}
\includegraphics[width=0.5\textwidth]{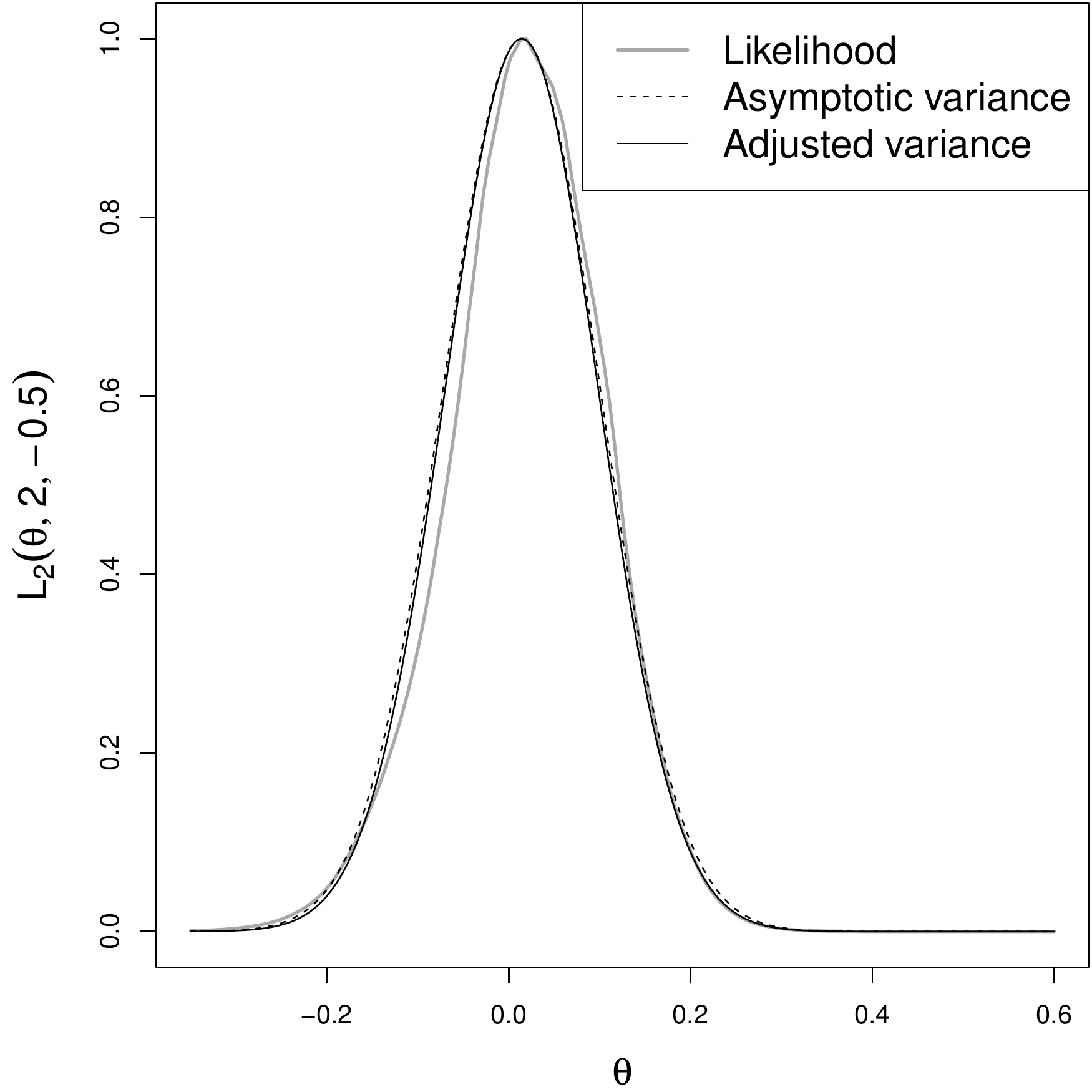} &
\includegraphics[width=0.5\textwidth]{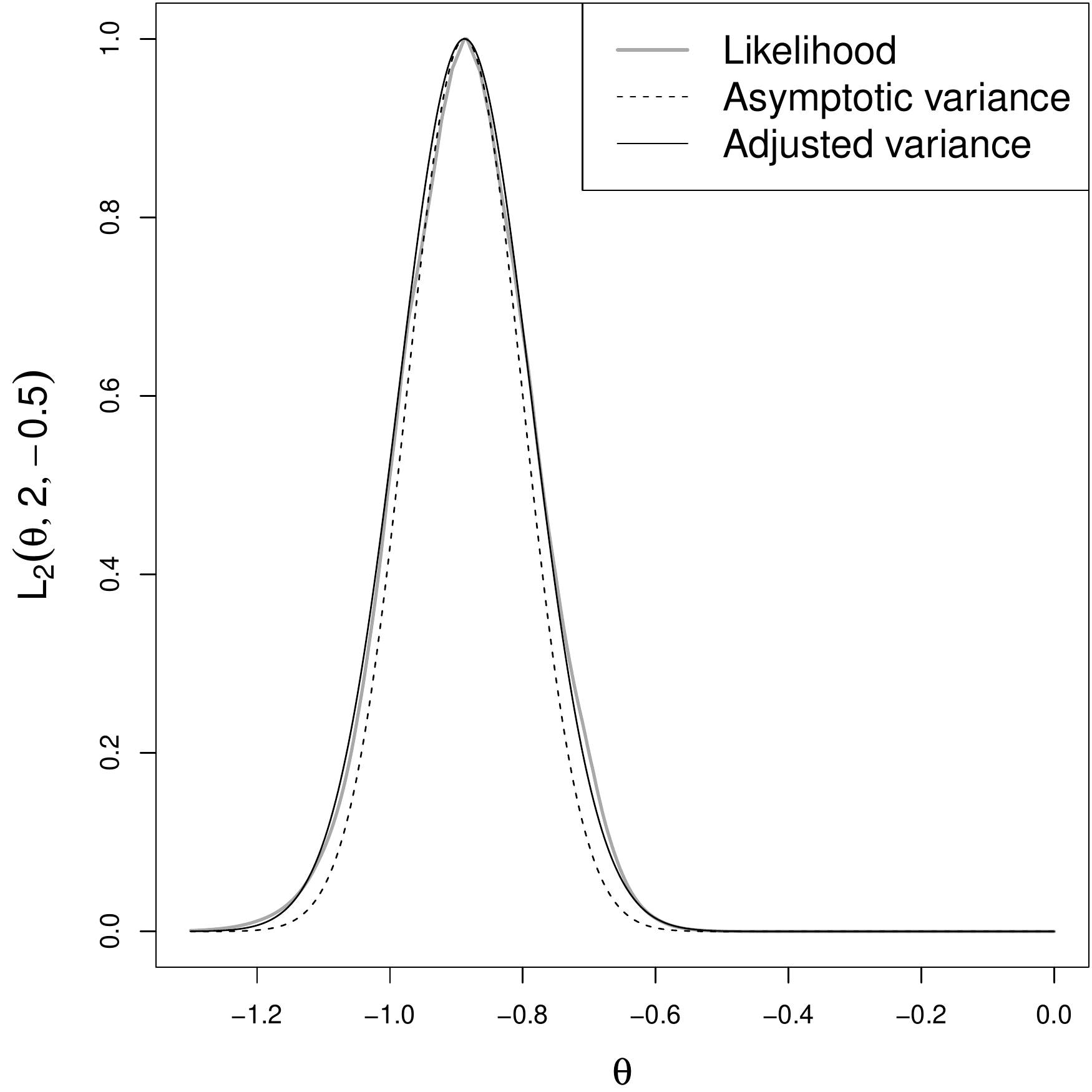}
\end{tabular}
\end{center}
\caption{Quadratic approximation to $L_2$ (solid grey) with $p=1,n=200$ from asymptotic covariance
(dotted black) and least-squares adjustment (solid black).
Left: $\epsilon_i \sim \mbox{AL}(0,2,-0.5)$; Right: $\epsilon_i \sim N(0,2)$.}
\label{fig:quadapprox}
\end{figure}

\section{Supplementary results}
\label{sec:supplres}

\begin{table}
\begin{center}
\begin{tabular}{|c|c|c|c|c|}\hline
 \multicolumn{5}{|c|}{Simulation truth $\epsilon_i \sim \mbox{AN}(0,4,\alpha)$} \\\cline{1-5}
Fitted model & $\alpha=0$ & $ \alpha=-0.25$ & $\alpha= -0.5$ & $\alpha= -0.75$ \\ \hline
Normal & 76.99 & 98.84 & 103.84 & 101.25 \\
ANormal & 92.08 & 86.10 & 102.60 & 115.64 \\
Laplace & 90.58 & 92.84 & 97.90 & 93.13 \\
ALaplace & 122.64 & 121.12 & 124.69 & 131.50 \\ \hline
 \multicolumn{5}{|c|}{Simulation truth $\epsilon_i \sim \mbox{AL}(0,4,\alpha)$} \\\cline{1-5}
Fitted model & $\alpha=0$ & $ \alpha=-0.25$ & $\alpha= -0.5$ & $\alpha= -0.75$ \\ \hline
Normal & 76.62 & 96.05 & 99.85 & 97.78 \\
ANormal & 81.77 & 82.74 & 92.67 & 104.76 \\
Laplace & 90.29 & 93.42 & 92.88 & 91.25 \\
ALaplace & 117.30 & 113.69 & 115.08 & 122.79 \\
\hline
\end{tabular}
\end{center}
\caption{CPU time ($10^{-4}$ seconds) on 3.4GHz Intel i7, 32Gb RAM, Windows 10.
$p=6$, $\vartheta=4$, $\theta=(0,0.5,0.75,1,0,\ldots,0)$, $n=100$, $\rho_{ij}=0.5$.
}
\label{tab:runtime_sim1}
\end{table}

\begin{table}
\begin{center}
\begin{tabular}{|c|c|c|c|c|}\hline
 & \multicolumn{4}{|c|}{Simulation truth} \\\cline{2-5}
 & $N(0,4)$ & $\mbox{AN}(0,4,-0.5)$ & $L(0,4)$ & $\mbox{AL}(0,4,-0.5)$ \\ \hline
Normal  &  6.9 & 29.7 &  6.4  & 32.9 \\
ANormal & 52.9 & 21.9 & 41.3  & 22.2 \\
Laplace & 17.0 & 28.2 & 14.6  & 26.6 \\
ALaplace& 57.7 & 26.4 &  26.7 & 22.4 \\
Inferred&  6.1 & 22.6 &  13.5 & 23.0 \\
\hline
\end{tabular}
\end{center}
\caption{CPU time (seconds) on 8GB RAM Mac laptop with 1.6GHz Intel i5 processors running OS X 10.11.6
$p=100$, $\vartheta=2$, $\theta=(0,0.5,0.75,1,0,\ldots,0)$, $n=100$, $\rho_{ij}=0.5$.
}
\label{tab:runtime_sim}
\end{table}

\subsection{Simulation study with identically distributed errors}
\label{ssec:simstudy_supplres}

\begin{table}
\begin{center}
\begin{tabular}{|l|cccc|} \hline
Truth & \multicolumn{4}{c|}{Average $p(\gamma_{p+1},\gamma_{p+2} \mid y)$} \\
 & $\gamma_{p+1}=\gamma_{p+2}=0$ & $\gamma_{p+1}=1,\gamma_{p+2}=0$
& $\gamma_{p+1}=0,\gamma_{p+2}=1$ & $\gamma_{p+1}=\gamma_{p+2}=0$ \\ \hline \hline
& \multicolumn{4}{c|}{$p=6$, $g_\alpha=0.357$, Laplace $p(\gamma \mid y)$} \\
$N(0,2)$              &  0.91    &         0.02 &            0.06  &           0.00  \\
$\mbox{AN}(0,2,-0.5)$ &  0.11    &         0.81 &            0.01  &           0.06 \\
$L(0,2)$              &  0.14    &         0.00 &            0.84  &           0.02  \\
$\mbox{AL}(0,2,-0.5)$ &  0.02    &         0.12 &            0.01  &           0.85 \\
\hline \hline
& \multicolumn{4}{c|}{$p=6$, $g_\alpha=0.357$, Monte Carlo $p(\gamma \mid y)$} \\
$N(0,2)$              &    0.91  &           0.02  &           0.06   &          0.00 \\
$\mbox{AN}(0,2,-0.5)$ &    0.11  &           0.81  &           0.01   &          0.07 \\
$L(0,2)$              &    0.12  &           0.01  &           0.85   &          0.02 \\
$\mbox{AL}(0,2,-0.5)$ &    0.02  &           0.12  &           0.01   &          0.85 \\
\hline \hline
& \multicolumn{4}{c|}{$p=6$, $g_\alpha=0.087$, Laplace $p(\gamma \mid y)$} \\
$N(0,2)$              &   0.87  &           0.07  &           0.06  &           0.01 \\
$\mbox{AN}(0,2,-0.5)$ &   0.07  &           0.86  &           0.01  &           0.07 \\
$L(0,2)$              &   0.13  &           0.01  &           0.79  &           0.07 \\
$\mbox{AL}(0,2,-0.5)$ &   0.01  &           0.13  &           0.01  &           0.85 \\ \hline \hline
\end{tabular}
\end{center}
\caption{Simulation study for $p=6$. Posterior probability of the 4 error distributions under
$\vartheta=2$, $\theta=(0,0.5,1,1.5,\ldots,0)$, $n=100$, $\rho_{ij}=0.5$.}
\label{tab:perror_p5}
\end{table}

\begin{table}
\begin{center}
\begin{tabular}{|l|cccc|} \hline
Truth & \multicolumn{4}{c|}{Average $p(\gamma_{p+1},\gamma_{p+2} \mid y)$} \\
& $\gamma_{p+1}=\gamma_{p+2}=0$ & $\gamma_{p+1}=1,\gamma_{p+2}=0$
& $\gamma_{p+1}=0,\gamma_{p+2}=1$ & $\gamma_{p+1}=\gamma_{p+2}=0$ \\ \hline \hline
& \multicolumn{4}{c|}{$p=101$, $\vartheta=1$} \\
$N(0,2)$              &   0.91   &          0.01  &           0.08 &            0.00 \\
$\mbox{AN}(0,2,-0.5)$ &   0.03   &          0.86  &           0.00 &            0.11 \\
$L(0,2)$              &   0.15   &          0.01  &           0.83 &            0.02 \\
$\mbox{AL}(0,2,-0.5)$ &   0.00   &          0.13  &           0.01 &            0.86 \\
\hline \hline
& \multicolumn{4}{c|}{$p=101$, $\vartheta=2$} \\
$N(0,2)$              &    0.89  &           0.01  &           0.10  &           0.00 \\
$\mbox{AN}(0,2,-0.5)$ &    0.02  &           0.89  &           0.00  &           0.09 \\
$L(0,2)$              &    0.15  &           0.01  &           0.82  &           0.02 \\
$\mbox{AL}(0,2,-0.5)$ &    0.00  &           0.16  &           0.01  &           0.83 \\
\hline \hline
& \multicolumn{4}{c|}{$p=501$, $\vartheta=1$} \\
$N(0,2)$              &      0.85 &            0.00   &          0.14  &           0.00 \\
$\mbox{AN}(0,2,-0.5)$ &      0.01 &            0.85   &          0.01  &           0.14 \\
$L(0,2)$              &      0.18 &            0.00   &          0.80  &           0.02 \\
$\mbox{AL}(0,2,-0.5)$ &      0.00 &            0.15   &          0.00  &           0.84 \\
\hline \hline
& \multicolumn{4}{c|}{$p=501$, $\vartheta=2$} \\
$N(0,2)$              &     0.83  &           0.00    &         0.16   &          0.00 \\
$\mbox{AN}(0,2,-0.5)$ &     0.00  &           0.87    &         0.00   &          0.12 \\
$L(0,2)$              &     0.19  &           0.00    &         0.79   &          0.01 \\
$\mbox{AL}(0,2,-0.5)$ &     0.00  &           0.22    &         0.00   &          0.77 \\
\hline
\end{tabular}
\end{center}
\caption{Simulation study for $p=101,501$. Posterior probability of the 4 error distributions under
$g_\alpha=0.357$, $\theta=(0,0.5,1,1.5,\ldots,0)$, $n=100$, $\rho_{ij}=0.5$.
Laplace approximation to $p(y \mid \gamma)$ was used.}
\label{tab:perror_largep}
\end{table}

\begin{figure}
\begin{center}
\begin{tabular}{cc}
$\epsilon_i \sim N(0,4)$ & $\epsilon_i \sim \mbox{AN}(0,4,-0.5)$ \\
\includegraphics[width=0.48\textwidth]{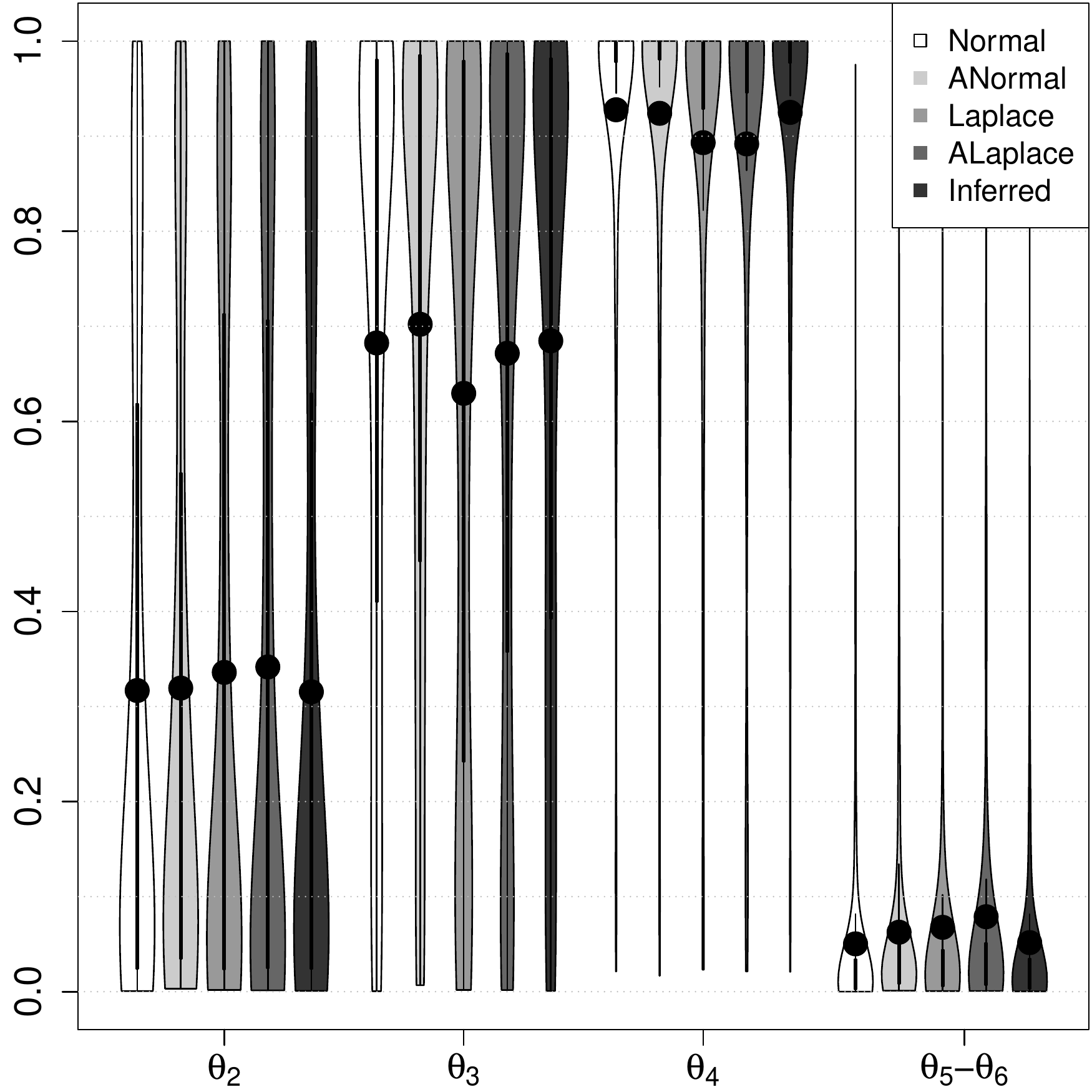} &
\includegraphics[width=0.48\textwidth]{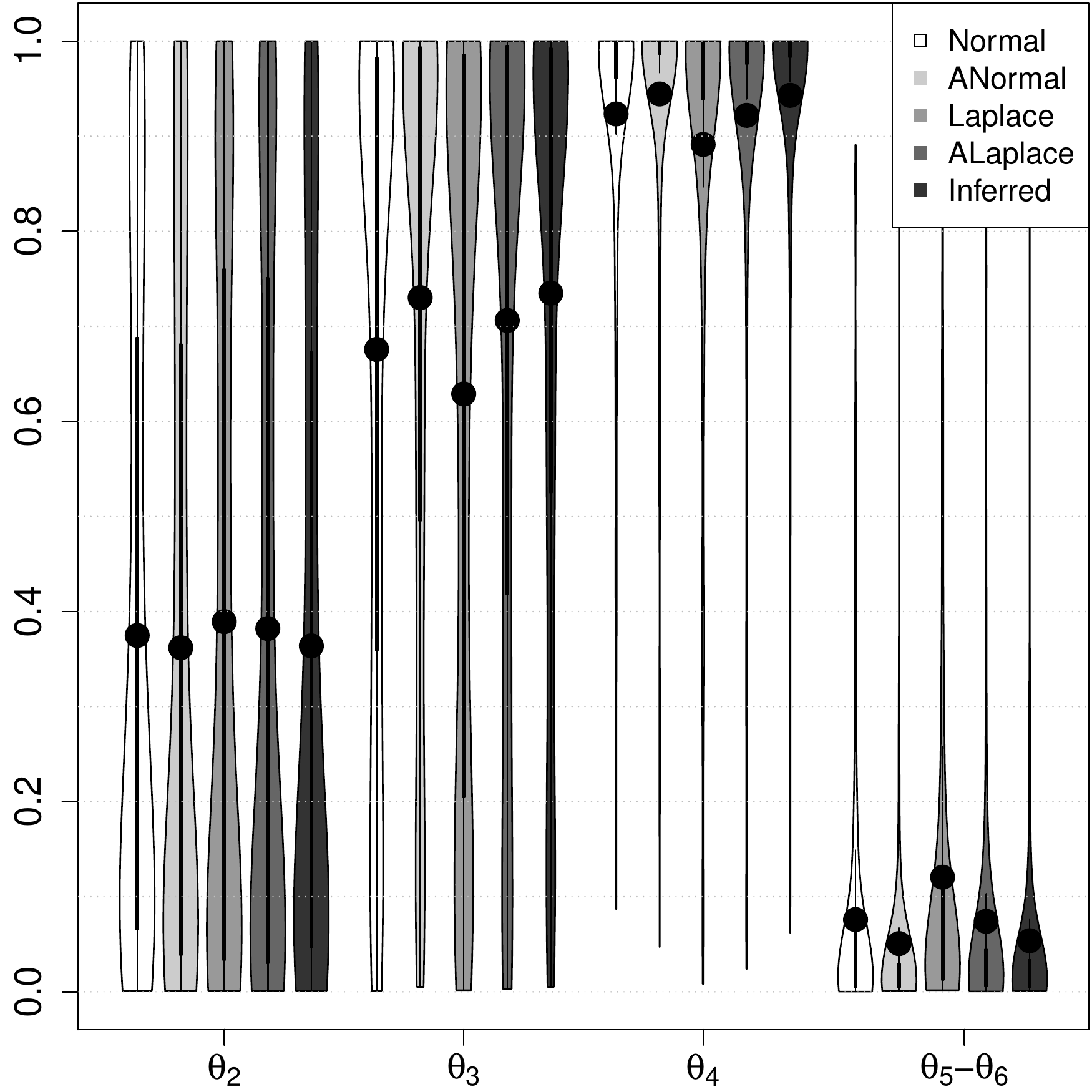} \\
$\epsilon_i \sim L(0,4)$ & $\epsilon_i \sim \mbox{AL}(0,4,-0.5)$ \\
\includegraphics[width=0.48\textwidth]{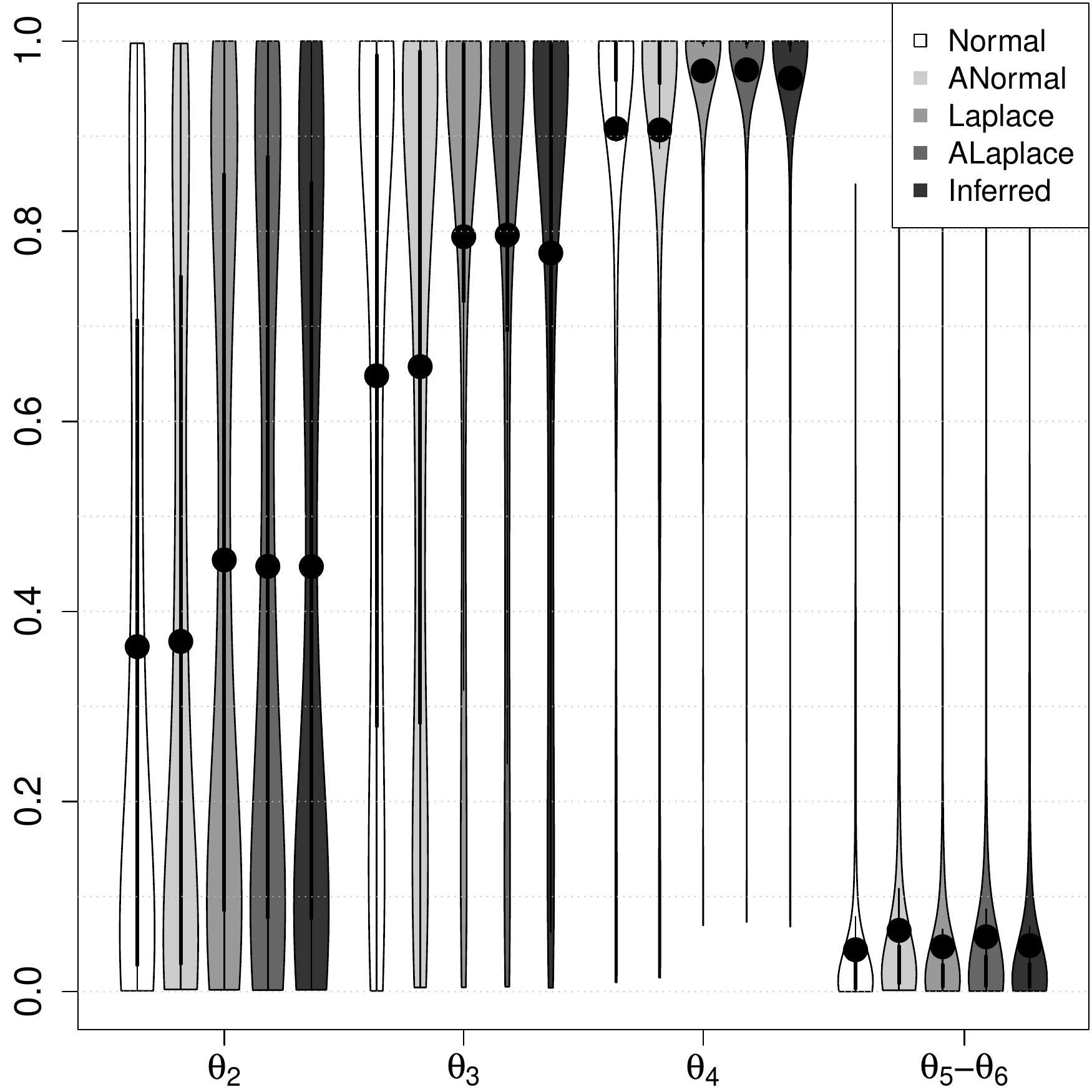} &
\includegraphics[width=0.48\textwidth]{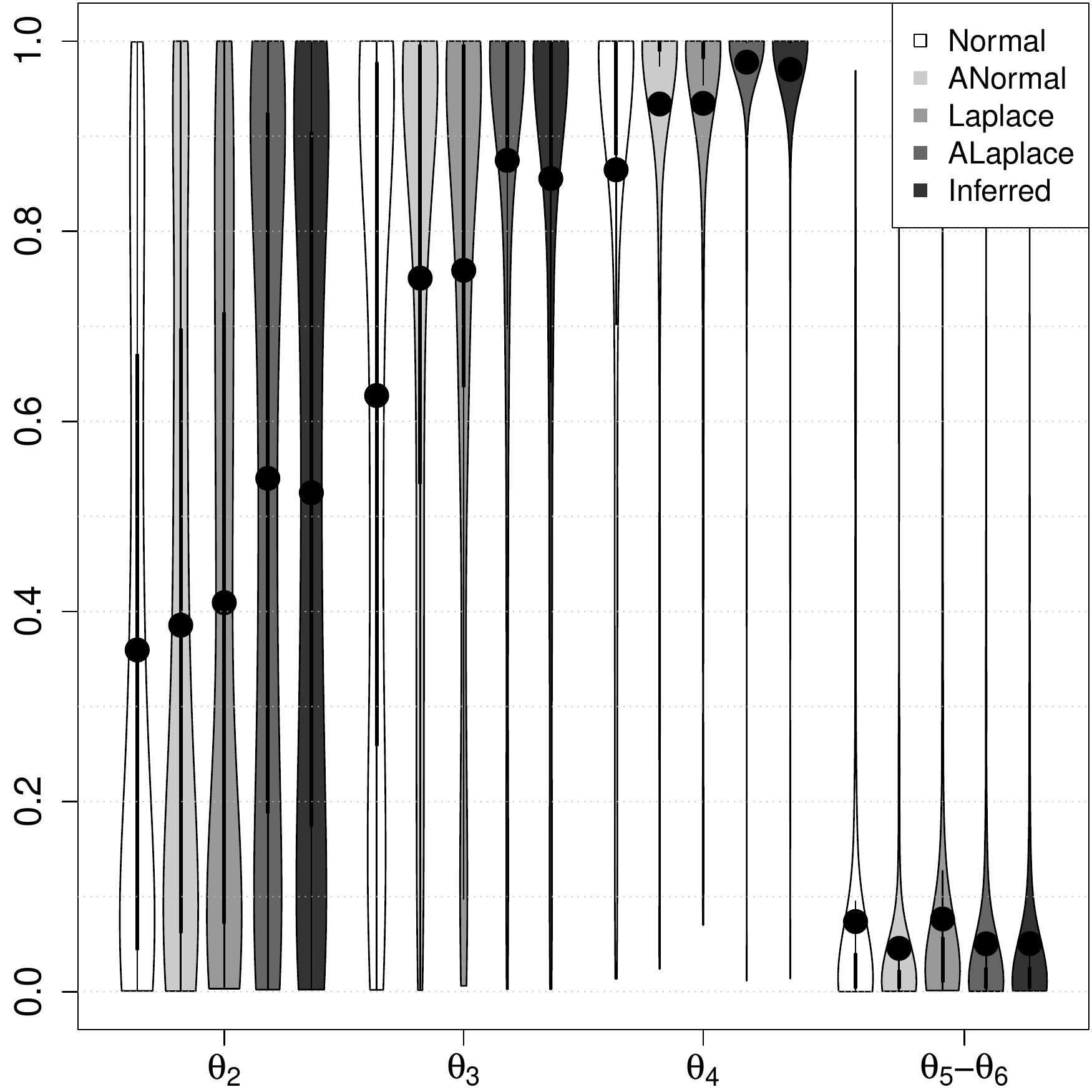} \\
\end{tabular}
\end{center}
\caption{Sensitivity analysis with $g_\alpha=0.087$.
$P(\theta_i \neq 0 \mid y)$ for $p=5$, $\vartheta=2$, $\theta=(0.5,1,1.5,0,0)$,
$n=100$, $\rho_{ij}=0.5$. Black circles show the mean.}
\label{fig:simres_margpp_priorskew1}
\end{figure}

\begin{figure}
\begin{center}
\begin{tabular}{cc}
$\epsilon_i \sim N(0,4)$ & $\epsilon_i \sim \mbox{AN}(0,4,-0.5)$ \\
\includegraphics[width=0.48\textwidth]{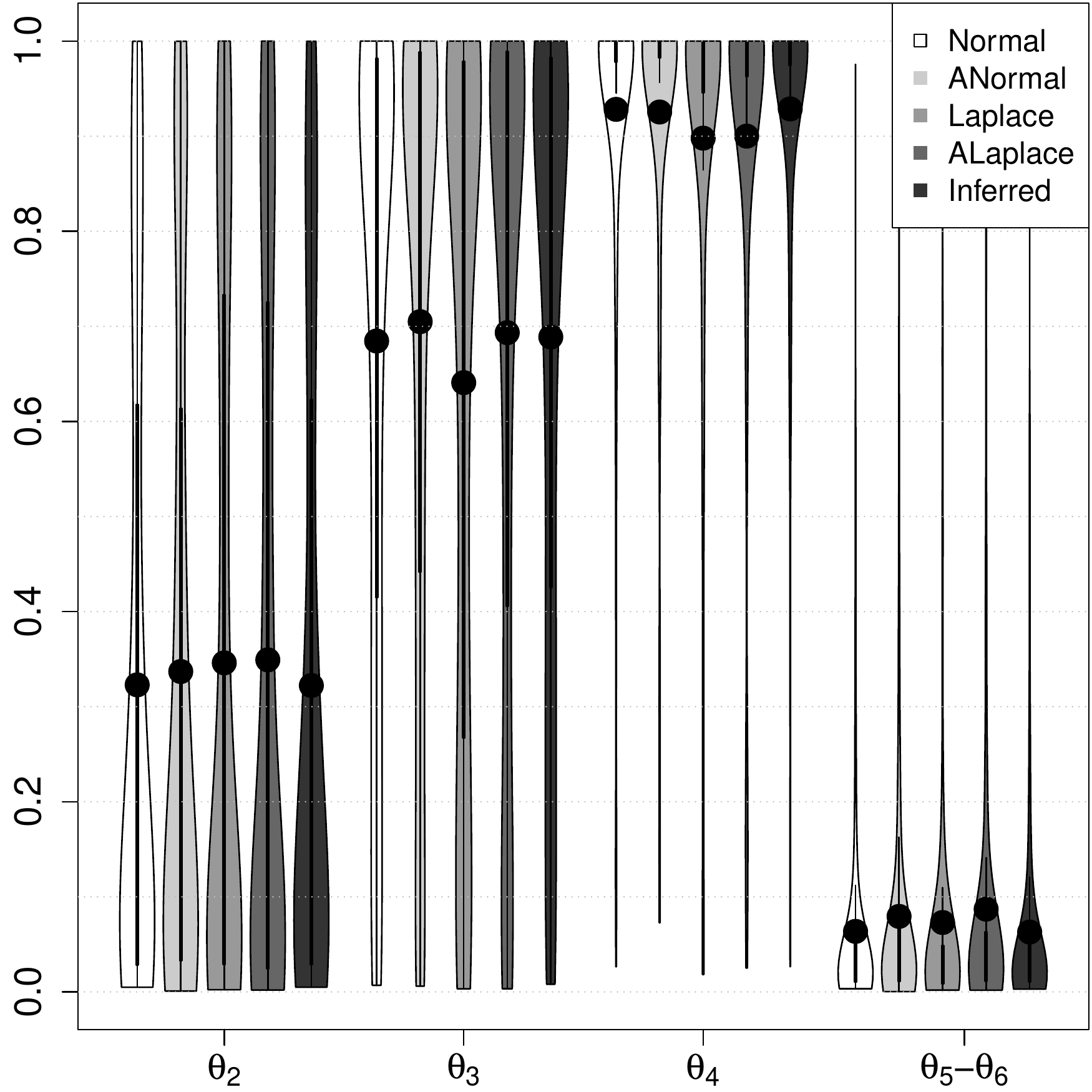} &
\includegraphics[width=0.48\textwidth]{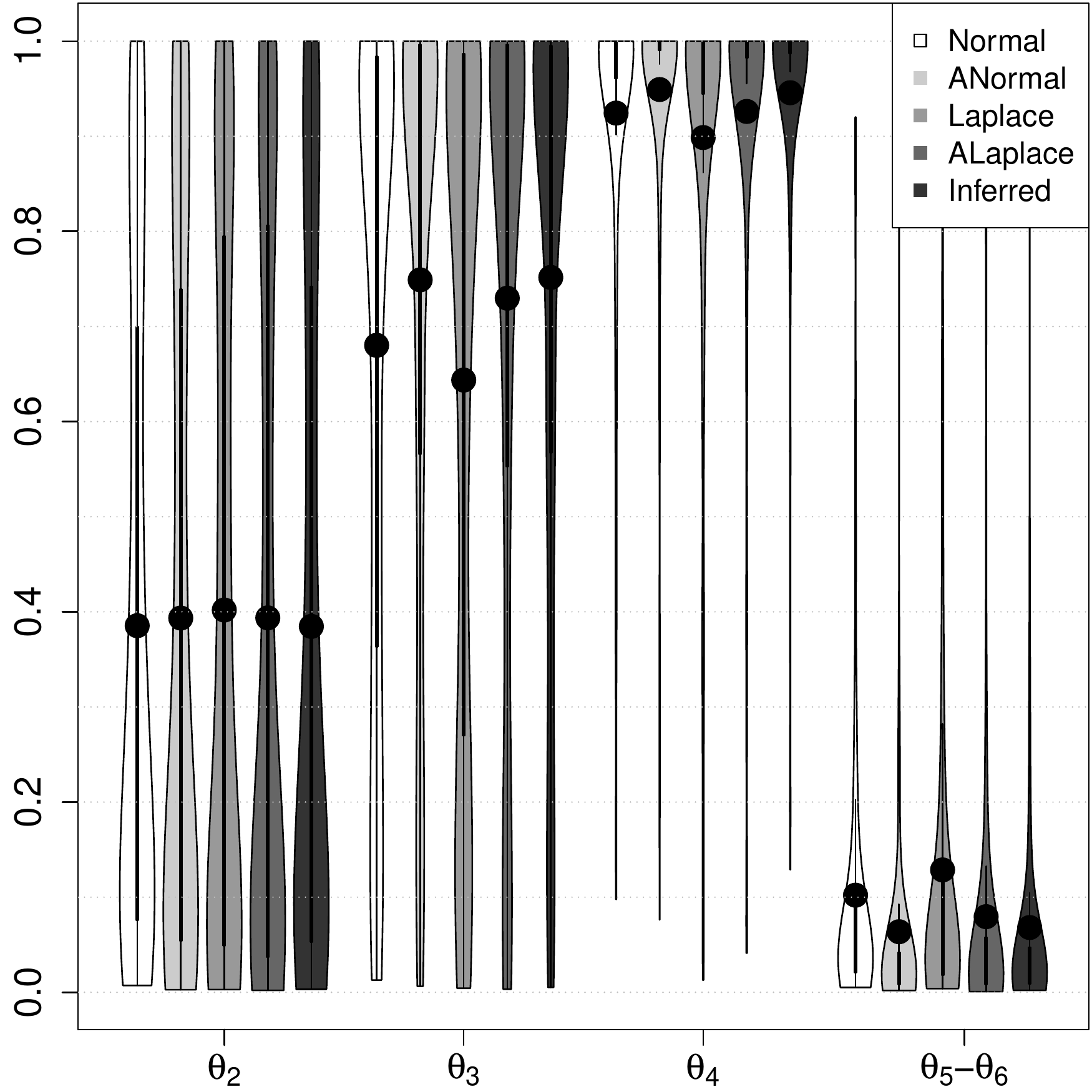} \\
$\epsilon_i \sim L(0,4)$ & $\epsilon_i \sim \mbox{AL}(0,4,-0.5)$ \\
\includegraphics[width=0.48\textwidth]{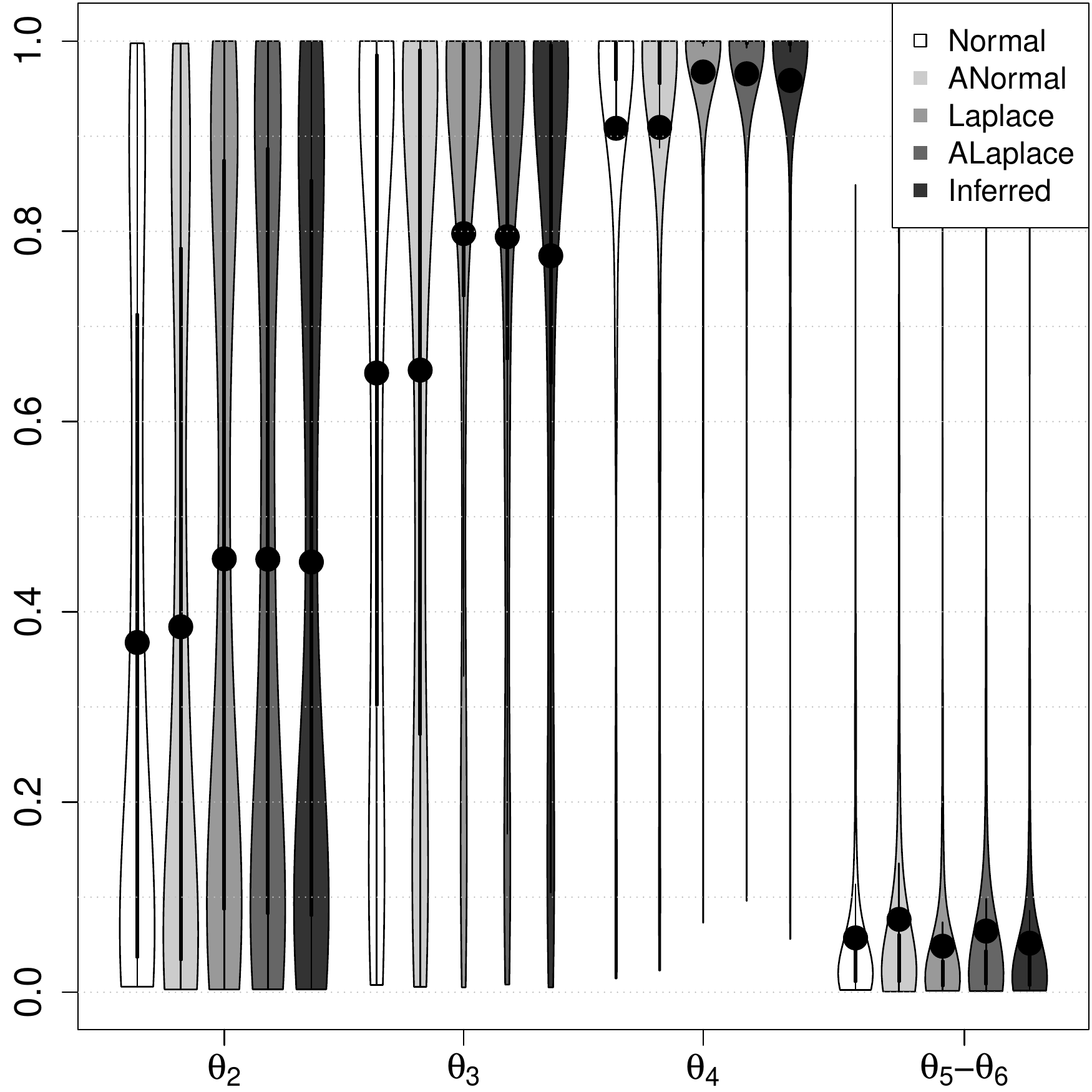} &
\includegraphics[width=0.48\textwidth]{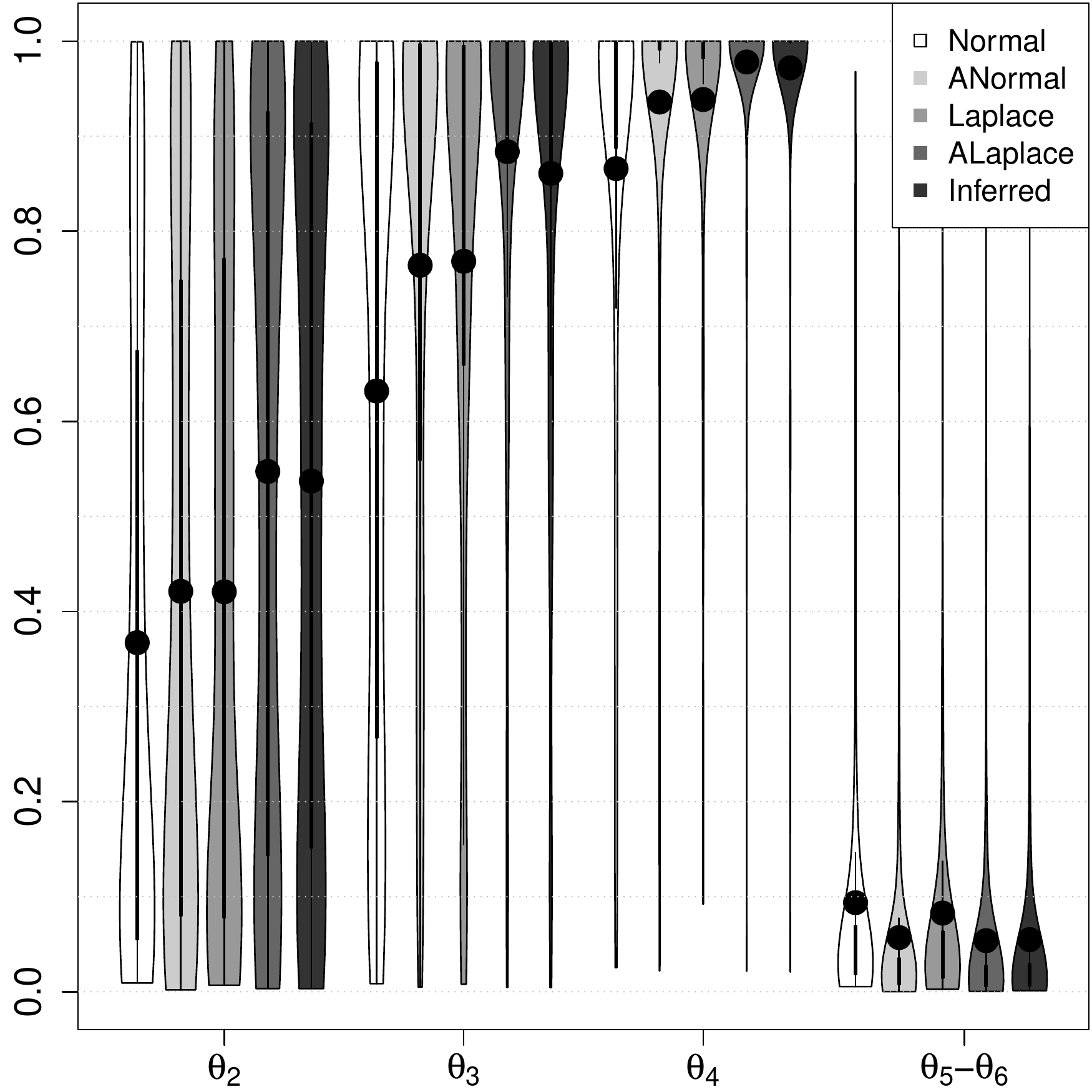} \\
\end{tabular}
\end{center}
\caption{Monte Carlo estimates ($B=10,000$) under $g_\alpha=0.357$.
$P(\theta_i \neq 0 \mid y)$ for $p=5$, $\vartheta=2$, $\theta=(0.5,1,1.5,0,0)$,
$n=100$, $\rho_{ij}=0.5$. Black circles show the mean.}
\label{fig:simres_margpp_mc}
\end{figure}

\begin{table}
\begin{center}
\begin{tabular}{|l|cccc|cccc|} \hline
 & \multicolumn{4}{c|}{$p=100$} & \multicolumn{4}{c|}{$p=500$} \\
 & $p(\gamma_0 \mid y)$ & $p(\widehat{\gamma}=\gamma_0)$ & FP & TP & $p(\gamma_0 \mid y)$ & $p(\widehat{\gamma}=\gamma_0)$ & FP & TP \\ \hline \hline
 & \multicolumn{8}{|c|}{Truly $\epsilon \sim N(0,1)$} \\
  Normal & 0.46 & 0.63 & 0.1 & 2.7 & 0.26 & 0.37 & 0.2 & 2.4 \\
  Two-piece Normal & 0.43 & 0.63 & 0.2 & 2.7 & 0.24 & 0.38 & 0.3 & 2.4 \\
  Laplace & 0.26 & 0.42 & 0.5 & 2.6 & 0.12 & 0.19 & 0.8 & 2.3 \\
  Two-piece Laplace & 0.23 & 0.39 & 0.7 & 2.6 & 0.12 & 0.21 & 0.9 & 2.3 \\
  Inferred & 0.45 & 0.62 & 0.2 & 2.7 & 0.25 & 0.37 & 0.2 & 2.4 \\
  LASSO-LS &  & 0.00 & 12.4 & 3.0 &  & 0.00 & 20.4 & 2.9 \\
  LASSO-LAD &  & 0.00 & 10.2 & 2.9 &  & 0.00 & 18.7 & 2.6 \\
  LASSO-QR &  & 0.00 & 10.2 & 2.9 &  & 0.00 & 18.7 & 2.6 \\
  SCAD &  & 0.07 & 4.2 & 2.9 &  & 0.01 & 7.3 & 2.8 \\
\hline \hline
 & \multicolumn{8}{|c|}{Truly $\epsilon \sim AN(0,1,-0.5)$} \\
  Normal & 0.38 & 0.55 & 0.2 & 2.6 & 0.21 & 0.34 & 0.5 & 2.4 \\
  Two-piece Normal & 0.59 & 0.73 & 0.1 & 2.8 & 0.40 & 0.55 & 0.4 & 2.6 \\
  Laplace & 0.20 & 0.35 & 0.7 & 2.5 & 0.07 & 0.14 & 1.2 & 2.4 \\
  Two-piece Laplace & 0.33 & 0.48 & 0.5 & 2.7 & 0.18 & 0.32 & 1.1 & 2.5 \\
  Inferred & 0.57 & 0.72 & 0.1 & 2.8 & 0.38 & 0.52 & 0.4 & 2.6 \\
  LASSO-LS &  & 0.00 & 12.4 & 3.0 &  & 0.00 & 21.9 & 2.9 \\
  LASSO-LAD &  & 0.00 & 9.8 & 2.8 &  & 0.00 & 18.1 & 2.6 \\
  LASSO-QR &  & 0.00 & 9.0 & 2.9 &  & 0.00 & 15.1 & 2.7 \\
  SCAD &  & 0.07 & 4.0 & 2.9 &  & 0.03 & 7.3 & 2.8 \\
\hline \hline
 & \multicolumn{8}{|c|}{Truly $\epsilon \sim L(0,1)$} \\
  Normal & 0.11 & 0.14 & 0.3 & 2.0 & 0.03 & 0.02 & 0.6 & 1.6 \\
  Two-piece Normal & 0.11 & 0.15 & 0.3 & 2.1 & 0.04 & 0.04 & 1.1 & 1.7 \\
  Laplace & 0.29 & 0.38 & 0.2 & 2.4 & 0.13 & 0.19 & 0.4 & 2.0 \\
  Two-piece Laplace & 0.28 & 0.35 & 0.3 & 2.4 & 0.12 & 0.18 & 0.5 & 2.0 \\
  Inferred & 0.28 & 0.38 & 0.2 & 2.4 & 0.12 & 0.18 & 0.4 & 2.0 \\
  LASSO-LS &  & 0.00 & 11.3 & 2.8 &  & 0.00 & 21.4 & 2.5 \\
  LASSO-LAD &  & 0.01 & 9.7 & 2.8 &  & 0.00 & 17.8 & 2.5 \\
  LASSO-QR &  & 0.01 & 9.7 & 2.8 &  & 0.00 & 17.8 & 2.5 \\
  SCAD &  & 0.02 & 5.0 & 2.7 &  & 0.00 & 9.0 & 2.4 \\
\hline \hline
 & \multicolumn{8}{|c|}{Truly $\epsilon \sim AL(0,-0.5)$} \\
  Normal & 0.07 & 0.10 & 0.4 & 1.9 & 0.02 & 0.02 & 1.1 & 1.5 \\
  Two-piece Normal & 0.21 & 0.27 & 0.2 & 2.2 & 0.11 & 0.15 & 0.3 & 2.0 \\
  Laplace & 0.16 & 0.19 & 0.4 & 2.1 & 0.05 & 0.07 & 0.7 & 1.8 \\
  Two-piece Laplace & 0.43 & 0.51 & 0.2 & 2.5 & 0.27 & 0.34 & 0.4 & 2.3 \\
  Inferred & 0.41 & 0.48 & 0.2 & 2.5 & 0.25 & 0.33 & 0.4 & 2.2 \\
  LASSO-LS &  & 0.00 & 11.6 & 2.8 &  & 0.00 & 20.1 & 2.5 \\
  LASSO-LAD &  & 0.00 & 9.9 & 2.7 &  & 0.00 & 17.5 & 2.3 \\
  LASSO-QR &  & 0.00 & 9.0 & 2.8 &  & 0.00 & 15.2 & 2.5 \\
  SCAD &  & 0.01 & 5.2 & 2.6 &  & 0.01 & 9.4 & 2.3 \\
\hline
\end{tabular}
\end{center}
\caption{Simulation results under $\vartheta=1$.
$\gamma_0$: true predictors. $\widehat{\gamma}$: selected variables.
CC: number of correctly classified variables ($\sum_{j=1}^{p} \mbox{I}(\widehat{\gamma}_j=\gamma_{0j})$).
FP: number of false positives; TP: number of true positives.
LASSO-LAD and LASSO-QR are equivalent when $\alpha=0$}
\label{tab:simres_vtheta1}
\end{table}

\begin{table}
\begin{center}
\begin{tabular}{|l|cccc|cccc|} \hline
 & \multicolumn{4}{c|}{$p=100$} & \multicolumn{4}{c|}{$p=500$} \\
 & $p(\gamma_0 \mid y)$ & $p(\widehat{\gamma}=\gamma_0)$ & FP & TP & $p(\gamma_0 \mid y)$ & $p(\widehat{\gamma}=\gamma_0)$ & FP & TP \\ \hline \hline
 & \multicolumn{8}{|c|}{Truly $\epsilon \sim N(0,1)$} \\
 Normal & 0.01 & 0.01 & 0.4 & 1.2 & 0.00 & 0.00 & 0.8 & 0.9 \\
 Two-piece Normal & 0.01 & 0.01 & 0.5 & 1.2 & 0.00 & 0.00 & 0.9 & 0.8 \\
 Laplace & 0.00 & 0.00 & 0.7 & 1.1 & 0.00 & 0.00 & 1.0 & 0.8 \\
 Two-piece Laplace & 0.00 & 0.01 & 0.8 & 1.1 & 0.00 & 0.00 & 1.1 & 0.8 \\
 Inferred & 0.01 & 0.01 & 0.5 & 1.2 & 0.00 & 0.00 & 0.7 & 0.9 \\
 LASSO-LS &  & 0.00 & 11.9 & 2.5 &  & 0.00 & 18.0 & 2.0 \\
 LASSO-LAD &  & 0.00 & 8.9 & 2.0 &  & 0.00 & 15.6 & 1.4 \\
 LASSO-QR &  & 0.00 & 8.9 & 2.0 &  & 0.00 & 15.6 & 1.4 \\
 SCAD &  & 0.00 & 6.3 & 2.3 &  & 0.01 & 10.4 & 1.8 \\
\hline \hline
 & \multicolumn{8}{|c|}{Truly $\epsilon \sim AN(0,1,-0.5)$} \\
  Normal & 0.00 & 0.00 & 0.5 & 1.2 & 0.00 & 0.00 & 0.7 & 0.9 \\
  Two-piece Normal & 0.01 & 0.01 & 0.4 & 1.4 & 0.00 & 0.01 & 0.7 & 1.1 \\
  Laplace & 0.00 & 0.00 & 0.9 & 1.0 & 0.00 & 0.00 & 1.4 & 0.7 \\
  Two-piece Laplace & 0.01 & 0.01 & 0.7 & 1.2 & 0.00 & 0.00 & 1.5 & 1.0 \\
  Inferred & 0.01 & 0.01 & 0.4 & 1.4 & 0.00 & 0.01 & 0.9 & 1.0 \\
  LASSO-LS &  & 0.00 & 11.0 & 2.4 &  & 0.00 & 19.4 & 1.9 \\
  LASSO-LAD &  & 0.00 & 8.6 & 1.8 &  & 0.00 & 15.3 & 1.4 \\
  LASSO-QR &  & 0.00 & 8.1 & 2.1 &  & 0.00 & 12.8 & 1.5 \\
  SCAD &  & 0.00 & 6.1 & 2.1 &  & 0.00 & 10.1 & 1.8 \\
\hline \hline
 & \multicolumn{8}{|c|}{Truly $\epsilon \sim L(0,1)$} \\
  Normal & 0.01 & 0.01 & 0.4 & 1.3 & 0.00 & 0.00 & 0.8 & 0.9 \\
  Two-piece Normal & 0.01 & 0.01 & 0.5 & 1.3 & 0.00 & 0.00 & 0.9 & 1.0 \\
  Laplace & 0.05 & 0.06 & 0.4 & 1.7 & 0.01 & 0.01 & 0.7 & 1.2 \\
  Two-piece Laplace & 0.05 & 0.07 & 0.4 & 1.7 & 0.01 & 0.01 & 0.8 & 1.2 \\
  Inferred & 0.04 & 0.04 & 0.3 & 1.7 & 0.01 & 0.01 & 0.7 & 1.2 \\
  LASSO-LS &  & 0.00 & 10.8 & 2.5 &  & 0.00 & 20.4 & 2.0 \\
  LASSO-LAD &  & 0.01 & 9.3 & 2.5 &  & 0.00 & 17.1 & 2.0 \\
  LASSO-QR &  & 0.01 & 9.3 & 2.5 &  & 0.00 & 17.1 & 2.0 \\
  SCAD &  & 0.00 & 5.9 & 2.2 &  & 0.00 & 10.3 & 1.8 \\
\hline \hline
 & \multicolumn{8}{|c|}{Truly $\epsilon \sim AL(0,-0.5)$} \\
  Normal & 0.00 & 0.00 & 0.5 & 1.1 & 0.00 & 0.00 & 0.9 & 0.8 \\
  Two-piece Normal & 0.02 & 0.01 & 0.4 & 1.5 & 0.01 & 0.01 & 0.6 & 1.2 \\
  Laplace & 0.02 & 0.01 & 0.6 & 1.3 & 0.00 & 0.01 & 0.8 & 1.0 \\
  Two-piece Laplace & 0.09 & 0.12 & 0.3 & 1.9 & 0.04 & 0.05 & 0.7 & 1.5 \\
  Inferred & 0.09 & 0.10 & 0.3 & 1.8 & 0.04 & 0.05 & 0.6 & 1.4 \\
  LASSO-LS &  & 0.00 & 10.9 & 2.3 &  & 0.00 & 18.0 & 1.8 \\
  LASSO-LAD &  & 0.00 & 9.4 & 2.3 &  & 0.00 & 15.6 & 1.7 \\
  LASSO-QR &  & 0.00 & 8.3 & 2.5 &  & 0.00 & 14.0 & 2.0 \\
  SCAD &  & 0.01 & 5.7 & 2.1 &  & 0.00 & 10.2 & 1.6 \\
\hline
\end{tabular}
\end{center}
\caption{Simulation results under $\vartheta=2$.
$\gamma_0$: true predictors. $\widehat{\gamma}$: selected variables.
CC: number of correctly classified variables ($\sum_{j=1}^{p} \mbox{I}(\widehat{\gamma}_j=\gamma_{0j})$).
FP: number of false positives; TP: number of true positives.
LASSO-LAD and LASSO-QR are equivalent when $\alpha=0$}
\label{tab:simres_vtheta2}
\end{table}

\begin{figure}
\begin{center}
\begin{tabular}{cc}
$\epsilon_i \sim N(0,4)$ & $\epsilon_i \sim \mbox{AN}(0,4,-0.5)$ \\
\includegraphics[width=0.48\textwidth]{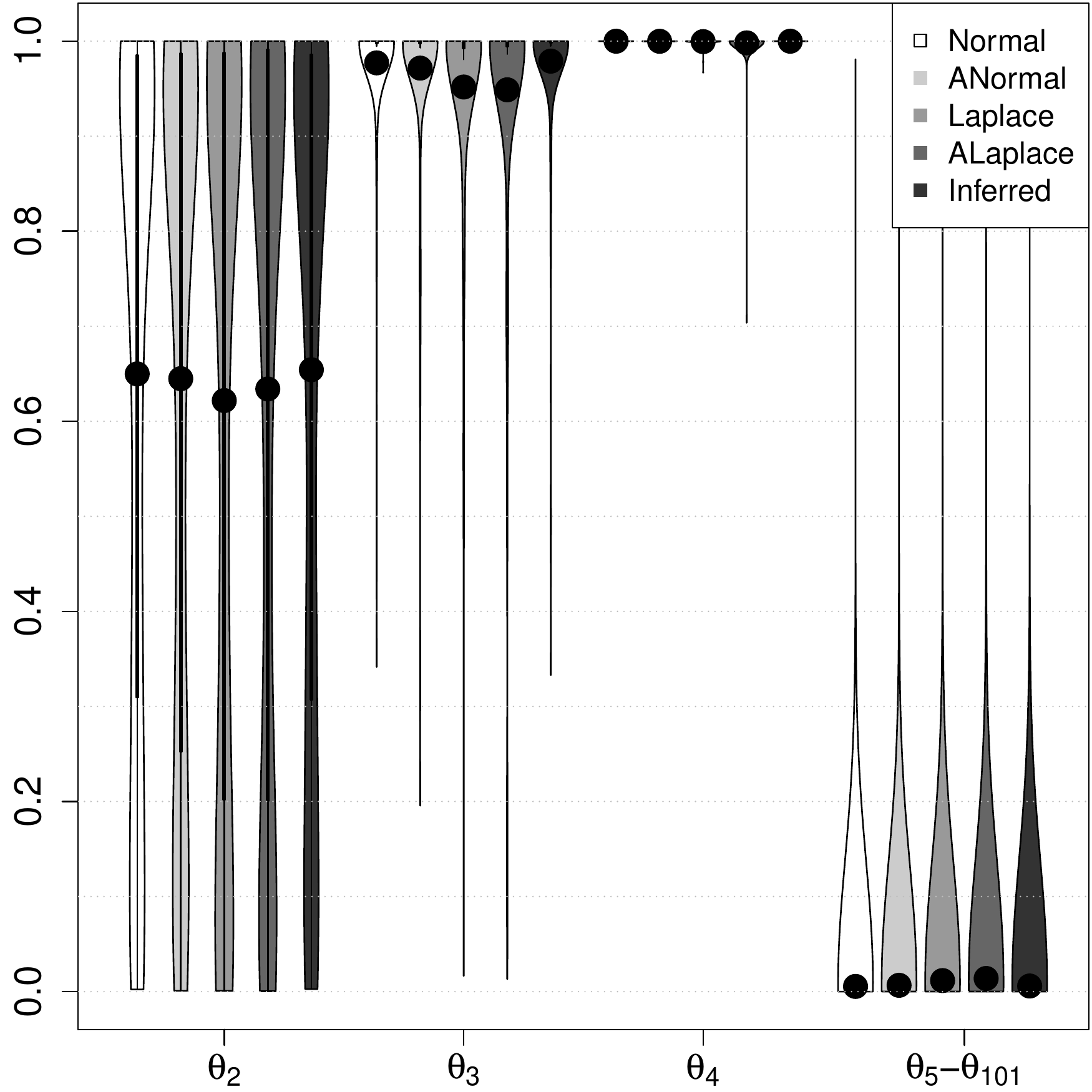} &
\includegraphics[width=0.48\textwidth]{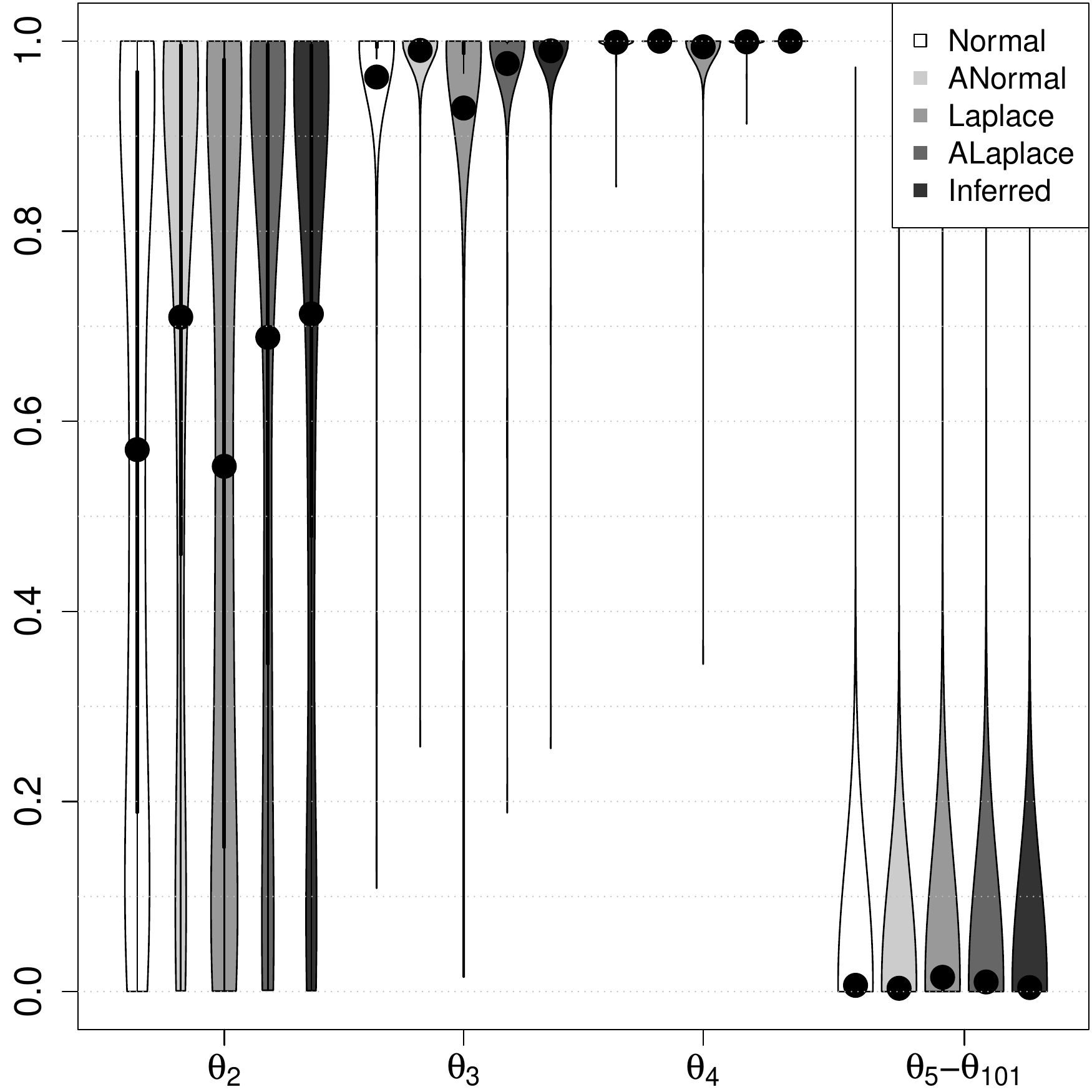} \\
$\epsilon_i \sim L(0,4)$ & $\epsilon_i \sim \mbox{AL}(0,4,-0.5)$ \\
\includegraphics[width=0.48\textwidth]{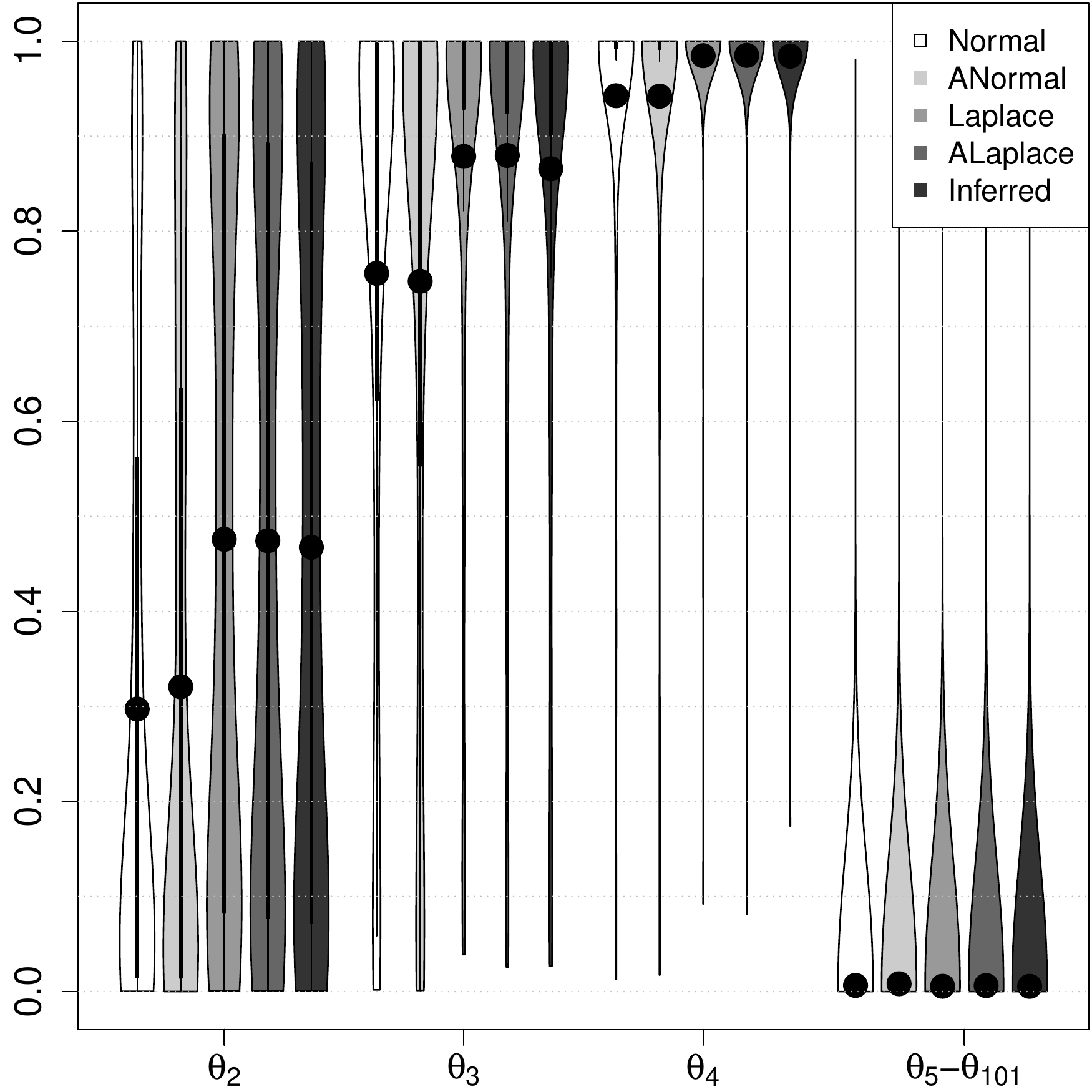} &
\includegraphics[width=0.48\textwidth]{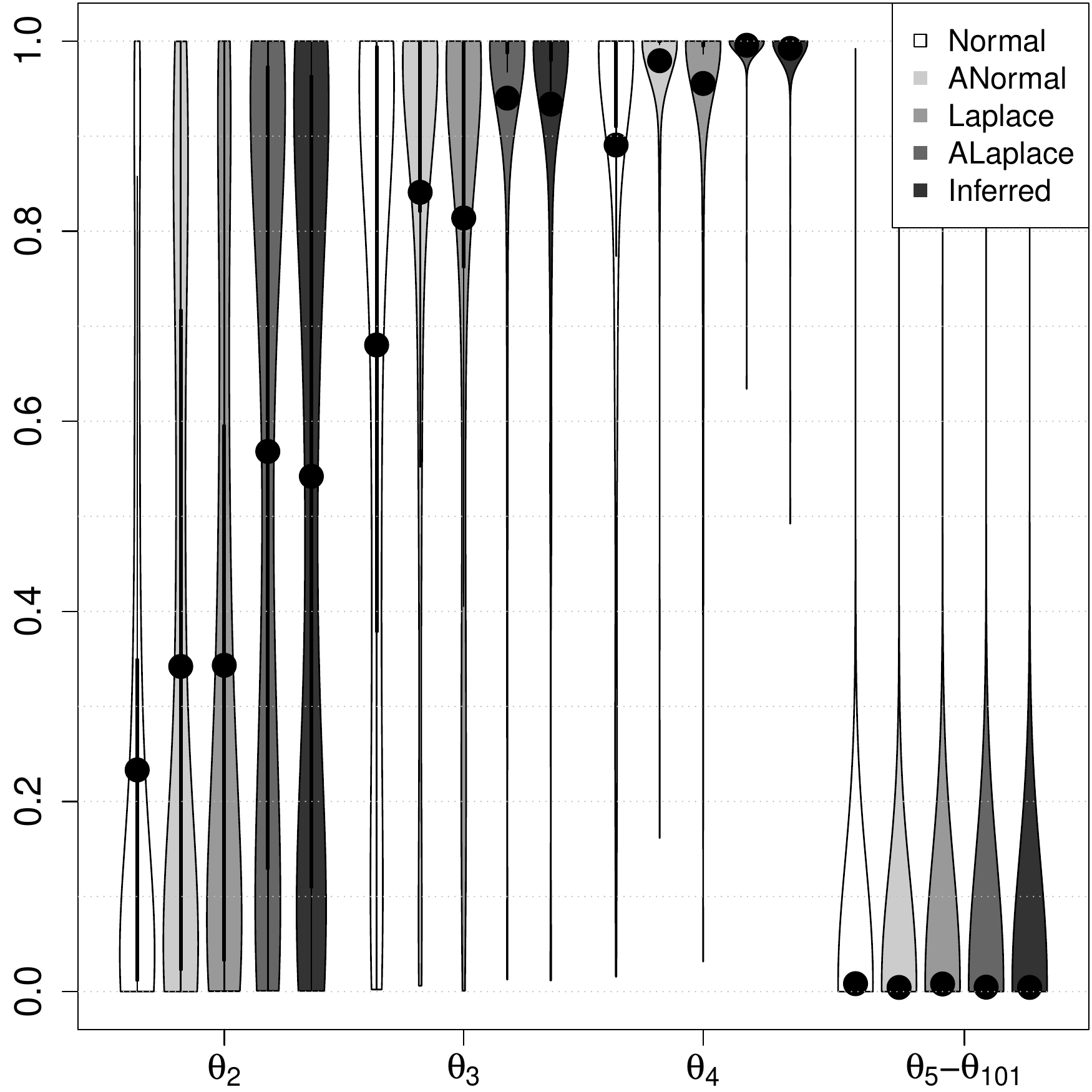} \\
\end{tabular}
\end{center}
\caption{$P(\theta_i \neq 0 \mid y)$ for $p=100$, $\vartheta=1$, $\theta=(0,0.5,1,1.5,0,\ldots,0)$,
$n=100$, $\rho_{ij}=0.5$. Black circles show the mean.}
\label{fig:simres_margpp_p100_vartheta1}
\end{figure}

\begin{figure}
\begin{center}
\begin{tabular}{cc}
$\epsilon_i \sim N(0,4)$ & $\epsilon_i \sim \mbox{AN}(0,4,-0.5)$ \\
\includegraphics[width=0.48\textwidth]{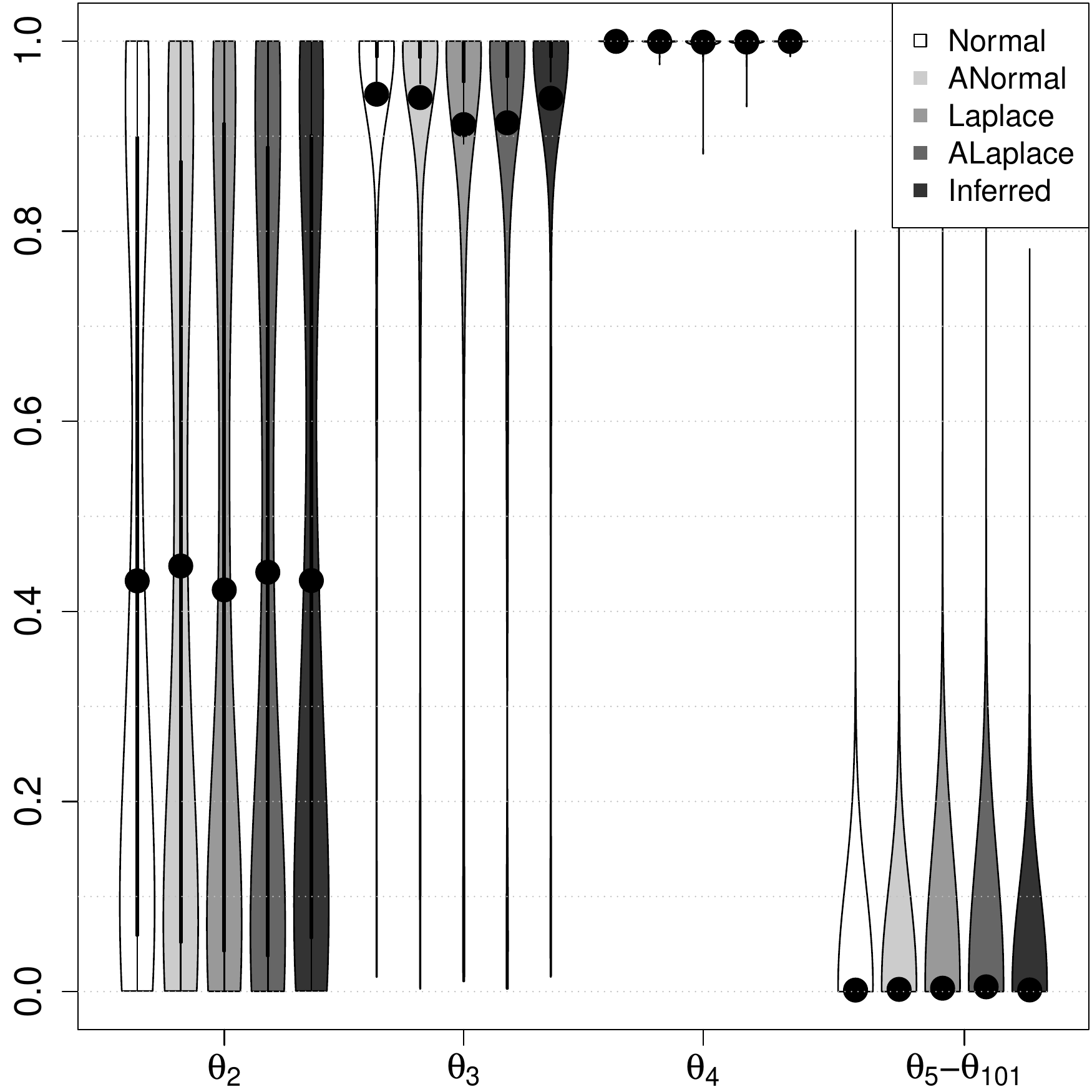} &
\includegraphics[width=0.48\textwidth]{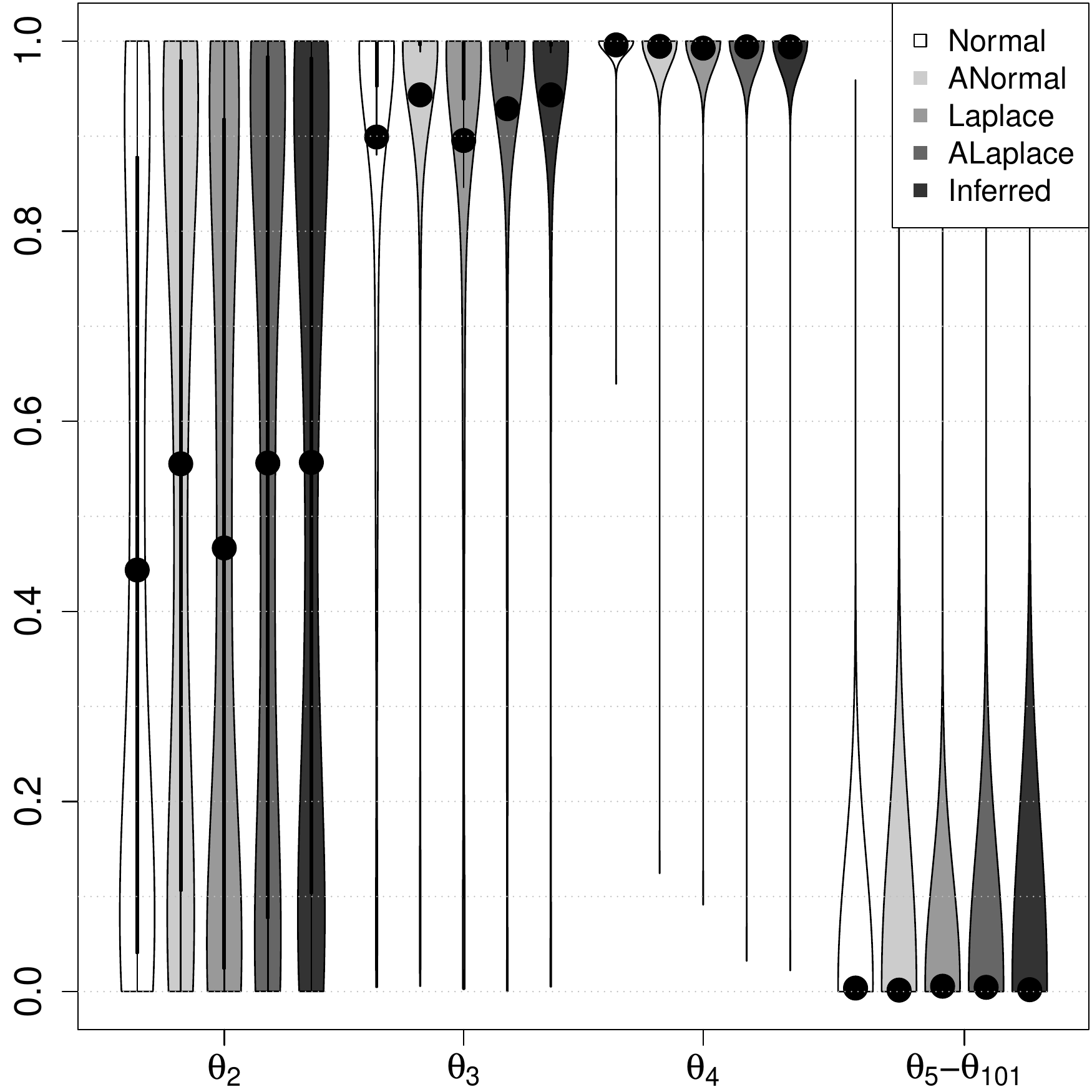} \\
$\epsilon_i \sim L(0,4)$ & $\epsilon_i \sim \mbox{AL}(0,4,-0.5)$ \\
\includegraphics[width=0.48\textwidth]{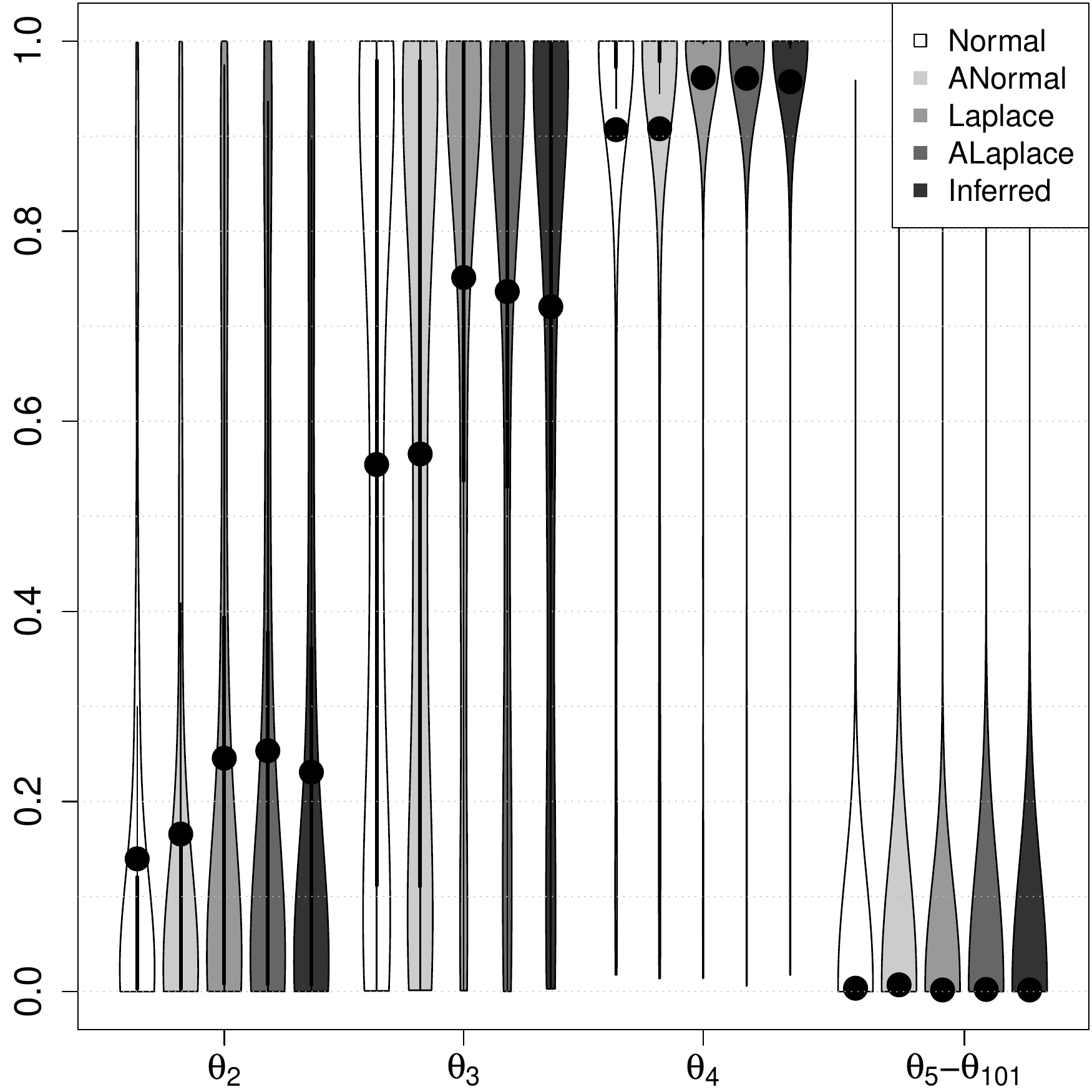} &
\includegraphics[width=0.48\textwidth]{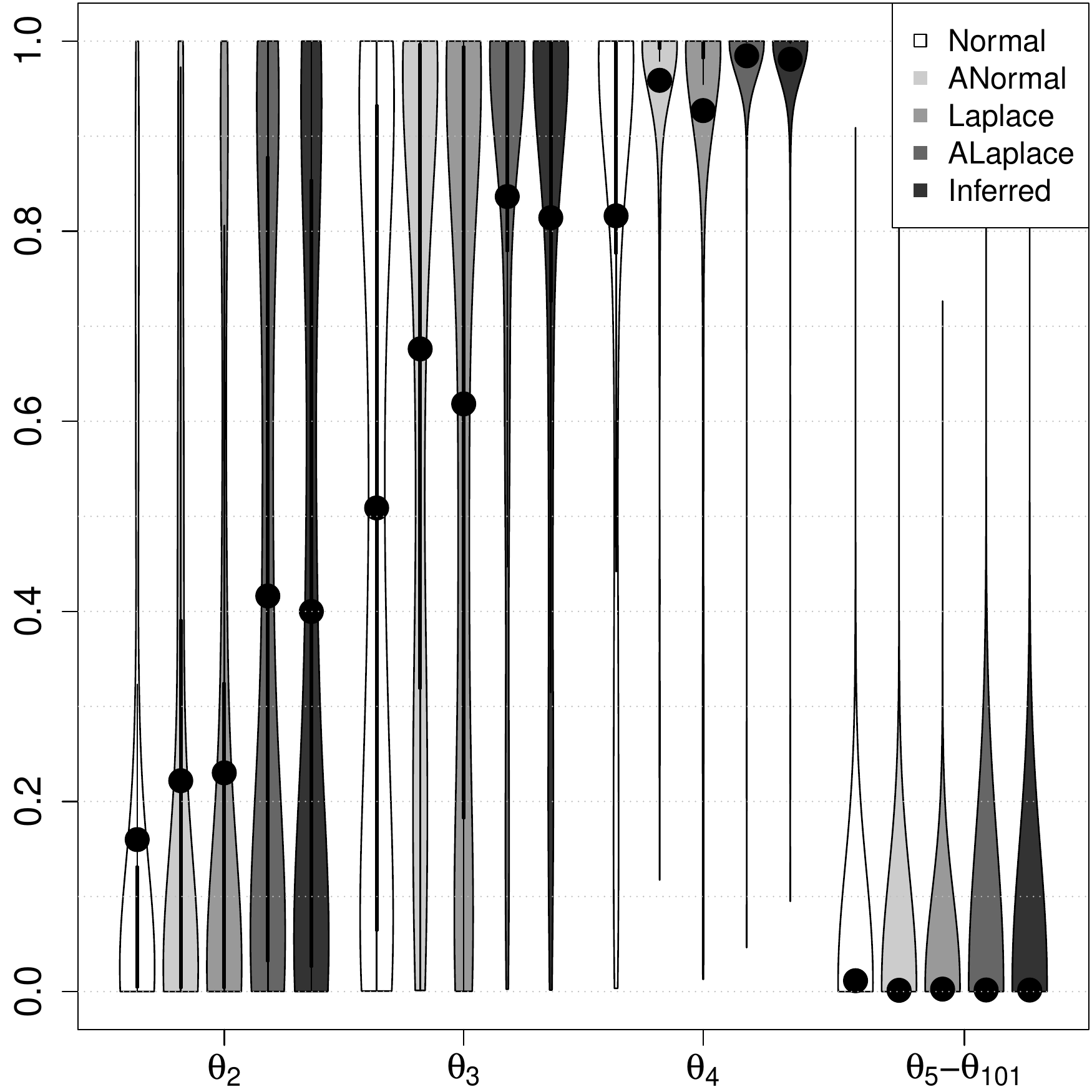} \\
\end{tabular}
\end{center}
\caption{$P(\theta_i \neq 0 \mid y)$ for $p=500$, $\vartheta=1$, $\theta=(0,0.5,1,1.5,0,\ldots,0)$,
$n=100$, $\rho_{ij}=0.5$. Black circles show the mean.}
\label{fig:simres_margpp_p500_vartheta1}
\end{figure}

\begin{figure}
\begin{center}
\begin{tabular}{cc}
$\epsilon_i \sim N(0,4)$ & $\epsilon_i \sim \mbox{AN}(0,4,-0.5)$ \\
\includegraphics[width=0.48\textwidth]{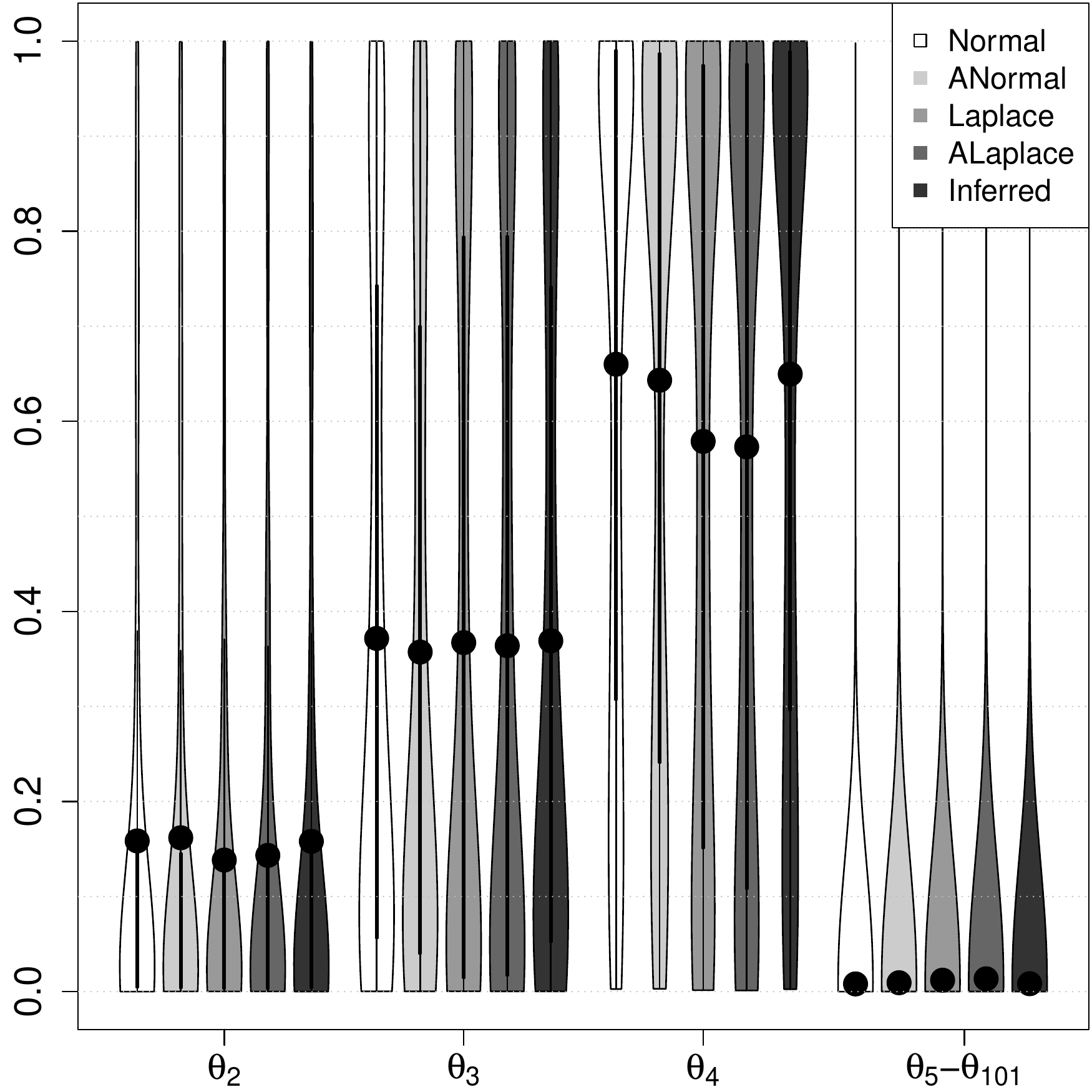} &
\includegraphics[width=0.48\textwidth]{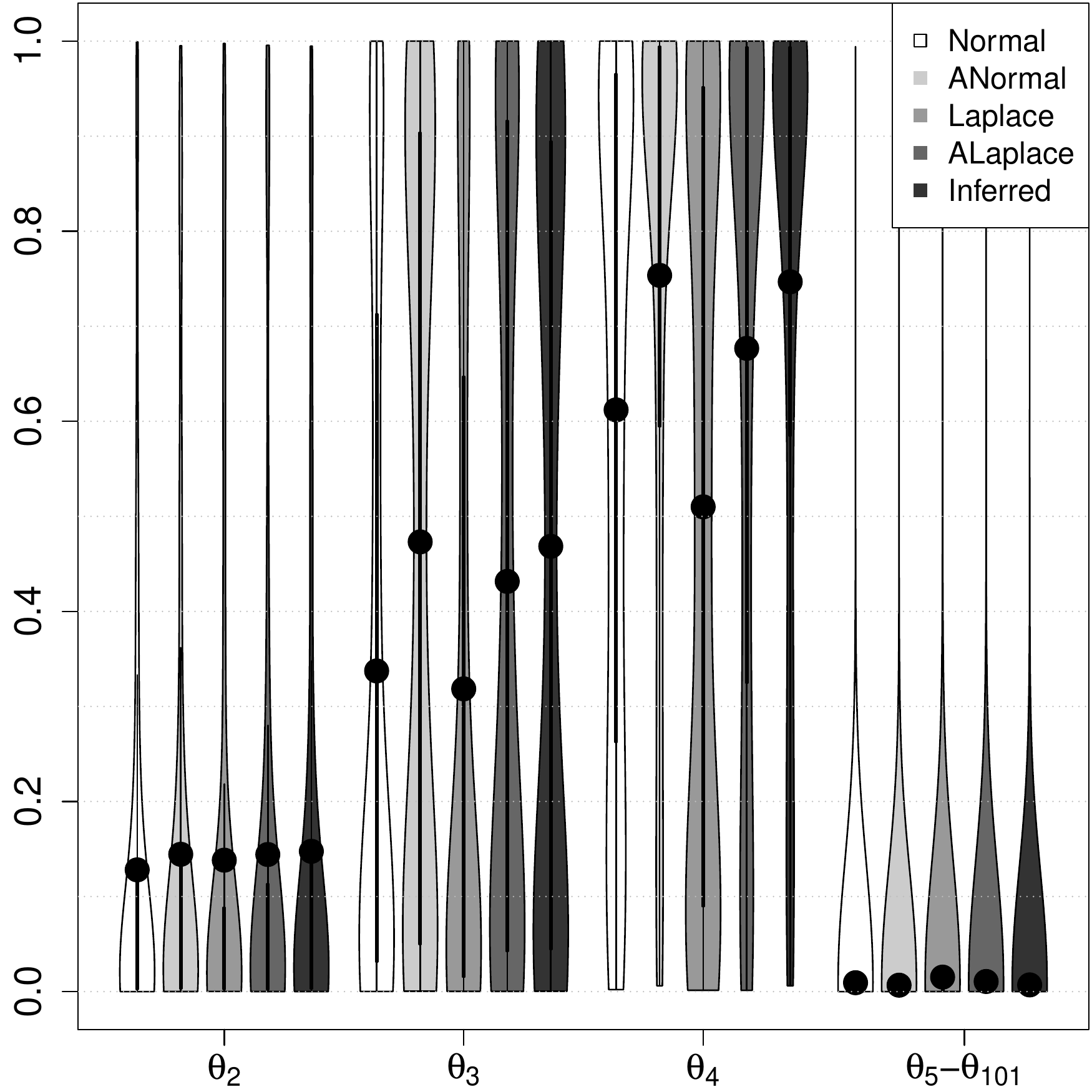} \\
$\epsilon_i \sim L(0,4)$ & $\epsilon_i \sim \mbox{AL}(0,4,-0.5)$ \\
\includegraphics[width=0.48\textwidth]{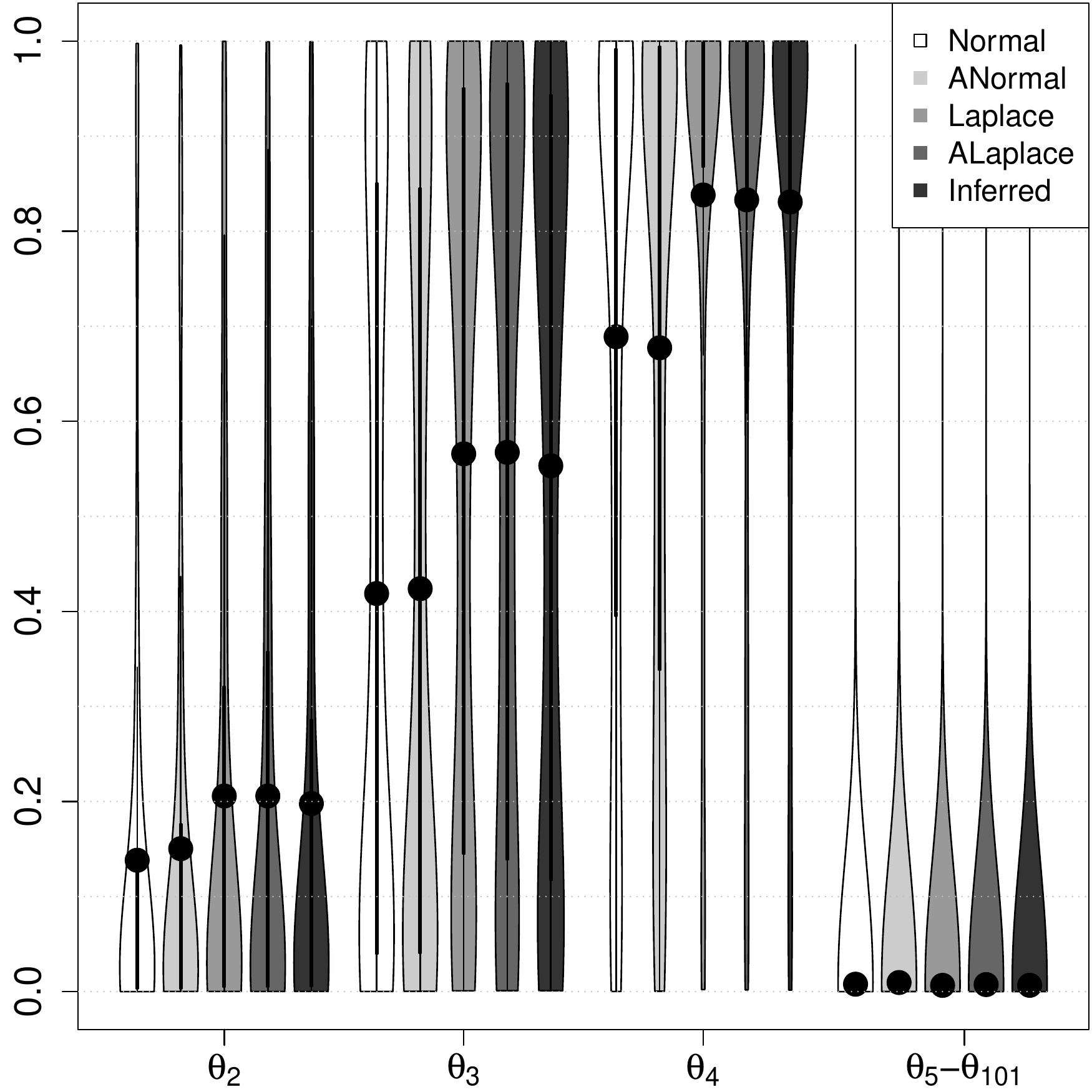} &
\includegraphics[width=0.48\textwidth]{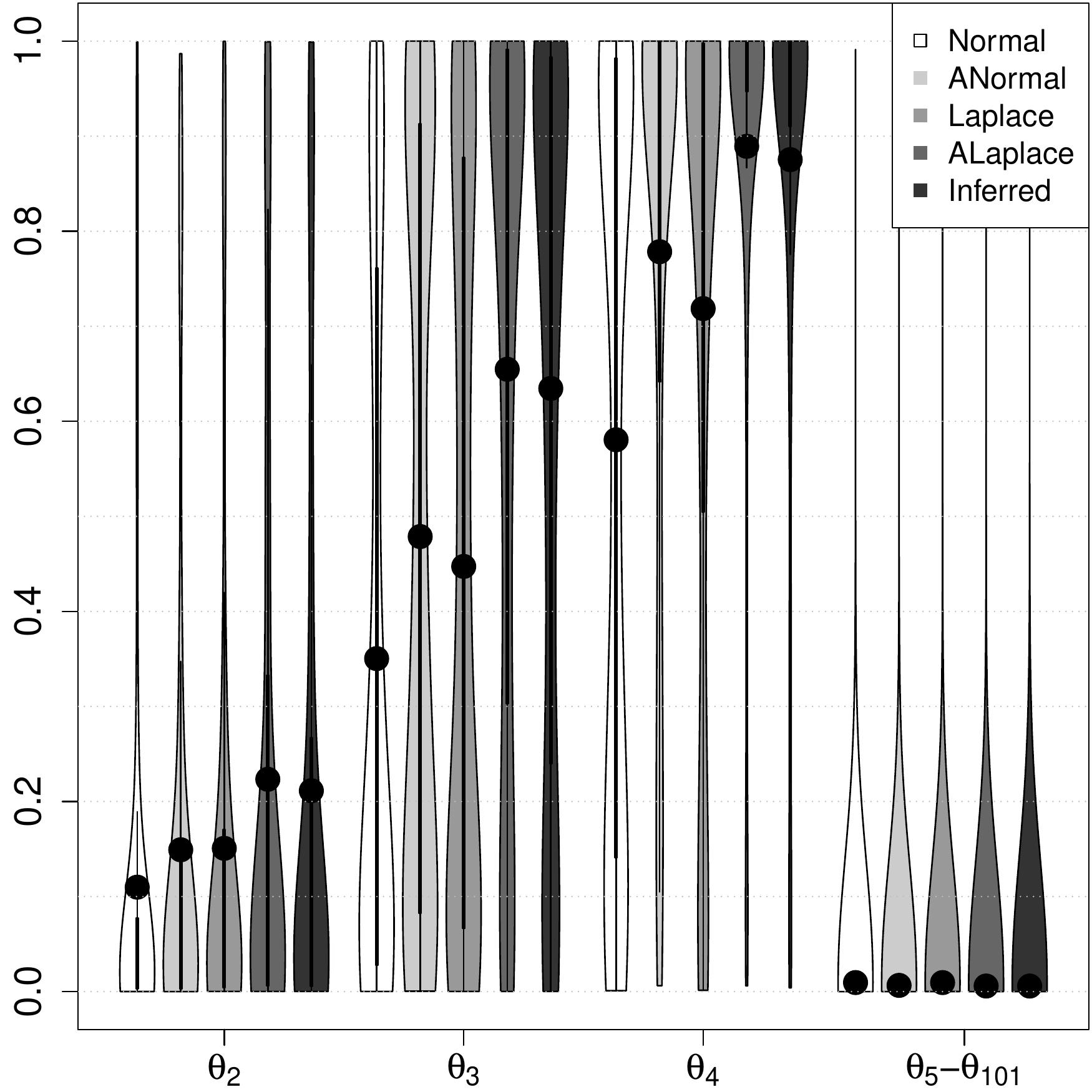} \\
\end{tabular}
\end{center}
\caption{$P(\theta_i \neq 0 \mid y)$ for $p=100$, $\vartheta=2$, $\theta=(0,0.5,1,1.5,0,\ldots,0)$,
$n=100$, $\rho_{ij}=0.5$. Black circles show the mean.}
\label{fig:simres_margpp_p100_vartheta2}
\end{figure}

\begin{figure}
\begin{center}
\begin{tabular}{cc}
$\epsilon_i \sim N(0,4)$ & $\epsilon_i \sim \mbox{AN}(0,4,-0.5)$ \\
\includegraphics[width=0.48\textwidth]{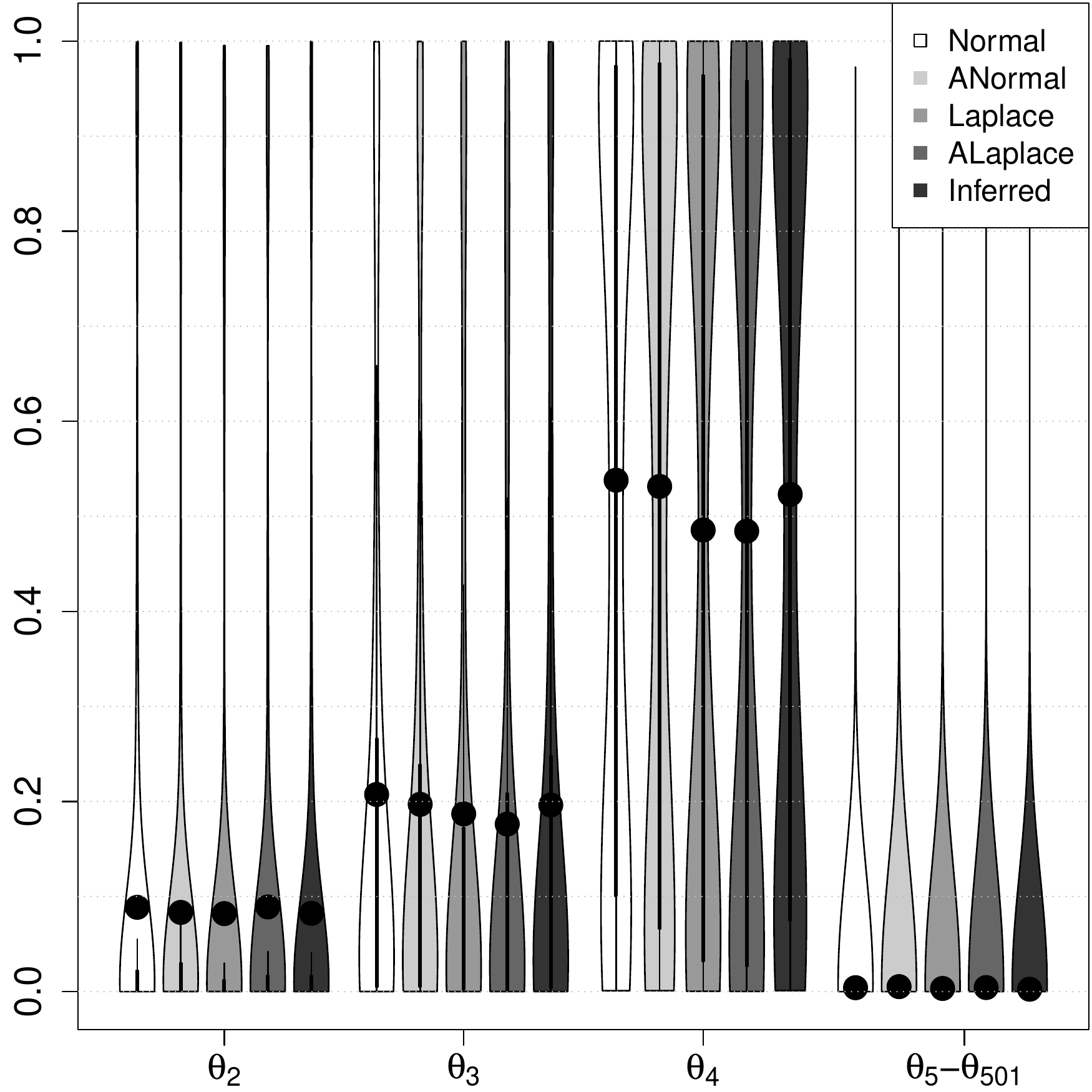} &
\includegraphics[width=0.48\textwidth]{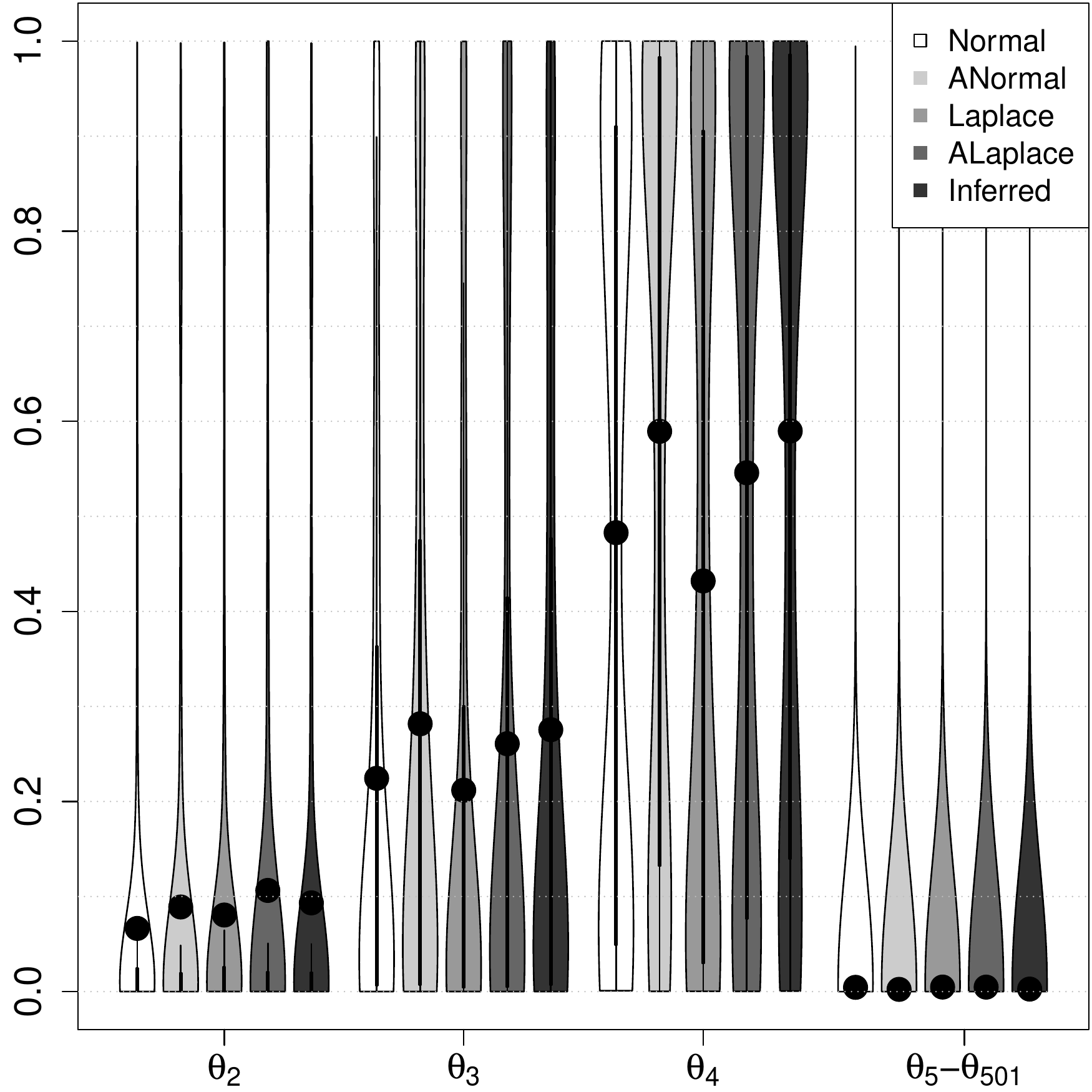} \\
$\epsilon_i \sim L(0,4)$ & $\epsilon_i \sim \mbox{AL}(0,4,-0.5)$ \\
\includegraphics[width=0.48\textwidth]{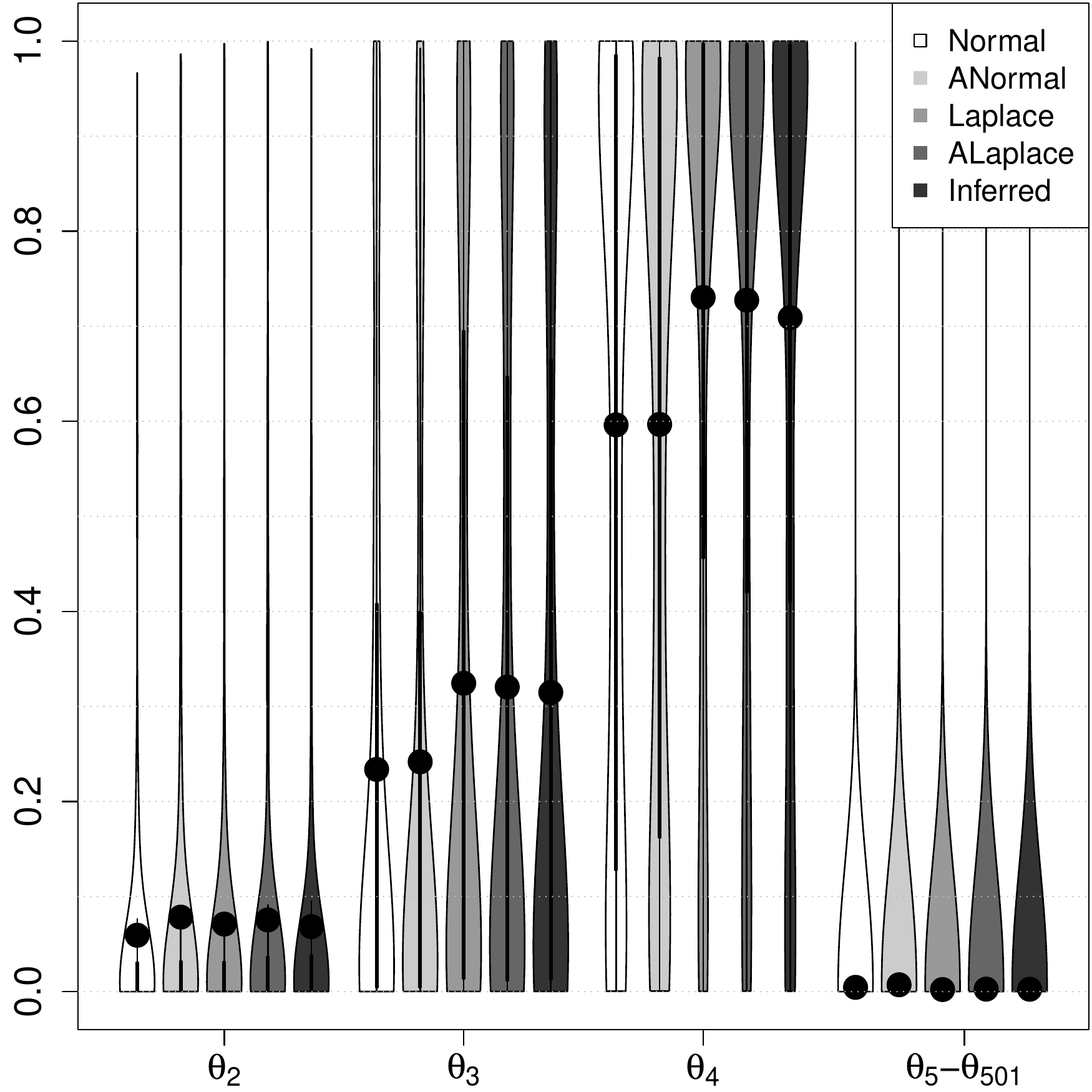} &
\includegraphics[width=0.48\textwidth]{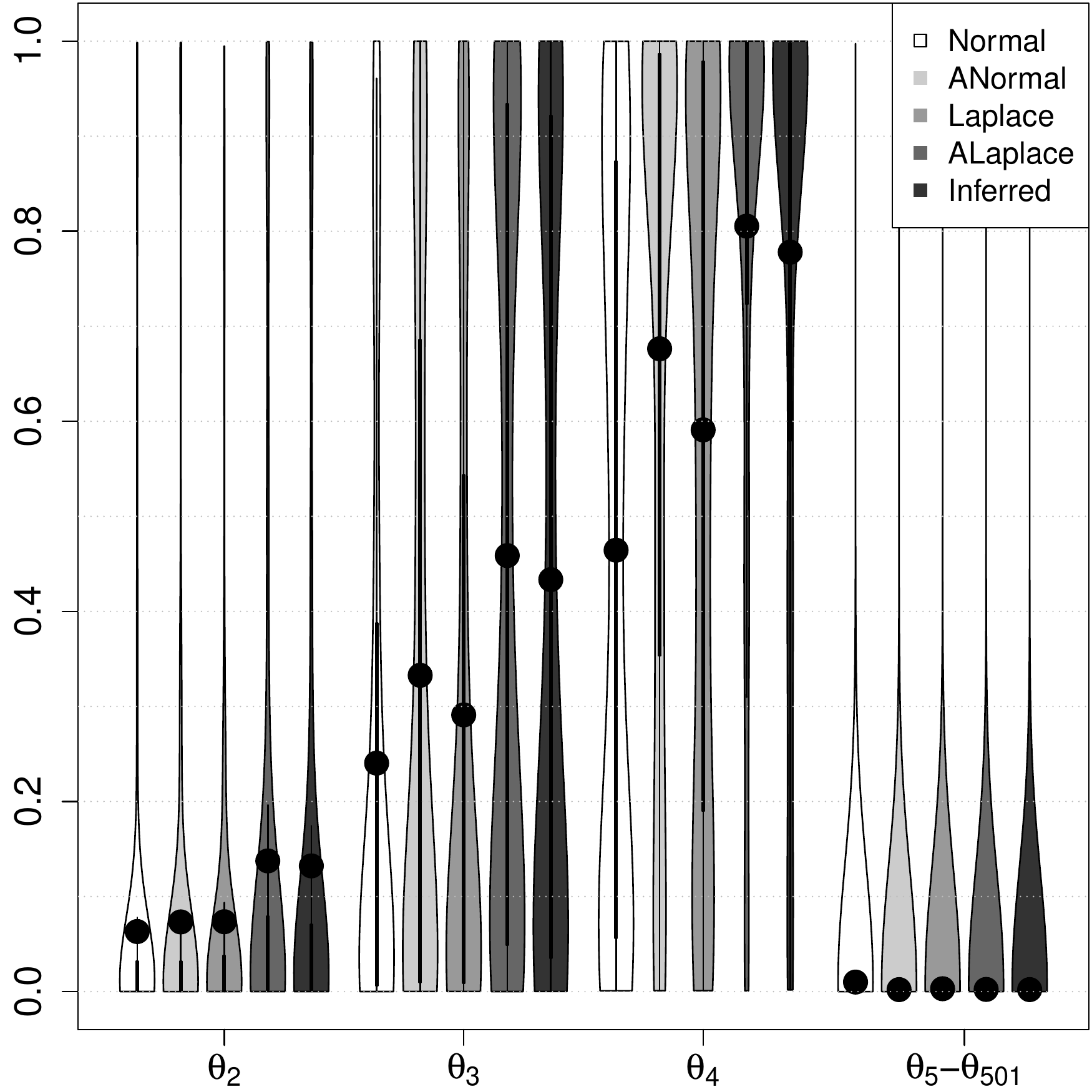} \\
\end{tabular}
\end{center}
\caption{$P(\theta_i \neq 0 \mid y)$ for $p=500$, $\vartheta=2$, $\theta=(0,0.5,1,1.5,0,\ldots,0)$,
$n=100$, $\rho_{ij}=0.5$. Black circles show the mean.}
\label{fig:simres_margpp_p500_vartheta2}
\end{figure}

We assessed the sensitivity of the results of the $p=6$ simulation study in Section \ref{ssec:lowdim_sim} of the main paper
to the prior on the asymmetry coefficient by setting $g_\alpha$ such that $P(|\alpha|>0.1)=0.99$.
Supplementary Table \ref{tab:perror_p5} summarizes the inference on the error distribution and
Supplementary Figure \ref{fig:simres_margpp_priorskew1} the marginal variable inclusion probabilities.
The latter were virtually identical to those in Figure \ref{fig:simres_margpp} obtained under $g_\alpha$ such that $P(|\alpha|>0.2)=0.99$,
showing that variable inclusion is robust to moderate changes in $g_\alpha$.

We also assessed the accuracy of the Laplace approximations to the integrated likelihood $p(y \mid \gamma)$
by comparing the results with those obtained with the importance sampling estimates with $B=10,000$ draws described in Section \ref{sec:modsel}
of the main paper.
Supplementary Figure \ref{fig:simres_margpp_mc} displays the results for $g_\alpha=0.357$.
These are extremely similar to those based on Laplace approximation in Figure \ref{fig:simres_margpp}.

Supplementary Figure \ref{fig:simres_margpp_p100_vartheta2} shows analogous results for $p=100$,
with $g_\alpha=0.357$ and $p(y \mid \gamma)$ estimated via Laplace approximations.

\subsection{Simulation study with non-identically distributed errors}
\label{ssec:simstudy_nonid_supplres}

\begin{table}
\begin{center}
\begin{tabular}{|c|cccc|} \hline
Truth & \multicolumn{4}{c|}{Average $p(\gamma_{p+1},\gamma_{p+2} \mid y)$} \\
& $\gamma_{p+1}=\gamma_{p+2}=0$ & $\gamma_{p+1}=1,\gamma_{p+2}=0$
& $\gamma_{p+1}=0,\gamma_{p+2}=1$ & $\gamma_{p+1}=\gamma_{p+2}=0$ \\ \hline \hline
$\mbox{N}(0,\vartheta_i)$       & 0.000 & 0.000 & 0.914 & 0.086 \\
$\mbox{AN}(0,\vartheta_i,-0.5)$ & 0.000 & 0.003 & 0.096 & 0.901 \\
$\mbox{L}(0,\vartheta_i)$       & 0.000 & 0.000 & 0.906 & 0.094 \\
$\mbox{AL}(0,\vartheta_i,-0.5)$ & 0.000 & 0.000 & 0.053 & 0.947 \\
\hline
\end{tabular}
\end{center}
\caption{Inference on the error distribution under the $p=6$ simulation and heteroskedastic
$\vartheta_i \propto e^{x_i^T \theta}$ errors}
\label{tab:infererror_heterosk}
\end{table}

\begin{table}
\begin{center}
\begin{tabular}{|c|ccccc|} \hline
 & $P(\gamma_2=1 \mid \by)$ & $P(\gamma_3=1 \mid \by)$ & $P(\gamma_4=1 \mid \by)$ & $P(\gamma_5=1 \mid \by)$ & $P(\gamma_6=1 \mid \by)$ \\  \hline
$q=0.05$ & 0.425 & 0.834 & 0.961 & 0.017 & 0.015 \\
$q=0.25$  & 0.751 & 0.950 & 0.996 & 0.016 & 0.015 \\
$q=0.5$  & 0.796 & 0.970 & 0.999 & 0.020 & 0.016 \\
$q=0.75$  & 0.769 & 0.969 & 0.999 & 0.016 & 0.012 \\
$q=0.95$  & 0.473 & 0.912 & 0.987 & 0.016 & 0.016 \\ \hline
\end{tabular}
\end{center}
\caption{Average marginal $P(\gamma_j=1 \mid \by)$ at multiple quantiles $q=0.05,0.25,0.5,0.75,0.95$
(i.e. conditioning on asymmetric Laplace errors with fixed $\alpha=2q-1$)
under the $p=6$ simulation and heteroskedastic $\epsilon_i \sim N(0,\vartheta_i)$,
$\vartheta_i \propto e^{x_i^T \theta}$ errors. Simulation truth is $\theta=(0,0.5,1,1.5,0,0)$}
\label{tab:multiplequantile_heterosk}
\end{table}

\begin{figure}
\begin{center}
\begin{tabular}{cc}
$\epsilon_i \sim $ Normal & $\epsilon_i \sim $ ANormal \\
\includegraphics[width=0.48\textwidth]{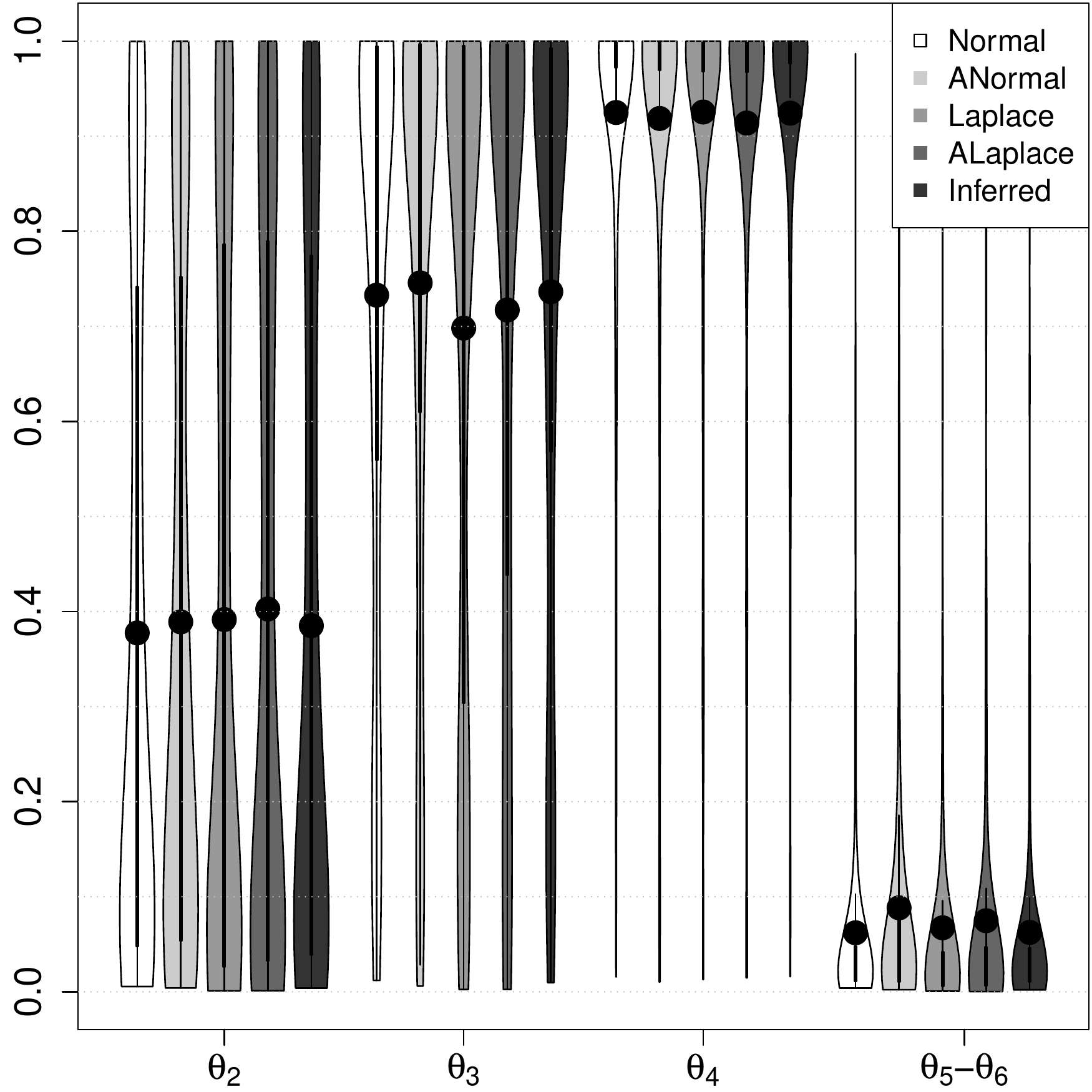} &
\includegraphics[width=0.48\textwidth]{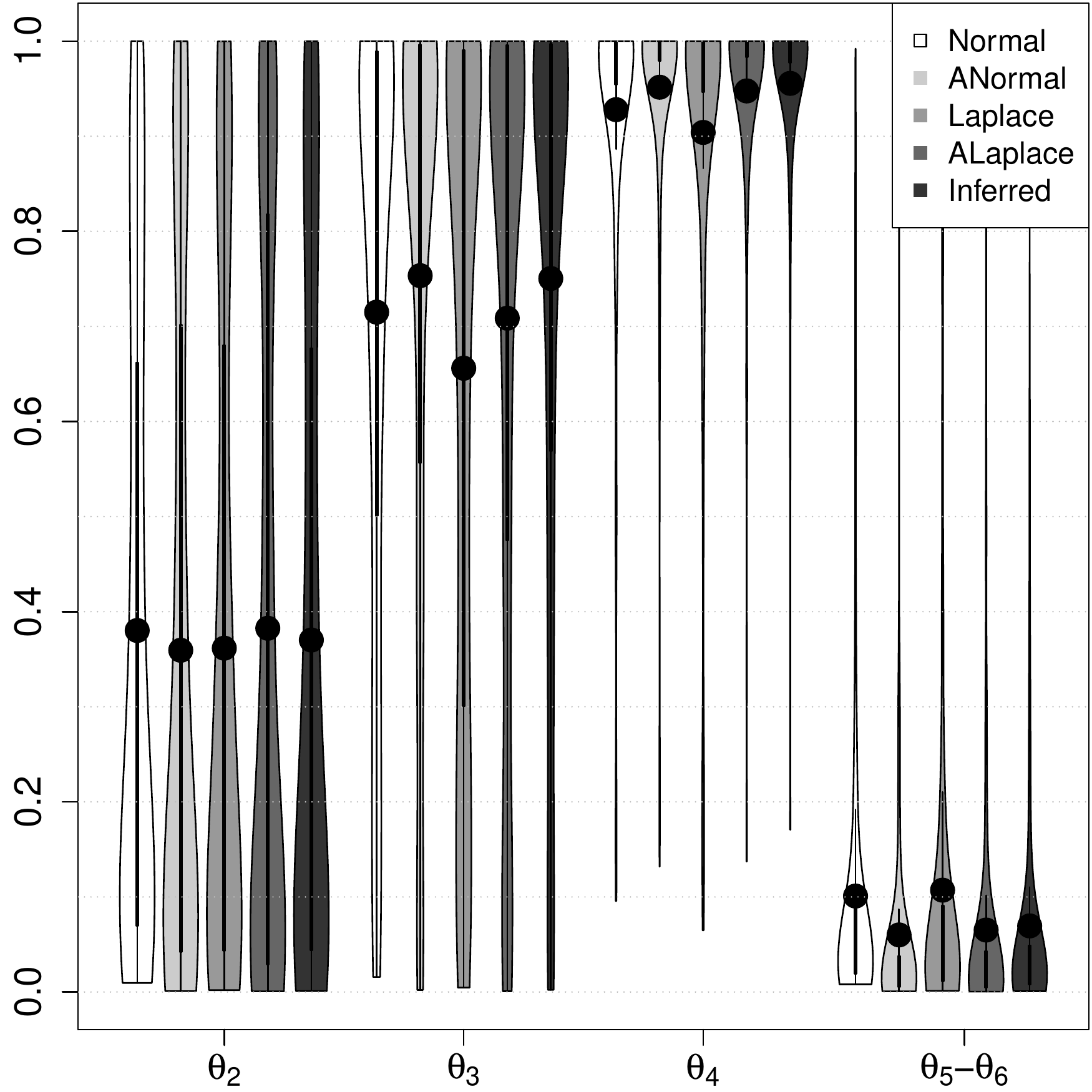} \\
$\epsilon_i \sim $ Laplace & $\epsilon_i \sim $ ALaplace \\
\includegraphics[width=0.48\textwidth]{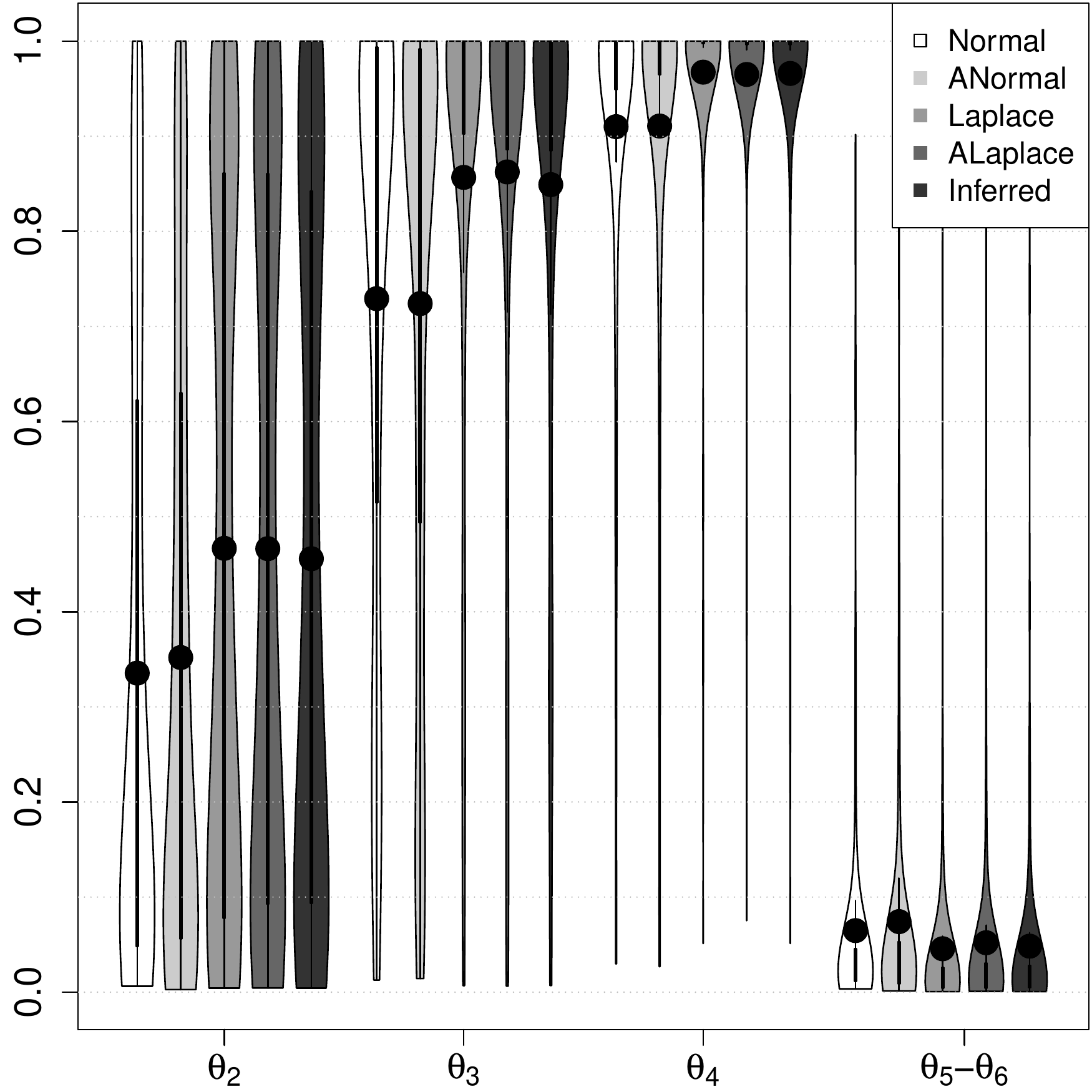} &
\includegraphics[width=0.48\textwidth]{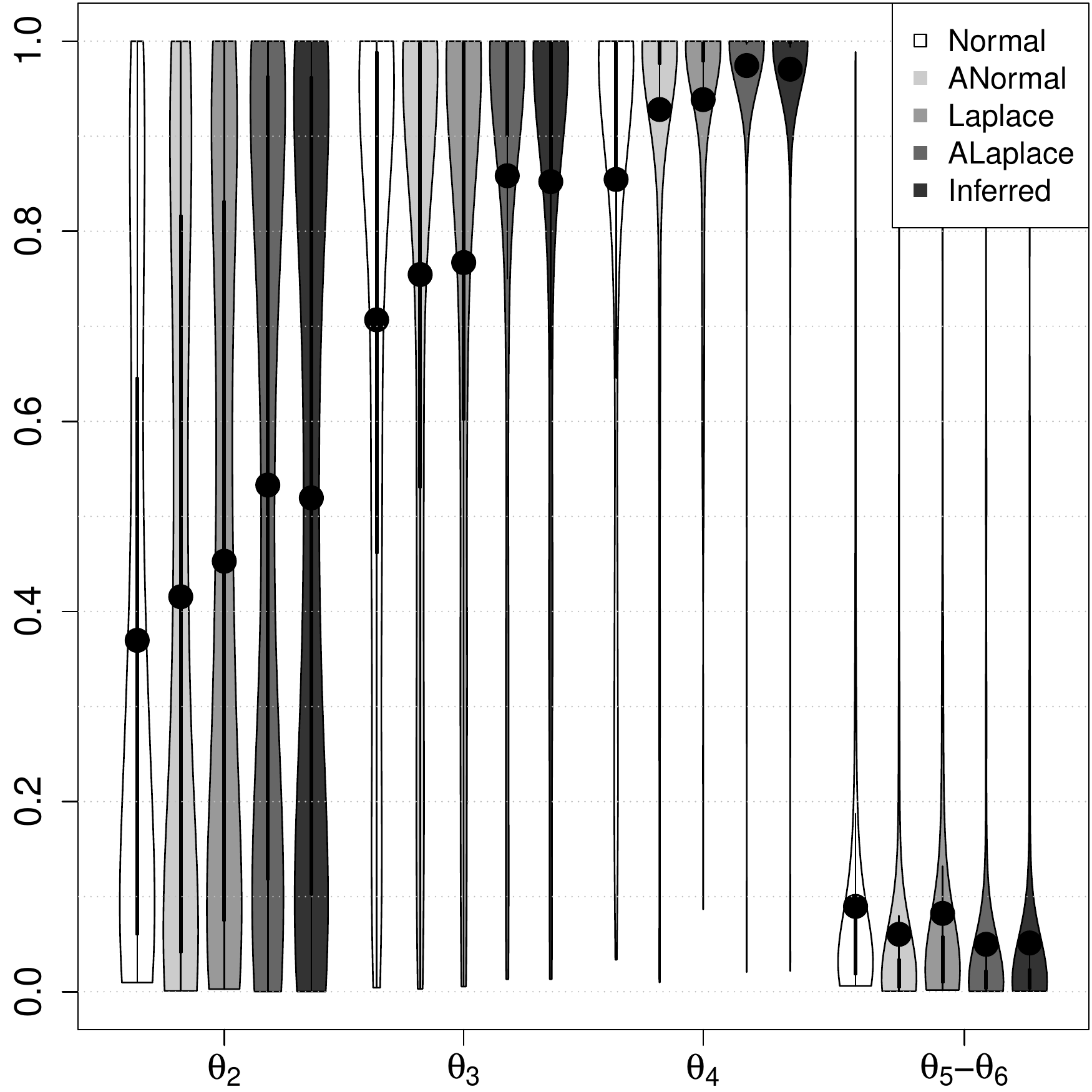} \\
\end{tabular}
\end{center}
\caption{$P(\theta_i \neq 0 \mid y)$ for simulation with constant $\vartheta=0$
and varying $\mbox{tanh}(\alpha_i) \sim N(\mbox{atanh}(\bar{\alpha},1/4^2))$,
where $\bar{\alpha}=0$ for Normal and Laplace and $\bar{\alpha}=-0.5$ for ANormal and ALaplace.
$P(\theta_i \neq 0 \mid y)$ for $p=6$, $\theta=(0,0.5,1,1.5,0,0)$,
$n=100$, $\rho_{ij}=0.5$. Black circles show the mean.}
\label{fig:simres_hetsym_margpp}
\end{figure}

\begin{table}
\begin{center}
\begin{tabular}{|c|cc|} \hline
                         & TP & FP \\
Zellner, Normal errors   & 2.8 & 21.3 \\
pMOM, Normal errors      & 3.0 & 12.0 \\
pMOM, inferred errors    & 2.8 & 10.5 \\
peMOM, Normal errors     & 1.9 & 2.9 \\ \hline
\end{tabular}
\end{center}
\caption{Number of true and false positives in non-id example with 0.5 probability of degenerate $(y_i,x_i)=(0,\ldots,0)$.
$p=n=50$, $\theta^*=(0.1,0.1,0.1,0.1,0.1,0,\ldots,0)$, $\vartheta^*=2$}
\label{tab:gruenwald}
\end{table}

Supplementary Table \ref{tab:infererror_heterosk} shows the mean average posterior probability
assigned to the Normal, asymmetric Normal, Laplace and asymmetric Laplace models under the heteroskedastic
simulation (Section \ref{ssec:nonid_sim}, main manuscript).

Supplementary Figure \ref{fig:simres_hetsym_margpp} shows marginal variable inclusion probabilities
under the hetero-asymmetric simulation.

Supplementary Table \ref{tab:gruenwald} reports true and false positives for our simulation study mimicking \cite{gruenwald:2014}
described in Section \ref{ssec:nonid_sim} of the main manuscript.

\subsection{DLD data}
\label{ssec:suppl_dld}

\begin{table}
\begin{center}
\begin{tabular}{|c|c|c|}\hline
 Gene symbol & Normal & Inferred \\
 C6orf226  & 1.000 & 1.000\\
     ECH1  & 1.000 & 1.000\\
   CSF2RA  & 1.000 & 1.000\\
    RRP1B  & 0.944 & 0.999\\
    FBXL19 & 0.993 & 0.658\\
     MTMR1 & 0.183 & 0.467\\
   SLC35B4 & 0.209 & 0.332\\
 RAB3GAP2  & 0.007 & 0.040\\
\hline
\end{tabular}
\end{center}
\caption{Six genes with largest $p(\gamma_j=1\mid y)$ in the DLD dataset under assumed normality and inferred error distribution.}
\label{tab:dld_margpp}
\end{table}

\begin{table}
\begin{center}
\begin{tabular}{|c|c|} \hline
\multicolumn{2}{|c|}{$\alpha=-0.5$} \\ \hline
 Model & $P(\gamma \mid y)$ \\
 C6orf226, ECH1, CSF2RA, FBXL19, RRP1B & 0.384 \\
 SLC35B4, C6orf226, ECH1, CSF2RA, RRP1B & 0.349 \\
 SLC35B4, C6orf226, MTMR1, ECH1, CSF2RA, RRP1B & 0.127 \\
 C6orf226, MTMR1, ECH1, CSF2RA, FBXL19, RRP1B & 0.049 \\
 C6orf226, MTMR1, RAB3GAP2, ECH1, CSF2RA, RRP1B & 0.023 \\
\hline
\multicolumn{2}{|c|}{$\alpha=0$} \\ \hline
 Model & $P(\gamma \mid y)$ \\
 C6orf226, MTMR1, ECH1, CSF2RA, FBXL19, RRP1B & 0.454 \\
 C6orf226, ECH1, CSF2RA, FBXL19, RRP1B & 0.258 \\
 SLC35B4, C6orf226, MTMR1, ECH1, CSF2RA, RRP1B & 0.108 \\
 SLC35B4, C6orf226, ECH1, CSF2RA, RRP1B & 0.061 \\
 C6orf226, MTMR1, RAB3GAP2, ECH1, CSF2RA, RRP1B & 0.016 \\
\hline
\multicolumn{2}{|c|}{$\alpha=0.5$} \\ \hline
 Model & $P(\gamma \mid y)$ \\
 C6orf226, ECH1, CSF2RA, FBXL19, RRP1B & 0.399 \\
 SLC35B4, C6orf226, ECH1, CSF2RA, RRP1B & 0.359 \\
 SLC35B4, C6orf226, MTMR1, ECH1, CSF2RA, RRP1B & 0.120 \\
 C6orf226, MTMR1, ECH1, CSF2RA, FBXL19, RRP1B & 0.051 \\
 SLC35B4, C6orf226, RAB3GAP2, ECH1, CSF2RA, RRP1B & 0.008 \\
\hline
\end{tabular}
\end{center}
\caption{DLD data. Top 5 models when conditioning on asymmetric Laplace residuals
and fixed $\alpha=-0.5,0,0.5$}
\label{tab:dld_top5models}
\end{table}

Supplementary Table \ref{tab:dld_margpp} shows the six genes with largest marginal inclusion probabilities $p(\gamma_j=1\mid y)$
when conditioning on Normal errors and when inferring the error distribution.
The figures were similar for the four top genes, but the Normal model assigned somewhat higher probability to FBXL19
substantially lower probability to MTMR1.

\clearpage

\bibliographystyle{plainnat}
\bibliography{references}

\end{document}